\documentclass[12pt,onecolumn,a4paper]{article}
\usepackage[ breaklinks,
			 colorlinks = true,
             linkcolor = blue,
             urlcolor  = blue,
             citecolor = red,
             anchorcolor = blue,
]{hyperref}
\usepackage[left=1in,right=1in,top=1in,bottom=1in]{geometry}
\usepackage{cite}

\usepackage[noblocks]{authblk}
\usepackage[dvipsnames]{xcolor}
\usepackage{comment}
\usepackage{bm}

\usepackage{subcaption}
\usepackage{hyperref}
\usepackage{enumerate}

\usepackage{amsfonts}
\usepackage{amsmath}
\usepackage{graphicx}
 
\usepackage{mathtools}
\mathtoolsset{showonlyrefs}

\usepackage{dcolumn}

\usepackage{enumerate}
\usepackage{amssymb}
\usepackage{amsthm}
\usepackage{calc} 
\usepackage{accents}

\newcommand{\Mod}[1]{\ (\mathrm{mod}\ #1)}

\newcommand{\dtilde}[1]{\tilde{\raisebox{0pt}[0.90\height]{$\tilde{#1}$}}}

\newcommand{\appref}[1]{\hyperref[#1]{{Appendix~\ref*{#1}}}}
\newcommand{\be}{\begin{eqnarray} \begin{aligned}}
\newcommand{\ee}{\end{aligned} \end{eqnarray} }
\newcommand{\benn}{\begin{eqnarray*} \begin{aligned}}
\newcommand{\eenn}{\end{aligned} \end{eqnarray*}}
\newcommand*{\textfrac}[2]{{{#1}/{#2}}}

\newcommand*{\cA}{\mathcal{A}} 
\newcommand*{\cB}{\mathcal{B}}
\newcommand*{\cC}{\mathcal{C}}
\newcommand*{\cE}{\mathcal{E}}
\newcommand*{\cF}{\mathcal{F}}

\newcommand*{\cH}{\mathcal{H}}

\newcommand*{\cN}{\mathcal{N}}

\newcommand*{\cR}{\mathcal{R}}
\newcommand*{\cO}{\mathcal{O}}

\newcommand*{\cS}{\mathcal{S}}

\newcommand*{\cX}{\mathcal{X}}

\newcommand*{\tr}{\mathop{\mathrm{tr}}\nolimits}

\newcommand*{\supp}{\mathrm{supp}}
\newcommand*{\myspan}{\mathrm{span}}

\newcommand{\bc}{\begin{center}}
\newcommand{\ec}{\end{center}}

\makeatletter
\newtheorem*{rep@theorem}{\rep@title}
\newcommand{\newreptheorem}[2]{%
\newenvironment{rep#1}[1]{%
 \def\rep@title{#2 \ref{##1}}%
 \begin{rep@theorem}}%
 {\end{rep@theorem}}}
\makeatother

\newtheorem{theorem}{Theorem}[section]
\newreptheorem{theorem}{Theorem}
\newtheorem{lemma}[theorem]{Lemma}
\newreptheorem{lemma}{Lemma}

\newreptheorem{claim}{Claim}
\newtheorem{definition}[theorem]{Definition}
\newreptheorem{definition}{Definition}

\newreptheorem{remark}{Remark}
\newtheorem{corollary}[theorem]{Corollary}
\newreptheorem{corollary}{Corollary}

\newreptheorem{observation}{Observation}

\newcommand{\myacknowledgements}{\begin{center}{\bf Acknowledgments}\end{center}\par}

\usepackage{amsfonts}

\def\01{\{0,1\}}

\newcommand{\ket}[1]{|#1\rangle}

\newcommand{\kket}[1]{|#1\rangle\hspace{-0.5ex}\rangle}
\newcommand{\bbra}[1]{\langle\hspace{-0.5ex}\langle#1|}
\newcommand{\sspr}[2]{\langle\hspace{-0.5ex}\langle#1|#2\rangle\hspace{-0.5ex}\rangle}

\newcommand*{\mylabel}[1]{\label{#1}}
\newcommand{\bra}[1]{\langle#1|}

\newcommand{\proj}[1]{|#1\rangle\langle#1|}

\newcommand{\rank}{\operatorname{rank}}

\newcommand*{\physical}{\mathsf{p}}

\newcommand*{\bond}{D}

\newcommand*{\systemsize}{n}

\newcommand*{\hspin}{\mathbf{s}}  
\newcommand*{\halfj}{\mathbf{j}}  
\newcommand*{\identityoperator}{I}

\usepackage[symbol]{footmisc}

\begin{document}

\title{\LARGE Quantum Error-Detection at Low Energies}

\author[1]{Martina Gschwendtner\thanks{martina.gschwendtner@tum.de}}
\author[1,2]{Robert K\"onig\thanks{robert.koenig@tum.de}}
\author[3]{Burak \c{S}ahino\u{g}lu\thanks{sahinoglu@caltech.edu}}
\author[3]{Eugene Tang\thanks{eugene.tang@caltech.edu}}
\affil[1]{\footnotesize Zentrum Mathematik, Technical University of Munich, 85748 Garching, Germany}
\affil[2]{\footnotesize Institute for Advanced Study, Technical University of Munich, 85748 Garching, Germany}
\affil[3]{\footnotesize Department of Physics and Institute for Quantum Information and Matter, California Institute of Technology, Pasadena, CA 91125, USA}

\date{\today}
\maketitle
\begin{abstract}

Motivated by the close relationship between quantum error-correction, topological order, the holographic AdS/CFT duality, and tensor networks, we initiate the study of approximate quantum error-detecting codes in matrix product states (MPS). We first show that using open-boundary MPS to define boundary to bulk encoding maps yields at most constant distance error-detecting  codes. These are degenerate ground spaces of gapped local Hamiltonians. To get around this no-go result, we consider excited states, i.e., we use the excitation ansatz to construct encoding maps: these yield error-detecting codes with distance $\Omega(n^{1-\nu})$ for any $\nu\in (0,1)$ and $\Omega(\log n)$ encoded qubits. This shows that gapped systems contain -- within isolated energy bands -- error-detecting codes spanned by momentum eigenstates. We also consider the gapless Heisenberg-XXX model, whose energy eigenstates can be described via Bethe ansatz tensor networks. We show that it contains -- within its low-energy eigenspace -- an error-detecting code with the same parameter scaling.  All these codes detect arbitrary $d$-local (not necessarily geometrically local) errors even though they are not permutation-invariant. This suggests that a wide range of naturally occurring many-body systems possess intrinsic error-detecting features.

\end{abstract}
\newcommand*{\ad}[1]{\mathsf{Ad}_{#1}}
\newcommand{\Vod}{V_{\text{od}}}
\newcommand{\Vd}{V_{\text{d}}}
\newcommand*{\flip}{\mathbb{F}}
\newcommand*{\hP}{\hat{P}}
\newcommand*{\E}{\mathbb{E}}
\newcommand*{\diam}{\mathsf{diam}}
\newcommand*{\delt}{p_{\mathsf{dec}}}
\newcommand*{\mpo}{\mathcal{O}}
\newcommand*{\F}{\mathbb{F}}

\newcommand*{\err}{\mathsf{err}}
\newcommand*{\errapprox}{\epsilon_{\mathsf{approx}}}

\newpage
\tableofcontents

\newpage

\section{Introduction}
Quantum error-correcting codes are fundamental for achieving robust quantum memories and fault-tolerant quantum computation. Following seminal work by Shor~\cite{shor1995scheme} and others~\cite{calderbank1996good, kitaev1997quantum, knill1997theory, knill1998resilient},  the study of quantum error-correction has seen tremendous progress both from both the theoretical and the experimental point of view. Beyond its operational implications for the use of faulty quantum hardware, quantum error-correction is closely connected to fundamental physics, as shown early on by the work of Kitaev~\cite{kitaev2003fault}: the ground space of a topologically ordered model constitutes a quantum error-correcting code whose dimension depends on the topology of the underlying surface containing the physical degrees of freedom.  In addition to giving rise to a new field called topological quantum computing~\cite{freedman2003topological,ogburn1999topological,dennis2002topological, nayak2008non,raussendorf2006fault,stern2013topological,terhal2015quantum}, this work has had a significant impact on the problem of classifying topologically ordered phases in two spatial dimensions~\cite{kitaev2006anyons, levin2005string}. Motivated by the success of this program, follow-up work has pursued the classification of gapped phases of matter  with or without global symmetries, starting from one spatial dimension~\cite{kitaev2009periodic, fidkowski2011topological, chen2011classification, chen2011complete} up to arbitrarily high dimensions~\cite{chen2011two, chen2013symmetry}.

More recently, concepts from quantum error-correction have helped to resolve conceptual puzzles in AdS/CFT holographic duality. Almheiri, Dong, and Harlow~\cite{almheiri15}  have proposed that subspaces of holographic conformal field theories (CFTs) which are dual to perturbations around a particular classical bulk AdS geometry constitute a quantum error-correcting code robust against erasure errors. In this proposal, the bulk and boundary degrees of freedom correspond to the logical and the physical degrees of freedom of the code, respectively. Puzzling features such as subregion-subregion duality and radial commutativity can naturally be understood in this language, under the hypothesis that the duality map works as a code which recovers, from erasure, part of the boundary degrees of freedom. Related to this picture, Ryu-Takayanagi type  formulas have been shown to hold in any quantum error-correcting code that corrects against erasure~\cite{harlow2017ryu}. 

Key to many of these results in the context of topological order and the AdS/CFT holographic duality is the language of tensor networks. The latter, originating in work by Fannes, Nachtergaele, and Werner on finitely correlated states~\cite{fannes1992finitely} and the density matrix renormalization group~\cite{white1992density, white1993density}, has  seen a revival in the last 15~years. Major conceptual contributions include the introduction of matrix product states by~\cite{vidal2003efficient, vidal2004efficient,perez2006matrix, verstraete2004renormalization, verstraete2008matrix}, the introduction of the multi-scale entanglement renormalization ansatz (MERA)~\cite{vidal2008class} by Vidal, and various projected entangled-pair states (PEPS) techniques~\cite{verstraete2004renormalization,verstraete2006projected, perez2007peps, kraus2010fermionic, schwarz2012preparing, fishman2018faster} for higher dimensional systems. 

It has been shown that tensor network techniques provide exact descriptions of topologically ordered states~\cite{buerschaper2009explicit, gu2009tensor, konig2009exact}, and furthermore, tensor networks have been instrumental in the characterization and classification of topological order~\cite{schuch2010peps, buerschaper2014twisted,csahinouglu2014characterizing, williamson2016matrix, bultinck2017anyons}. This approach has also been generalized to higher dimensions, clarifying the   connections to topological quantum field theories~\cite{TN-TQFT}.

A similar success story for the use of tensor networks is emerging in the area of AdS/CFT duality. Aspects of holographic duality have been explored in terms of toy models based on tensor networks~\cite{swingle2012entanglement, pastawski2015holographic, hayden2016holographic}. Indeed, many (though not all) conjectured features of this duality can be recovered in these examples. This field, while still in its infancy, has provided new appealing conjectures which point to a potentially more concrete understanding of the yet to be uncovered physics of quantum gravity~\cite{akers2018holographic, dong2018flat}.

Given the existing close connections between quantum error-correction and a variety of physical systems ranging from topological order to AdS/CFT, it is natural to ask how generic the appearance of error-correcting features is  in naturally occurring quantum many-body systems. A first step towards showing the ubiquity of such features is the work of Brandao, et. al.~\cite{brandao2017quantum}. There, it is shown that quantum chaotic systems satisfying the Eigenstate Thermalization Hypothesis (ETH) have energy eigenstates that form approximate quantum error-correcting codes. Nearby extensive energy eigenstates of 1D translation invariant Hamiltonians, as well as ground spaces of certain gapless systems (including the Heisenberg and Motzkin models), also contain approximate quantum error-correcting codes. Motivated by this work, we ask if one can demonstrate the existence of error-correcting codes within the low-energy eigenspaces of generic Hamiltonians, whether or not they are gapped or gapless. Specifically, we ask this question for 1D~systems.

Our work goes beyond earlier work by considering errors (that is, noise) of a more general form: existing studies of error-correction in the context of entanglement renormalization and/or holography have primarily concentrated on qubit loss, modeled by so-called erasure errors (see e.g.,~\cite{pastawski2015holographic,Kim2017, pastawski2017towards}). This erasure noise model has several theoretical advantages. In particular, it permits one to argue about the existence of recovery maps in terms of entanglement entropies of the associated erased regions. This can be connected to well-known results on entanglement entropies in critical 1D~systems. Furthermore, the appearance of entanglement entropies in these considerations is natural in the context of the AdS/CFT duality, where these quantities are involved in the connection of the boundary field theory to the bulk geometry via the Ryu-Takayangi formula. However, compared to other forms of errors typically studied in the quantum fault-tolerance community, erasure is quite a restricted form of noise: it is, in a certain sense, much easier to correct than, e.g.,~depolarizing noise. As an example to illustrate this point, we recall that the toric code can recover from loss of half its qubits~\cite{toriccodequbitloss}, whereas it can only tolerate depolarizing noise up to a noise rate of 11\% even  given perfect syndrome measurements~\cite{dennis2002topological}. Motivated by this, we aim to analyze error-correcting properties with respect to more generic noise even though this precludes the use of entanglement entropies. Again, the work~\cite{brandao2017quantum} provides first results in this direction by considering errors on a fixed, connected subset of sites (that is,  geometrically localized errors). The restriction to a connected subset was motivated in part by the consideration of permutation-invariant subspaces (note other previous works on permutation-invariant code spaces~\cite{pollatsek2004permutationally, ouyang2014permutation}). In our work, we lift the restriction to permutation-invariant codes and instead analyze arbitrary weight-$d$ errors with potentially disconnected supports. Furthermore, we study an operational task -- that of error-detection -- with respect to a noise model where errors can occur on any subset of qubits of a certain size, instead of only a fixed subset.

We find that the language of matrix product states (MPS) and the related excitation ansatz states provides a powerful analytical tool for studying error-detection in 1D systems. In particular, we relate properties of  transfer operators  to error-detection features: for MPS describing (degenerate) ground spaces of gapped Hamiltonians, injectivity of the transfer operators gives rise to a no-go theorem. For excitation ansatz states describing the low-energy excitations of gapped systems, we use injectivity and a certain normal form to establish error-correction properties.  Finally, for a gapless integrable model, we analyze the Jordan structure of (generalized) transfer matrices to find bounds on code parameters. In this way, our work connects locally defined features of tensor networks to global error-correction properties. This can be seen as a first step in an organized program of studying approximate quantum error-correction in tensor network states.

\section{Our contribution}
We focus on error-detection, a natural primitive in fault-tolerant quantum computation. Contrary to full error-correction, where the goal is to recover the initial encoded state from its corrupted version, error-detection  merely permits one to decide whether or not an error has occurred. Errors (such as local observables) detected by an error-detecting code have expectation values independent of the particular logical state. In the context of topological order, where local errors are considered, error-detection has been referred to as TQO-$1$ (topological quantum order condition~$1$); see, e.g.,~\cite{bravyihastingsmichalakis}. An approximate version of the latter is discussed in~\cite{hastings11}.  

A code, i.e., a subspace of the physical Hilbert space, is said to be error-detecting (for a set of errors) if the projection back onto the code space after the application of an error results in the original encoded state, up to normalization. Operationally, this means that one can ensure that no error occurred by performing a binary-outcome POVM consisting of the projection onto the code space or its complement. This notion of an error-detecting code is standard, though quite stringent: unless the code is constructed algebraically (e.g., in terms of Pauli operators), it is typically not going to have this property.

Our first contribution is  a relaxed, yet  still operationally meaningful  definition for approximate error-detection. It relaxes the former notion in two directions: first, the post-measurement state is only required to approximate the original encoded state. Second, we only demand that this approximation condition is satisfied if the projection onto the code space occurs with non-negligible probability. This is motivated by the fact that if this projection does not succeed with any significant probability, the error-detection measurement has little effect (by the gentle measurement lemma~\cite{gentlemeasurementlemma}) and may as well be omitted. More precisely, we consider a CPTP map $\cN:\cB((\mathbb{C}^\physical)^{\otimes n})\rightarrow\cB((\mathbb{C}^\physical)^{\otimes n})$ modeling noise on~$n$ physical qudits (of dimension~$\physical$). Here the Kraus operators of~$\cN$ take the role of errors (considered in the original definition). We define the following notion:
\begin{repdefinition}{def:errdet1}[Approximate quantum error-detecting code] 
A subspace $\cC\subset(\mathbb{C}^\physical)^{\otimes n}$ (with associated projection $P$) is an $(\epsilon,\delta)$-approximate error-detecting code for $\cN$ if for any state~$\ket{\Psi}\in\cC$ the following holds:
\begin{align}
\textrm{if}\qquad \tr(P\cN(\proj{\Psi}))\geq \delta\qquad\textrm{then}\qquad  \bra{\Psi}\rho_{\cN,P}\ket{\Psi}\geq 1-\epsilon\ ,
\end{align}
where $\rho_{\cN,P}=\tr(P\cN(\proj{\Psi}))^{-1}\cdot P\cN(\proj{\Psi})P$.
\end{repdefinition}
This definition ensures that the post-measurement state~$\rho_{\cN,P}$ is close (as quantified by~$\epsilon$) to the initial code state when the outcome of the POVM is~$P$. Furthermore, we only demand this in the case where $\cN(\proj{\Psi})$ has an overlap with the code space of at least~$\delta$. 

In the following, we often consider families of codes~$\{\cC_n\}_n$ indexed by the number~$n$ of physical spins. In this case, we demand that both approximation parameters $\epsilon_n$ and $\delta_n$ tend to zero as $n\rightarrow\infty$. This is how we make sure that we have a working error-detecting code in the asymptotic or thermodynamic limit of the physical Hilbert space.

Of particular interest are errors of weight~$d$, i.e., errors which only act non-trivially on a subset of~$d$ of the $n$~subsystems in the product space~$(\mathbb{C}^\physical)^{\otimes n}$. We call this subset the \it support \rm of the error, and refer to the error as {\em  $d$-local}. We emphasize that throughout this paper, $d$-local only refers to the weight of the errors: they do not need to be geometrically local, i.e., their support may be disconnected. In contrast, earlier work on approximate error-correction such as~\cite{brandao2017quantum} only considered errors with support on a (fixed) connected subset of $d$~sites. We then define the following:
\begin{repdefinition}{def:errdet2}[Error-detection for $d$-local errors]
A subspace $\cC\subset(\mathbb{C}^\physical)^{\otimes n}$  is called an $(\epsilon,\delta)[[n,k,d]]$-approximate quantum error-detecting code (AQEDC) if $\dim\cC=\physical^k$ and if~$\cC$ is an $(\epsilon,\delta)$-approximate error-detecting code for any CPTP map $\cN:\cB((\mathbb{C}^\physical)^{\otimes n})\rightarrow\cB((\mathbb{C}^\physical)^{\otimes n})$ of the form
\begin{align}
\cN(\rho)=\sum_{j\in [J]} p_j F_j\rho F_j^{\dagger}\ ,\label{eq:cnrhoxzv}
\end{align}
where each $F_j$ is a $d$-local operator with $\|F_j\|\leq 1$ and $\{p_j\}_{j\in [J]}$ is a probability distribution. We refer to~$d$ as the {\em distance} of the code.
\end{repdefinition}
In other words, an $(\epsilon,\delta)[[n,k,d]]$-AQEDC deals with error channels which are convex combinations of $d$-local errors. This includes for example the commonly considered case of random Pauli noise (assuming the distribution is supported on errors having weight at most~$d$). However, it does not cover the most general case of (arbitrary) $d$-local errors/error channels because of the restriction to convex combinations. The consideration of convex combinations of $d$-local errors greatly facilitates our estimates and allows us to consider settings that go beyond earlier work. We leave it as an open problem to lift this restriction, and only provide some tentative statements in this direction.

To exemplify in what sense our definition of AQEDC for $d$-local errors extends earlier considerations, consider the case where the distribution over errors in~\eqref{eq:cnrhoxzv} is the uniform distribution over all $d$-qudit Pauli errors on $n$~qubits. In this case, the number of Kraus operators in the representation~\eqref{eq:cnrhoxzv} is polynomial in~$n$ even for constant distance~$d$. In particular, arguments involving the number of terms in~\eqref{eq:cnrhoxzv} cannot be used to establish bounds on the code distance as in~\cite{brandao2017quantum}, where instead, only Pauli errors acting on $d$~fixed sites were considered: The number of such operators is only~$4^d$ instead of the number $\binom{n}{d}4^d$ of all weight-$\leq d$-Paulis, and, in particular, does not depend on the system size~$n$.

We establish the following approximate Knill-Laflamme type conditions which are sufficient for error-detection: 
\begin{repcorollary}{cor:errordectionapproximate}
Let $\cC\subset (\mathbb{C}^{\physical})^{\otimes n}$ be a code with orthonormal basis $\{\psi_\alpha\}_{\alpha\in [\physical^k]}$ such that (for some~$\gamma>0$),
\begin{align}
\big|\bra{\psi_\alpha} F\ket{\psi_\beta}-\delta_{\alpha,\beta}\bra{\psi_1} F\ket{\psi_1}\big|\leq \gamma\cdot \|F\|\qquad\textrm{ for all }\alpha,\beta\in [\physical^k]\ ,
\end{align}
for every $d$-local operator $F$ on~$(\mathbb{C}^{\physical})^{\otimes n}$. Let $\delta > \physical^{5k}\gamma^2$. Then $\cC$ is an $(\epsilon=\physical^{5k}\gamma^2 \delta^{-1},\delta)[[n,k,d]]$-AQEDC.
\end{repcorollary}
This condition, which is applicable for ``small'' code space dimension, i.e., $k=O(\log n)$, allows us to reduce the consideration of approximate error-detection to the estimation of matrix elements of local operators. We also establish a partial converse to this statement: if a subspace~$\cC\subset(\mathbb{C}^{\physical})^{\otimes n}$ contains two orthonormal vectors whose reduced $d$-local density operators (for some subset of $d$~sites) are almost orthogonal, then~$\cC$ cannot be an error-detecting code with distance~$d$ (see Lemma~\ref{lem:distancenecessary} for a precise statement).

Equipped with these notions of approximate error-detection, we study quantum many-body systems in terms of their error-detecting properties using tensor network techniques. More specifically, we consider two types of code families, namely:
\begin{enumerate}[(i)]
\item\label{it:tnscodesdef}
codes that are degenerate ground spaces of local Hamiltonians and permits a description in terms of tensor networks, and 
\item\label{it:generalcodesdefmps}
codes defined by low-energy eigenstates of (geometrically) local Hamiltonians, with the property that these can be efficiently described  in terms of tensor networks.
\end{enumerate}
As we explain below,~\eqref{it:tnscodesdef} and~\eqref{it:generalcodesdefmps} are closely connected via the parent Hamiltonian construction. For~\eqref{it:tnscodesdef}, we follow a correspondence between tensor networks and codes which is implicit in many existing constructions: we may think of a tensor as a map from certain virtual to physical degrees of freedom. To define this map, consider a tensor network given by a graph~$G=(V,E)$ and a collection of tensors~$A$. Let us say that an edge~$e\in E$ is a dangling edge if one of its vertices has degree~$1$, and let us call the corresponding vertices the dangling vertices of the tensor network. An edge $e\in E$ is an internal edge if it is not a dangling edge; we use an analogous notion for vertices. We assume that each internal edge~$e\in E$ is associated a virtual space of fixed bond dimension~$D$, and each dangling edge with a physical degree of dimension~$\physical$.  Then the tensor network associates a tensor~$T$ of degree~$\mathsf{deg}(v)$ to each internal vertex~$v$ of~$G$, where it is understood that indices corresponding to internal edges are contracted. The tensor network is fully specified by the family~$A$ of such tensors. We partition the set of dangling vertices into a two subsets $M$ and $M^c$.  Then the tensor network defines a map~$\Gamma(A,G):(\mathbb{C}^\physical)^{\otimes |M|}\rightarrow (\mathbb{C}^\physical)^{\otimes |M^c|}$ as each fixing of the  degrees of freedom in $M$ defines an element of the Hilbert space associated with the degrees of freedom in~$M^c$ by tensor contraction. That is, the map depends on the graph~$G$ specifying the structure of the tensor network, as well as the family~$A$ of local tensors.  In particular, fixing a subspace of~$(\mathbb{C}^\physical)^{\otimes |M|}$,  its image under the map~$\Gamma(A,G)$ defines a subspace~$\cC\subset(\mathbb{C}^\physical)^{\otimes |M^c|}$ which we will think of as an error-correcting code.  In the following, we also allow the physical and virtual (bond) dimensions to vary (depending on the location in the tensor network); however, this description captures the essential construction.

This type of construction is successful in two and higher spatial dimensions, yielding error-correcting codes with macroscopic distance: examples are the ground states of the toric code~\cite{aguadovidal08,schuch2010peps} and other topologically ordered models~\cite{konig2009exact,buerschaper2009explicit,buerschaper2014twisted}. However, in 1D, it seems a priori unlikely that the very same setup can generate any nontrivial quantum error-detection code, at least for gapped systems. This is because of the exponential decay of correlations~\cite{hastings2006spectral, brandao2013area, brandao2015exponential} and the lack of topological order without symmetry protection~\cite{chen2011classification,schuch2011classifying}. We make this precise by stating and proving a no-go theorem. 

More precisely, we follow the above setup provided by the boundary-to-bulk tensor network map $\Gamma(A)=\Gamma(A,G)$. Here, $G$ is the $1D$~line graph with dangling edges attached to internal vertices, which is equivalent to considering the ground space of 1D local gapped Hamiltonians with open boundary conditions. The associated tensor network is a matrix product state.

Generic MPS satisfy a condition called injectivity, which is equivalent to saying that the transfer matrix of the MPS is gapped. Exploiting this property allows us to prove a lower bound on the distinguishability of $d$-local reduced density operators for any two orthogonal states in the code space. This bound is expressed in terms of the virtual bond dimension~$D$ of the  MPS  tensor $A$. In particular, the bound implies the following no-go theorem for codes generated by MPS as described above.

\begin{reptheorem}{thm:nogo}
Let $\mathcal{C}\subset (\mathbb{C}^{\physical})^{\otimes n}$ be an approximate quantum error-detecting code generated by $\Gamma(A)$, i.e., a translation-invariant injective MPS of constant bond dimension $D$ by varying boundary conditions. Then the distance of $\mathcal{C}$ is constant.
\end{reptheorem}

The physical interpretation of this theorem is as follows: for every injective MPS with periodic boundary conditions, there exists a strictly $\log D$-geometrically local gapped Hamiltonian such that the MPS is the unique ground state~\cite{perez2006matrix}. One can further enlarge the ground space by leaving out a few Hamiltonian terms near the boundary. The degeneracy then depends on the number of terms omitted, and the ground states are described by open boundary condition MPS. Then, our no-go theorem implies that the ground space of any such parent Hamiltonian arising from such a constant bond-dimension MPS is a trivial code, i.e., it can have at most a constant distance. This result is equivalent to saying that there is no topological quantum order in the ground space of 1D gapped systems.\footnote{More precisely, this statement holds for systems whose ground states can be approximated by constant bond dimension MPS. It is not clear whether this is sufficient to  make a statement about general 1D local gapped Hamiltonians. The identification of ground states of 1D local gapped Hamiltonians with constant bond dimension MPS is sometimes made in the literature, as for example in the context of classifying phases~\cite{chen2011classification, chen2011complete,schuch2011classifying}.} 

To get around this no-go result, we extend our considerations beyond the ground space and include low-energy subspaces in the code space. We show that this indeed leads to error-detecting codes with macroscopic distance. We identify two ways of constructing nontrivial codes by either considering single-particle excitations of varying momenta, or by considering multi-particle excitations above the ground space. Both constructions provide us with codes having distance scaling asymptotically significantly better than what can be achieved in the setup of our no-go theorem. In fact, the code distance is a polynomial arbitrarily close to  linear in the system size (i.e.,~$n$) in both cases.

Our first approach, using states of different momenta, involves the formalism of the excitation ansatz (see Section~\ref{sec:QEDClow} for a review). This gives a tensor network parametrization of momentum eigenstates associated with a Hamiltonian having quasi-particle excitations. We show the following:
\begin{reptheorem}{thm:excitationansatzparams}
Let $\nu \in (0,1)$ and let $\kappa,\lambda>0$ be such that 
\begin{align}
5\kappa+\lambda < \nu \ .
\end{align}
Let $A,B(p)$ be tensors associated with an injective excitation ansatz state~$\ket{\Phi_p(B;A)}$, where $p$ is the momentum of the state.  Then there is a subspace~$\cC\subset (\mathbb{C}^\physical)^{\otimes n}$ spanned by excitation ansatz states~$\{\ket{\Phi_p(B;A)}\}_p$ with different momenta~$p$ such that $\cC$ is an $(\epsilon,\delta)[[n,k,d]]$-AQEDC with parameters
\begin{align}
k&=\kappa \log_{\physical} n\ ,\\
d&=n^{1-\nu}\ ,\\
\epsilon &=\Theta(n^{-(\nu-(5\kappa+\lambda))})\ ,\\
\delta&=n^{-\lambda}\ .
\end{align}
\end{reptheorem} 
The physical interpretation of this result stems from the fact that excitation ansatz states approximate quasi-particle excitations: given a local gapped Hamiltonian, assuming a good MPS approximation to its ground state, we can construct an arbitrarily good approximation to its isolated quasi-particle excitation bands by the excitation ansatz. This approximation guarantee is shown using Lieb-Robinson type bounds~\cite{hastings2006spectral, hastings2006solving} based on  a previous result~\cite{haegeman2013elementary} which employs the method of energy filtering operators. Thus our result demonstrates that generic low-energy subspaces contain approximate error-detecting codes with the above parameters. Also, note that unlike the codes considered in \cite[Theorem 1]{brandao2017quantum}, the excitation ansatz codes are comprised of finite energy states, and not finite energy \it density \rm states.

We remark that the choice of momenta is irrelevant for this result; it is not necessary to restrict to nearby momenta. Instead, any subset of momentum eigenstates can be used. The only limitation here is that the number of different momenta is bounded by the dimension of the code space. This is related to the fact that localized wave functions (which would lead to a non-extensive code distance) are a superposition of many different momenta, a fact formally expressed by the position-momentum uncertainty relation.

Our second approach for side-stepping the no-go theorem is to consider multi-particle excitations.  We consider a specific model, the periodic Heisenberg-XXX spin chain Hamiltonian $H$ on $n$~qubits. We find that there are good error-detecting codes within the low-energy subspace of this system. For this purpose, we consider the state
\begin{align}
\ket{\Psi}=\sum_{m=1}^n \omega^m \hspin_m^-\ket{1}^{\otimes n}\qquad\textrm{ where }\qquad \omega=e^{2\pi i/n},\label{eq:onemagnonstateintrox} 
\end{align}
and where~$\hspin_m^-=\ket{0}\bra{1}$ changes the state of the $m$-th spin from $\ket{1}$ to $\ket{0}$. This has  energy $O(1/n^2)$ above the ground state energy of~$H$. The corresponding eigenspace is degenerate and contains all ``descendants'' $S_-^r\ket{\Psi}$ for $r\in \{0,\ldots,n-2\}$, where $S_-=\sum_{m=1}^n \hspin_m^-$ is the (total) spin lowering operator. We also note that each state~$S_-^r\ket{\Psi}$ has fixed momentum~$2\pi/n$, and that~$r$ directly corresponds to its total magnetization. We emphasize that these states are, in particular, not permutation-invariant. Our main result concerning these states is the following: 

\begin{reptheorem}{thm:magnoncodemain}
Let $\nu \in (0,1)$ and $\kappa,\lambda>0$ be such that
\begin{align}
6\kappa+\lambda<\nu\ .
\end{align}
Then there is a subspace~$\cC$ spanned by descendant states $\{S_-^r\ket{\Psi}\}_r$ with magnetization~$r$ pairwise differing by at least~$2$ such that  $\cC$ is an $(\epsilon,\delta)[[n,k,d]]$-AQEDC with parameters
\begin{align}
k&=\kappa \log_2 n\ ,\\
d&=n^{1-\nu}\ ,\\
\epsilon &=\Theta(n^{-(\nu-(6\kappa+\lambda))})\ ,\\
\delta&=n^{-\lambda}\ .
\end{align}
\end{reptheorem}

This code, which we call the  magnon-code, can also be seen to be realized by tensor networks. The state~\eqref{eq:onemagnonstateintrox} has an MPS description with bond dimension~$2$ and the  descendants~$S_-^r\ket{\Psi}$ can be expressed using a matrix-product operator (MPO) description of the operator~$S_-$.  More generally, it is known that these states form an example of the algebraic Bethe ansatz, and the latter have a natural tensor network description~\cite{verstraetekorepin}. This suggests that our results may generalize to other exactly solvable models.

\subsubsection*{Outline}
The paper is organized as follows. 
We discuss our notion of approximate error-detection and establish sufficient and necessary conditions in Section~\ref{AQED}. In Section~\ref{sec:expectationmpsops}, we review the basics of matrix product states. We also establish bounds on expectation values in terms of properties of the associated transfer operators. In Section~\ref{sec:nogo}, we prove our no-go theorem and show the limits of error-detection for code spaces limited to the ground space of a gapped local Hamiltonian. We then consider low-energy eigenstates of local Hamiltonians and show how they perform asymptotically better than the limits of the no-go theorem. We first consider  single-particle momentum eigenstates of generic local gapped Hamiltonians in Section~\ref{sec:QEDClow}. In Section~\ref{sec:qedcintegrable}, we consider codes defined by many-particle eigenstates of the Heisenberg-XXX model.

\section{Approximate Quantum Error-Detection\label{AQED}}
Here we introduce our notion of approximate quantum error-detection. In Section~\ref{sec:operationaldefaqed}, we give an operational definition of this notion. In Section~\ref{sec:sufficientconderrordetect}, we provide sufficient conditions for approximate quantum error-detection which are analogous to the Knill-Laflamme conditions for quantum error-correction~\cite{knill1997theory}. Finally, in Section~\ref{sec:necessaryconditions}, we give necessary conditions for a subspace to be an approximate quantum error-detecting code. 

\subsection{Operational definition of approximate error-detection\label{sec:operationaldefaqed}}
Let $\cN:\cB((\mathbb{C}^\physical)^{\otimes n})\rightarrow\cB((\mathbb{C}^\physical)^{\otimes n})$ be a CPTP map modeling noise on $n$ physical qubits. We introduce the following notion:
 
\begin{definition}\label{def:errdet1}
A subspace $\cC\subset (\mathbb{C}^\physical)^{\otimes n}$ (with associated projection $P$) is an $(\epsilon,\delta)$-approximate error-detection code for $\cN$ if for any pure state~$\ket{\Psi}\in\cC$ the following holds:
\begin{align}
\textrm{if}\qquad \tr(P\cN(\proj{\Psi}))\geq \delta\qquad\textrm{then}\qquad  \bra{\Psi}\rho_{\cN,P}\ket{\Psi}\geq 1-\epsilon\ ,
\end{align}
where $\rho_{\cN,P}=\tr(P\cN(\proj{\Psi}))^{-1}\cdot P\cN(\proj{\Psi})P$.
\end{definition}

In this definition, $\rho_{\cN,P}$ is the post-measurement state when applying the POVM~$\{P,I-P\}$ to~$\cN(\proj{\Psi})$. Roughly speaking, this definition ensures that the post-measurement state is $\epsilon$-close to the initial code state if the outcome of the POVM is~$P$. Note, however, that we only demand  this in the case where $\cN(\proj{\Psi})$ has an overlap with the code space of at least~$\delta$. The idea behind this definition is that if this overlap  is negligible, then the outcome~$P$ does not occur with any significant probability and the error-detection measurement may as well be omitted.
 
Definition~\ref{def:errdet1} is similar in spirit to operationally defined notions of approximate quantum error-correction considered previously. In~\cite{crepeauetal05}, approximate error-correction was defined in terms of the ``recoverable fidelity'' of any encoded pure state affected by noise. The restriction to pure states in the definition is justified by means of an earlier result  by Barnum, Knill, and Nielsen~\cite{barnumknillnielsen}.

We note that, by definition, an $(\epsilon,\delta)$-approximate error-detection code for~$\cN$ is also an $(\epsilon',\delta')$-approximate error-detection code for any $\epsilon\leq \epsilon'$ and $\delta\leq \delta'$. The traditional ``exact'' notion of a quantum error-detecting code~$\cC$ (see e.g.,~\cite{errorcorrectionoldintro}) demands that for a set~$\cF\subset \cB((\mathbb{C}^\physical)^{\otimes n})$ of {\em detectable errors}, we have 
\begin{align}
\bra{\Psi}E\ket{\Phi}&=\lambda_E\bra{\Psi}\Phi\rangle\qquad\textrm{ for all }\ket{\Psi},\ket{\Phi}\in\cC
\end{align}
for some scalar $\lambda_E\in\mathbb{C}$ depending only on~$E$, for all $E\in\cF$. It is straightforward to see that such a code defines a~$(0,0)$-approximate error-detecting code of any CPTP map~$\cN$ whose Kraus operators belong to~$\cF$.

\subsection{Sufficient conditions for approximate quantum error-detection\label{sec:sufficientconderrordetect}}
The following theorem shows that certain approximate Knill-Laflamme-type conditions are sufficient for approximate error-detection.
\begin{theorem}\mylabel{thm:AQEDConditions}
Let $\cN(\rho)=\sum_{j\in [J]} R_j\rho R_j^{\dagger}$ be a CPTP map
on~$\cB((\mathbb{C}^{\physical})^{\otimes n})$.   Let $\cC\subset (\mathbb{C}^{\physical})^{\otimes n}$ be a subspace
with orthonormal basis 
$\{\psi_\alpha\}_{\alpha\in [K]}$. Define
\begin{align}\label{eq:errapproxdef}
\errapprox:=\max_{\alpha,\beta\in [K]}\sum_{j\in [J]}\big|\bra{\psi_\alpha} R_j\ket{\psi_\beta}-\delta_{\alpha,\beta}\bra{\psi_1} R_j\ket{\psi_1}\big|^2\ .
\end{align}
Let $\delta>K^5\errapprox$ be arbitrary. Then the subspace $\cC$ is an $(\epsilon,\delta)$-approximate quantum error-detection code for~$\cN$ with $\epsilon = K^5\errapprox \delta^{-1}$.
\end{theorem}
This theorem deals with cases where the code dimension~$K$ is ``small'' compared to other quantities. We will later apply this theorem to the case where $K$ is polynomial, and where $\errapprox$ and $\delta$ are inverse polynomial in the system size~$n$.

We note that  the conditions of Theorem~\ref{thm:AQEDConditions} may appear more involved than  e.g., the Knill-Laflamme type conditions (see~\cite{knill1997theory}) for (exact) quantum error-correction: the latter involve one or two error operators (interpreted as Kraus operators of the channel), whereas in expression~\eqref{eq:errapproxdef}, we sum over all Kraus operators.  It appears that this is, to some extent, unavoidable when going from exact to approximate error-correction/detection in general. We note that (tight) approximate error-correction conditions~\cite{benyoreshkov}  obtained by considering the decoupling property of the complementary (encoding plus noise) channel similarly depend on the entire noise channel. Nevertheless, we show below that, at least for probabilistic noise, simple sufficient conditions for quantum error-detection involving only individual Kraus operators can be given. 

\begin{proof} Let us define 
\begin{align}
\err^\psi(R,\alpha,\beta):= \bra{\psi_\alpha}R\ket{\psi_\beta}- \delta_{\alpha,\beta}\bra{\psi_1}R\ket{\psi_1}\ .
\end{align}
Consider an arbitrary orthonormal basis  $\{\varphi_\alpha\}_{\alpha\in [K]}\in\cC\subset(\mathbb{C}^\physical)^{\otimes n}$ of~$\cC$. Let~$U$ be a unitary matrix such that 
\begin{align}
\varphi_\alpha &=\sum_{\beta\in [K]}U_{\alpha,\beta}\psi_\beta\qquad\textrm{ for all }\alpha\in [K]\ .
\end{align}
Because $\sum_{\gamma\in [K]} (U^\dagger)_{\alpha,\gamma}U_{\gamma,\beta}=\delta_{\alpha,\beta}$, we obtain by straightforward computation
\begin{align}
\bra{\varphi_\alpha}R\ket{\varphi_\beta}-\delta_{\alpha,\beta}\bra{\psi_1}R\ket{\psi_1}&=\sum_{\gamma,\delta\in [K]}\overline{U_{\alpha,\gamma}}U_{\beta,\delta}\,\err^\psi(R,\gamma,\delta)\ .
\end{align}
We conclude that
\begin{align}
|\bra{\varphi_\alpha}R\ket{\varphi_\beta}|&\leq \sum_{\gamma,\delta\in [K]} |\err^{\psi}(R,\gamma,\delta)|\leq K\cdot \sqrt{\sum_{\gamma,\delta\in [K]} |\err^{\psi}(R,\gamma,\delta)|^2}\qquad\textrm{ for }\alpha\neq \beta
\end{align}
because $\max_{\gamma,\delta} |\overline{U}_{\alpha,\gamma}U_{\beta,\delta}|\leq 1$ for a unitary matrix~$U$ and by using the Cauchy-Schwarz inequality. By definition of $\err$ and $\errapprox$, this implies that 
\begin{align}
\bra{\varphi_\alpha}\cN(\proj{\varphi_\beta})\ket{\varphi_\alpha}&\leq   K^4 \errapprox\qquad\textrm{ for }\alpha\neq \beta \mylabel{eq:upperboundphialphactwo}
\end{align} 
for any orthonormal basis $\{\varphi_\alpha\}_{\alpha\in [K]}$ of $\cC$.

Let now $\delta>0$ be given and let $\Psi\in\cC$ be an arbitrary state in the code space such that 
\begin{align}
\tr(P\cN(\proj{\Psi}))\geq \delta\ .\mylabel{eq:assumptionpsindelta}
\end{align}  
Let us pick an orthonormal basis  $\{\varphi_\alpha\}_{\alpha\in [K]}\in\cC\subset(\mathbb{C}^\physical)^{\otimes n}$ of~$\cC$ such that $\varphi_1=\Psi$. Then
  \begin{align}
1-  \bra{\Psi}\rho_{\cN,P}\ket{\Psi}&=1- \frac{\bra{\Psi}\cN(\proj{\Psi})\ket{\Psi}}{\tr(P\cN(\proj{\Psi}))}\\
&=\frac{1}{\tr(P\cN(\proj{\Psi}))}\cdot \left(\tr(P\cN(\proj{\Psi}))-
\bra{\Psi}\cN(\proj{\Psi})\ket{\Psi}\right)\\
&= \frac{1}{\tr(P\cN(\proj{\Psi}))}\cdot \sum_{\alpha=2}^K \bra{\varphi_\alpha}\cN(\proj{\varphi_1})\ket{\varphi_\alpha}\\
&\leq \frac{1}{\delta }\cdot K^5\errapprox
  \end{align}
because of~\eqref{eq:assumptionpsindelta}  and~\eqref{eq:upperboundphialphactwo}. The claim follows. 

 \end{proof}
 
If there are vectors~$\{\eta_{\alpha,\beta}\}_{\alpha,\beta\in [K]}$ such  that 
\begin{align}
\big|\bra{\psi_\alpha} R_j\ket{\psi_\beta}-\delta_{\alpha,\beta}\bra{\psi_1} R_j\ket{\psi_1}\big|&\leq \|R_j\eta_{\alpha,\beta}\|\qquad\textrm{ for all }j\in [J]\ , \mylabel{eq:retaalphabeta}
\end{align}
then this implies the bound
\begin{align}
\errapprox\leq \max_{\alpha,\beta}\tr(\cN(\proj{\eta_{\alpha,\beta}}))
&=\max_{\alpha,\beta}\|\eta_{\alpha,\beta}\|^2\ .
\end{align}
Unfortunately, good bounds of the form~\eqref{eq:retaalphabeta} are not straightforward to establish in the cases considered here. Instead, we consider a slightly weaker condition (see equation~\eqref{eq:upperbounduniversal}) which still captures many cases of interest. In particular, it provides a simple criterion for establishing that a code can detect probabilistic Pauli errors with a  certain maximum weight. Correspondingly, we introduce the following definition:
\begin{definition}\label{def:errdet2}
An $(\epsilon,\delta)[[n,k,d]]$-AQEDC~$\cC$ is a
$\physical^k$-dimensional subspace of $(\mathbb{C}^\physical)^{\otimes n}$ such that~$\cC$ is an $(\epsilon,\delta)$-error-detecting code for any CPTP map of the form
\begin{align}
\cN(\rho)=\sum_{j\in [J]} p_j F_j\rho F_j^{\dagger}\ ,\mylabel{eq:nrhozchannel}
\end{align}
where each $F_j$ is a $d$-local operator with $\|F\|\leq 1$ and $\{p_j\}_{j\in [J]}$ is a probability distribution. 
\end{definition}
We then have the following sufficient condition: 
\begin{corollary}\mylabel{cor:errordectionapproximate}
Let $K=\physical^k$ and $\cC\subset (\mathbb{C}^{\physical})^{\otimes n}$ be a code with orthonormal basis $\{\psi_\alpha\}_{\alpha\in [K]}$ satisfying (for some $\gamma>0$),  

\begin{align}
\big|\bra{\psi_\alpha} F\ket{\psi_\beta}-\delta_{\alpha,\beta}\bra{\psi_1} F\ket{\psi_1}\big|\leq \gamma\cdot \|F\|\qquad\textrm{ for all }\alpha,\beta\in [K]\ ,\mylabel{eq:upperbounduniversal}
\end{align}
for every $d$-local operator $F$ on~$(\mathbb{C}^{\physical})^{\otimes n}$. Let $\delta > K^5\gamma^2$. Then $\cC$ is an $(\epsilon=K^5\gamma^2 \delta^{-1},\delta)[[n,k,d]]$-AQEDC.
\end{corollary}
\begin{proof}
Defining $R_j=\sqrt{p_j} F_j$, the claim follows immediately from Theorem~\ref{thm:AQEDConditions}.
\end{proof}
Note that the exponents in this statement are not optimized, and could presumably be improved. We have instead opted for the presentation of a simple proof, as this ultimately provides the same qualitative statements.

We also note that the setting considered in Corollary~\ref{cor:errordectionapproximate}, i.e., our notion of $(\epsilon,\delta)[[n,k,d]]$-error-detecting codes, goes beyond existing work on approximate error-detection/correction~\cite{brandao2017quantum,Kim2017, pastawski2017towards}, where typically only noise  channels with Kraus (error) operators acting on a fixed, contiguous (i.e., geometrically local) set of $d$~physical spins are considered. At the same time, our results are limited to convex combinations of the form~\eqref{eq:nrhozchannel}. It remains an open problem whether these codes also detect noise given by more general (coherent) channels.

\subsection{Necessary conditions for approximate quantum error-detection\label{sec:necessaryconditions}}
Here we give a partial converse to Corollary~\ref{cor:errordectionapproximate}, which shows that a condition of the form~\eqref{eq:upperbounduniversal} is indeed necessary for approximate quantum error-detection. 
\begin{lemma}\mylabel{lem:necessaryconditiondetect}
Let $\psi_1,\psi_2\in (\mathbb{C}^{\physical})^{\otimes n}$ be two orthonormal states in the code space $\mathcal C$ and $F=F_{\cS}\otimes I_{[n]\backslash \cS}\in \cB((\mathbb{C}^{\physical})^{\otimes d})$ an orthogonal projection acting on~$d$~sites~$\cS\subset [n]$ such that
\begin{align}
|\bra{\psi_1}F\ket{\psi_1}-\bra{\psi_2}F\ket{\psi_2}|&=\eta
\end{align}
for some $\eta\in [0,1]$, with $1-\eta\ll 1$.  Then any subspace~$\cC\subset (\mathbb{C}^{\physical})^{\otimes n}$ of dimension~$\physical^k$ is not an $(\epsilon,\delta)[[n,k,d]]$-code for
\begin{align}
\epsilon&< 1-10(1-\eta),\quad\text{and}\\
\delta &< \eta^2\ .
\end{align}
\end{lemma}
\begin{proof}
Let 
\begin{align}
F_{j,k}&:=\bra{\psi_j} F\ket{\psi_k}\qquad\textrm{ for } j,k\in \{1,2\}\ .
\end{align}
By choosing the phase of~$\ket{\psi_1}$ appropriately, we may assume that $F_{1,2}\geq 0$. 
Note that  $F_{1,2}=F_{2,1}\leq \|F\psi_2\|=\sqrt{\bra{\psi_2}F\ket{\psi_2}}$ by the Cauchy-Schwarz inequality and because~$F$ is a projection. Let us denote the entries of~$F$ by 
\begin{align}
F&=\begin{pmatrix}
p & r \\
r & q
\end{pmatrix}
\end{align}
where $q\in [0,1-\eta]$, $p=q+\eta$, and $r\in [0,\sqrt{q}]$. Let us define a CPTP map~$\cN$ of the form~\eqref{eq:nrhozchannel} by 
\begin{align}
\cN(\rho)&=e^{i \pi F}\rho e^{-i\pi F}\qquad\textrm{ where }\qquad F=F_{\cS}\otimes I_{[n]\backslash \cS}\ .
\end{align}
Let $\hat{P}=\sum_{j=1}^2 \proj{\psi_j}$.  
Consider the normalized vector~$\ket{\Psi}=\frac{1}{\sqrt{2}}(\ket{\psi_1}+\ket{\psi_2})$.  Then 
\begin{align}
\hat{P}\cN(\proj{\Psi})\hat{P}&=\frac{1}{2}\sum_{i,j,k,\ell}W_{k,i}\overline{W_{\ell,j}} \ket{\psi_k}\bra{\psi_\ell}\ , \mylabel{eq:identitypnpsi}
\end{align}
where 
\begin{align}
W_{j,k}&:=\bra{\psi_j} e^{i\pi F}\ket{\psi_k}\qquad\textrm{ for } j,k\in \{1,2\}\ .
\end{align}
Observe that since $F^2=F$ is a projection, we have
$e^{i\pi F}= I-2F$, thus  the entries of~$W$ are
\begin{align}
W_{j,k}&=\delta_{j,k}-2F_{j,k}\qquad\textrm{ for }j,k\in \{1,2\}\ .
\end{align}
In particular, from~\eqref{eq:identitypnpsi} we obtain  for the projection $P$ onto~$\cC$
\begin{align}
\tr\left(P\cN(\proj{\Psi})P\right)
&\geq \tr\left(
\hat{P}\cN(\proj{\Psi})\hat{P}\right)\\
&=\frac{1}{2}\sum_{i,j,k}
W_{k,i}\overline{W_{k,j}}\nonumber\\
&=1-2p+2p^2-2q+2q^2+4r(p+q-1+r)\nonumber\\
&\geq (p-q)^2=\eta^2\ ,\label{eq:lbnddeltax}
\end{align}
where we used that the last expression is minimal (and equal to $(p-q)^2$) for $r=1/2 (1-p-q)$.
We also have
\begin{align}
\bra{\Psi}\cN(\proj{\Psi})\ket{\Psi}&=\frac{1}{4}\sum_{i,j,k,\ell}W_{k,i}\overline{W_{\ell,j}}\\
&=(2r+p+q-1)^2\\
&=(2(r+q)-(1-\eta))^2\ .
\end{align}
This expression is maximal for $(r,q)$ each maximal (since both are non-negative), hence for $(r,q)=(\sqrt{1-\eta},1-\eta)$ and we obtain the upper bound

\begin{align}
 \bra{\Psi}\cN(\proj{\Psi})\ket{\Psi}&\leq
 (1-\eta+2\sqrt{1-\eta})^2\leq 9 (1-\eta)\ ,
  \end{align}
  where we used that $x\leq \sqrt{x}$ for $x\in [0,1]$.
 This implies with~\eqref{eq:lbnddeltax} that for $\rho_{\cN,P}=\tr(P\cN(\proj{\Psi}))^{-1}\cdot P\cN(\proj{\Psi})P$ we have 
 \begin{align}
 \bra{\Psi}\rho_{\cN,P}\ket{\Psi}&
 \leq \frac{9(1-\eta)}{\eta^2}=\frac{9(1-\eta)}{(1-(1-\eta))^2}\leq 10(1-\eta)
 \end{align}
 for $1-\eta\ll 1$. Thus   
 \begin{align}
 1- \bra{\Psi}\rho_{\cN,P}\ket{\Psi}\geq 1-10(1-\eta)\ .
 \end{align}
 With~\eqref{eq:lbnddeltax}, this implies the claim. 
 
 \end{proof}
 
We reformulate Lemma~\ref{lem:necessaryconditiondetect}, by stating it in terms of reduced density matrices, as follows:
\begin{lemma}\mylabel{lem:distancenecessary}
Let $\psi_1,\psi_2\in (\mathbb{C}^\physical)^{\otimes n}$ be two orthonormal vectors  in a subspace~$\cC\subset  (\mathbb{C}^\physical)^{\otimes n}$ of dimension~$\physical^k$. Fix a region~$R\subset [n]$ of size $|R|=d$ and let $\rho_j=\tr_{[n]\backslash R}\proj{\psi_j}$, $j=1,2$ be the reduced density matrices on~$R$. 
Then~$\cC$ is not a $(\epsilon,\delta)[[n,k,d]]$-error-detecting code for
\begin{align}
\epsilon &< 1-10\zeta(\rho_1,\rho_2)\ ,\qquad \text{and}\\
\delta &< (1-\zeta(\rho_1,\rho_2))^2\ ,
\end{align}
where $\zeta(\rho_1,\rho_2):=\max\{\rank \rho_1,\rank \rho_2\}^2 \cdot\tr(\rho_1\rho_2)$.
\end{lemma}
\begin{proof}
By definition of the trace distance, the projection~$F$ onto the positive part of $\rho_1-\rho_2$ satisfies
\begin{align}
\eta:=\frac{1}{2}\|\rho_1-\rho_2\|_1&=\tr(F(\rho_1-\rho_2))\ .
\end{align} 
With the inequality $\|A\|_1\leq \sqrt{\rank(A)}\|A\|_F$ we get the bound
\begin{align}
F(\rho_1,\rho_2)=\|\sqrt{\rho_1}\sqrt{\rho_2}\|_1^2\leq D^2 \|\sqrt{\rho_1}\sqrt{\rho_2}\|^2_F=D^2\tr(\rho_1\rho_2)
\end{align}
on the fidelity of $\rho_1$ and $\rho_2$, where $D=\rank(\sqrt{\rho_1}\sqrt{\rho_2})\leq \max\{\rank\rho_1,\rank\rho_2\}$.
Inserting this into the inequality $\frac{1}{2}\|\rho_1-\rho_2\|_1\geq 1-F(\rho_1,\rho_2)$  yields 
\begin{align}
\eta&\geq 1- \max\{\rank \rho_1,\rank \rho_2\}^2 \cdot\tr(\rho_1\rho_2)\ .
\end{align}
The claim then follows from Lemma~\ref{lem:necessaryconditiondetect} and the fact that if~$\cC$ is not an $(\epsilon,\delta)[[n,k,d]]$-code, then it is not an $(\epsilon',\delta')[[n,k,d]]$-code for any $\delta'\leq \delta$ and $\epsilon'\leq \epsilon$. 
\end{proof}
We will use Lemma~\ref{lem:distancenecessary} below to establish our no-go result for codes based on injective MPS with open boundary conditions.

\section{On expectation values of local operators in MPS\label{sec:expectationmpsops}}
Key to our analysis are expectation values of local observables in MPS, and more generally, matrix elements of local operators with respect to different MPS. These directly determine whether or not the considered subspace satisfies the approximate quantum error-detection conditions. To study these quantities, we first review the terminology of transfer operators (and, in particular, injective MPS) in Section~\ref{sec:reviewmps}.  In Section~\ref{sec:transfermatrixtechniques}, we establish bounds on the matrix elements and the norms of transfer operators. These will subsequently be applied in all our derivations.

\subsection{Review of matrix product states\label{sec:reviewmps}}
A matrix product state (or MPS) of bond dimension~$\bond$ is a state $\ket{\Psi}$ on $(\mathbb{C}^\physical)^{\otimes \systemsize}$ which is parametrized by a collection of $D\times D$ matrices.
In this paper, we focus on uniform, site-independent MPS. Such a state is fully specified by a family $A=\{A_j\}_{j=1}^{\physical}$
of $D\times D$ matrices describing the ``bulk properties'' of the state, together with a ``boundary condition'' matrix $X\in \cB(\mathbb{C}^{\bond})$. We write~$\ket{\Psi}=\ket{\Psi(A,X,\systemsize)}$ for such a state, where we often suppress the defining parameters $(A,X,\systemsize)$ for brevity.

Written in the standard computational basis, the state $\ket{\Psi(A,X,\systemsize)}$ is expressed as
\begin{align}
\ket{\Psi}&=\sum_{i_1,\ldots,i_\systemsize\in [\physical]} \tr\left(A_{i_1}\cdots A_{i_\systemsize}X\right) \ket{i_1\cdots i_\systemsize}\ \mylabel{eq:mpssiteindependent}
\end{align}
for a family $\{A_j\}_{j=1}^\physical\subset \cB(\mathbb{C}^{\bond})$ of matrices. 
The number of sites $\systemsize\in \mathbb{N}$ is called the system size, and each site is of local dimension $\physical\in \mathbb{N}$, which is called the physical dimension of the system. The parameter $\bond\in\mathbb{N}$ is called the bond, or virtual, dimension. This state can be represented graphically as a tensor network as in Figure~\ref{fig:matrixproductstates}.

\begin{figure}
\centering
\includegraphics[scale=0.18]{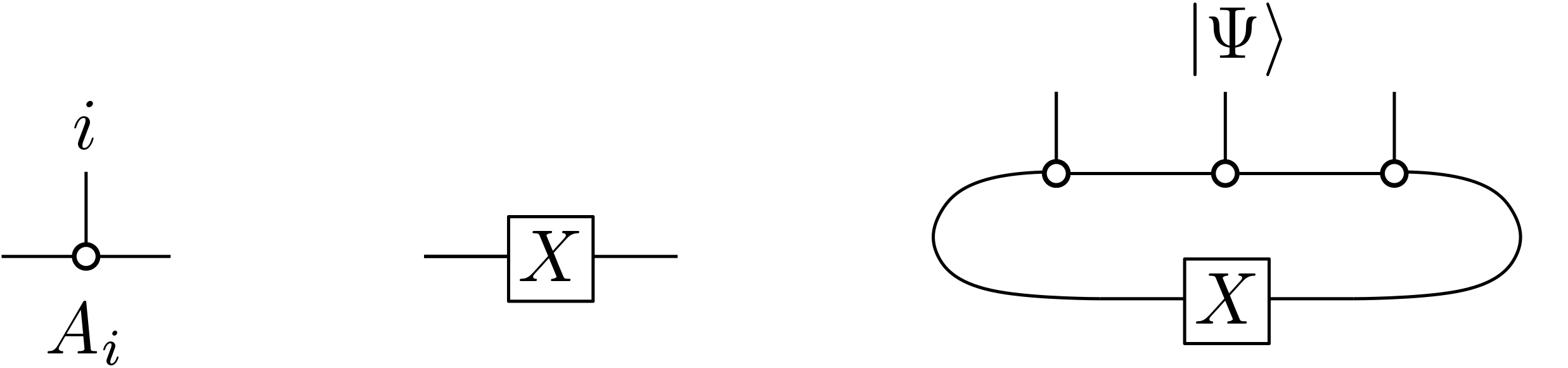}
\caption{This figure illustrates an MPS with $n=3$ physical spins, defined in terms of a family $\{A_j\}_{j=1}^\physical$ of matrices and a matrix~$X$.\mylabel{fig:matrixproductstates}}
\end{figure}

Note that the family of matrices $A=\{A_j\}^\physical_{j=1}$ of a site-independent MPS equivalently defines a three-index tensor $(A_j)_{\alpha\beta}$ with one ``physical'' ($j$) and two ``virtual'' ($\alpha,\beta$) indices. We call this the local MPS tensor associated to $\ket{\Psi(A,X,n)}$.

The matrices $\{A_j\}_{j=1}^\physical$ defining a site-independent MPS~$\ket{\Psi(A,X,n)}$ give rise to a completely positive (CP) linear map $\cE:\cB(\mathbb{C}^D) \rightarrow \cB(\mathbb{C}^D)$ which acts on $Y\in \cB(\mathbb{C}^D)$ by
\begin{equation}\cE(Y) = \sum_{i=1}^\physical A_iYA^\dagger_i\ .\mylabel{eq:transferchannel}\end{equation}
Without loss of generality (by suitably normalizing the matrices~$\{A_j\}_{j=1}^\physical$), we assume  that~$\cE$ has spectral radius~$1$.  This implies that~$\cE$ has a positive semi-definite fixed point $r\in \cB(\mathbb{C}^\bond)$ by the Perron-Frobenius Theorem, see~\cite[Theorem~2.5]{evanshoegh}. We say that the MPS~$\ket{\Psi(A,X,n)}$ is {\em injective}\footnote{Injective MPS are known to be ``generic". More precisely, consider the space $\mathbb{C}^D\otimes \mathbb{C}^D \otimes \mathbb{C}^\physical$ of all defining tensors with physical dimension $\physical$ and bond dimension $D$. Then the set of defining tensors with a primitive transfer operator forms an open, co-measure zero set. The definition of injective that we use here differs from the one commonly used in the literature (cf. \cite{perez2006matrix}), but is ultimately equivalent. For a proof of equivalence, see Definition 8, Lemma 6, and Theorem 18 of \cite{ruizthesis}.} if the associated map~$\cE$ is {\em primitive}, i.e., if the fixed-point $r$ is positive definite (and not just positive semi-definite), and if the eigenvalue~$1$ associated to $r$ is the only eigenvalue on the unit circle, including multiplicity~\cite[Theorem~6.7]{wolfchannel}.

From expression~\eqref{eq:mpssiteindependent}, we can see that there is a  gauge freedom of the form
\begin{align}
\tilde{A}_j&=P^{-1}A_jP,\qquad \tilde{X} = P^{-1}XP,\quad\textrm{ for }j=1,\ldots,\physical, \mylabel{eq:gaugefreedommps}
\end{align}
for every invertible matrix~$P\in GL(\mathbb{C}^D)$, for which $\ket{\Psi(A,X,n)}=\ket{\Psi(\tilde{A},\tilde{X},n)}$. Given an injective MPS, the defining tensors can be brought into a \it canonical form \rm by exploiting this gauge freedom in the definition of the MPS.\footnote{The canonical form holds for non-injective MPS as well, see~\cite{perez2006matrix}. We only consider the injective case here.}

One proceeds as follows: given an injective MPS, let~$r$ denote the unique fixed-point of the transfer operator~$\cE$. We can apply the gauge freedom~\eqref{eq:gaugefreedommps} with $P=\sqrt{r}$ to obtain an equivalent MPS description by matrices $\{\tilde{A}_j:=r^{-1/2}A_j r^{1/2}\}_{j=1}^\physical$, where the associated map~$\tilde{\cE}$  is again primitive with spectral radius~$1$, but now with the identity operator~$\tilde{r}=\identityoperator_{\mathbb{C}^D}$ as the unique fixed-point.
 
Similarly, one can show that the adjoint
\[\cE^\dagger(Y)=\sum_{i=1}^\physical A_i^\dagger Y A_i\]
of a primitive map~$\cE$ is also primitive.\footnote{Note that the adjoint is taken with respect to the Hilbert-Schmidt inner product on~$\cB(\mathbb{C}^D)$. One way to see that the adjoint of a primitive map is primitive is to note that an equivalent definition for primitivity given in ~\cite[Theorem~6.7(2)]{wolfchannel} is in terms of irreducible maps. A map is irreducible if and only if its adjoint is irreducible (see the remarks in~\cite{wolfchannel} after Theorem 6.2). This in turn means that a map is primitive if and only if its adjoint is primitive.}  Since the spectrum of~$\cE^\dagger$ is given by $\mathsf{spec}(\cE^\dagger)=\overline{\mathsf{spec}(\cE)}$, this implies that the map~$\cE^\dagger$ has a unique positive fixed-point~$\ell$ with eigenvalue~$1$, with all other eigenvalues having magnitude less than~$1$. 

Now, let $\tilde{\ell}$ denote the unique fixed-point of the previously defined $\tilde{\cE}$. Since $\tilde{\ell}$ is positive definite, it is unitarily diagonalizable:
\[ \tilde{\ell} = U\Lambda U^\dagger, \]
with $U$ being a unitary matrix, and $\Lambda$ being a diagonal matrix with all diagonal entries positive. Using the gauge freedom~\eqref{eq:gaugefreedommps} in the form
\begin{align}
  \tilde{A}_j \mapsto \dtilde{A}_j = U^\dagger\tilde{A}_jU\qquad\textrm{ for }j=1,\ldots,\physical\ ,\mylabel{eq:tildeakdefgauge}
\end{align} 
we obtain an equivalent MPS description such that the associated channel~$\dtilde{\cE}^\dagger$ has a fixed-point given by a positive definite diagonal matrix~$\Lambda$. We may without loss of generality take $\Lambda$ to be normalized as $\tr(\Lambda) = 1$. It is also easy to check that the identity operator $\mathbb{I}_D$ remains the unique fixed-point of $\dtilde{\cE}$.
 
In summary, given an injective MPS with associated map~$\cE$, one may, by using the gauge freedom, assume without loss of generality that:
\begin{enumerate}[(i)]
    \item The unique  fixed-point~$r$ of $\mathcal{E}$ is equal to the identity, i.e.,~$r=I_{\mathbb{C}^D}$.\mylabel{it:proponeinjectv}
    \item The unique  fixed-point~$\ell$ of $\mathcal{E}^\dagger$ is given by a positive definite diagonal matrix $\ell=\Lambda$, normalized so that $\tr(\Lambda) = 1$.\mylabel{it:proptwoinjectv}
\end{enumerate}
An MPS with defining tensors $A$ satisfying these two properties above is said to be in \it canonical form\rm.

In the following, after fixing a standard orthonormal basis~$\{\ket{\alpha}\}_{\alpha=1}^\bond$ of $\mathbb{C}^D$, we identify elements~$X\in\cB(\mathbb{C}^D)$ with vectors $\kket{X}\in\mathbb{C}^D\otimes\mathbb{C}^D$ via the vectorization isomorphism
\[X\mapsto \kket{X}:=(X^T\otimes I) \sum_{\alpha=1}^D \ket{\alpha}\otimes\ket{\alpha}=\sum_{\alpha,\beta=1}^D X_{\alpha,\beta}\ket{\beta}\otimes\ket{\alpha},\] 
where $X=\sum_{\alpha,\beta=1}^D X_{\alpha,\beta}\ket{\alpha}\bra{\beta}$. It is easy to verify that $\sspr{X}{Y}=\tr(X^\dagger Y)$, i.e., the standard inner product on~$\mathbb{C}^D\otimes\mathbb{C}^D$ directly corresponds to the Hilbert-Schmidt inner product of operators in~$\cB(\mathbb{C}^D)$ under this identification. Furthermore, under this isomorphism, a super-operator~$\cE : \cB(\mathbb{C}^D)\rightarrow \cB(\mathbb{C}^D)$ becomes a linear map $E:\mathbb{C}^D\otimes\mathbb{C}^D\rightarrow\mathbb{C}^D\otimes\mathbb{C}^D$ defined by 
\[\kket{\cE(X)}=E\kket{X}\] 
for all $X\in\cB(\mathbb{C}^D)$. The matrix~$E$ is simply the matrix representation of~$\cE$, thus~$E$ has the same spectrum as~$\cE$. Explicitly, for a map of the form~\eqref{eq:transferchannel}, it is given by   
\begin{align}
E&=\sum_{i = 1}^\physical \overline{A_i}\otimes A_i\ .\mylabel{eq:transferoperatorbasicdef}
\end{align}
The fixed-point equations for a fixed-point~$r$ of~$\cE$ and a fixed-point~$\ell$ of~$\cE^\dagger$ become
\begin{align}
\bbra{\ell}E&=\bbra{\ell}\qquad\textrm{ and }\qquad E\kket{r}=\kket{r}\ ,\mylabel{eq:leftrighteveqs}
\end{align}
i.e., the corresponding vectors are left and right eigenvectors of $E$, respectively.

For a site-independent MPS~$\ket{\Psi(A,X,n)}$, defined by matrices $\{A_j\}_{j=1}^\physical$, 
we call the associated matrix ~$E$ (cf.~\eqref{eq:transferoperatorbasicdef}) the {\em transfer matrix}.  
Many key properties of a site-independent MPS are captured by its transfer matrix. For example, the normalization of the state is given by
\begin{align}
\|\Psi\|^2&=\langle\Psi|\Psi\rangle=\tr(E^{n}(\overline{X}\otimes X))\ .
\end{align}
If the MPS is injective, then, according to~\eqref{it:proponeinjectv}--\eqref{it:proptwoinjectv}, it has a Jordan decomposition of the form 
\begin{align}
    E = \kket{I}\bbra{\Lambda} \oplus \tilde{E}.\mylabel{eq:transferoperatorilambdae}
\end{align}
In this expression, $\kket{I}\bbra{\Lambda}$ is the ($1$-dimensional) Jordan block corresponding to eigenvalue~$1$, whereas~$\tilde{E}$ is a direct sum of Jordan blocks with eigenvalues of modulus less than~$1$. 
The second largest eigenvalue~$\lambda_2$ of~$E$ has a direct interpretation in terms of the correlation length~$\xi$ of the state, which determines two-point correlators $|\langle \sigma_j\sigma_{j'}\rangle-\langle \sigma_j\rangle \cdot\langle\sigma_{j'}\rangle|\sim e^{-|j-j'|/\xi}$. The latter is given by $\xi=\log(1/\lambda_2)$.

For an injective MPS, the fact that~$\kket{I}$ is the unique right-eigenvector of~$E$ with eigenvalue~$1$ implies the normalization condition
\begin{align}
\sspr{\Lambda}{I}=\tr(\Lambda)=1\ .\mylabel{eq:normalizationconditionlambdaI}
\end{align}
We will represent these identities diagrammatically, which is convenient for later reference. The matrix~$\Lambda$ will be shown by a square box, the identity matrix corresponds to a straight line. That is, the normalization condition~\eqref{eq:normalizationconditionlambdaI} takes the form
\begin{align}
\raisebox{-.54\height} {\includegraphics[width=2.5cm]{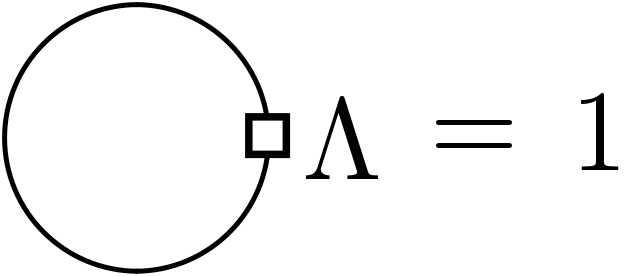}} \end{align}
and the left and right eigenvalue equations~\eqref{eq:leftrighteveqs} 
\begin{align}
\raisebox{-.54\height} {\includegraphics[width=3.7cm]{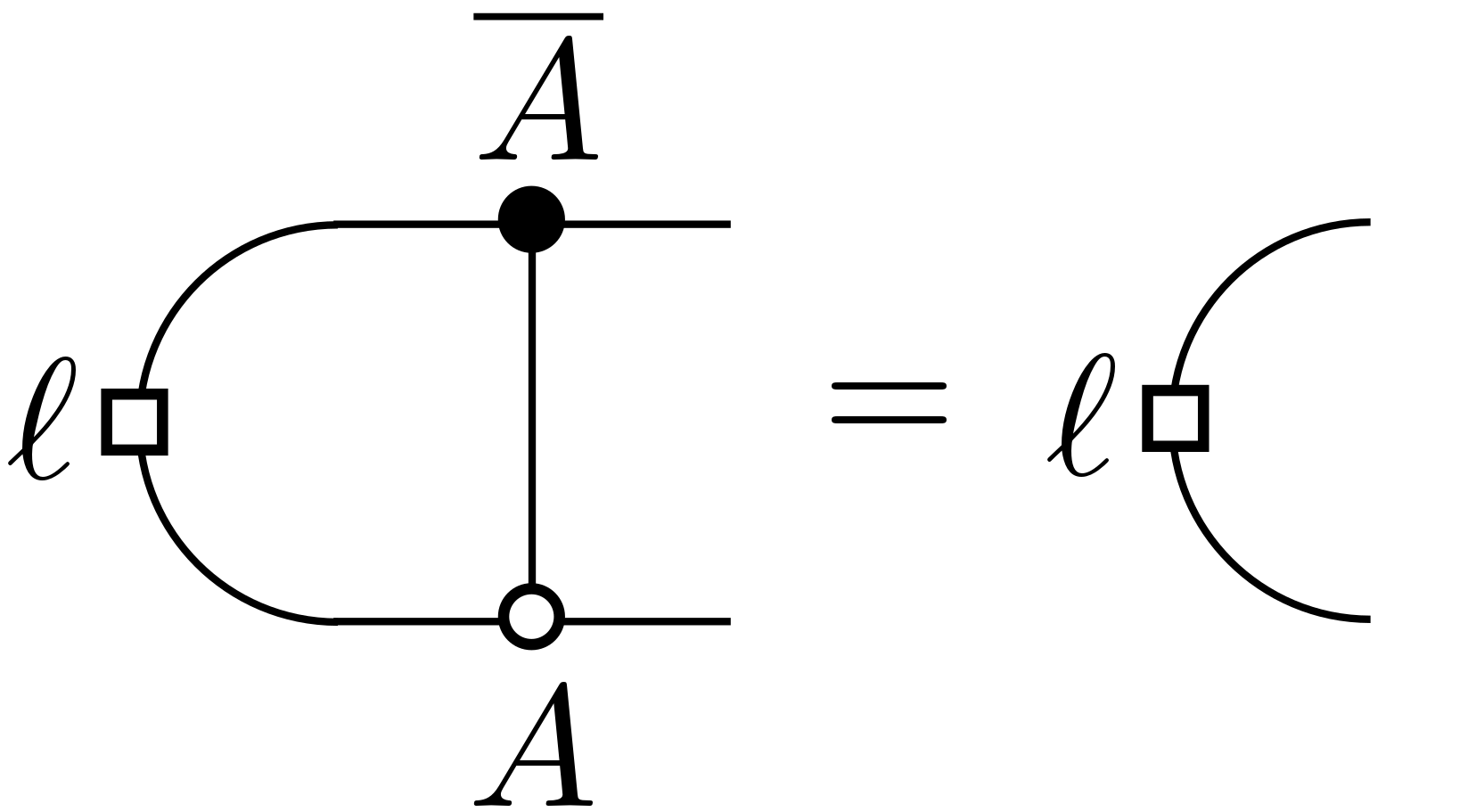}}\mylabel{eq:lefteigenvalueeq}\qquad\qquad\qquad 
\raisebox{-.54\height} {\includegraphics[width=4.4cm]{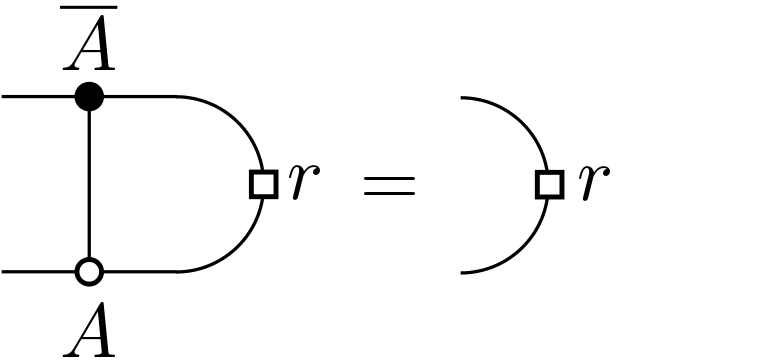}}
\end{align}

\subsection{Transfer matrix techniques\label{sec:transfermatrixtechniques}}
Here we establish some 
essential statements for the analysis of transfer operators. In Section~\ref{sec:generalmixedtransferops}, we  introduce generalized (non-standard) transfer operators: these can be used to express the matrix elements
of the form $\bra{\Psi}F\ket{\Psi'}$ of local operators~$F$ with respect to pairs of MPS $(\Psi,\Psi')$. In Section~\ref{sec:normboundsgeneralizedtrsfop}, we establish bounds on the norm of such operators. Relevant quantities
appearing in these bounds are the second largest eigenvalue~$\lambda_2$ of the transfer matrix, as well as the sizes of its Jordan blocks. 

\subsubsection{More general and mixed transfer operators\label{sec:generalmixedtransferops}}
Consider a single-site operator $Z\in\cB(\mathbb{C}^\physical)$. The generalized transfer matrix $E_Z\in \cB(\mathbb{C}^D\otimes\mathbb{C}^D)$ is defined as 
\begin{align}
E_Z&=\sum_{n,m}  \bra{m}Z\ket{n}\overline{A_m}\otimes A_n\ .\mylabel{eq:ezdeffirst}
\end{align}
We further generalize this as follows:
if $Z_1,\ldots,Z_d\in \cB(\mathbb{C}^\physical)$, then $E_{Z_1\otimes\cdots \otimes Z_d}\in\cB(\mathbb{C}^D\otimes\mathbb{C}^D)$ is the operator 
\begin{align}
E_{Z_1\otimes\cdots \otimes Z_d}&=E_{Z_1}\cdots E_{Z_d}\ .
\end{align}
This definition extends by linearity
to any operator $F\in\cB((\mathbb{C}^\physical)^{\otimes d})$, and gives a corresponding operator $E_F\in\cB(\mathbb{C}^D\otimes\mathbb{C}^D)$. The tensor network  diagrams for these definitions are given in Figure~\ref{fig:varioustransferops}, and the composition of the corresponding maps is illustrated in Figure~\ref{fig:productoperationtransfer}.

\begin{figure}
\centering
\includegraphics[width=10cm]{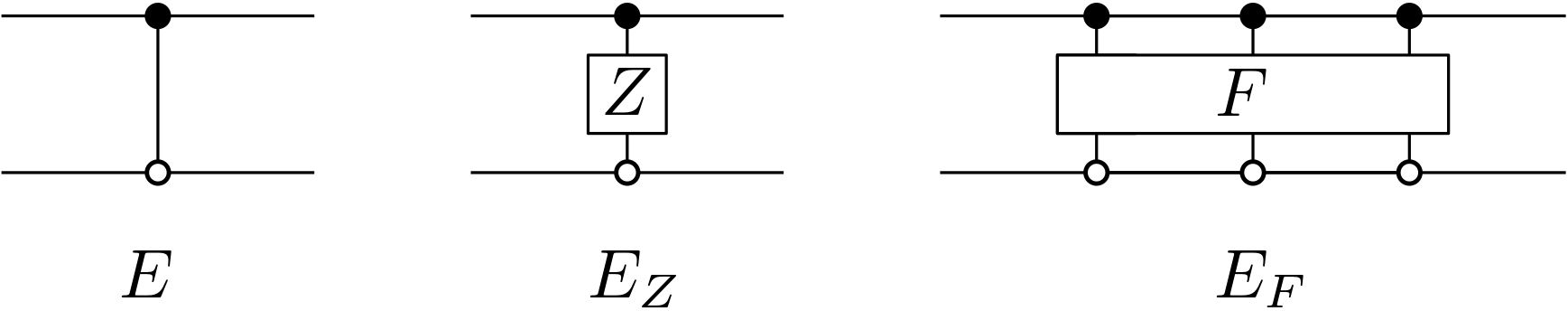}
\caption{The transfer operator $E$, as well as $E_Z$ for  $Z\in \cB(\mathbb{C}^\physical)$, and $E_F$ for  $F\in (\mathbb{C}^\physical)^{\otimes 3}$.}
\mylabel{fig:varioustransferops}
\end{figure}

\begin{figure}
\centering
\includegraphics[width=4cm]{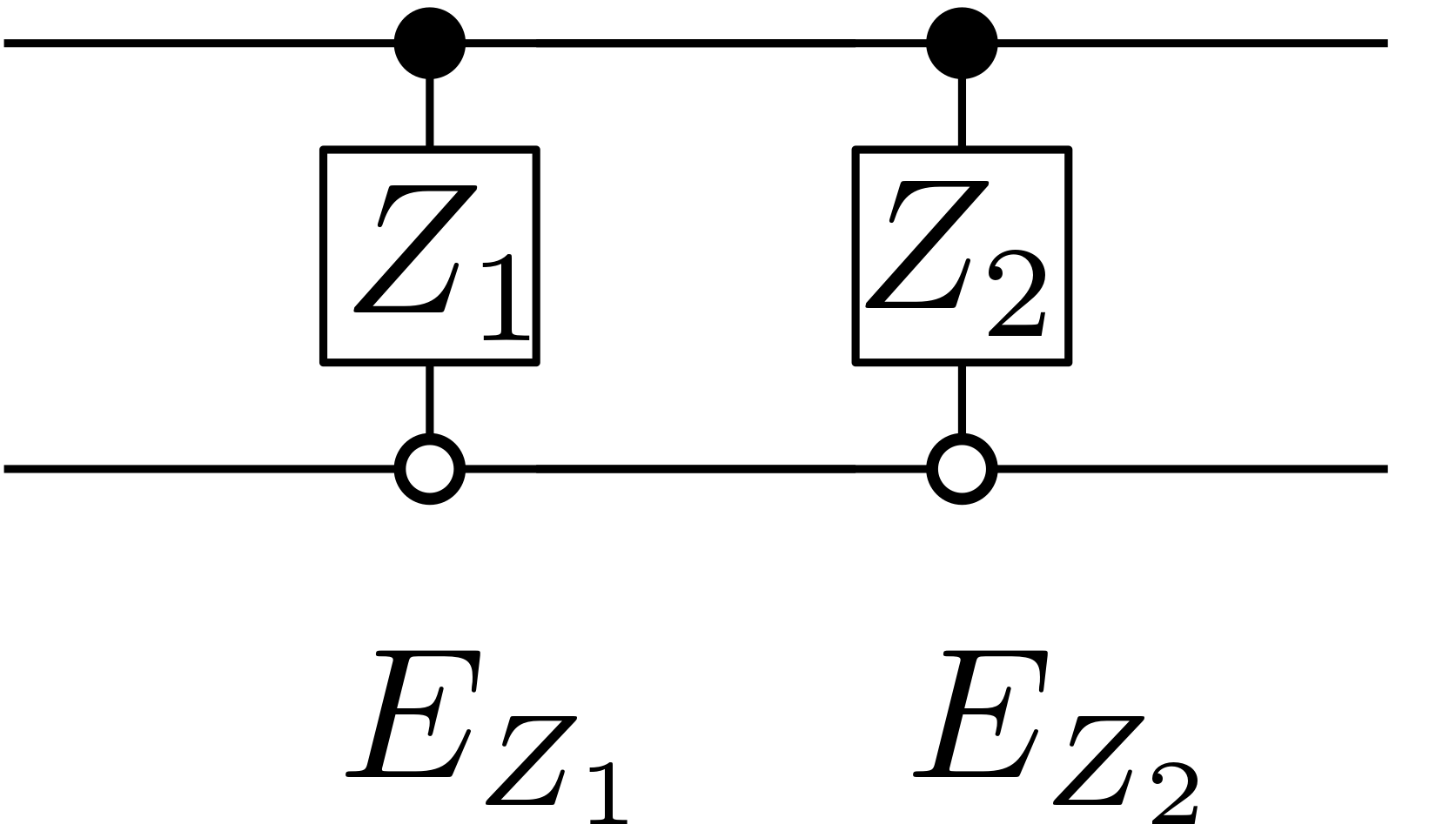}
\caption{The product~$E_{Z_1\otimes Z_2}=E_{Z_1}E_{Z_2}$ of two transfer operators. Left-multiplication by an operator corresponds to attaching the corresponding diagram on the left. 
\mylabel{fig:productoperationtransfer}}
\end{figure}

In the following, we are interested in 
inner products $\bra{\Psi(A,X,n)}\Psi(B,Y,n)\rangle$ of two MPS, defined by local tensors $A$ and $B$, with boundary matrices $X$ and $Y$, which may have different bond dimensions~$D_1$ and $D_2$, respectively. To analyze these, it is convenient to introduce an ``overlap'' transfer operator $E=E(A,B)$ which now depends on both MPS tensors $A$ and $B$.
First we define $E\in\cB(\mathbb{C}^{D_1}\otimes\mathbb{C}^{D_2})$ by
\begin{align}
E&=\sum_{m=1}^\physical \overline{A_m}\otimes B_m\ .
\end{align}
 The definition of $E_Z$ for $Z\in\cB(\mathbb{C}^\physical)$ is analogous to  equation~\eqref{eq:ezdeffirst}, but with appropriate substitutions. We set 
\begin{align}
E_Z&=\sum_{n,m}  \bra{m}Z\ket{n}\overline{A_m}\otimes B_n\ .
\end{align}
Starting from this definition, 
the expression  $E_F\in\cB(\mathbb{C}^{D_1}\otimes\mathbb{C}^{D_2})$ for $F\in\cB((\mathbb{C}^\physical)^{\otimes d})$ is then defined analogously as before.

\subsubsection{Norm bounds on generalized transfer operators\label{sec:normboundsgeneralizedtrsfop}}

A first key observation is that the (operator) norm of powers of any transfer operator scales (at most) as a polynomial in the number of physical spins, with the degree of the polynomial determined by the size of the largest Jordan block.  We need these bounds explicitly and start with the following simple bounds. 

Below, we often consider families of parameters depending on the system size~$n$, i.e., the total number of spins. We write $m\gg h$ as a shorthand for a parameter~$m$ ``being sufficiently large'' compared to another parameter~$h$.  More precisely, this signifies that we assume that~$|h/m|\rightarrow 0$ for $n\rightarrow\infty$, and that by a corresponding choice of a sufficiently large~$n$, the term $|h/m|$ can be made  sufficiently small for a given bound to hold. Oftentimes $h$ will in fact be constant, with $m\rightarrow\infty$ as $n\rightarrow\infty$. 
\begin{lemma} \mylabel{lem:jordanblocksr}
For $m>h$, the Frobenius norm of the $m$-th power $(\lambda I + N)^m$ of a Jordan block~$\lambda I+N\in \cB(\mathbb{C}^h)$  with eigenvalue $\lambda$, such that $|\lambda|\leq 1$, and size $h$ is bounded by
\begin{align}
\left\Vert \left(\lambda I+N\right)^{m}\right\Vert _{F}\le 3h^{3/2}m^{h-1}|\lambda|^{m-(h-1)}\ .\mylabel{eq:upperboundlambdanpowm}
\end{align}
Furthermore,
\begin{align}
\|\left(\lambda I+N\right)^{m}\|&\leq 4 m^{h-1}\qquad\textrm{ for  }m\gg h\ .\label{eq:simplejordanblockpower}
\end{align}
\end{lemma}
\begin{proof}
For $h=0$ the claim is trivial. Assume that $h>1$. Because $N^h=0$ and $N^r$ has exactly $h-r$ non-zero entries for $r<h$, we have
\begin{align}
\left\Vert \left(\lambda I_{h}+N\right)^{m}\right\Vert _{F}&\leq\sum_{r=0}^{h-1}\binom{m}{r}|\lambda|^{m-r}\left\|N^{r}\right\|_{F}\\
&\leq |\lambda|^m \cdot |\lambda|^{-(h-1)}\sum_{r=0}^{h-1}\binom{m}{r}(h-r)^{1/2}\\
&\leq h^{1/2} |\lambda|^m \cdot |\lambda|^{-(h-1)}\sum_{r=0}^{h-1}\binom{m}{r}\ .
\end{align}
Since the right hand side is maximal for $r=h-1$, and the binomial coefficient can be bounded from above by $\binom{m}{r}\leq \left(\frac{e m}{r}\right)^r\leq 3 \cdot m^{r}$
we obtain
\begin{align}
\sum_{r=0}^{h-1}\binom{m}{r}\leq 3h\cdot m^{h-1}\ ,
\end{align}
hence, the first claim follows.

For the second claim, recall that the entries of the $m$-th power of a Jordan block are 
\begin{align}
\left((\lambda I_h+N)^m\right)_{p,q}
&=\begin{cases}
\binom{m}{q-p}\lambda^{m+(p-q)}
 \qquad&\textrm{ if }q\geq p\\
 0 &\textrm{ otherwise}\ 
\end{cases}\label{eq:jordanblockpowermatrixelement}
\end{align} 
for $p,q\in [h]$. This means that if $|\lambda|=1$, the maximum matrix element $|\left((\lambda I_h+N)^m\right)_{p,q}|=\binom{m}{h-1}$ is attained for $(p,q)=(1,h)$. Using the Cauchy-Schwarz inequality, it is straightforward to check that 
\begin{align}
\|\left(\lambda I+N\right)^{m}\|&\leq h\max_{p,q}|\left(\lambda I+N\right)^{m}_{p,q}|=h\cdot \binom{m}{h-1}=\frac{h}{(h-1)!} \frac{m!}{(m-(h-1))!}\ .
\end{align}
Since $\frac{h}{(h-1)!}\leq 2$ for $h\in\mathbb{N}$ and
\begin{align}
\frac{m!}{(m-(h-1))!} &=m^{h-1}(1+O(h/m))\leq 2 m^{h-1}\quad\textrm{ for }m\gg h\ ,
\end{align}
the claim follows. 
\end{proof}

Now, we apply Lemma~\ref{lem:jordanblocksr} to (standard and mixed) transfer operators. It is convenient to state these bounds as follows. The first two statements are about the scaling of the norms of powers of~$E$; the last statement is about the magnitude of matrix elements in powers of~$E$. 
\begin{lemma}\mylabel{lem:normscalingtransferop}
Let $\rho(E)$ denote the spectral radius of a matrix~$E\in \cB(\mathbb{C}^{D_1}\otimes\mathbb{C}^{D_2})$. 
\begin{enumerate}[(i)]
\item\mylabel{it:firstjordan}  Suppose $\rho(E)\leq 1$. Let $h^*$ be the size of the largest Jordan block(s) of~$E$. Then
\begin{align}
\|E^m\|\leq 4 m^{h-1}\qquad\textrm{ for  }m\gg h\ .
\end{align}
\item\mylabel{it:secondjordan}
If $\rho(E)<1$, then
\begin{align}
\|E^m\|_F&\leq \rho(E)^{m/2}\qquad\textrm{ for }m\gg D_1D_2\ .
\end{align}
We will often use $\|E^m\|_F \le 1$ as a coarse bound.
\item\mylabel{it:thirdjordan}
Suppose that $\rho(E)=1$. Let $h^*$ denote the size of the largest Jordan block(s) in~$E$.  For $p,q\in [D_1D_2]$, let $(E^m)_{p,q}$ denote 
the matrix element of $E^m$ with respect to the standard computational basis $\{\ket{p}\}_{p=1}^{D_1D_2}$. Then the following holds: for all $p,q\in [D_1 D_2]$, there is a constant $c_{p,q}$ with $c_{p,q}=O(1)$ as $m\rightarrow\infty$ and some $\ell\in \{1,\ldots,h^*\}$ such that 
\begin{align}
|(E^m)_{p,q}|&=c m^{\ell-1}(1+O(m^{-1}))\ .
\end{align}
\end{enumerate}
\end{lemma}
\begin{proof}
For $\lambda\in\mathsf{spec}(E)$, let us denote by $\lambda I_{h(\lambda)}+N_{h(\lambda)}$ the associated Jordan block, where $h(\lambda)$ is its size. Then
\begin{align}
\|E^m\|&=\max_{\lambda\in\mathsf{spec}(E)}\|(\lambda I_{h(\lambda)}+N_{h(\lambda)})^m\|\leq \max_{\lambda\in\mathsf{spec}(E)}\|(\lambda I_{h(\lambda)}+N_{h(\lambda)})^m\|_F
\end{align}
where we assumed that $m\gg h^*\ge h(\lambda)$, $|\lambda|\leq 1$ and~\eqref{eq:simplejordanblockpower}. This shows claim~\eqref{it:firstjordan}.

Claim~\eqref{it:secondjordan} immediately follows from~\eqref{eq:upperboundlambdanpowm} and the observation that $m^{D_1D_2-1}\rho(E)^m=O(\rho(E)^{m/2})$.

For the proof of statement~\eqref{it:thirdjordan}, observe that matrix elements~\eqref{eq:jordanblockpowermatrixelement} of a Jordan block (matrix) with eigenvalue~$\lambda$ (with $|\lambda|=1$) of size $h$ scale as
\begin{align}
\big|((\lambda I+N)^m)_{p,q}\big|&=\frac{1}{(q-p)!}\frac{m!}{(m-(q-p))!}=\frac{1}{(q-p)!}m^{q-p}\left(1+O(1/m)\right)
\end{align}
for $q>p$ and are constant otherwise. Because $q-p\in\{1,\ldots,h-1\}$ when $q>p$, this is of the form $m^{\ell}(1+O(1/m))$ for some $\ell\in [h-1]$.  Since~$E^m$ is similar (as a matrix) to a direct sum of such powers of Jordan blocks, and the form of this scaling does not change under linear combination of matrix coefficients, the claim follows. 
\end{proof}

Now let us consider the case where  $E=\kket{\ell}\bbra{r}\oplus\tilde{E}$ is the transfer operator of an injective MPS, normalized with maximum eigenvalue~$1$. Let $\lambda _2 < 1$ denote the second largest eigenvalue. Without loss of generality, we can take the MPS to be in canonical form, so that $E$ has a unique right fixed-point given by the identity matrix $I$ and a unique left fixed-point given by some positive-definite diagonal matrix $\Lambda$ with unit trace. We can then write the Jordan decomposition of the transfer matrix as
\begin{equation}
    E = \kket{I}\bbra{\Lambda} \oplus \tilde{E}\ , \label{eq:transferoperatorsumexpr}
\end{equation}
where $\kket{I}$ and $\kket{\Lambda}$ denotes the vectorization of $I$ and $\Lambda$ respectively, and where $\tilde{E}$ denotes the remaining Jordan blocks of $E$. Note that powers of $E$ can then be expressed as
\begin{align}
    E^m = \kket{I}\bbra{\Lambda} \oplus \tilde{E}^m\ . \label{eq:transfer_power}
\end{align}
We can bound the Frobenius norm of the transfer matrix as
\begin{equation}
\|E^m\|^2_F = \|\kket{I}\bbra{\Lambda} \oplus \tilde{E}^m\|_F^2 = \tr(I)\tr(\Lambda^2) + \|\tilde {E}^m\|_F^2 \le D + \|\tilde {E}^m\|_F^2,
\end{equation}
where $\tr(I) = D$ and $\tr(\Lambda^2) \le \tr(\Lambda)^2 = 1$. In particular, since $\rho(\tilde{E})=\lambda_2$, we obtain from Lemma~\ref{lem:normscalingtransferop}\eqref{it:secondjordan} that
\begin{align}
\|\tilde{E}^m\|_F\leq \lambda_2^{m/2}\quad\textrm{ for }m\gg D\  .\mylabel{eq:tildeEpowersnormbnd}
\end{align}
This implies  the following statement:
\begin{lemma}\mylabel{lem:normboundinjectivemps}
The transfer operator~$E$ of an injective MPS satisfies
\begin{align}
\|E^m\|_F\leq \sqrt{D+1}\qquad \text{for } m\gg D.\label{eq:bnddimembnd}
\end{align}
\end{lemma}

We also need a bound on the norm $\|E_F^{\dagger}(\psi_1\otimes\psi_2)\|$, where~$E$ is a mixed transfer operator, $F\in \cB((\mathbb{C}^\physical)^{\otimes d})$ is an operator acting on $d$~sites, and where $\psi_j\in\mathbb{C}^{D_j}$ for $j=1,2$. 

\begin{lemma} \mylabel{lem:efboundsimplified}
Let  $E_1=E(A)$ and $E_2=E(B)$ be the transfer operators associated with the tensors~$A$ and~$B$, respectively, with bond dimensions $D_1$ and $D_2$. Let $E=E(A,B)\in \cB(\mathbb{C}^{D_1}\otimes\mathbb{C}^{D_2})$ be the combined transfer operator. 
Let $\psi_1\in \mathbb{C}^{D_1}$ and $\psi_2\in\mathbb{C}^{D_2}$ be unit vectors. Then 
\begin{align}
 \|(E_F)^{ \dagger} (\psi_1\otimes\psi_2)\| &\leq \|F\| \sqrt{\|E_1^d\|\cdot\|E_2^d\|}\ , \label{eq:claimfirstefbound}\\
 \|E_F (\psi_1\otimes\psi_2)\| &\leq \|F\| \sqrt{\|E_1^d\|\cdot\|E_2^d\|}\ , \label{eq:claimtwoefbound}\\
 \|E_F\|_F &\leq D_1D_2 \|F\| \sqrt{\|E_1^d\|\cdot \|E_2^d\|}\ , \label{eq:frobeniusnormefbound}
  \end{align}
 for all $F\in\cB((\mathbb{C}^{\physical})^{\otimes d})$. 
  \end{lemma}
 \begin{proof}
   Writing matrix elements in the computational basis as 
  $$F_{j_1\cdots j_d,i_1\cdots i_d}:=\bra{j_1\cdots j_d}F\ket{i_1\cdots i_d},$$ 
  we have
  \begin{align}
  E_F &=\sum_{(i_1,\ldots,i_d), (j_1,\ldots,j_d)}
  F_{j_1\cdots j_d,i_1\cdots i_d} (\overline{A}_{j_1}\otimes B_{j_1}
  )(\overline{A}_{j_2}\otimes B_{j_2})\cdots (\overline{A}_{j_d}\otimes B_{j_d}) .
  \end{align}
 Therefore,
 \begin{align}
 (E_F)^{\dagger} &=\sum_{(i_1,\ldots,i_d), (j_1,\ldots,j_d)}
  \overline{F_{j_1\cdots j_d,i_1\cdots i_d}}(\overline{A^{\dagger}_{j_d}}\otimes B^{\dagger}_{j_d})\cdots (\overline{A^{\dagger}_{j_2}}\otimes B^{\dagger}_{j_2}) (\overline{A^{\dagger}}_{j_1}\otimes B^{\dagger}_{j_1}
  )\\
  &=\sum_{(i_1,\ldots,i_d), (j_1,\ldots,j_d)}
 (\pi \overline{F}\pi^{\dagger})_{j_d\cdots j_1,i_d\cdots i_1}(\overline{A^{\dagger}_{j_d}}\otimes B^{\dagger}_{j_d})\cdots (\overline{A_{j_2}^{\dagger}}\otimes B^{\dagger}_{j_2}) (\overline{A_{j_1}^{\dagger}}\otimes B^{\dagger}_{j_1})\ ,
 \end{align}
where $\pi$ is the permutation which maps the $j$-th factor in the tensor product $(\mathbb{C}^\physical)^{\otimes n}$ to the $(n-j+1)$-th factor,
and where $\overline{F}$ is obtained by complex conjugating the 
matrix elements in the computational basis. This means that 
\begin{align}
(E_F)^{\dagger}=E^{\dagger}_{\pi \overline{F}\pi^{\dagger}}\ ,\label{eq:relationefstarpi}
\end{align}
with $E^{\dagger}$ being the mixed transfer operator $E^{\dagger}=E(A^{\dagger},B^{\dagger})$ 
obtained by replacing each $A_j$ respectively $B_j$ with its adjoint. 

Now consider 
\begin{align}
\|(E_F)^{\dagger} (\psi_1\otimes\psi_2)\|^2 &= (\bra{\psi_1}\otimes\bra{\psi_2})
E_F (E_F)^{\dagger}(\ket{\psi_1}\otimes\ket{\psi_2})\\
&= (\bra{\psi_1}\otimes\bra{\psi_2})
E_F E^{\dagger}_{\pi \overline{F}\pi^{\dagger}}(\ket{\psi_1}\otimes\ket{\psi_2})\ .
\end{align}
This can be represented diagrammatically as 
\begin{align}
\|(E_F)^{\dagger} (\psi_1\otimes\psi_2)\|^2 &=\centering   
\raisebox{-.38\height} {\includegraphics[scale=0.08]{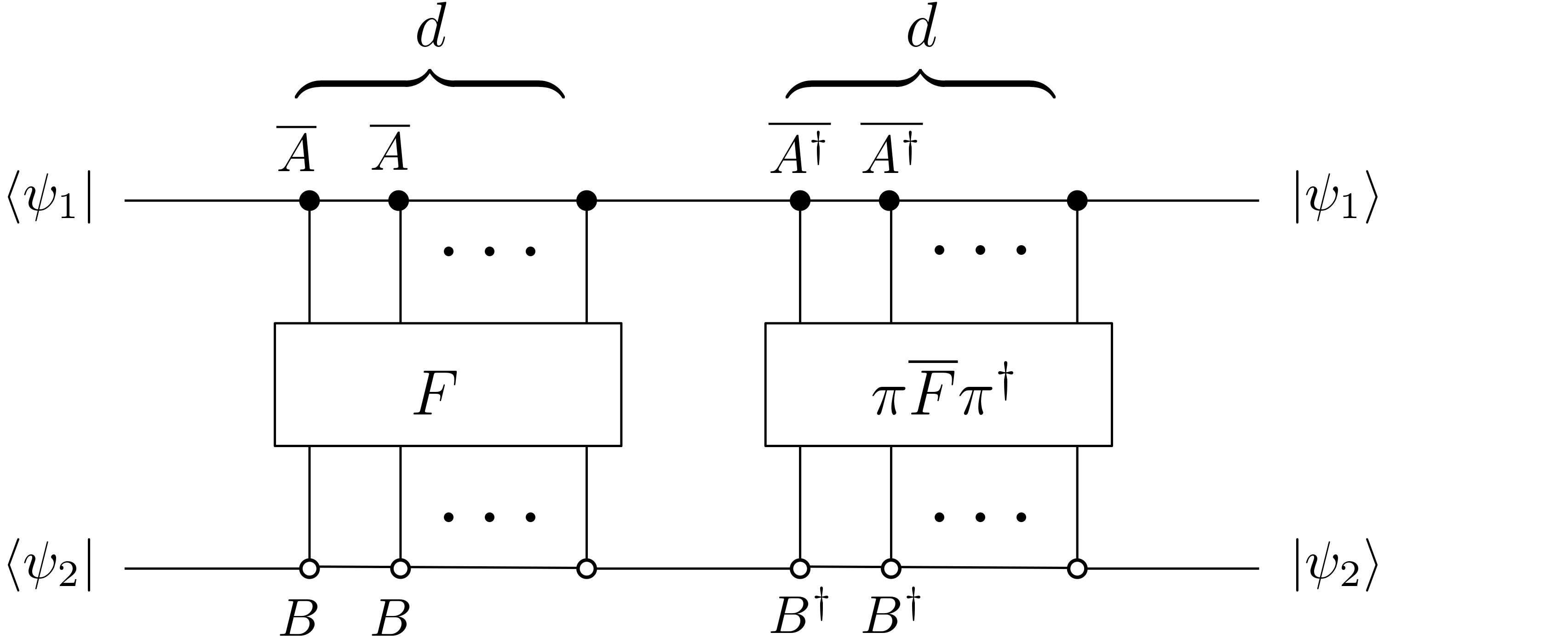}}
\end{align}

In particular, we have 
\begin{align}
\|(E_F)^{\dagger} (\psi_1\otimes\psi_2)\|^2&=\ \bra{\chi}(F\otimes I^{\otimes d})(I^{\otimes d}\otimes \pi \overline{F}\pi^{\dagger})\ket{\varphi}\ ,\mylabel{eq:efstarnormexpression}
\end{align}
where $\varphi,\chi\in (\mathbb{C}^{\physical})^{\otimes 2d}$ are defined as
\begin{align}
\ket{\phi}&=\ \raisebox{-0.42\height}{\includegraphics[scale=0.12]{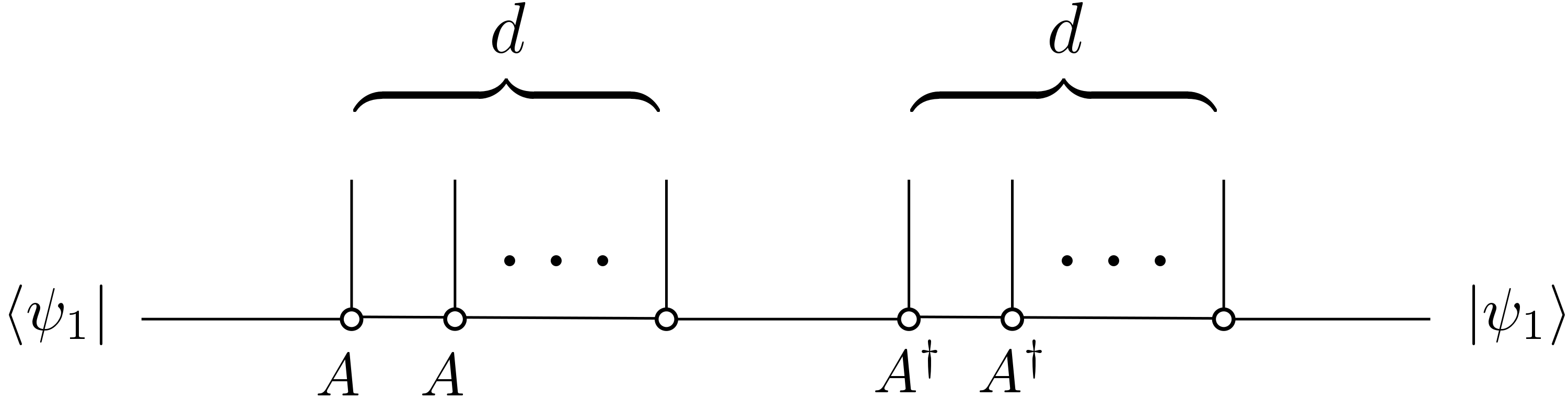}}\ ,\\
\ket{\chi} &=\ \raisebox{-0.42\height}{\includegraphics[scale=0.12]{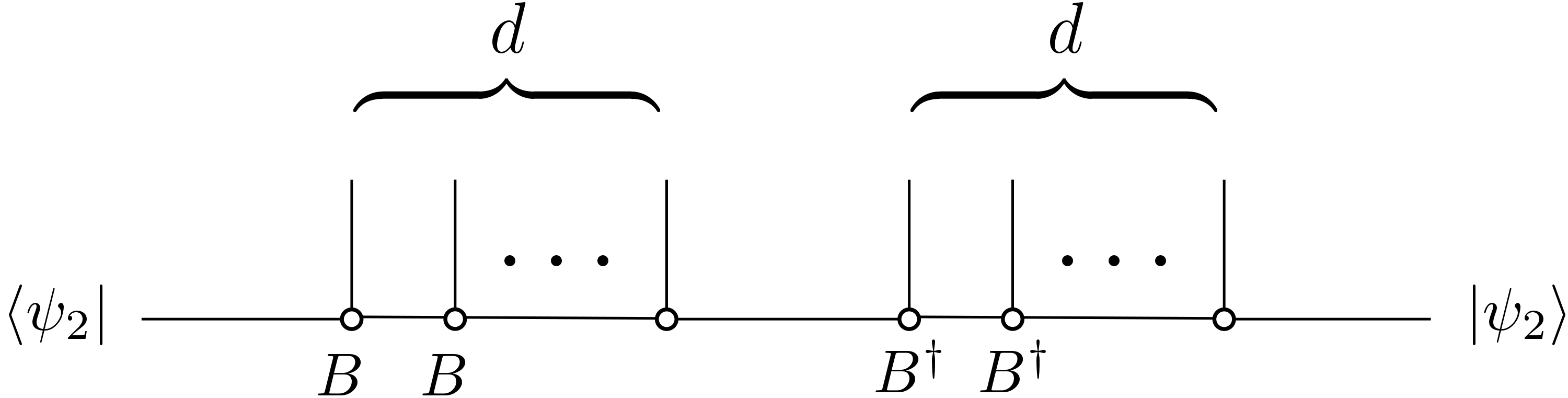}}\ . 
\end{align}
It is straightforward to check that
\begin{align}
\|\chi\|^2 &= (\bra{\psi_1}\otimes\bra{\psi_1})E_1^d (E_1^{\dagger})^d (\ket{\psi_1}\otimes\ket{\psi_1})\ ,\\
\|\varphi\|^2 &= (\bra{\psi_2}\otimes\bra{\psi_2})E_2^d (E_2^{\dagger})^d (\ket{\psi_2}\otimes\ket{\psi_2})\ .
\end{align} 
Since $\|(E_j^{\dagger})^d\|=\|E^d_j\|$ for $j=1,2$,  it follows with the submultiplicativity of the operator norm that  
\begin{align}
\|\chi\|^2&\leq \|E_1^d\|^{2}\ ,\\
\|\varphi\|^2 &\leq \|E_2^d\|^{2}\ .
\mylabel{eq:normboundvarphichix}
\end{align}
Applying the Cauchy-Schwarz inequality to~\eqref{eq:efstarnormexpression} yields
\begin{align}
\|(E_F)^{\dagger} (\psi_1\otimes\psi_2)\|^2&\leq \|(F^{\dagger}\otimes I^{\otimes d})\chi\|\cdot \|(I^{\otimes d}\otimes \pi \overline{F}\pi^{\dagger})\varphi\|\\
&\leq \|F\|^2\cdot \|\chi\|\cdot\|\varphi\|\ ,\mylabel{eq:cauchyschwarzappliedhmv}
\end{align}
where we used the fact that the operator norm satisfies~$\|F^{\dagger}\|=\|\overline{F}\|=\|F\|$ and $\|I\otimes A\|=\|A\|$. The claim~\eqref{eq:claimfirstefbound} follows from this and~\eqref{eq:normboundvarphichix}.
  
The claim~\eqref{eq:claimtwoefbound} follows analogously by using equation~\eqref{eq:relationefstarpi}. Finally, the claim~\eqref{eq:frobeniusnormefbound} follows from~\eqref{eq:claimtwoefbound} and 
\begin{align}
\|E_F\|^2 &=\sum_{\alpha_1,\alpha_2\in [D_1]}\sum_{\beta_1,\beta_2\in [D_2]}
|(\bra{\alpha_1}\otimes\bra{\beta_1})E_F(\ket{\alpha_2}\otimes\ket{\beta_2})|^2\\
&\leq \sum_{\alpha_1,\alpha_2\in [D_1]}
\sum_{\beta_1,\beta_2\in [D_2]} \|E_F(\ket{\alpha_2}\otimes\ket{\beta_2})\|^2\\
&\leq D_1^2D_2^2\max_{\alpha,\beta}\|E_F(\ket{\alpha_2}\otimes\ket{\beta_2})\|^2\ ,
\end{align}
where we employed the orthonormal basis $\{\ket{\alpha}\}_{\alpha\in [D_1]}$ and $\{\ket{\beta}\}_{\beta\in [D_2]}$ for $\mathbb{C}^{D_1}$ and $\mathbb{C}^{D_2}$, respectively, and applied the Cauchy-Schwarz inequality.
\end{proof}

The main result of this section is the following upper bound on the matrix elements of geometrically $d$-local operators with respect to two MPS.
\begin{theorem}\mylabel{thm:basicboundmatrixelements}
Let $\ket{\Psi_1}=\ket{\Psi(A_1,X_1,n)},\ket{\Psi_2}=\ket{\Psi(A_2,X_2,n)}\in (\mathbb{C}^\physical)^{\otimes n}$ be two MPS with bond dimensions $D_1$ and $D_2$, where
\begin{align}
X_j=\ket{\varphi_j}\bra{\psi_j},\qquad\textrm{ with }\qquad \|\varphi_j\|=\|\psi_j\|=1\qquad \textrm{ for }j=1,2\ , 
\end{align}
are rank-one operators. Let $E=E(A_1,A_2)\in \cB(\mathbb{C}^{D_1}\otimes\mathbb{C}^{D_2})$ denote the combined transfer operator defined by the MPS tensors $A_1$ and $A_2$, $h_j^*$ the size of the largest Jordan block of $E_j=E(A_j)$ for $j=1,2$, and $h^*$ the size of the largest Jordan block of $E=E(A_1,A_2)$. 
Assume that the spectral radii $\rho(E)$, $\rho(E_1)$, and $\rho(E_2)$ are contained in~$[0,1]$.
Then, for any  $F\in\cB((\mathbb{C}^\physical)^{\otimes d})$, we have 
\begin{align}
|\bra{\Psi_1}(F\otimes I_{(\mathbb{C}^\physical)^{\otimes n-d}})\ket{\Psi_2}|\leq  16 \cdot \|F\|\cdot d^{(h_1^*+h_2^*-2)/2}(n-d)^{h^*-1}\ 
\end{align}
for $d\gg D_1, D_2$ and $(n-d)\gg D_1D_2$. 
\end{theorem}
\begin{proof}
The matrix elements $\alpha=\bra{\Psi_1}(F\otimes I_{(\mathbb{C}^\physical)^{\otimes n-d}})\ket{\Psi_2}$ of interest can be written as 
\begin{align}
\alpha &= \left(\bra{\psi_1}\otimes\bra{\psi_2}\right)
E_F E^{n-d} \left(\ket{\varphi_1}\otimes\ket{\varphi_2}\right)\ .
\end{align}
By the Cauchy-Schwarz inequality, we have 
\begin{align}
|\alpha|&\leq \|E_F^{\dagger}(\ket{\psi_1}\otimes\ket{\psi_2})\|\cdot
\|E^{n-d} (\ket{\varphi_1}\otimes\ket{\varphi_2}\|\\
&\leq \|F\| \sqrt{\|E_1^d\|\cdot \|E_2^d\|}\cdot \|E^{n-d}\|\ ,
\end{align}
by  the definition of the operator norm and Lemma~\ref{lem:efboundsimplified}. Then, the claim follows from Lemma~\ref{lem:normscalingtransferop}~\eqref{it:firstjordan}, which provides the bounds
\begin{align}
\|E_j^d\|&\leq 4d^{h_j^*-1}\quad\textrm{ for }j=1,2\ ,\\
\|E^{n-d}\|&\leq  4(n-d)^{h^*-1}
\end{align}
by our assumptions$\colon$ $\rho(E_j)\in [0,1]$, $\rho(E)\in [0,1]$, and $d\gg D_j\geq h_j^*$ for $j=1,2$, as well as $n-d\gg D_1D_2\geq h^*$.  
\end{proof}

\section{No-Go Theorem: Degenerate ground spaces of gapped Hamiltonians are constant-distance AQEDC \label{sec:nogo}}

In this section we prove a no-go result regarding the error-detection performance of the ground spaces of local gapped Hamiltonians: their distance can be no more than constant.  We prove this result by employing the necessary condition for approximate error-detection from Lemma~\ref{lem:distancenecessary} for the code subspaces generated by varying the boundary conditions of an (open-boundary) injective MPS. Note that, given a translation invariant MPS with periodic boundary conditions and bond dimension~$D$, there exists a local gapped Hamiltonian, called the parent Hamiltonian, with a unique ground state being the MPS~\cite{perez2006matrix}.

We need the following bounds which follow from the orthogonality and normalization of states in such codes.
\begin{lemma}\mylabel{lem:orthogonailityxy}
Let $A$ be the MPS  tensor of an injective MPS with bond dimension~$D$, and let $X,Y\in \cB(\mathbb{C}^D)$ be such that the states $\ket{\Psi_X}=\ket{\Psi(A,X,n)}$ and $\ket{\Psi_Y}=\ket{\Psi(A,Y,n)}$  are normalized and  orthogonal. Let us write the transfer operator as $E=\kket{I}\bbra{\Lambda}\oplus\tilde{E}$ (cf.~equation~\eqref{eq:transferoperatorsumexpr}). Assume $n\gg D$. Then
\begin{enumerate}[(i)]
\item\label{it:firstmpsinj}
The Frobenius norm of $X$ (and similarly the norm of $Y$) is bounded by 
\begin{align}
\|X\|_F &=O(1)\ .
\end{align}
\item\mylabel{it:secondmpsinj}
We have 
\begin{align}
|\bbra{\Lambda}(\overline{X}\otimes Y)\kket{I}|&= O(\lambda_2^{n/2})\ ,\\
|\bbra{\Lambda}(\overline{Y}\otimes X)\kket{I}|&= O(\lambda_2^{n/2})\ .
\end{align}
\end{enumerate}
\end{lemma}

In the following proofs, we repeatedly use
the inequality
\begin{align}
|\tr(M_1\ldots M_k)|&\leq \|M_1\|_F\cdot \|M_2\|_F\cdots \|M_k\|_F\ \mylabel{eq:matrixboundmultiple}
\end{align}
for $D\times D$-matrices $\{M_j\}_{j=1}^k$. 
Note that the inequality~\eqref{eq:matrixboundmultiple} is simply the Cauchy-Schwarz inequality for $k=2$.
For $k>2$, the inequality follows from the inequality for $k=2$ and the  submultiplicativity of the Frobenius-norm because
\begin{align}
|\tr(M_1\ldots M_k)|&\leq \|M_1\|_F\cdot \|M_2 \cdots M_k\|_F\leq  \|M_1\|_F\cdot \|M_2\|_F \cdots \|M_k\|_F\ .
\end{align}

\begin{proof}
The proof of~\eqref{it:firstmpsinj} follows from the fact that the state~$\Psi_X$ is normalized, i.e.,
\begin{align}
1&=\|\Psi_X\|^2\\
&=\tr\left( E^n (\overline{X}\otimes X)\right)\\
&=\tr\left(\kket{I}\bbra{\Lambda} (\overline{X}\otimes X)\right)+
\tr(\tilde{E}^n (\overline{X}\otimes X))\\
&=\tr(\Lambda X X^\dagger)+
\tr(\tilde{E}^n (\overline{X}\otimes X))\\
&\geq \lambda_{\min}(\Lambda)\cdot \|X\|_F^2+\tr(\tilde{E}^n (\overline{X}\otimes X))\ ,
\end{align}
where $\lambda_{\min}(\Lambda)$ denotes the smallest eigenvalue of~$\Lambda$, and we make use of the fact that $XX^{\dagger}$ is positive with trace $\tr(XX^{\dagger})=\|X\|_F^2$. Since  
\begin{align}
|\tr(\tilde{E}^n (\overline{X}\otimes X))|&\leq \|\tilde{E}^n\|_F\cdot \|\overline{X}\otimes X\|_F\leq \lambda_2^{n/2} \|X\|_F^2\quad\textrm{ for } n\gg D
\end{align}
by~\eqref{eq:matrixboundmultiple} and~\eqref{eq:tildeEpowersnormbnd}, we conclude 
\begin{align}
\|X\|_F^2 \leq \left(\lambda_{\min}(\Lambda)-\lambda_2^{n/2}\right)^{-1}=\lambda_{\min}(\Lambda)^{-1}(1+O(\lambda_2^{n/2}))\ .
\end{align}
Then the claim~\eqref{it:firstmpsinj} follows since $\lambda_{\min}(\Lambda)^{-1}$ is a constant. 

Now, consider the first inequality in~\eqref{it:secondmpsinj}
(the bound for~$|\bbra{\Lambda} (\overline{Y}\otimes X)\kket{I}|$ is shown analogously). Using the orthogonality of the states $\ket{\Psi_X}$ and $\ket{\Psi_Y}$, we obtain
\begin{align}
0&=\langle \Psi_X|\Psi_Y\rangle=\tr\left( E^n (\overline{X}\otimes Y)\right)\\
&=\tr\left((\kket{I}\bbra{\Lambda}+\tilde{E}^n)(\overline{X}\otimes Y)\right)\\
&=\bbra{\Lambda} (\overline{X}\otimes Y)\kket{I}+\tr(\tilde{E}^n (\overline{X}\otimes Y))
\end{align}
hence
\begin{align}
|\bbra{\Lambda} (\overline{X}\otimes Y)\kket{I}|=|\tr(\tilde{E}^n (\overline{X}\otimes Y))|&\le \|\tilde{E}\|_F\cdot\|\overline{X}\otimes Y\|_F\\
&\le \lambda_2^{n/2} \|X\|_F\cdot\|Y\|_F\ ,
\end{align}
using~\eqref{eq:tildeEpowersnormbnd}. The claim~\eqref{it:secondmpsinj} then follows from~\eqref{it:firstmpsinj}.
\end{proof}

With the following lemma, we prove an upper bound on the overlap of the reduced density matrices $\rho_X$ and $\rho_Y$, supported on $2\Delta$-sites surrounding the boundary, of the global states $\ket{\Psi_X}$ and $\ket{\Psi_Y}$, respectively. 

\begin{lemma}\mylabel{lem:overlaplowinjectivemps}
Let $A$ be an MPS tensor of an injective MPS with bond dimension~$D$, and let $X,Y\in \cB((\mathbb{C}^\physical)^{\otimes n})$ be such that the states $\ket{\Psi_X}=\ket{\Psi(A,X,n)}$ and $\ket{\Psi_Y}=\ket{\Psi(A,Y,n)}$ are normalized and  orthogonal. Let $\Delta\gg D$. Let $\cS=\{1,2,\ldots,\Delta\}\cup \{n-\Delta+1,n-\Delta+2,\ldots,n\}$ be the subset of~$2\Delta$ spins consisting of $\Delta$ systems at the left and and $\Delta$ systems at the right boundary. Let $\rho_X=\tr_{[n]\backslash\cS }\ket{\Psi_X} \bra{\Psi_X}$ and $\rho_Y=\tr_{[n]\backslash\cS }\ket{\Psi_Y} \bra{\Psi_Y}$ be the reduced density operators on these subsystems. Then 
\begin{align}
\tr(\rho_X\rho_Y)\leq c \lambda_2^{\textfrac{\Delta}{2}}
\end{align}
where~$\lambda_2$ is the second largest eigenvalue of the  transfer operator~$E=E(A)$ and where~$c$ is  a constant depending  only on the minimal eigenvalue of~$E$ and the bond dimension~$D$. 
\end{lemma}

\begin{proof}
\begin{figure}
    \centering
    \includegraphics[scale=0.25]{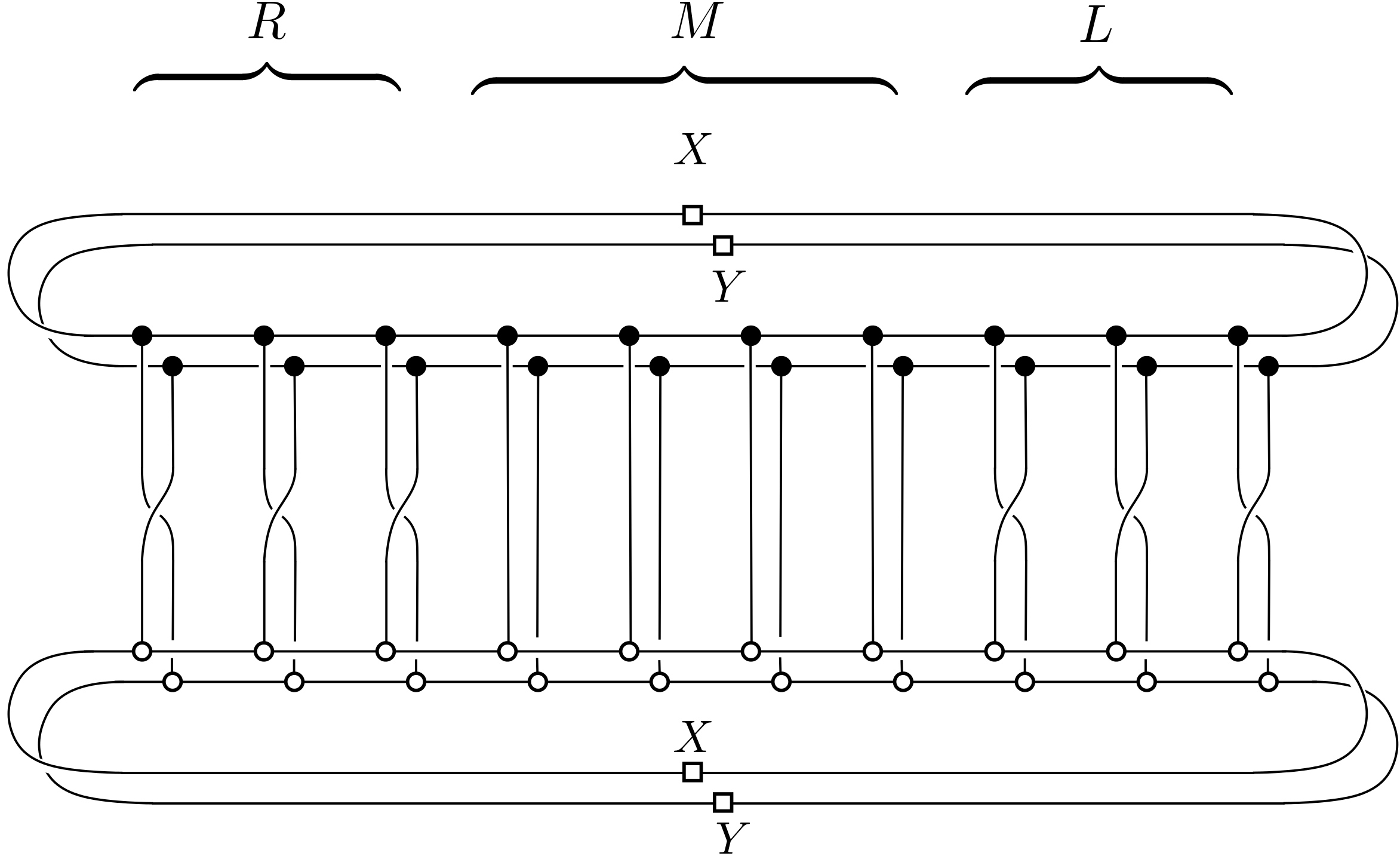}\qquad \includegraphics[scale=0.25]{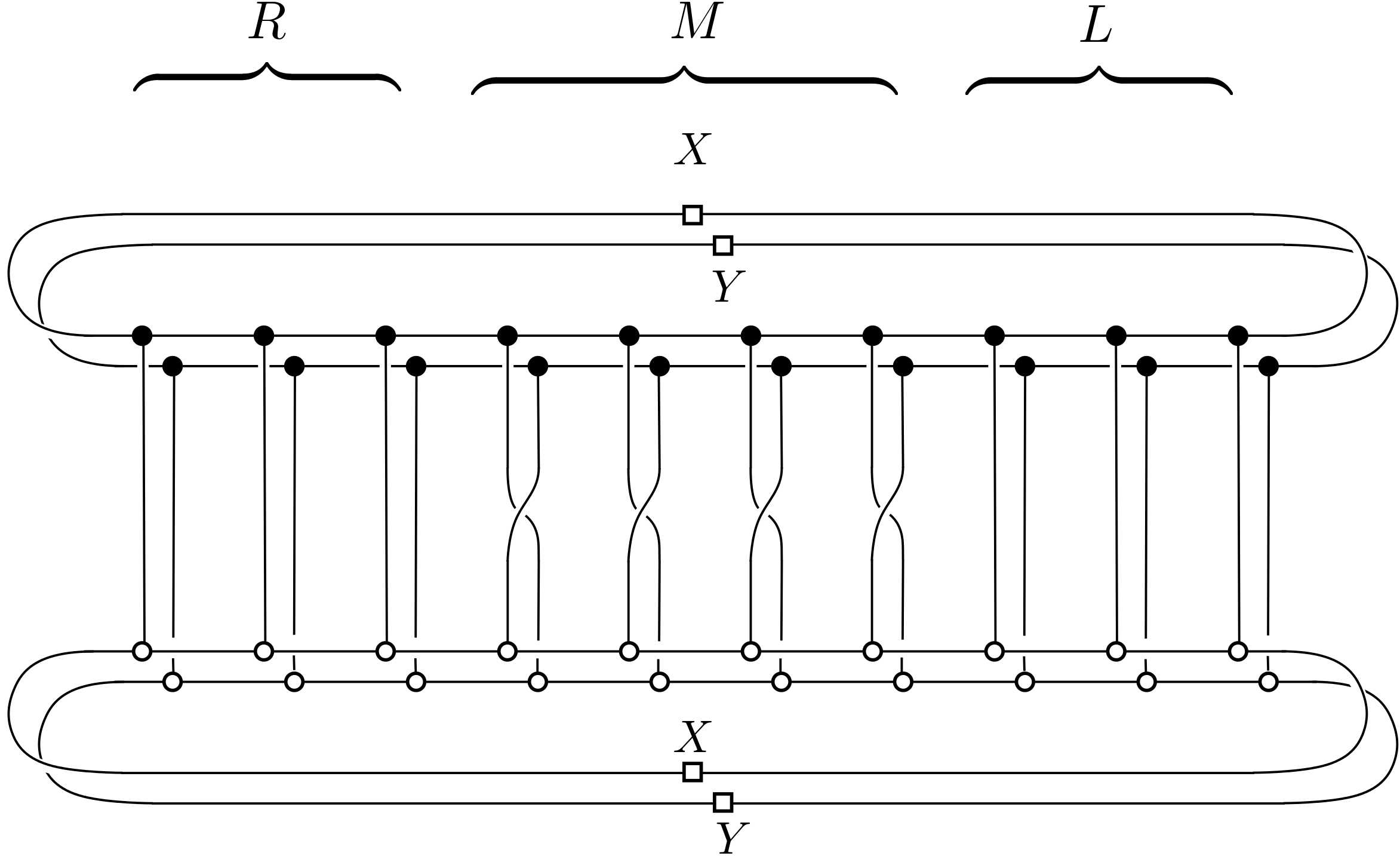}
    \caption{The two expressions in equation~\eqref{eq:swapfmmprime}, where 
    $L$, $M$ and $R$ are used to denote the sites defined in~\eqref{eq:siteslmr}.\mylabel{fig:matrixelementswapops}}
\end{figure}
For convenience, let us relabel the systems as
\begin{align}
\begin{matrix}
(L_1,\ldots,L_\Delta)&=&(1,2,\ldots,\Delta)\\
(M_1,\ldots,M_{n-2\Delta})&=&(\Delta+1,\Delta+2,\ldots,n-\Delta)\\
(R_1,\ldots,R_\Delta)&=&(n-\Delta+1,n-\Delta+2,\ldots,n)\ 
\end{matrix}\label{eq:siteslmr}
\end{align}
indicating their location on the left, in the middle, and on the right, respectively. 
For the tensor product~$\cH_A\otimes\cH_B$ of two isomorphic Hilbert spaces, we denote by $\F_{AB}\in \cB(\cH_A\otimes\cH_B)$ the flip-operator which swaps the two systems. The following expressions are visualized in Figure~\ref{fig:matrixelementswapops}. We have
\begin{align}
\tr(\rho_X\rho_Y)&=\tr((\rho^{L_1\cdots L_\Delta R_1\cdots R_\Delta}_X\otimes\rho^{L_1'\cdots L'_\Delta R'_1\cdots R'_\Delta}_Y)
(\F_{LL'}\otimes \F_{RR'}))\ ,\qquad\textrm{ where }\\
\F_{LL'}&=\F_{L_1L_1'}\otimes\F_{L_2L_2'}\otimes \cdots\otimes \F_{L_\Delta L'_\Delta}\ ,\\
\F_{RR'}&=\F_{R_1R_1'}\otimes\F_{R_2R_2'}\otimes \cdots \otimes \F_{R_\Delta R'_{\Delta}}\ .
\end{align}
Defining $\F_{MM'}$ analogously, $I_{MM'}=\identityoperator_{M_1M_1'}\otimes\cdots\otimes\identityoperator_{M_{n-2\Delta}M'_{n-2\Delta}}$, and similarly $\identityoperator_{LL'}$ and $\identityoperator_{RR'}$, this can be rewritten (by the definition of the partial trace) as 
\begin{align}
\tr(\rho_X\rho_Y)&=(\bra{\Psi_X^{LMR}}\otimes\bra{\Psi_Y^{L'M'R'}})
(\F_{LL'}\otimes \identityoperator_{MM'}\otimes \F_{RR'})
(\ket{\Psi_X^{LMR}}\otimes\ket{\Psi_Y^{L'M'R'}})\nonumber\\
&=(\bra{\Psi_X^{LMR}}\otimes\bra{\Psi_Y^{L'M'R'}})
(\identityoperator_{LL'}\otimes \F_{MM'}\otimes \identityoperator_{RR'})
(\ket{\Psi_Y^{LMR}}\otimes\ket{\Psi_X^{L'M'R'}})\ .\label{eq:swapfmmprime}
\end{align}
In the last identity, we have used that $\F^2=\identityoperator$ is the identity.

Reordering and regrouping the systems as 
\begin{align}
(L_1L_1')(L_2L_2')\cdots (L_\Delta L_\Delta')(M_1M_1')(M_2M_2')\cdots (M_{n-2\Delta}M'_{n-2\Delta})(R_1R_1')(R_2R_2')\cdots (R_\Delta R_\Delta')\ ,
\end{align}
we observe that 
$\ket{\Psi_X^{LMR}}\otimes \ket{\Psi_Y^{L'M'R'}}$
is an MPS with MPS tensor $A\otimes A$ and boundary tensor $X\otimes Y$ and 
$\ket{\Psi_Y^{LMR}}\otimes \ket{\Psi_X^{L'M'R'}}$
is an MPS with MPS tensor $A\otimes A$ and boundary tensor $Y\otimes X$.  Let us denote the virtual systems of the first MPS by $V_1V_2$, and those of the second MPS by $W_1W_2$, such that the boundary tensors are  $X^{V_1}\otimes Y^{V_2}$ and $Y^{W_1}\otimes X^{W_2}$ respectively. Let $\hat{E}=E^{V_1W_1}\otimes E^{V_2W_2}$ be the associated transfer operator. 
Then we have from~\eqref{eq:swapfmmprime}
\begin{align}
\tr(\rho_X\rho_Y)&= \tr\left(\hat{E}^\Delta \hat{E}_{\F^{\otimes n-2\Delta}}\hat{E}^\Delta \left[(\overline{X}^{V_1}\otimes \overline{Y}^{V_2})\otimes (Y^{W_1}\otimes X^{W_2})\right]\right)\ . \mylabel{eq:sumexpressionrhoxrhoy}
\end{align}
Recall that $E^\Delta =\kket{I}\bbra{\Lambda}\oplus\tilde{E}^\Delta$, where we have
\begin{align}
\|\tilde{E}^\Delta\|_F&\leq  \sqrt{D^2}\cdot\|\tilde{E}^\Delta\|\leq D\cdot \lambda_2^{\Delta/2}\quad\textrm{ for }\Delta\gg D\ ,\\
\|\kket{I}\bbra{\Lambda}\|_F&=\|\kket{I}\|^2\cdot\|\kket{\Lambda}\|^2\leq D^2\ .
\end{align}
In the second line, we use the fact that
$\|\kket{\Lambda}\|^2=\tr(\Lambda^{\dagger} \Lambda) =\sum_{i} {\lambda_i}^2 \le 1 $ and $\|\kket{I}\|^2 =D^2$.  Therefore, we have
\begin{align}
E^\Delta &=\sum_{b\in \{0,1\}} H_b\ ,\end{align}
where $H_0=\kket{I}\bbra{\Lambda}$ and $H_1=\tilde{E}^\Delta$ satisfy
\begin{align}
\|H_0\|_F\leq D^2\ ,\qquad\textrm{ and }\qquad \|H_1\|_F\leq D\cdot \lambda_2^{\Delta/2}\qquad\textrm{ for }\Delta\gg D\ . \label{eq:hzeroneupperbound}
\end{align}
Note that
\begin{align}
\hat{E}^\Delta&=E^\Delta\otimes E^\Delta=\sum_{b_1,b_2\in \{0,1\}} H_{b_1}\otimes H_{b_2}\ .
\end{align} 
Inserting this into~\eqref{eq:sumexpressionrhoxrhoy} gives a sum of 16~terms
\begin{align}
\tr(\rho_X\rho_Y)&\le\sum_{b_1,b_2,b_3,b_4\in \{0,1\}}| \alpha_{b_1,b_2,b_3,b_4}|\ ,
\end{align}
where
\begin{align}
    \alpha_{b_1,b_2,b_3,b_4}&=\tr\left( (H_{b_1}^{V_1W_1}\otimes H_{b_2}^{V_2W_2}) \hat{E}_{\F^{\otimes n-2\Delta}} (H^{V_1W_1}_{b_3}\otimes H^{V_2W_2}_{b_4})\left[(\overline{X}^{V_1}\otimes \overline{Y}^{V_2})\otimes (Y^{W_1}\otimes X^{W_2})\right]\right)\ .
\end{align}
Consider the term with $b_j=0$ for all $j\in \{1,\ldots,4\}$. This is given by
\begin{align}
\alpha_{0,0,0,0}&=\tr\bigg((\kket{I}\bbra{\Lambda}^{V_1W_1} \otimes \kket{I}\bbra{\Lambda}^{V_2W_2} )\hat{E}_{\F^{\otimes n-2\Delta}} (\kket{I}\bbra{\Lambda}^{V_1W_1}\otimes \kket{I}\bbra{\Lambda}^{V_2W_2})\\
&\qquad\qquad\cdot\left[(\overline{X}^{V_1}\otimes \overline{Y}^{V_2})\otimes (Y^{W_1}\otimes X^{W_2})\right]\bigg)\\
&=\bbra{\Lambda}(\overline{X}\otimes Y)\kket{I}\cdot \bbra{\Lambda} (\bar{Y}\otimes X)\kket{I}\cdot 
(\bbra{\Lambda}\otimes \bbra{\Lambda}) \hat{E}_{\F^{\otimes n-2\Delta}}
(\kket{I}\otimes \kket{I})\ .
\mylabel{eq:alpha0000}
\end{align}
By inserting this into~\eqref{eq:alpha0000} we get
with Lemma~\ref{lem:orthogonailityxy}~\eqref{it:secondmpsinj} and the 
 Cauchy-Schwarz inequality
\begin{align}
|\alpha_{0,0,0,0}|&= O(\lambda_2^{n})\cdot 
\Big|\big(\bbra{\Lambda}\otimes \bbra{\Lambda}) \hat{E}_{\F^{\otimes n-2\Delta}}
(\kket{I}\otimes \kket{I}\big)\Big|\\
&=O(\lambda_2^n)\cdot  \|\kket{\Lambda}\otimes\kket{\Lambda}\|\cdot
\|\hat{E}_{\F^{\otimes n-2\Delta}}
(\kket{I}\otimes \kket{I}\big)\|\ .
\end{align}
With Lemma~\ref{lem:efboundsimplified} this can further be bounded as 
\begin{align}
|\alpha_{0,0,0,0}|&=O(\lambda_2^n)\cdot \|\kket{\Lambda} \|^2 \cdot \|\kket{I}\|^2\cdot \|\F^{\otimes n-2\Delta}\|\cdot \|E^{n-2\Delta}\|\ .
\end{align} 
Since $\|\F\|=1$ and $ \|\kket{\Lambda} \|=O(1)$, $ \|\kket{I}\|=O(1)$
and $\|E^{n-2\Delta}\|=O(1)$ (cf.~\eqref{eq:bnddimembnd}), we conclude that 
\begin{align}
|\alpha_{0,0,0,0}|&=O(\lambda_2^n)\ .\label{eq:alphazzzzz}
\end{align}
 The remaining terms~$|\alpha_{b_1,b_2,b_3,b_4}|$  with $(b_1,b_2,b_3,b_4) \neq (0,0,0,0)$ can be bounded as follows using inequality~\eqref{eq:matrixboundmultiple}: We have 
\begin{align}
|\alpha_{b_1,b_2,b_3,b_4}|&=\big|\tr\left((H_{b_1}\otimes H_{b_2})\hat{E}_{\F^{\otimes N-2\Delta}}(H_{b_3} \otimes H_{b_4})\left[(\overline{X}\otimes\overline{Y})\otimes (Y\otimes X)\right]\right)\big|\\
&\le  \|H_{b_1}\otimes H_{b_2}\|_F \cdot \|E_{\F^{\otimes N-2n}} \|_F \cdot \|H_{b_3} \otimes H_{b_4}\|_F \cdot  \|\overline{X}\otimes\overline{Y}\otimes Y\otimes X\|_F\\ 
&= \|X\|_F^2\cdot \|Y\|_F^2\cdot \left(\prod_{j=1}^4\|H_{b_j}\|_F\right)\cdot \|E_{\F^{\otimes n-2\Delta}}\|_F\\
&= O(\lambda_2^{\Delta/2})\cdot \|X\|_F^2\cdot \|Y\|_F^2 \cdot \|E_{\F^{\otimes n-2\Delta}}\|_F\ ,
\end{align}
where we use~\eqref{eq:hzeroneupperbound}
and the assumption that $(b_1,b_2,b_3,b_4)\neq (0,0,0,0)$. 
We use Lemma~\ref{lem:efboundsimplified} and~\eqref{eq:bnddimembnd} to get the upper bound  $\|E_{\F^{\otimes n-2\Delta}}\|\leq D^2 \|F^{\otimes n-2\Delta}\|\cdot \|E^{n-D}\|=O(1)$.
Thus 
\begin{align}
|\alpha_{b_1,b_2,b_3,b_4}|&=O(\lambda_2^{\Delta/2})\qquad\textrm{ for }(b_1,b_2,b_3,b_4)\neq (0,0,0,0)\ .\label{eq:bseqneqzero}
\end{align}
Combining~\eqref{eq:bseqneqzero} with~\eqref{eq:alphazzzzz}, 
we conclude that 
\begin{align}
|\tr(\rho_X \rho_Y)|&\le \sum_{b_1,b_2,b_3,b_4 \in \{0,1\}} |\alpha_{b_1,b_2,b_3,b_4}| \le |\alpha_{0,0,0,0}|+15\max_{(b_1,b_2,b_3,b_4)\neq (0,0,0,0)} |\alpha_{b_1,b_2,b_3,b_4}|\\
&=O(\lambda_2^{\Delta/2})\ .
\end{align}
The claim follows from this. 
\end{proof} 
Recall that we call (a family of subspaces)  $\cC\subset(\mathbb{C}^\physical)^{\otimes n}$ an approximate error-detection code if it is an $(\epsilon,\delta)[[n,k,d]]$-code with $\epsilon\rightarrow 0$ and $\delta\rightarrow 0$ for $n\rightarrow\infty$. Our main result is the following:

\begin{theorem}\label{thm:nogo}
Let $\mathcal{C}\subset (\mathbb{C}^{\physical})^{\otimes n}$ be an approximate error-detecting code generated from a translation-invariant  injective MPS of constant bond dimension $D$ by varying boundary conditions. Then the distance of $\mathcal{C}$ is constant.
\end{theorem}
\begin{proof}
Let $\mathcal{C}=\mathcal{C}_n\subset (\mathbb{C}^{\physical})^{\otimes n}$ be a (family of) subspace(s) of dimension~$\physical^k$
defined by an MPS tensor $A$ by choosing different boundary conditions, i.e., 
\begin{align}
\cC_n=\{\ket{\Psi(A,X,n)}\ |\ X\in\cX\}\subset (\mathbb{C}^\physical)^{\otimes n}
\end{align}
for some (fixed) subspace~$\cX\subset\cB(\mathbb{C}^D)$. 
For the sake of contradiction, assume that 
$\mathcal{C}_n$ is an $(\epsilon_n,\delta_n)[[n,k,d_n]]$-code with
\begin{align}
\epsilon_n,\delta_n\rightarrow 0\qquad\textrm{ and }\qquad \textrm{code distance }d_n\rightarrow\infty\qquad\textrm{ for } n\rightarrow\infty\ .\label{eq:epsdeltadasymptotics}
\end{align}
Let $\ket{\Psi_X}=\ket{\Psi(A,X,n)},\ket{\Psi_Y}=\ket{\Psi(A,Y,n)}\in\cC$ be two orthonormal states defined by choosing different boundary conditions~$X,Y\in\cX$. From Lemma~\ref{lem:overlaplowinjectivemps}, we may choose $\Delta$ sufficiently large such that the reduced density operators~$\rho_X,\rho_Y$ on $d$~sites surrounding the boundary satisfies
\begin{align}
\tr(\rho_X\rho_Y)&\leq c \lambda_2^{d/4}\qquad\textrm{ for all}\quad d\ge 2\Delta\ .\label{eq:rhoxyzer} 
\end{align}
We note that $\Delta$ only depends on the transfer operator and is independent of~$n$. Fix any constant $\epsilon,\delta\in (0,1)$ and choose some $d\ge 2\Delta$ sufficiently large such that
\begin{align}
\zeta(\rho_X,\rho_Y):=c D^2\lambda_2^{d/4}\ .
\end{align}
satisfies
\begin{align}
\begin{matrix}
    \epsilon<1-10\zeta \quad \text{and} \quad \delta <(1-\zeta)^2\end{matrix}\ .\label{eq:epszx}
\end{align}
Since by assumption $d_n \rightarrow \infty$, there exists some~$N_0\in\mathbb{N}$ such that
\begin{align}
    d_n>d \qquad\textrm{ for all }n\geq N_0\ .\label{eq:dnlowerb}
\end{align}
Combining~\eqref{eq:rhoxyzer},~\eqref{eq:epszx}, and~\eqref{eq:dnlowerb} with  Lemma~\ref{lem:distancenecessary}, we conclude that~$\cC_n$ is not an $(\epsilon,\delta)[[n,k,d_n]]$-code for any~$n\geq N_0$.

By assumption~\eqref{eq:epsdeltadasymptotics}, there exists some $N_1\in\mathbb{N}$ such that 
\begin{align}
    \epsilon_n<\epsilon\qquad\textrm{ and }\qquad \delta_n<\delta\qquad\textrm{ for all }n\geq N_1\ .
\end{align}
Let us set $N=\max\{N_0,N_1\}$. Then we obtain that
$\cC_n$ is not an $(\epsilon_n,\delta_n)[[n,k,d_n]]$-code for any $n\geq N$, a contradiction.
\end{proof}
In terms of the TQO-1 condition (cf.~\cite{bravyihastingsmichalakis}), Theorem~\ref{thm:nogo} shows the absence of topological order in 1D~gapped systems.
The theorem also tells us that we should not restrict our attention to the ground space of a local Hamiltonian when looking for quantum error-detecting codes.\footnote{Note that this conclusion is only valid for local gapped Hamiltonians in one dimension. When the spatial dimension $d\geq 2$, there are ground spaces that have topological order, e.g. Toric code, and even for higher dimensions good quantum LDPC codes are shown to exist in the ground space of frustration free Hamiltonians~\cite{bohdanowicz2018good}.} In the following sections, we bypass this no-go result by extending our search for codes to low-energy states. In particular, we show that single quasi-particle momentum eigenstates of local gapped Hamiltonians and multi-particle excitations of the gapless  Heisenberg model constitute error-detecting codes. See Sections~\ref{sec:QEDClow} and~\ref{sec:qedcintegrable}, respectively.

\section{AQEDC at low energies: The excitation ansatz\label{sec:QEDClow}}

In this section, we employ tangent space methods for the matrix product state formalism, i.e., the excitation ansatz~\cite{haegemanetal13, haegeman2013elementary, haegeman2014geometry}, in order to show that quasi-particle momentum eigenstates of local gapped Hamiltonians yield an error-detecting code with distance $\Omega(n^{1-\nu})$ for any $\nu \in (0,1)$ and $\Omega(\log n)$ encoded qubits.

In order to render the formalism accessible to an unfamiliar reader, we review the definition of the excitation ansatz in Section~\ref{sec:excitationansatzreview}.  We then develop the necessary calculational ingredients in order to prove the error-detection properties. In Section~\ref{sec:normasymptoticsexcitationansatz}, we compute the norm of the excitation ansatz states to lowest order. In Section~\ref{sec:normboundstransferopexcitation}, we establish (norm) bounds on the transfer operators associated with the excitation ansatz. Then, in Section~\ref{sec:matrixelementlocalexcitation}, we provide estimates on matrix elements of local operators with respect to states appearing in the definition of the excitation ansatz states. Finally, in Section~\ref{sec:excitationansatzcode}, we combine these results to obtain the parameters of quantum error-detecting codes based on the excitation ansatz.\footnote{A simple yet illustrative example of the excitation ansatz states is the following. Consider the $n$-fold product state~$\ket{0}^{\otimes n}$, the  $n$-body $W$-state $$\frac{|10\cdots0\rangle + \cdots + |00\cdots1\rangle}{\sqrt{n}},$$ as well as other $W$-like states with position dependent phase, such as $$\frac{|10\cdots0\rangle + e^{ip}|01\cdots0\rangle+ \cdots +  e^{ip(n-1)}|00\cdots1\rangle}{\sqrt{n}}.$$ Here $p$ can be interpreted as the momentum of a single particle excitation. These states are the ground state and first excited states with different momenta of the non-interacting Hamiltonian~$H= -\sum_i Z_i$. One can represent them by a bond-dimension $D=2$ non-injective MPS which is obtained by expressing the excitation ansatz as a single MPS instead of a sum of injective MPS. One can also consider higher (multi-particle) excitations, which can again be treated by using non-injective MPS. 

We note that error-detecting properties of various subspaces of the low-energy space of this particular simple non-interacting Hamiltonian can be studied either with or without the formalism of MPS. The tangent space methods serve as a powerful tool that allow us to perform our error-detection analysis, not only for the non-interacting cases, but also for the most general interacting Hamiltonians.}

\subsection{MPS tangent space methods: The excitation ansatz\label{sec:excitationansatzreview}}

In~\cite{haegemanetal13}, the MPS ansatz was generalized to a variational class of states  which have non-zero momentum. The resulting states are called the {\em excitation ansatz}. 
An excitation ansatz state~$\ket{\Phi_p(B;A)}\in (\mathbb{C}^\physical)^{\otimes n}$
is specified by two MPS tensors $\{A_i\}_{i=1}^\physical$ and $\{B_i\}_{i=1}^\physical$ of the same bond and physical dimensions, together with a parameter $p\in \{2\pi k/n\ |\ k=0,\ldots,n\}$ indicating the momentum. It is defined as 
\begin{align}
\ket{\Phi_p(B;A)} &= e^{-ip}\sum_{j=1}^n e^{ipj} \sum_{i_1,\ldots,i_n\in [\physical]} \tr(A_{i_1}\cdots A_{i_{j-1}}B_{i_j}A_{i_{j+1}}\cdots A_{i_n}) \ket{i_1\ldots i_n}\ .\mylabel{eq:excitationansatz}
\end{align}
The definition of these states is  illustrated in Figure~\ref{fig:excitationansatz}.
\begin{figure}
\centering
\includegraphics[scale=0.18]{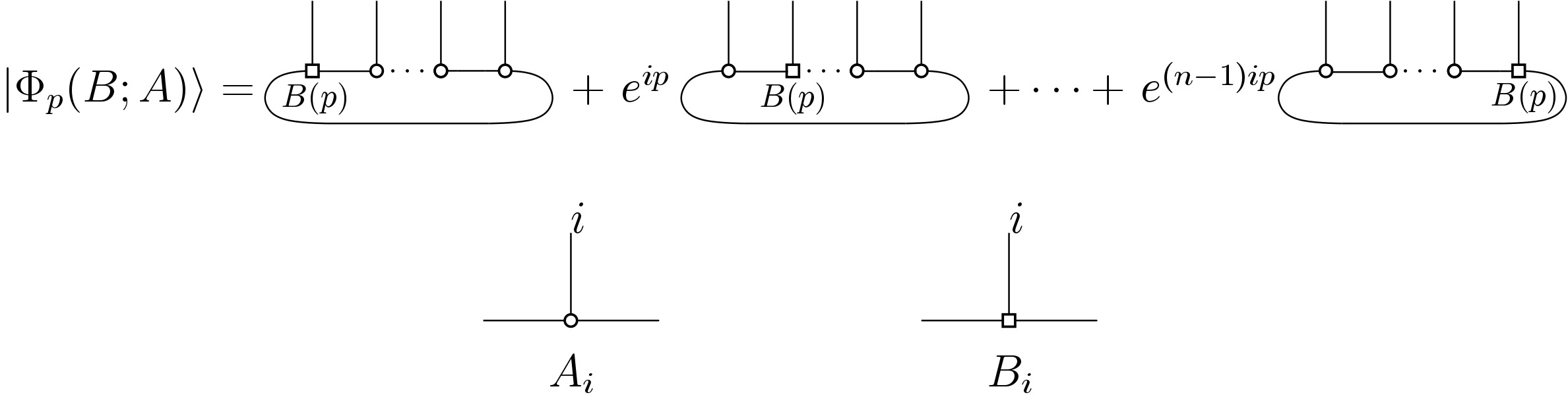}
\caption{This figure illustrates the excitation ansatz~$\ket{\Phi_p(B;A)}$ for $n$ physical spins. \mylabel{fig:excitationansatz}}
\end{figure}
Note that we allow the $B$ tensors themselves to depend on the momentum $p$, so we will sometimes write $B(p)$ when we feel the need to be explicit, and the notation $\ket{\Phi_p(B;A)}$ should really be read as a short-hand for $\ket{\Phi_p(B(p);A)}$. 

It is also useful to define the constituent ``position space" states
\begin{align}
\ket{\Phi_{j,p}(B;A)} &= \sum_{i_1,\ldots,i_n\in [\physical]}
\tr(A_{i_1}\cdots A_{i_{j-1}}B_{i_j}A_{i_{j+1}}\cdots A_{i_n}) \ket{i_1\ldots i_n},\\
&=\raisebox{-.52\height}{\includegraphics[scale=0.16]{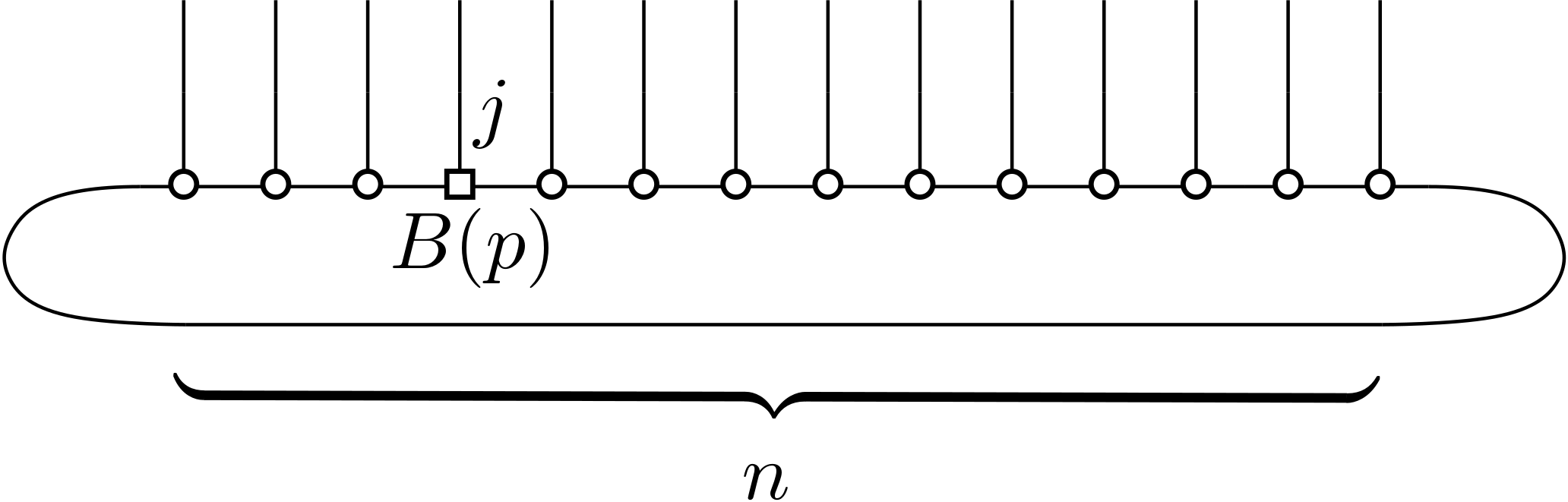}}\mylabel{eq:Phinpdefnitionmain}\ ,
\end{align}
which is the state with a ``single $B(p)$ excitation" at site $j$. Note that we retain the $p$ dependence in the definition of these ``position space" states since the $B$ tensors themselves are generally $p$ dependent.

We call an excitation ansatz state~$\ket{\Phi_p(B;A)}$ \it injective \rm if the transfer operator $\cE(A)$ associated with $\{A_j\}_{j=1}^\physical$ is primitive, which is the only case we consider in this work. Denoting the transfer matrix associated with $\cE(A)$ simply as $E$, it will also be useful to define several other mixed transfer matrices as follows:
$$
\begin{matrix}
E_{B(p)}&=&\sum_{j=1}^D \overline{A}_j\otimes B_j(p)&=&\raisebox{-.47\height}{\includegraphics[scale=0.05]{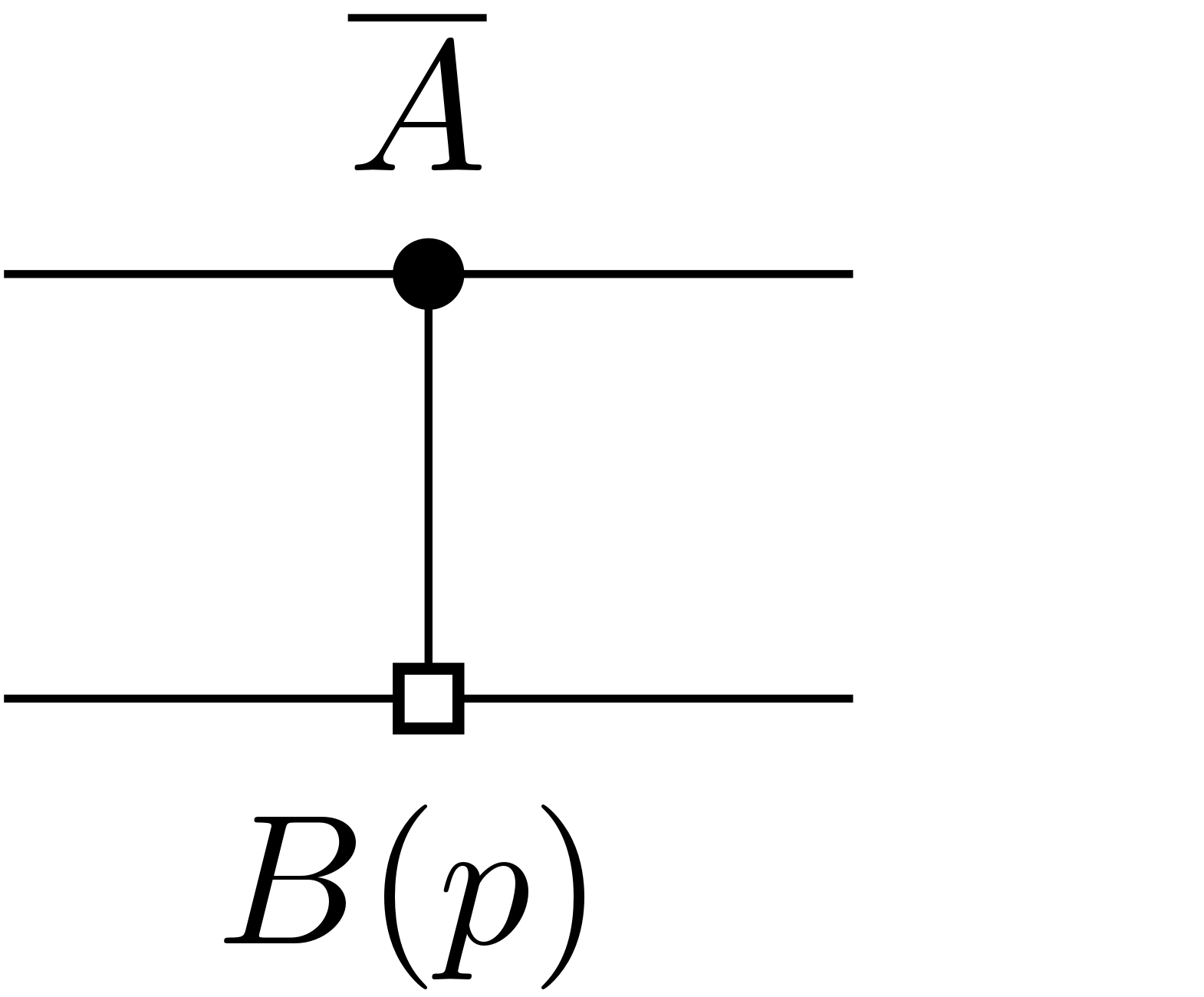}},\\
\\
E_{\overline{B(p)}}&=&\sum_{j=1}^D \overline{B_j(p)}\otimes A_j&=&\raisebox{-.47\height}{\includegraphics[scale=0.05]{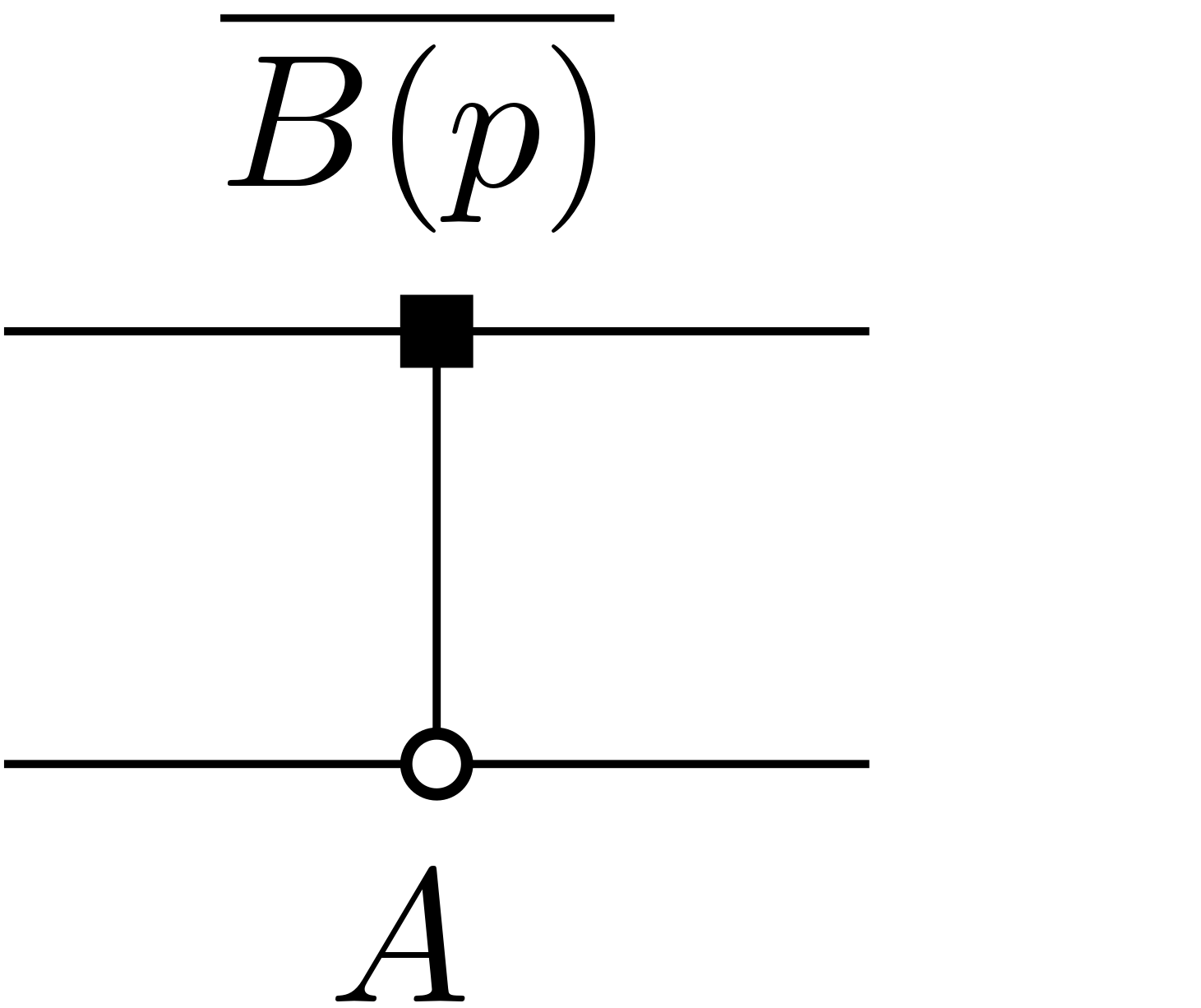}} \mylabel{eq:mixedtransferdef},\\
\\
E_{\overline{B(p')}B(p)}&=&\sum_{j=1}^D \overline{B_j(p')}\otimes B_j(p)&=&\raisebox{-.47\height}{\includegraphics[scale=0.05]{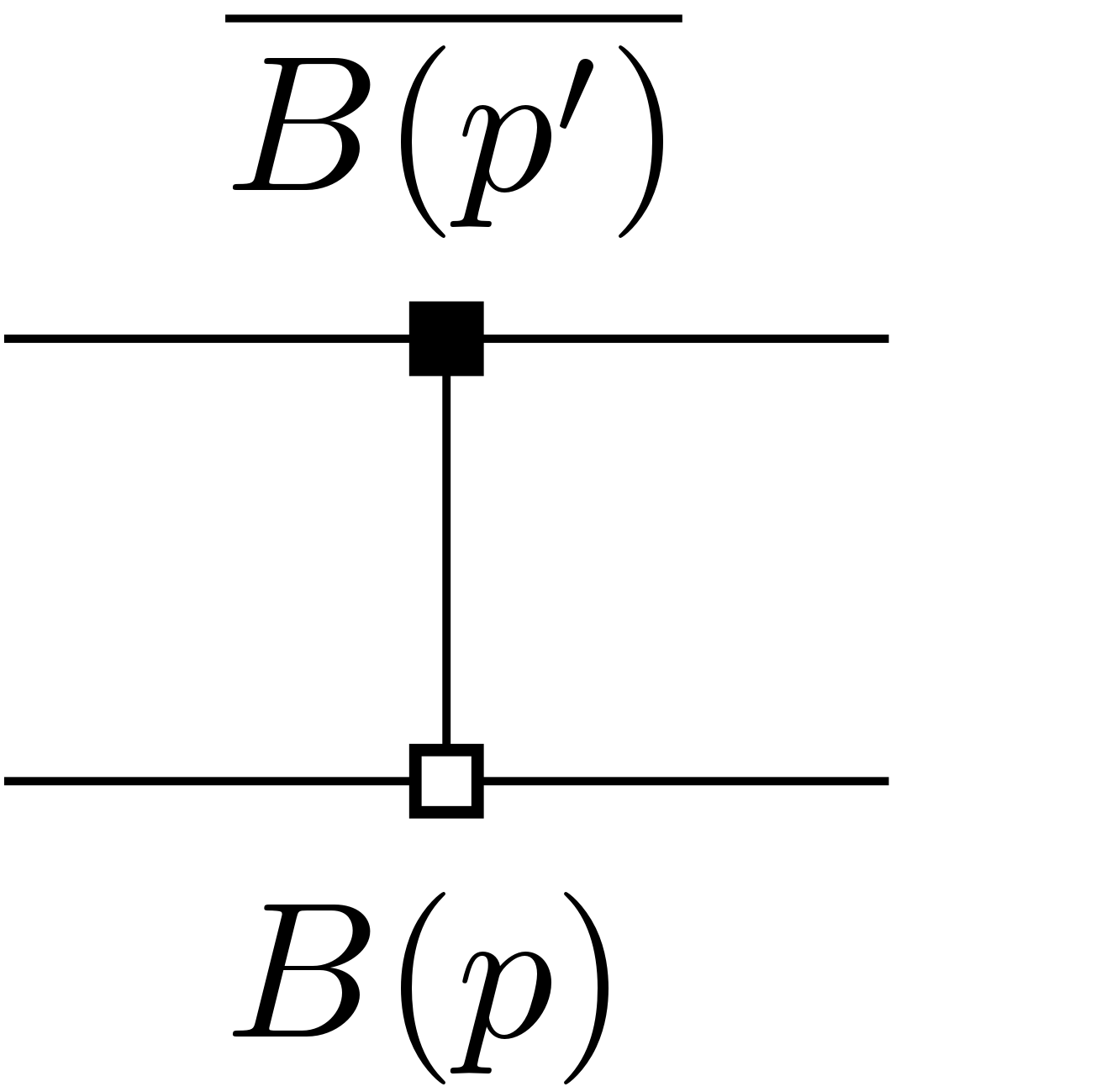}}.\\
\end{matrix}
$$
For brevity, we often suppress the dependence on $A$ and $B$ and simply write $\ket{\Phi_p} \equiv \ket{\Phi_p(B;A)}$ when no confusion is possible.

In addition to the multiplicative gauge freedom \eqref{eq:gaugefreedommps}, the excitation ansatz admits an additional additive gauge freedom. Exploiting this additive gauge freedom, the following statement can be shown (see~\cite[Equation~(154)]{haegeman2014geometry}):
\begin{lemma}\label{lem:canonicalformexcitationansatz}
Let $\ket{\Phi_p(B;A)}$ be an injective excitation ansatz state and assume that $A$ is normalized such that the transfer operator has spectral radius~$1$. Let $\ell$ and $r$ be the corresponding left- and right- eigenvectors corresponding to eigenvalue $1$. 
Assume $p\neq 0$.\footnote{We have made the $p\neq 0$ assumption here for simplicity. The gauge condition also holds for $p=0$ in the form $\bbra{\ell}E_{\tilde{B}(p)} = \bbra{\ell}E_{\overline{\tilde{B}}(p)}=O(\lambda_2^n)$. All of the results presented below for $p\neq 0$ also hold for $p=0$ up to an exponentially small error.} Then there exists a tensor $\tilde{B}$
such that $\ket{\Phi_p(B;A)}=\ket{\Phi_p(\tilde{B};A)}$, and such that
\begin{align}
\bbra{\ell} E_{\tilde{B}(p)}&=0\qquad\textrm{ and }\qquad \bbra{\ell} E_{\overline{\tilde{B}(p)}}=0\ .\mylabel{eq:gaugeconditionrewrittenexcit}
\end{align}
\end{lemma}

For completeness, we give a proof of this statement in Appendix~\ref{app:canonicalexcitation}. Below, we assume that all excitation ansatz states satisfy the gauge condition \eqref{eq:gaugeconditionrewrittenexcit}.

\subsection{The norm of an excitation ansatz state\label{sec:normasymptoticsexcitationansatz}}
For a family of excitation ansatz states~$\{\ket{\Phi_p(B;A)}\}_{p}$ we define the constants
\begin{align}
c_{pp'}&=\bbra{\ell} E_{\overline{B(p')}B(p)}\kket{r}=\raisebox{-.45\height}{\includegraphics[scale=0.055]{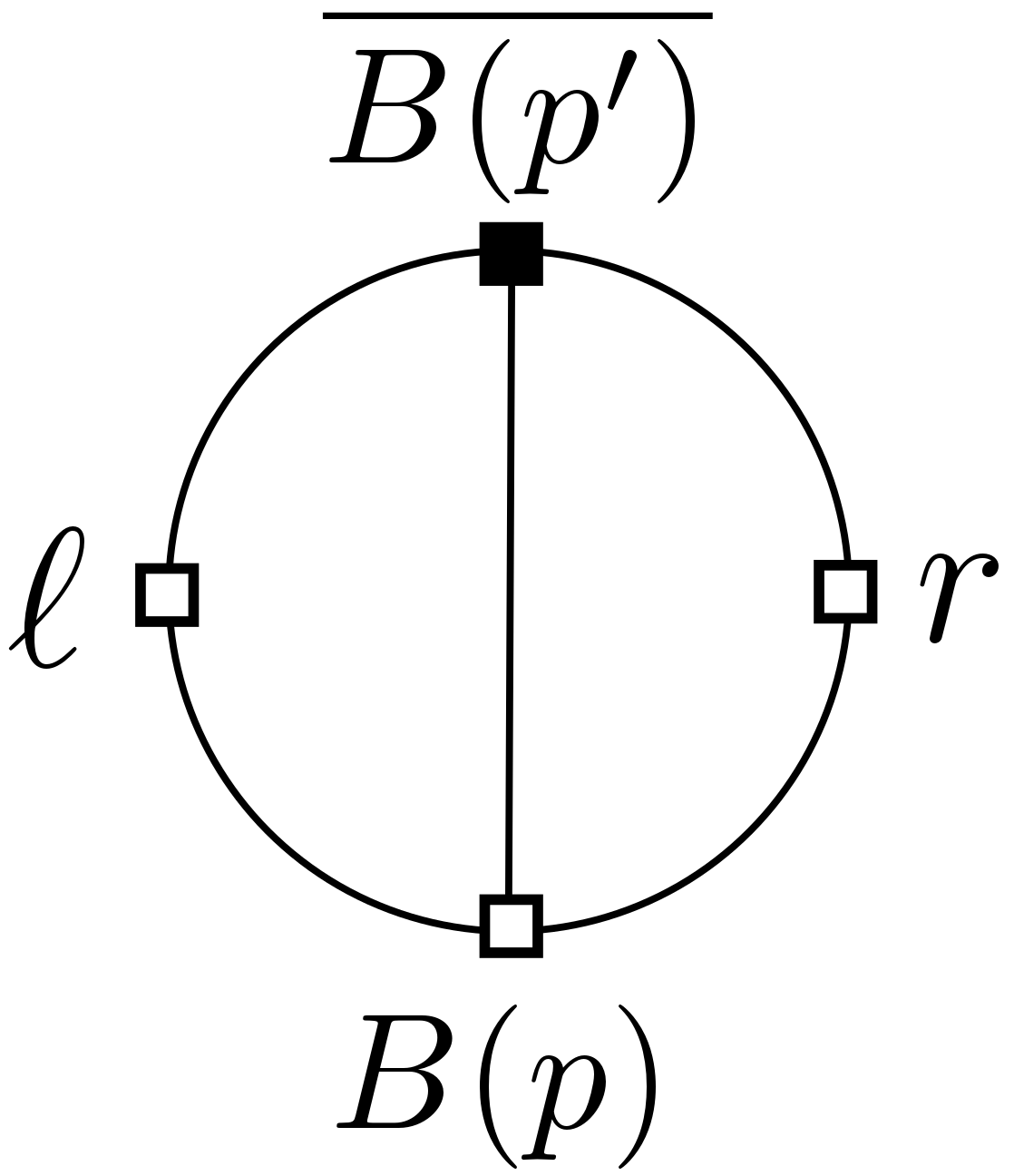}}.
\end{align}
We also write $c_p := c_{pp}$. These appear in the norm of the excitation ansatz states as follows:
\begin{lemma}\mylabel{lem:normofexcitationansatzstates}
The norm of an excitation ansatz state~$\ket{\Phi_{p}(B;A)}\in (\mathbb{C}^\physical)^{\otimes n}$ satisfies 
\begin{align}
\|\Phi_{p}(B;A)\| &= \sqrt{n c_p} + O(n^{3/2}\lambda_2^{n/6}), \mylabel{eq:normofexcitationansatzstates}
\end{align}
where $\lambda_2$ is the second largest eigenvalue of the transfer matrix $E$.
\end{lemma}
\begin{proof}
Using the mixed transfer operators defined in \eqref{eq:mixedtransferdef}, we can write the norm of the state $\ket{\Phi_{p}(B;A)}$ as a sum over pairs $(j,j')\in [n]^2$ satisfying $j<j'$, $j=j'$, and $j>j'$ respectively, as follows:
\begin{align}
\|\Phi_{p}(B;A)\|^2&=\sum_{j<j'} e^{ip(j-j')} \tr\left(E^{j-1}E_{B(p)}E^{j'-j-1}E_{\overline{B(p)}}E^{n-j'}\right)\nonumber\\
&\qquad +\sum_{j>j'} e^{ip (j-j')} \tr\left(E^{j-1}E_{\overline{B(p)}}E^{j'-j-1}E_{B(p)}E^{n-j'}\right)\nonumber\\
&\qquad +\sum_{j=1}^n \tr\left(E^{j-1}E_{\overline{B(p)}B(p)}E^{n-j}\right).\mylabel{eq:phinpsquarednormexpr}
\end{align}
Consider an individual term~$\tr\left(E^{j-1}E_{B(p)}E^{j'-j-1}E_{\overline{B(p)}}E^{n-j'}\right)$ in the first sum. By the cyclicity of the trace, it can be expressed as
\begin{align}
\tr\left(E^{j-1}E_{B(p)}E^{j'-j-1}E_{\overline{B(p)}}E^{n-j'}\right)
&=\tr\left(E_{B(p)}E^{\Delta-1}E_{\overline{B(p)}} E^{n-\Delta-1}\right)\ ,
\end{align}
where $\Delta=j'-j$. Clearly, one of the terms $\Delta-1$ or $n-\Delta-1$ must be lower bounded by $n/3$. Assume that it is the first (the argument for the other case is analogous), i.e., that 
\begin{align}
\Delta-1>n/3\ .\mylabel{eq:deltanthirdlowerbound}
\end{align}
Then we may substitute the Jordan decomposition of $E$ in the form
\begin{align}
E^{\Delta-1}&=\kket{r}\bbra{\ell}\oplus\tilde{E}^{\Delta-1}\ ,
\end{align}
which allows us to write
\begin{align}
\tr\left(E^{j-1}E_{B(p)}E^{j'-j-1}E_{\overline{B(p)}}E^{n-j'}\right)
&=\tr\left(E_{B(p)}\kket{r}\bbra{\ell}E_{\overline{B(p)}}E^{n-\Delta-1}\right)\\
&\quad+\tr\left(E_{B(p)}\tilde{E}^{\Delta-1}E_{\overline{B(p)}}E^{n-\Delta-1}\right).
\end{align}
By the gauge condition~\eqref{eq:gaugeconditionrewrittenexcit}, the first term vanishes. The magnitude of the second term can be bounded by inequality~\eqref{eq:matrixboundmultiple}, giving 
\begin{align}
\tr\left(E_{B(p)}\tilde{E}^{\Delta-1}E_{\overline{B(p)}}E^{n-\Delta-1}\right)&\leq O(1) \cdot \|\tilde{E}^{\Delta-1}\|_F\cdot \|E^{n-\Delta-1}\|_F\ .
\end{align}
Here we used the fact that $\|E_{\overline{B(p)}}\|_F=O(1)$ and $\|E_{B(p)}\|_F=O(1)$. With~\eqref{eq:deltanthirdlowerbound} and Lemma~\ref{lem:normscalingtransferop}(ii), we have $\|\tilde{E}^{\Delta-1}\|_F\leq \lambda_2^{n/6}$ and $\|E^{n-\Delta-1}\|_F=O(1)$. We conclude that 
\begin{align}
\left|\tr\left(E^{j-1}E_{B(p)}E^{j'-j-1}E_{\overline{B(p)}}E^{n-j'}\right)\right|= O(\lambda_2^{n/6})\ 
\end{align}
for all pairs $(j,j')$ with $j<j'$.

Identical reasoning gives us a bound of the form
\begin{align}
\tr\left(E^{j-1}E_{\overline{B(p)}}E^{j'-j-1}E_{B(p)}E^{n-j'}\right)&= O(\lambda_2^{n/6})
\end{align} 
for all pairs $(j,j')$ with $j>j'$. Inserting this into the sum~\eqref{eq:phinpsquarednormexpr}, we obtain
\begin{align}
\|\Phi_{p}(B;A)\|^2&=\sum_{j=1}^n \tr\left(E^{j-1}E_{\overline{B(p)}B(p)}E^{n-j}\right)+O(n^2\cdot\lambda_2^{n/6})\ .\mylabel{eq:phnpmv}
\end{align}
By the cyclicity of the trace and the Jordan decomposition of $E$, we have 
\begin{align}
\tr\left(E^{j-1}E_{\overline{B(p)}B(p)}E^{n-j}\right)&=
\tr(E_{\overline{B(p)}B(p)}E^{n-1})\\
&=\bbra{\ell} E_{\overline{B(p)}B(p)}\kket{r}+\tr(E_{\overline{B(p)}B(p)} \tilde{E}^{n-1})\\
&=c_p + \tr(E_{\overline{B(p)}B(p)} \tilde{E}^{n-1}).
\end{align}
Again using inequality~\eqref{eq:matrixboundmultiple} and Lemma~\ref{lem:normscalingtransferop}(ii), we get
\begin{align}
\left|\tr(E_{\overline{B(p)}B(p)} \tilde{E}^{n-1})\right|
&\leq \|E_{\overline{B(p)}B(p)}\|_F\cdot \|\tilde{E}^{n-1}\|_F = O\left(\lambda_2^{(n-1)/2}\right)\ .
\end{align}
Inserting this into~\eqref{eq:phnpmv} and noting that $\lambda_2^{(n-1)/2}\leq n\cdot\lambda_2^{n/6}$ gives us
\begin{align}
\|\Phi_{p}(B;A)\|^2 &= nc_p +O(n^2\cdot \lambda^{n/6})=nc_p(1+O(n\cdot \lambda^{n/6}))\ .
\end{align}
Taking the square root yields the desired claim. 
\end{proof}

\subsection{Bounds on transfer operators associated with the excitation ansatz\label{sec:normboundstransferopexcitation}}
For an operator $F\in (\mathbb{C}^\physical)^{\otimes L}$, sites $j,j'\in [L]$ and momenta $p,p'$, let us define operators
on~$\mathbb{C}^D\otimes\mathbb{C}^D$ by the diagrams
\begin{align}
E_F(j,p,j',p')=\raisebox{-.54\height}{\includegraphics[scale=0.1]{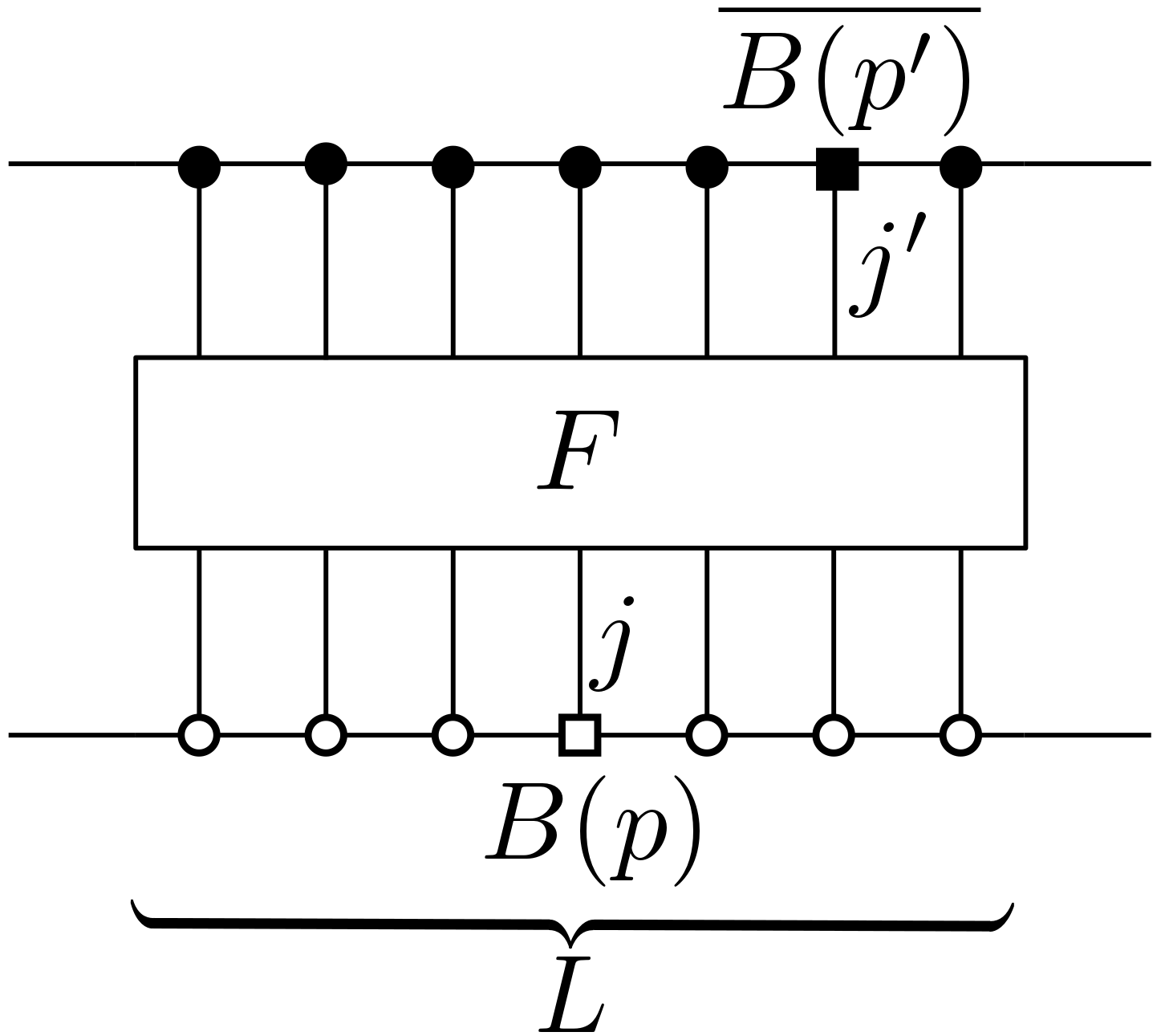}},\qquad \hbox{and} \qquad 
E_F(j,p)=\raisebox{-.60\height}{\includegraphics[scale=0.09]{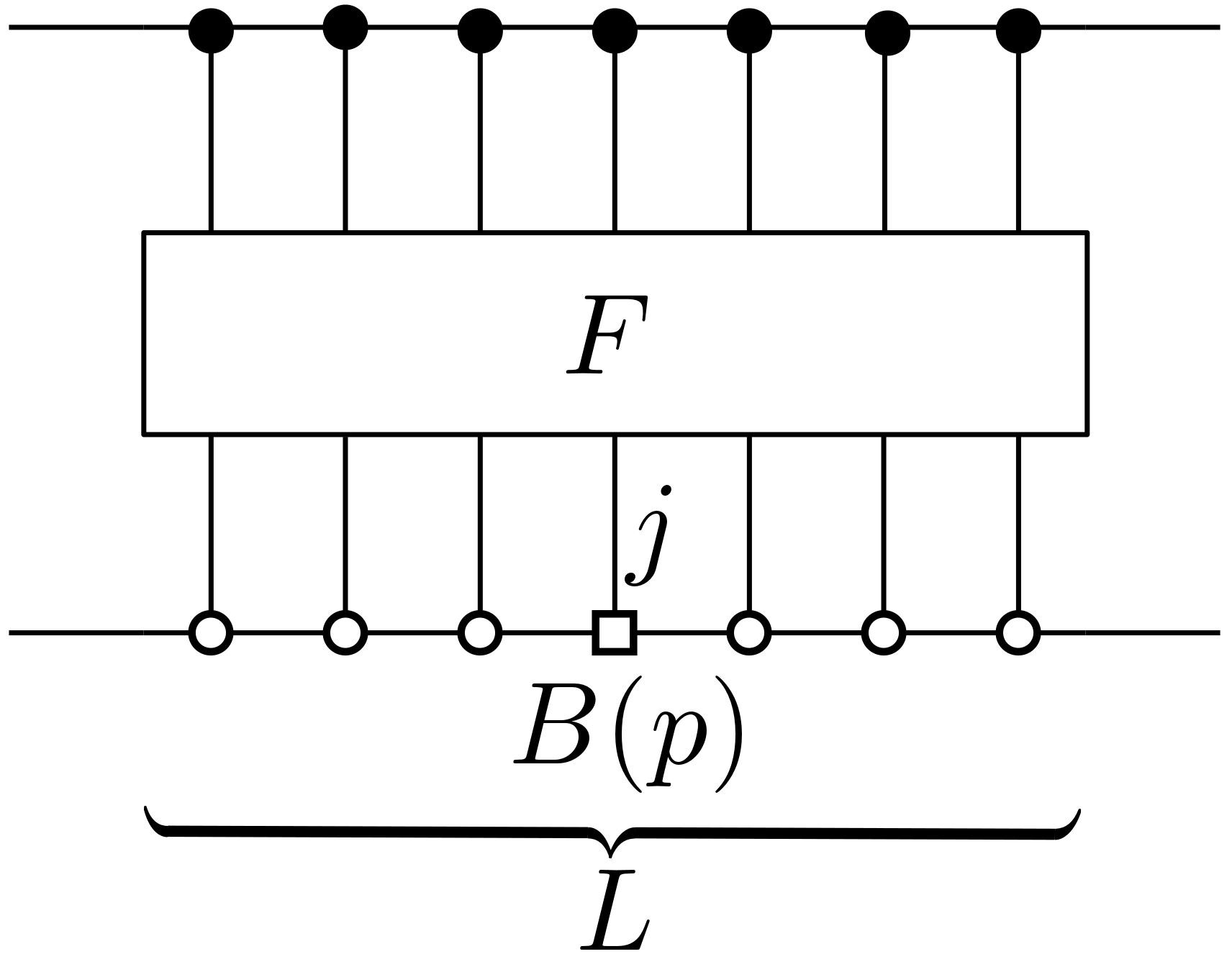}}\ .
\end{align}
We also denote by $E(j,p,j',p')$ the  expression $E_{I}(j,p,j',p')$.

We keep the dependence of $E_F$ on $L$ implicit, since none of our computations will explicitly depend on~$L$. Similar to the bounds discussed in Section~\ref{sec:normboundsgeneralizedtrsfop}, we require bounds on the norm (respectively matrix elements) of these transfer operators. These are given by the following:
\begin{lemma}\mylabel{lem:normboundtransopxnpnprime}
Let $F\in (\mathbb{C}^\physical)^{\otimes L}$, $j,j'\in [L]$, and momenta $p,p'$ be arbitrary. Then we have
\begin{align}
\bbra{\ell }E(j,p,j',p')\kket{r}&=\delta_{j,j'}c_{pp'}\ ,\mylabel{eq:ematrixelementlr}
\end{align}
and
\begin{align}
|\bbra{\ell}E_F(j,p,j',p')\kket{r}|&\leq \|F\| \sqrt{c_pc_{p'}}\mylabel{eq:ellrsandwitch}\ , \\
\|E_F (j,p,j',p')\|_F&\leq D^2 \|F\| \sqrt{\|E_{\overline{B(p)}B(p)}\|_F\, \|E_{\overline{B(p')}B(p')}\|_F}\ , \mylabel{eq:efdeltafrobeniusnorm}\\
\|E_F(j,p)\|_F &\le D^2\|F\|\sqrt{\|E_{\overline{B(p)}B(p)}\|_F}\ ,\mylabel{eq:efdeltafrobeniusnorm2}\\
\|E_F\|_F &\le D^2\|F\|\ .\mylabel{eq:efdeltafrobeniusnorm3}
\end{align}
\end{lemma}

For the proof of Lemma~\ref{lem:normboundtransopxnpnprime} (and other arguments below), we make repeated use of the following states. Let $L\in [n]$. Define
\begin{align}
\ket{\Phi^L_{j,p}}&=\raisebox{-.58\height}{\includegraphics[scale=0.07]{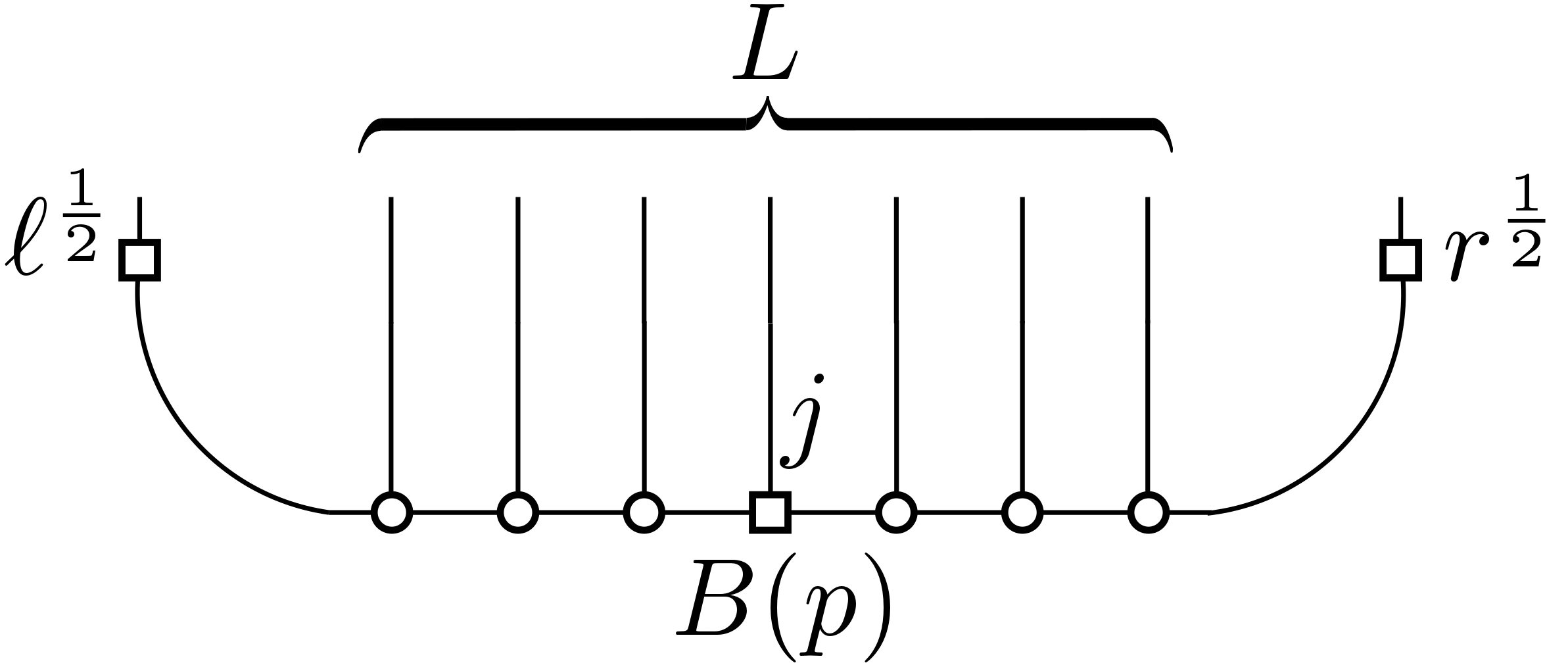}}\mylabel{eq:phinpaux} 
\end{align}
on $\mathbb{C}^D\otimes (\mathbb{C}^\physical)^{\otimes L}\otimes \mathbb{C}^D$. Despite the similar notation, these states are not to be confused with the ``position space" states~$\ket{\Phi_{j,p}}$ introduced in equation~\eqref{eq:Phinpdefnitionmain}. The key property of the states $\ket{\Phi^L_{j,p}}$ is the following:
\begin{lemma}
The states~\eqref{eq:phinpaux}  have inner product
\begin{align}
\langle \Phi^L_{j',p'}|\Phi^L_{j,p}\rangle &=\delta_{j,j'}c_{pp'}\ ,\label{eq:innerproductnpnpprimestatement}
\end{align}
independently of the value of $L$. 
\end{lemma}

\begin{proof}
First, consider the case where $j'=j$. Then we have \begin{align}
\langle \Phi^L_{j,p'}\ket{\Phi^L_{j,p}}
&=\bbra{\ell} E^{j-1}E_{\overline{B(p')}B(p)}E^{L-j} \kket{r}=\bbra{\ell} E_{\overline{B(p')}B(p)}\kket{r}=c_{pp'}\ ,
\end{align}
where we have used the fixed-point equations~\eqref{eq:leftrighteveqs}. That is, we have 
\begin{align}
\langle \Phi^L_{j,p'}\ket{\Phi^L_{j,p}}&=\raisebox{-.46\height}{\includegraphics[scale=0.11]{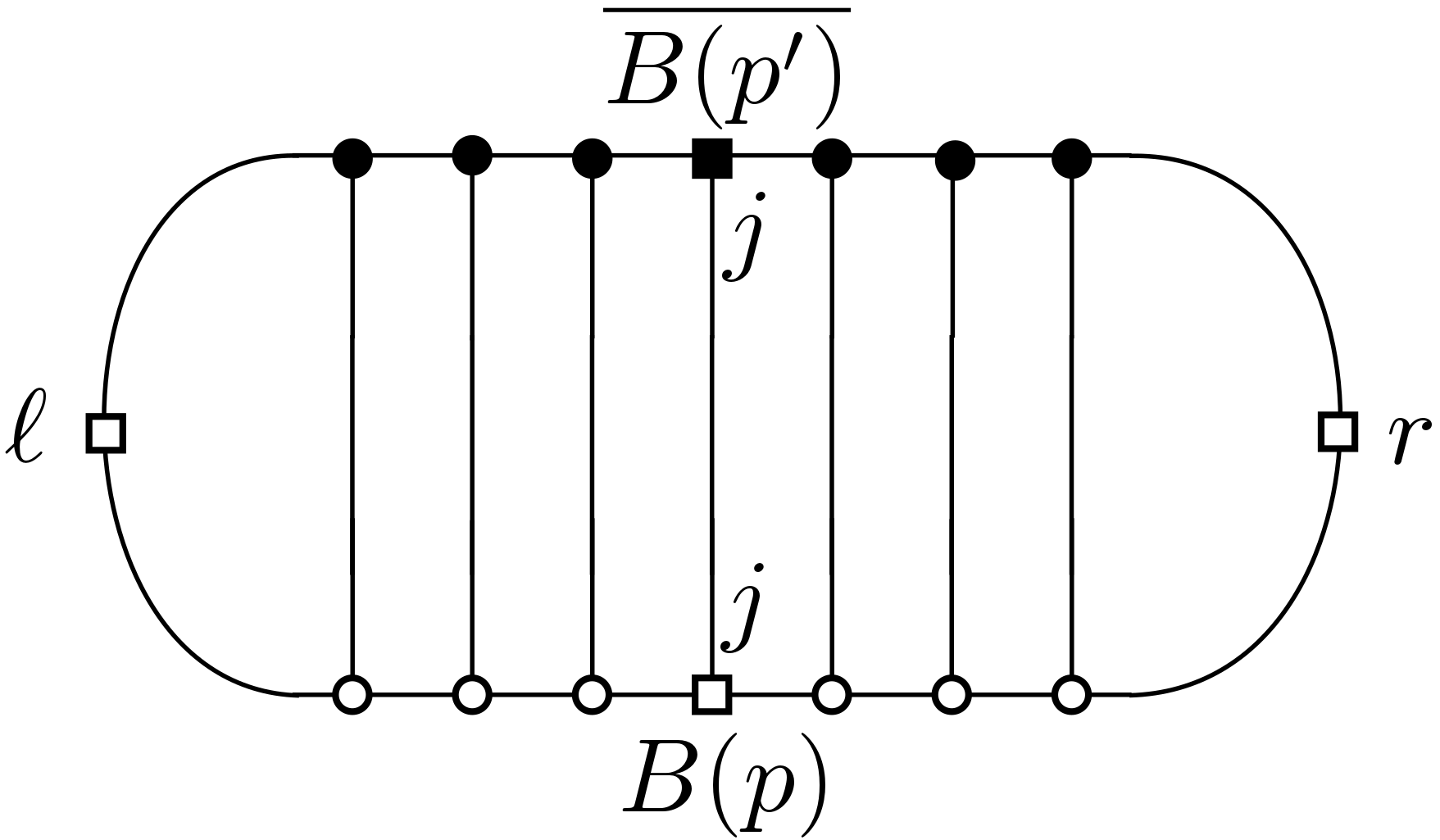}}=\raisebox{-.45\height}{\includegraphics[scale=0.055]{Xcpprimedef.png}}\ .
\end{align}
In a similar fashion, we can compute, for $j < j'$,
\begin{align}
\bra{\Phi^L_{j',p'}}\Phi^L_{j,p}\rangle &=\bbra{\ell}E^{j-1}E_{B(p)}E^{j'-j}E_{\overline{B(p')}}E^{L-j'}\kket{r}=\bbra{\ell}E_{B(p)}E^{j'-j}E_{\overline{B(p')}}\kket{r}=0,
\end{align}
where we have used the fixed-point equations~\eqref{eq:leftrighteveqs} and the gauge condition~\eqref{eq:gaugeconditionrewrittenexcit}. 
The proof for $j>j'$ is analogous. 
\end{proof}

\begin{proof}[Proof of Lemma~\ref{lem:normboundtransopxnpnprime}]
We first prove~\eqref{eq:ellrsandwitch}. The expression of interest can be written diagrammatically as
\begin{align}
\bbra{\ell}E_F (j,p,j',p')\kket{r}&=\raisebox{-.54\height}{\includegraphics[scale=0.1]{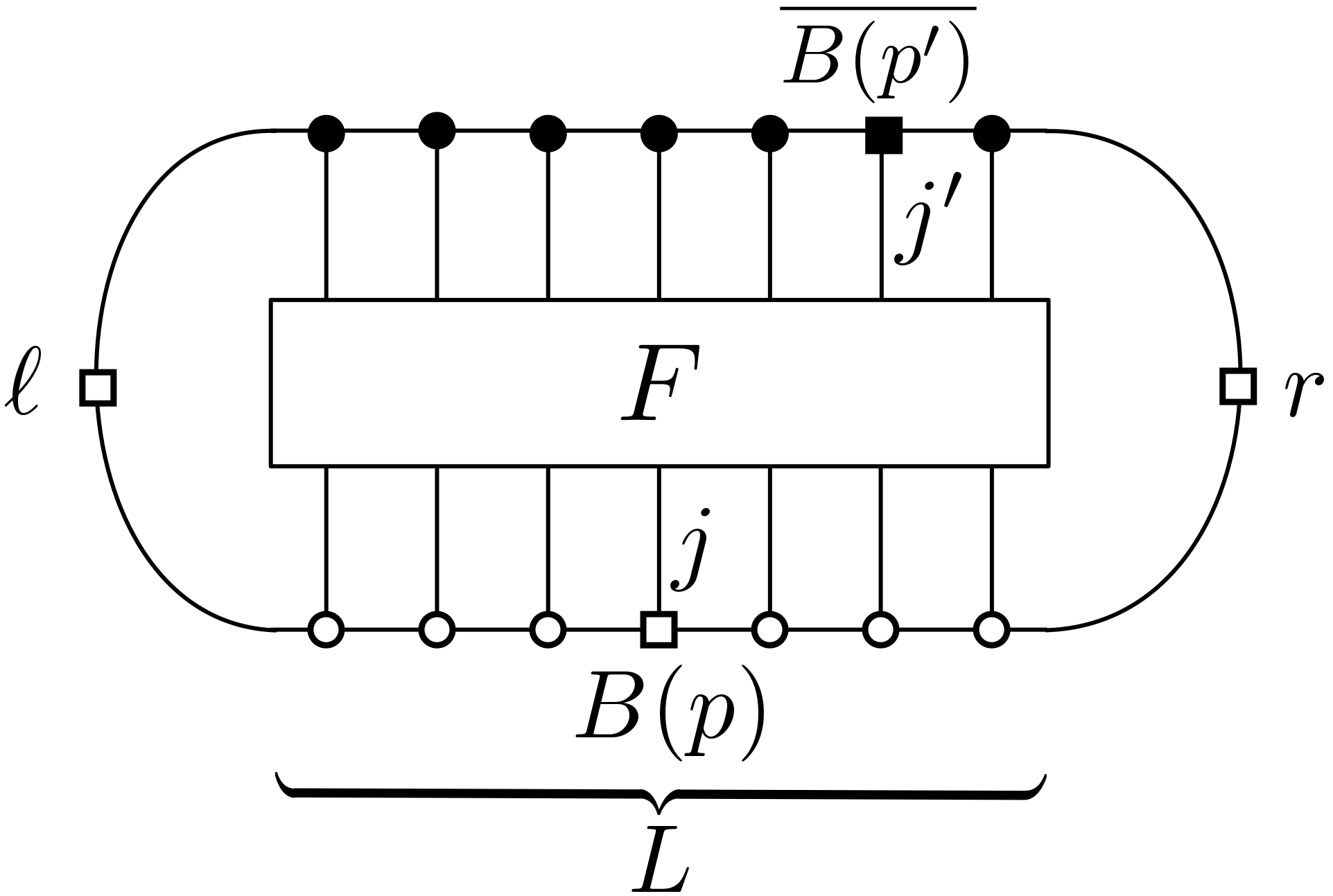}}=\bra{\Phi^L_{j',p'}}(I\otimes F\otimes I)\ket{\Phi^L_{j,p}}\  .
\end{align}
Equation~\eqref{eq:ematrixelementlr} follows by setting~$F$ to be equal to the identity on~$(\mathbb{C}^\physical)^{\otimes L}$ and using the orthogonality relation~\eqref{eq:innerproductnpnpprimestatement}. 
Furthermore, we have 
\begin{align}
\left|\bbra{\ell}E_F (j,p,j',p')\kket{r}\right|&=\left|\bra{\Phi^L_{j',p'}}(I\otimes F\otimes I)\ket{\Phi^L_{j,p}}\right| \leq \|F\|\cdot \|\Phi^L_{j,p}\|\cdot \|\Phi^L_{j',p'}\| .
\end{align}
The claim~\eqref{eq:ellrsandwitch} then follows from~\eqref{eq:innerproductnpnpprimestatement}.

Let us next prove~\eqref{eq:efdeltafrobeniusnorm}. By the definition of the Frobenius norm ~$\|\cdot\|_F$, we have 
\begin{align}
\|E_F (j,p,j',p')\|^2_F&=\sum_{\alpha_1,\alpha_2,\beta_1,\beta_2=1}^D \big|(\bra{\alpha_1}\otimes\bra{\alpha_2}) E_F(j,p,j',p')(\ket{\beta_1}\otimes\ket{\beta_2})\big|^2\ 
\end{align}
where $\{\ket{\alpha}\}_{\alpha=1}^D$ is an orthonormal basis of~$\mathbb{C}^D$.  The terms in the sum can be written diagrammatically as 
\begin{align}
(\bra{\alpha_1}\otimes\bra{\alpha_2}) E_F(j,p,j',p')(\ket{\beta_1}\otimes\ket{\beta_2})&=\raisebox{-.47\height}{\includegraphics[scale=0.095]{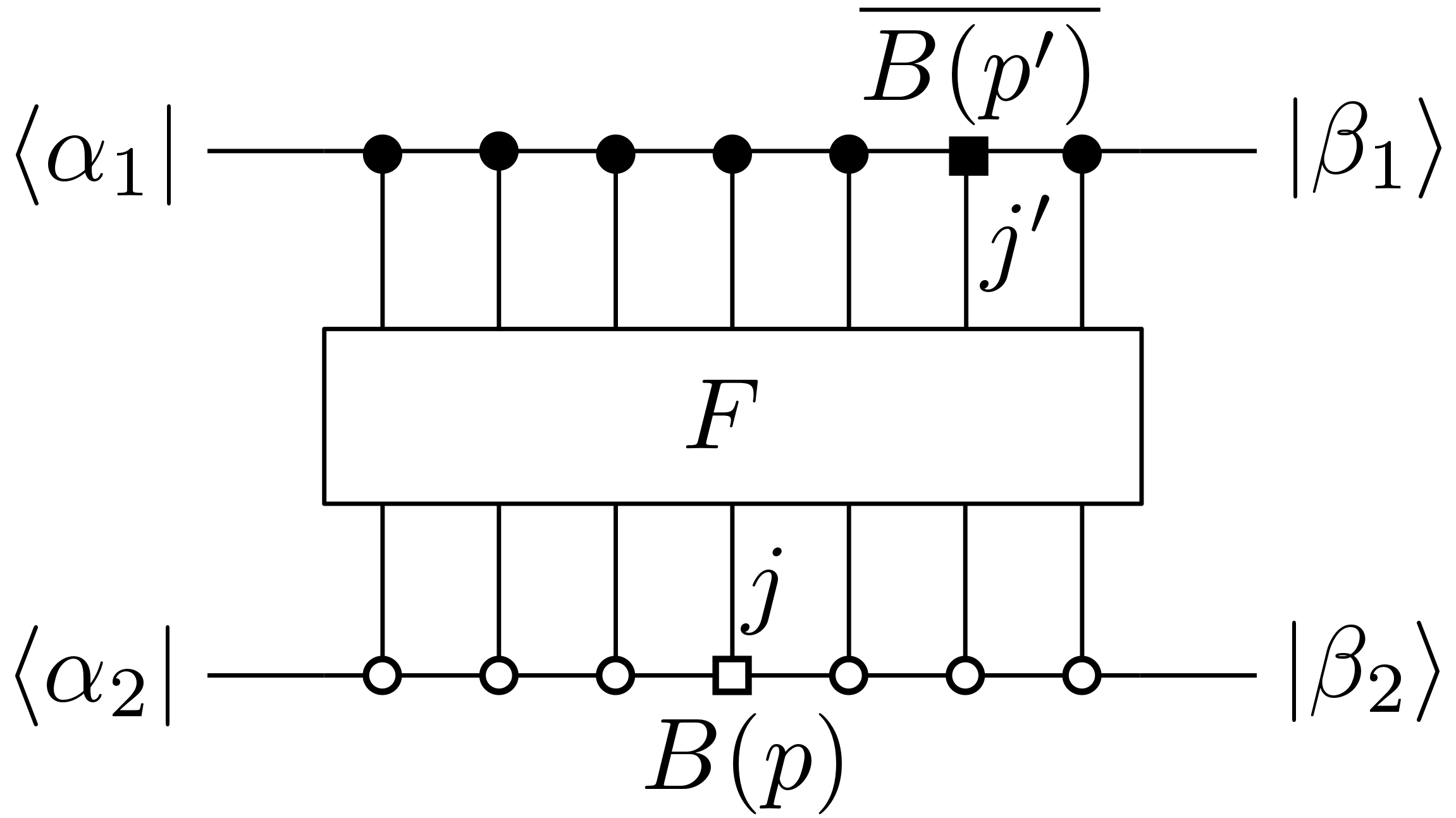}}.
\end{align}
Defining vectors
\begin{align}
\ket{\Psi_{j,p}(\alpha,\beta)}&=\raisebox{-.58\height}{\includegraphics[scale=0.07]{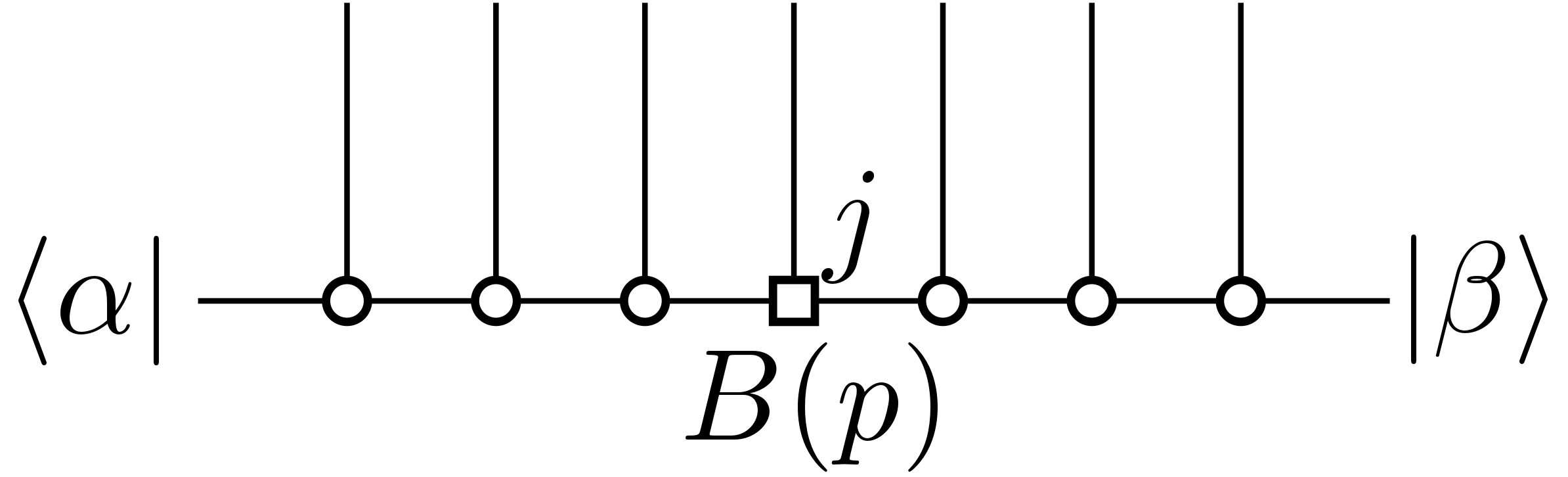}}\mylabel{eq:psinpjkvecs}
\end{align}
on $(\mathbb{C}^{\physical})^{\otimes L}$, we have
\begin{align}
|(\bra{\alpha_1}\otimes\bra{\alpha_2}) E_F(j,p,j',p')(\ket{\beta_1}\otimes\ket{\beta_2})|^2&=|\bra{\Psi_{j',p'}(\alpha_1,\beta_1)}F \ket{\Psi_{j,p}(\alpha_2,\beta_2)}|^2\\
&\leq  \|F\|^2\cdot  \|\Psi_{j,p}(\alpha_2,\beta_2)\|^2\cdot \|\Psi_{j',p'}(\alpha_1,\beta_1)\|^2\ .
\end{align}
The norm of the vector~\eqref{eq:psinpjkvecs} can be bounded as
\begin{align}
\|\Psi_{j,p}(\alpha,\beta)\|^2 &=\raisebox{-.45\height}{\includegraphics[scale=0.085]{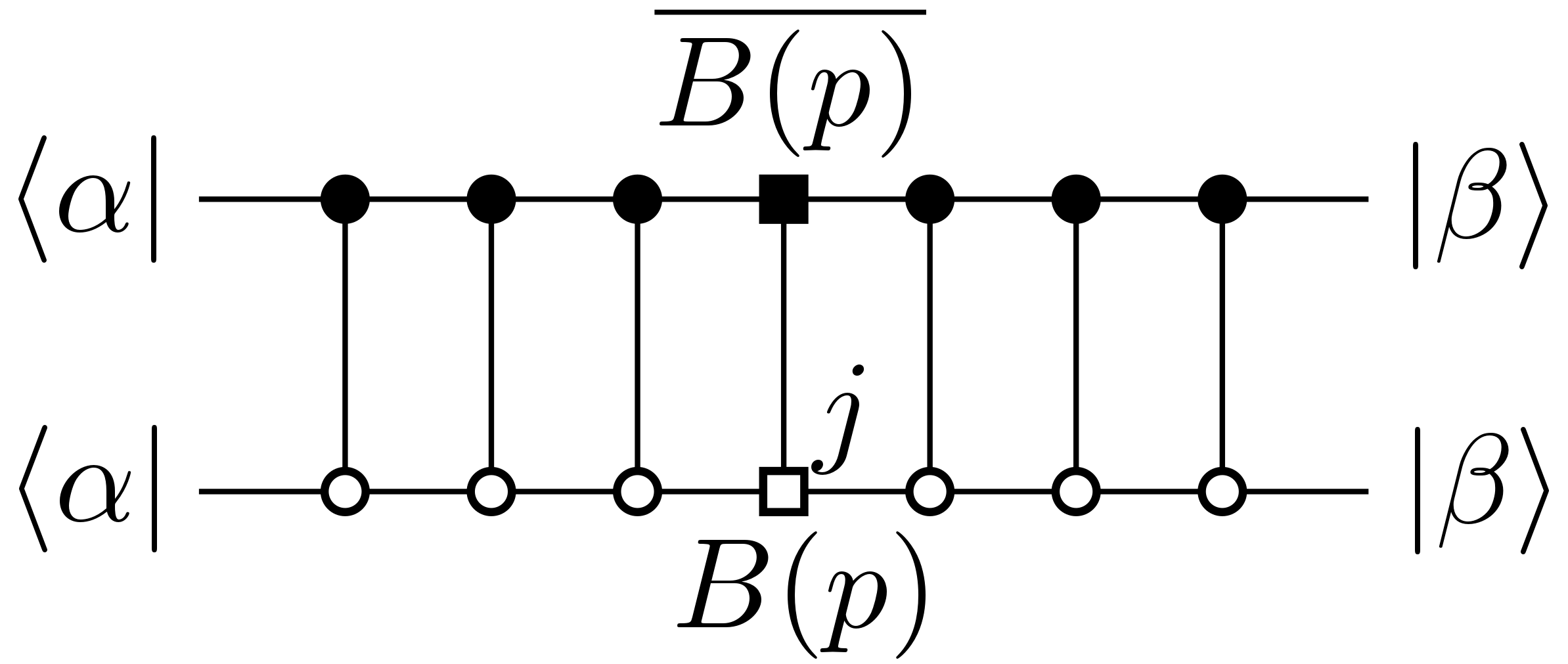}}\\
&= \big|\tr\left(E^{j-1}E_{\overline{B(p)}B(p)}E^{L-j}(\ket{\beta}\bra{\alpha}\otimes\ket{\beta}\bra{\alpha})\right)\big|\\
&\leq  \|E_{\overline{B(p)}B(p)}\|_F\cdot \|E^{j-1}\|_F\cdot \|E^{L-j}\|_F\\
&\leq \|E_{\overline{B(p)}B(p)}\|_F\ .
\end{align}
In the first inequality, we have used~\eqref{eq:matrixboundmultiple}, together with the fact that \begin{align}
    \|\ket{\beta}\bra{\alpha}\otimes\ket{\beta}\bra{\alpha}\|_F = 1\ .
\end{align}
In the second inequality, we have used Lemma~\ref{lem:normscalingtransferop}, along with the fact $\rho(E)=1$. The claim~\eqref{eq:efdeltafrobeniusnorm} follows from this. 

With a completely analogous proof, we also have
\begin{align}
    \|E_F(j,p)\|_F \le D^2\|F\|\sqrt{\|E_{\overline{B(p)}B(p)}\|_F},\qquad \text{and}\qquad \|E_F\|_F \le D^2\|F\|,
\end{align}
which are claims~\eqref{eq:efdeltafrobeniusnorm2} and~\eqref{eq:efdeltafrobeniusnorm3}.
\end{proof}

\subsection{Matrix elements of local operators in the excitation ansatz\label{sec:matrixelementlocalexcitation}}

\subsubsection{Overview of the proof}
Let us give a high-level overview of the argument used to establish our main technical result, Lemma~\ref{lem:excitationansatzmainlem}. The latter gives estimates on matrix elements $\bra{\phi_{p'}}F\ket{\phi_p}$ of a $d$-local operator~$F$ with respect to normalized excitation ansatz states~$\ket{\phi_p}$ and $\ket{\phi_{p'}}$, with possibly different momenta~$p$ and~$p'$.  More precisely, to apply the approximate Knill-Laflamme conditions for approximate error-detection, we need  to establish two kinds of bounds:
\begin{enumerate}
\item
For $p\neq p'$ (i.e., the non-diagonal elements), our aim is to argue that $|\bra{\phi_{p'}}F\ket{\phi_p}|$ vanishes as an inverse polynomial in~$n$. This is ultimately a consequence of the fact that in the Jordan decomposition $E=\kket{r}\bbra{\ell} \oplus \tilde{E}$ of the transfer matrix, the sub-dominant term $\tilde{E}$ has norm decaying exponentially with a rate determined by the second largest eigenvalue~$\lambda_2$.
\item\label{it:seconddiagonalelementsargument}
For the diagonal elements, our aim is to argue that $\bra{\phi_{p}}F\ket{\phi_p}$ is almost independent of $p$, that is, we want to show $|\bra{\phi_p}F\ket{\phi_p}-\bra{\phi_{p'}}F\ket{\phi_{p'}}|$ is small for different momenta~$p\neq p'$. For this purpose, we need to identify the leading order term in the expression~$\bra{\phi_p}F\ket{\phi_p}$. Higher order terms are again small by the properties of the transfer operator.
\end{enumerate} 
To establish these bounds, first observe that an unnormalized excitation ansatz state $\ket{\Phi_p(B;A)}$ is a superposition of the ``position space" states $\{\ket{\Phi_{j,p}}\}_{j=1}^n$, where each state~$\ket{\Phi_{j,p}}$ is given by a simple tensor network with an ``insertion'' of an operator at site~$j'$. Correspondingly, we first study matrix elements of the form $\bra{\Phi_{j',p'}}F\ket{\Phi_{j,p}}$. Bounds on these matrix elements are given in Lemma~\ref{lem:maincauchyschwarzestimates}. The idea of the proof of this statement is simple: in the tensor network diagram for the matrix element, subdiagrams associated with powers~$E^\Delta$ with sufficiently large~$\Delta$ may be replaced by the diagram associated with the map~$\kket{r}\bbra{\ell}$, with an error scaling term scaling as~$O(\lambda_2^{\Delta/2})$. This is due to the Jordan decomposition of the transfer operator. Thanks to the gauge condition~\eqref{eq:gaugeconditionrewrittenexcit}, the resulting diagrams then simplify, allowing us to identify the leading order term.

To realize this approach, a key step is to identify suitable 
subdiagrams corresponding to powers~$E^\Delta$ in the 
diagram associated with~$\bra{\Phi_{j',p'}}F\ket{\Phi_{j,p}}$. These are associated with connected regions of size~$\Delta$ where the operator~$F$ acts trivially, and there is no insertion of $B(p)$ (respectively $\overline{B(p')}$), meaning that $j$ and $j'$ do not belong to the region.  Lemma~\ref{lem:maincauchyschwarzestimates}  provides a careful case-by-case analysis depending on, at the coarsest level of detail, whether or not $j$ and $j'$ belong to a $\Delta$-neighborhood of the support of~$F$. 

Some subleties that arise are the following: to obtain estimates on the  leading-order terms for the diagonal matrix elements (see~\eqref{it:seconddiagonalelementsargument} above) as well as related expressions,  a bound on the magnitude of the matrix element~$\bra{\Phi_{j,p'}}F\ket{\Phi_{j,p}}$ only is not sufficient. The lowest-order approximating expression to $\bra{\Phi_{j,p'}}F\ket{\Phi_{j,p}}$ obtained by making the above substitutions of the transfer operators a priori seems to depend on the exact site location~$j$. This is awkward because the term $\bra{\Phi_{j,p'}}F\ket{\Phi_{j,p}}$ appears as a summand (with sum taken over~$j$) when computing matrix elements of excitation ansatz states. We argue that in fact, the leading order term of 
$\bra{\Phi_{j,p'}}F\ket{\Phi_{j,p}}$ is identical for all values of~$j$ not belonging to the support of~$F$. This statement is formalized in Lemma~\ref{lem:mainmatrixelementexcit} and allows us to subsequently estimate sums of interest without worry about the explicit dependence on~$j$.

Finally, we require a strengthening of the estimates obtained in Lemma~\ref{lem:maincauchyschwarzestimates} because we are ultimately interested in excitation ansatz states: these are superpositions of the states~$\ket{\Phi_{j,p}}$, with phases of the form~$e^{ipj}$. Estimating only the magnitude of matrix elements of the form~$\bra{\Phi_{j',p'}}F\ket{\Phi_{j,p}}$ is not sufficient to establish our results. Instead, we need to treat the phases ``coherently'', which leads to certain cancellations. The corresponding statement is given in Lemma~\ref{lem:strengthenedboundsum}.

\subsubsection{The proof}
We will envision the sites ~$\{1,\ldots,n\}$ as points on a ring, i.e., using periodic boundary conditions, and measure the distance between sites $j,j'$ by
\[\mathsf{dist}(j,j'):=\min_{k\in \mathbb{Z}}|j-j'+k\cdot n|.\]
For $\Delta\in \{0,\ldots,n\}$ and a subset $\cF\subset \{1,\ldots,n\}$, let 
\begin{align}
\cB^\Delta(\cF)&=\{j\in \{1,\ldots,n\} \mid \exists\, j'\in \cF\textrm{ such that }\mathsf{dist}(j,j')\leq \Delta \}\ 
\end{align}
be the {\em $\Delta$-thickening} of~$\cF$.

We say that $j'\in \{1,\ldots,n\}$ is a {\em left neighbor of} (or is {\em left-adjacent to})~$j\in\{1,\ldots,n\}$ if $j'= j-1$ for $j>1$, or $j'=n$ for $j=1$. A connected region~$\cR\subset \{1,\ldots,n\}$ is said to lie on the left of (or be left-adjacent to) $j\in \{1,\ldots,n\}$  if it is of the form  $\cR=\{j_1,\ldots,j_r\}$, with $j_{\alpha+1}$ left-adjacent to $j_\alpha$ for $\alpha\in \{0,\ldots,r-1\}$ with the convention that $j_0=j_r$. Analogous definitions hold for right-adjacency.

For an operator~$F$ acting on $(\mathbb{C}^{\physical})^{\otimes n}$, let
$\mathsf{supp}(F)\subset \{1,\ldots,n\}$ denote its support, i.e., the sites of the system that the operator acts on non-trivially. We say that $F$ is $d$-local if $|\supp(F)| = d$. Let us assume that 
$\mathsf{supp}(F)$ decomposes into $\kappa$ disjoint connected components
\begin{align}
\mathsf{supp}(F)&=\bigcup_{\alpha=0}^{\kappa-1} \cF_\alpha\ .\mylabel{eq:connectedcomponentstildeF}
\end{align}
We may, without loss of generality, assume that 
this gives a partition of~$\{1,\ldots,n\}$ into disjoint connected sets
\begin{align}
\{1,\ldots,n\}&=\cA_0\cup\cF_0\cup\cA_1\cup\cF_1\cup\cdots\cup\cA_{\kappa-1}\cup\cF_{\kappa-1}
\end{align}
where $\cA_\alpha$ is left-adjacent to $\cF_\alpha$ for $\alpha\in \{0,\ldots,\kappa-1\}$, $\cA_{\alpha+1}$ is right-adjacent to $\cF_\alpha$ for $\alpha\in \{0,\ldots,\kappa-2\}$, and $\cA_0$ is right-adjacent to~$\cF_{\kappa-1}$. We may then decompose the operator $F$ as
\begin{align}
    F = \sum_{i}\bigotimes_{\alpha=0}^{\kappa-1} (I_{\cA_\alpha}\otimes F_{i,\alpha}),
\end{align}
where we write $F$ as a sum of decomposable tensor operators (indexed by $i$), with each $F_{i,\alpha}$ being an operator acting on the component $\cF_\alpha$.

Let us define a function~$\tau:\{1,\ldots,n\}\backslash \mathsf{supp}(F)\rightarrow\{0,\ldots,\kappa-1\}$ which associates to every site $j\not\in\supp(\cF)$ the unique index~$\tau(j)$ for the component~$\cA_{\tau(j)}$ of the complement of $\mathsf{supp}(F)$ such that $j\in\cA_{\tau(j)}$.
 
It is also convenient to introduce the following operators $\{F(\tau)\}_{\tau=0}^{\kappa-1}$. The operator $F(\tau)$ is obtained by 
removing the identity factor on the sites~$\cA_\tau$ of $F$, and cyclically permuting the remaining components in such a way that $\cF_\tau$ ends up on the sites $\{1,\ldots,|\cF_\tau|\}$. More precisely, we define
$F(\tau)\in\cB((\mathbb{C}^{\physical})^{\otimes (n-|\cA_\tau|)})$ by 
\begin{align}
F(\tau)&=\sum_{i}F_{i,\tau}\otimes \left(\bigotimes_{\alpha=\tau+1}^{\tau+\kappa-1} I^{\otimes |\cA_{\alpha\Mod\kappa }|}_{\mathbb{C}^{\physical}}\otimes F_{i,\alpha\Mod\kappa} \right), \mylabel{eq:tildeftaudefass}
\end{align} 
for $\tau\in \{0,\ldots,\kappa-1\}$. We note that $j\mapsto F(\tau(j))$ associates a permuted operator to each site $j$ not belonging to the support of~$F$. Let us also define $\iota(j)$ to be the index of the site which gets cyclically shifted to the first site when defining~$F_{\tau(j)}$. An example is shown diagrammatically in Figure~\ref{fig:ftaudefexample}.

\begin{figure}
\begin{align}
F&=\raisebox{-.52\height}{\includegraphics[scale=0.082]{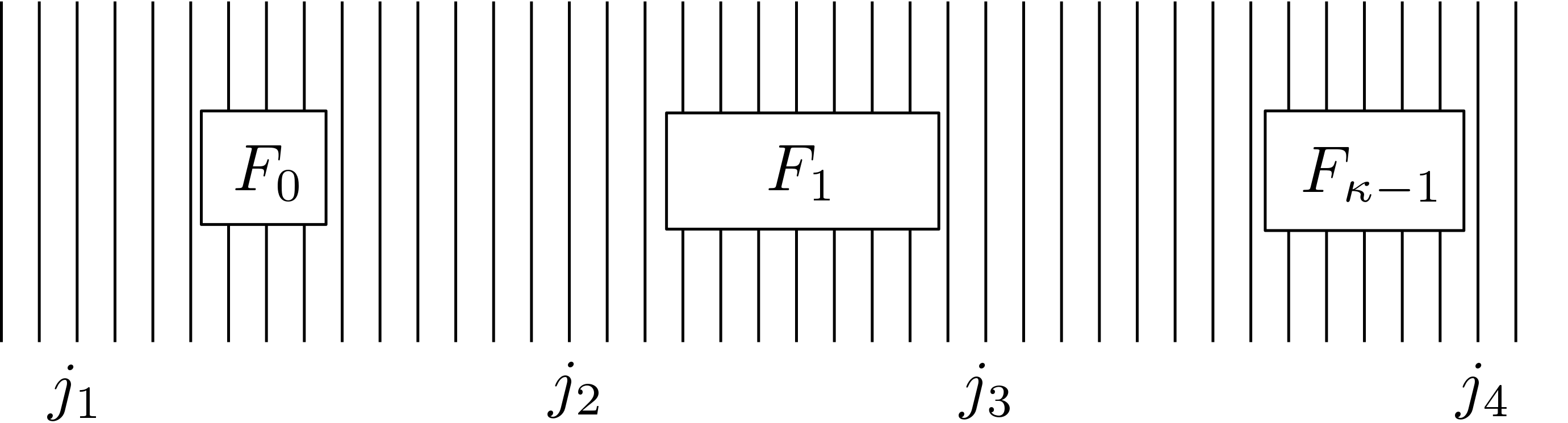}}\\
\\
F(\tau(j_1))=F(\tau(j_4))=F(0)&=\,\raisebox{-.44\height}{\includegraphics[scale=0.08]{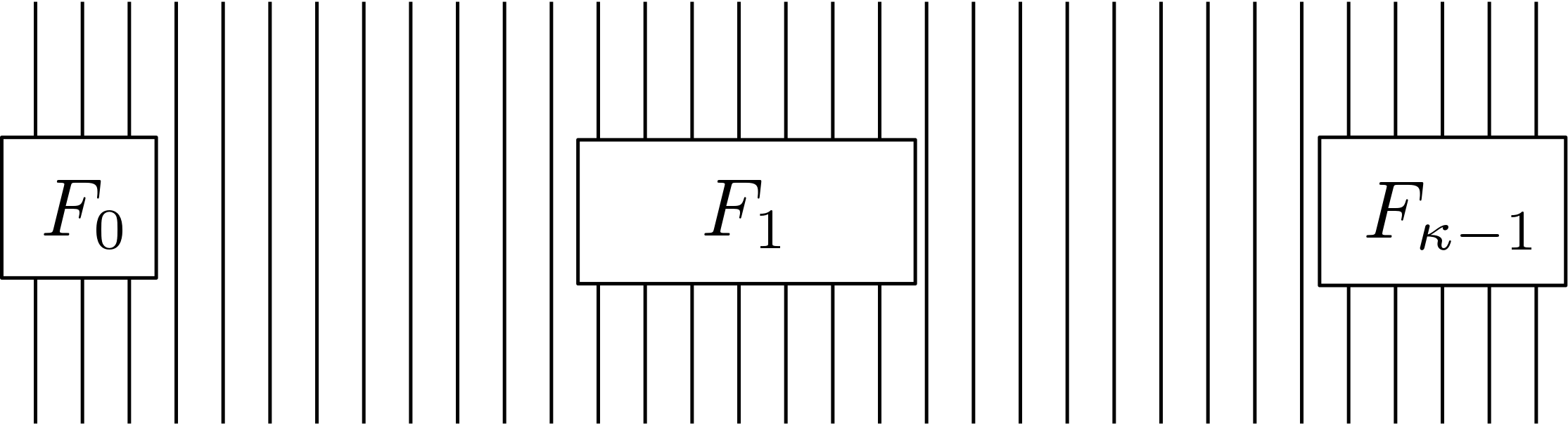}}\ ,\\
F(\tau(j_2))=F(1)&=\,\raisebox{-.44\height}{\includegraphics[scale=0.078]{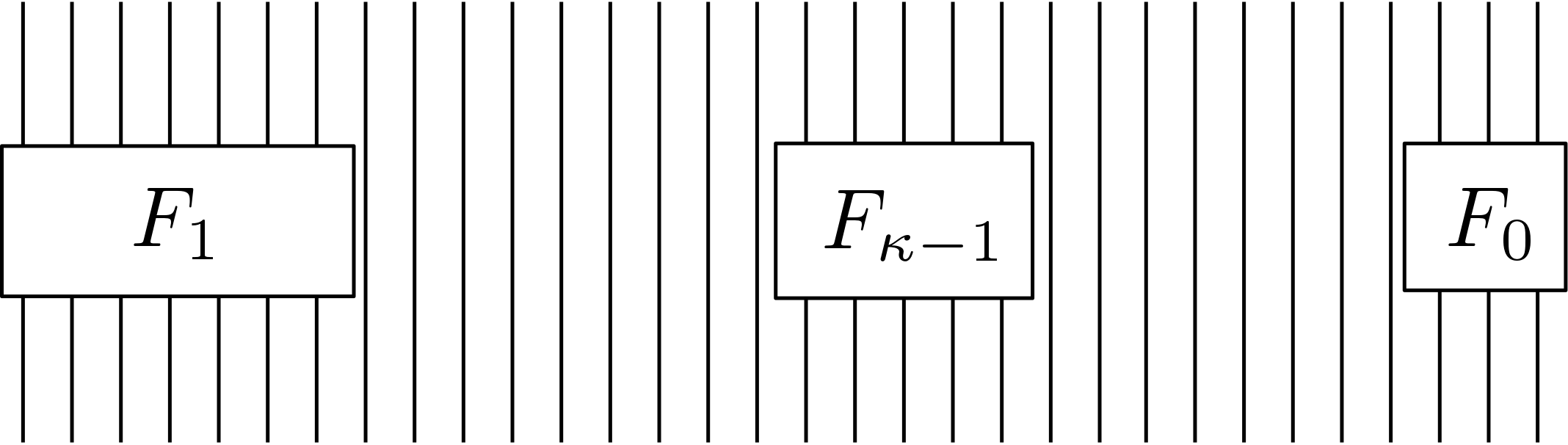}}\ ,\\
F(\tau(j_3))=F(\kappa-1)&=\,\raisebox{-.44\height}{\includegraphics[scale=0.08]{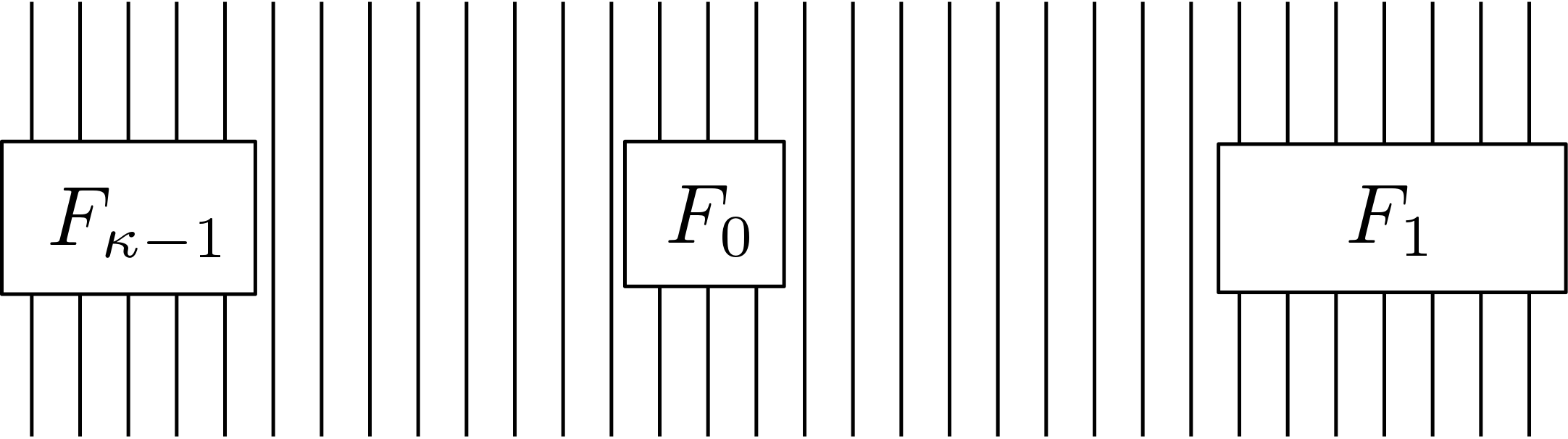}}\ .
\end{align}
\caption{Example for $F$ and sites $j_1,j_2,j_3,j_4 \in \{1,\ldots,n\}$ with $\iota(j_1)=\iota(j_4) = 7$, $\iota(j_2) = 19$, and $\iota(j_3) = 35$.}
\mylabel{fig:ftaudefexample}
\end{figure}

For two excitation ansatz states~$\ket{\Phi_p}$ and $\ket{\Phi_{p'}}$ ,
and an operator $F$ on $(\mathbb{C}^{\physical})^{\otimes n}$, we may write the corresponding matrix element as 
\begin{align}
\bra{\Phi_{p'}}F\ket{\Phi_{p}}&=
\sum_{j,j'=1}^n e^{i(pj-p'j')}
\bra{\Phi_{j',p'}}F\ket{\Phi_{j,p}}\ ,\mylabel{eq:sumphipprimephip}
\end{align}
where $\ket{\Phi_{j,p}}$ are the ``position space" states introduced in equation \eqref{eq:Phinpdefnitionmain}. We are interested in bounding the magnitude of this quantity.

We begin by bounding the individual terms in the sum~\eqref{eq:sumphipprimephip}.
\begin{lemma}\mylabel{lem:maincauchyschwarzestimates}
Let $j,j'\in\{1,\ldots,n\}$ and let $p,p'$ be arbitrary non-zero momenta. 
Consider the states~$\ket{\Phi_{j,p}}$ and $\ket{\Phi_{j',p'}}$ defined by~\eqref{eq:Phinpdefnitionmain}.  
Let $\Delta=\Delta(n)$ and $d=d(n)$ be monotonically increasing functions of $n$.  
Suppose further that we have
\begin{align}
10\Delta d < n\ .
\end{align}
Assume $F$ is a $d$-local operator of unit norm on 
$(\mathbb{C}^{\physical})^{\otimes n}$
whose support has~$\kappa$ connected components as in~\eqref{eq:connectedcomponentstildeF}. Then we have the following.\\

\begin{enumerate}[(i)]
\item\mylabel{it:a1bound} 
There is some fixed $q\in [n]$ such that for all $j,j'\in\cB^\Delta(\mathsf{supp}(F))$, we have
\begin{align}
\bra{\Phi_{j',p'}}F\ket{\Phi_{j,p}}&=
\bbra{\ell}
E_{I^{\otimes \Delta}\otimes F_{\tau(q)}\otimes I^{\otimes \Delta}}(\hat{j},p,\hat{j}',p')\kket{r}+ O(\lambda_2^{\Delta})\ , \mylabel{eq:overlapformulajfa1bound}
\end{align}
 where $\hat{j}=j-\iota(q)+\Delta+1 \Mod n$ and $\hat{j'}=j'-\iota(q)+\Delta+1\Mod n$. 
 
Furthermore,
\begin{align}
\langle \Phi_{j',p'}\ket{\Phi_{j,p}}=\Delta_{j,j'}c_{pp'}+O(\lambda_2^{\Delta})\ .
\mylabel{eq:overlapformulajprimepprimejp}
\end{align}
\item\mylabel{it:a2bound}
If $j,j'\not\in\cB^\Delta(\mathsf{supp}(F))$, then
\begin{enumerate}[(a)]
\item\mylabel{it:partaclaima2bound}
$|\bra{\Phi_{j',p'}}F\ket{\Phi_{j,p}}| = O(\lambda_2^{\Delta/2})$ if $j\neq j'$. 
\item\mylabel{it:partbclaima2bound}
$\bra{\Phi_{j',p'}}F\ket{\Phi_{j,p}}
=\bbra{\ell} E_{F(\tau(j))}\kket{r}
\cdot c_{pp'}+ O(\lambda_2^{\Delta/2})$. 

Here the operator~$F(\tau(j))$ is defined by equation~\eqref{eq:tildeftaudefass}. 
\end{enumerate}
\item\mylabel{it:a3bound}
If $j\in\cB^{\Delta}(\mathsf{supp}(F))$ and $j'\not\in\cB^\Delta(\mathsf{supp}(F))$, then
\begin{enumerate}[(a)]
\item\mylabel{it:a3bounda}
$|\bra{\Phi_{j',p'}}F\ket{\Phi_{j,p}}|= O( \lambda_2^{\Delta/2})$ if $j'\not\in\cB^{2\Delta}(\mathsf{supp}(F))$.
\item\mylabel{it:a3boundb}
There exists some fixed $q\in [n]$ such that, for all $j\in \cB^\Delta(\supp(F))$  \\
and $j' \in \cB^{2\Delta}(\supp(F))\backslash \cB^\Delta(\supp(F))$, we have
\begin{align}\bra{\Phi_{j',p'}}F\ket{\Phi_{j,p}}=
\bbra{\ell}E_{F(\tau(q))}(\hat{j},p,\hat{j'},p',2\Delta)\kket{r}+ O( \lambda_2^{2\Delta}),\end{align}
where $\hat{j}=j-\iota(q) + 2\Delta + 1 \Mod n$ and $\hat{j}'=j'-\iota(q) + 2\Delta + 1 \Mod n$. 
\end{enumerate}

\item\mylabel{it:a4bound}
If $j'\in\cB^{\Delta}(\mathsf{supp}(F))$ and $j\not\in\cB^\Delta(\mathsf{supp}(F))$, then
\begin{enumerate}[(a)]
\item\mylabel{it:a4bounda}
$|\bra{\Phi_{j',p'}}F\ket{\Phi_{j,p}}|= O(\lambda_2^{\Delta/2})$ if $j\not\in\cB^{2\Delta}(\mathsf{supp}(F))$.
\item\mylabel{it:a4boundb}
There exists some fixed $q\in [n]$ such that, for all $j'\in \cB^\Delta(\supp(F)))$ \\
and $j \in \cB^{2\Delta}(\supp(F))\backslash \cB^\Delta(\supp(F))$, we have
\begin{align}\bra{\Phi_{j',p'}}F\ket{\Phi_{j,p}}=
\bbra{\ell}E_{F(\tau(q))}(\hat{j},p,\hat{j'},p',2\Delta)\kket{r}+ O( \lambda_2^{2\Delta}),\end{align}
where $\hat{j}=j-\iota(q) + 2\Delta + 1 \Mod n$ and $\hat{j}'=j'-\iota(q) + 2\Delta + 1 \Mod n$. 
\end{enumerate}

\end{enumerate}
\end{lemma}

\begin{proof}
For the proof of~\eqref{it:a1bound}, suppose that $j,j'\in\cB^\Delta(\mathsf{supp}(F))$.  Pick any site $q\not\in\cB^{2\Delta}(\mathsf{supp}(F))$. We note that such a site always exists since 
\[|\cB^{2\Delta}(\mathsf{supp}(F))|\leq 5\Delta|\mathsf{supp}(F)|=5\Delta d < 10\Delta d < n\] 
by assumption. Let us define the shifted indices
\begin{align}
    \hat{j} = j - \iota(q) + \Delta + 1 \Mod n,\quad \text{and}\quad \hat{j}' = j' - \iota(q) + \Delta + 1 \Mod n.
\end{align}
Then we may write
\begin{align}
\bra{\Phi_{j',p'}}F\ket{\Phi_{j,p}}&=
\tr\left(E_F(j,p,j',p')\right)\\
&=\tr\left(E^s E_{I^{\otimes \Delta}\otimes F_{\tau(q)}\otimes I^{\otimes \Delta}}(\hat{j},p,\hat{j}',p')\right)\label{eq:phjpfjprm}
\end{align}
where $s\geq 2\Delta$. This is because by the choice of~$q$, there are at least~$2\Delta$ sites not belonging to~$\mathsf{supp}(F)$ both on the left and the right of~$q$.
Each of these~$4\Delta$ sites contributes a factor $E=E_I$ (i.e., a single transfer operator) to the expression
within the trace. 
The term~$E_{I^{\otimes \Delta}\otimes F_{\tau(q)}\otimes I^{\otimes \Delta}}(\hat{j},p,\hat{j}',p')$ incorporates~$\Delta$ of the  associated
transfer operators $E=E_I$
on the left- and right of $q$, respectively, such that at least $2\Delta$ factors of~$E$ remain. By the cyclicity of the trace, these can be consolidated into a single term~$E^s$ with~$s\geq 2\Delta$. The operator~$I^{\otimes \Delta}\otimes F_{\tau(q)}\otimes I^{\otimes \Delta}$ (i.e., the additional $I^{\otimes \Delta}$ factors) in the term $E_{I^{\otimes \Delta}\otimes F_{\tau(q)}\otimes I^{\otimes \Delta}}(\hat{j},p,\hat{j}',p')$ is used to ensure that~$j$ and $j'$ are correctly ``retained'' when going from the first to the second line in~\eqref{eq:phjpfjprm}. Inserting the Jordan decomposition~$E=\kket{r}\bbra{\ell}\oplus\tilde{E}$, we obtain
\begin{align}
\bra{\Phi_{j',p'}}F\ket{\Phi_{j,p}}
&=\bbra{\ell}
E_{I^{\otimes \Delta}\otimes F_{\tau(q)}\otimes I^{\otimes \Delta}}(\hat{j},p,\hat{j}',p')\kket{r}+\tr\left(\tilde{E}^{s} E_{I^{\otimes \Delta}\otimes F_{\tau(q)}\otimes I^{\otimes \Delta}}(\hat{j},p,\hat{j}',p')\right). \mylabel{eq:innerproductxybound}
\end{align}
By Lemma~\ref{lem:normscalingtransferop}(ii) and Lemma~\ref{lem:normboundtransopxnpnprime}, we have the bound
\begin{align}
\left|\tr\left(\tilde{E}^{s} 
E_{I^{\otimes \Delta}\otimes F_{\tau(q)}\otimes I^{\otimes \Delta}}(\hat{j},p,\hat{j}',p')\right)\right|&\leq  \|\tilde{E}^{s}\|_F\cdot \|
E_{I^{\otimes \Delta}\otimes F_{\tau(q)}\otimes I^{\otimes \Delta}}(\hat{j},p,\hat{j}',p')\|_F\\
&\leq \lambda_2^{{s/2}}\cdot D^2\|F\|\cdot\sqrt{\|E_{\overline{B(p')}B(p')}\|_F\|E_{\overline{B(p')}B(p)}\|_F}\\
& = O(\lambda_2^{\Delta}),
\end{align}
where we have used the fact that $\lambda_2^{s/2} \le \lambda_2^\Delta$ in the last line. We have also absorbed the dependence on the constants $D$, $\|F\|$, and $\sqrt{\|E_{\overline{B(p')}B(p')}\|_F\|E_{\overline{B(p')}B(p)}\|_F}$ into the big-O notation. Inserting this into~\eqref{eq:innerproductxybound} gives the first claim of~\eqref{it:a1bound}.

Now consider the inner product
$\langle \Phi_{j',p'}\ket{\Phi_{j,p}}=\tr(E(j,p,j',p'))$, which
corresponds to the case where $F$ is the identity. 
By the cyclicity of the trace, this can be written as $\langle \Phi_{j',p'}\ket{\Phi_{j,p}}=\tr(E^s E(\hat{j},p,\hat{j}',p'))$ for some $s\geq 2\Delta$ and suitably defined $\hat{j},\hat{j}'$. Repeating the same argument as above and using the fact that
\begin{align}
    \bbra{\ell}E(\hat{j},p,\hat{j}',p')\kket{r}&=\Delta_{\hat{j},\hat{j}'}c_{pp'}=\Delta_{j,j'}c_{pp'}
\end{align}
by definition of $E(\hat{j},p,\hat{j}',p')$, equation~\eqref{eq:leftrighteveqs} (i.e., the fact that $\kket{\ell}$ and $\kket{r}$ are left- respectively right eigenvectors of~$E$), and the gauge identities ~\eqref{eq:gaugeconditionrewrittenexcit} of $E_{B(p)}$ and $E_{\overline{B(p)}}$, we obtain the claim~\eqref{eq:overlapformulajprimepprimejp}.

Now consider claim~\eqref{it:a2bound}. Suppose  that $j,j'\not\in\cB^{\Delta}(\mathsf{supp}(F))$. We consider the following two cases:
\begin{enumerate}[(a)]
\item[\eqref{it:partaclaima2bound}]
If $j\neq j'$, then 
there is a connected region of at least~$\Delta$ sites not belonging to~$\mathsf{supp}(F)$ to either the left of~$j'$ and not containing $j$, or the left of~$j$ and not containing $j'$. Without loss of generality, we assume the former is the case. By the cyclicity of the trace, we may also assume without loss of generality that~$j'=\Delta+1$, $j>j'$, and that $F$ is supported on the sites~$\{2\Delta+2,\ldots,n\}$. Let $\hat{F}$ denote the restriction of~$F$ to the sites~$\{\Delta+2,\ldots,n\}$, and let $\hat{j}:=j-(\Delta+1)$. Then we may write
\begin{align}
\bra{\Phi_{j',p'}}F\ket{\Phi_{j,p}}&=\tr\left(E^{\Delta}E_{\overline{B(p')}}
E_{\hat{F}}(\hat{j},p)\right)\ .
\end{align}
Substituting the Jordan decomposition $E^\Delta=\kket{r}\bbra{\ell}\oplus\tilde{E}^\Delta$, we have
\begin{align}
\bra{\Phi_{j',p'}}F\ket{\Phi_{j,p}}&=\bbra{\ell}E_{\overline{B(p')}}E_{\hat{F}}(\hat{j},p)\kket{r}+
\tr\left(\tilde{E}^\Delta E_{\overline{B(p')}}E_{\hat{F}}(\hat{j},p)\right).
\end{align}
Since we assume that $p\neq 0$,  the gauge condition~\eqref{eq:gaugeconditionrewrittenexcit} states that~$\bbra{\ell}E_{\overline{B(p)}}=0$, hence the first term vanishes and it follows that
\begin{align}
 |\bra{\Phi_{j',p'}}F\ket{\Phi_{j,p}}|
 &= \left|\tr\left(\tilde{E}^\Delta E_{\overline{B(p')}}E_{\hat{F}}(\hat{j},p)\right)\right|\\
 &\le \|\tilde{E}^\Delta\|_F\cdot \|E_{\overline{B(p')}}\|_F\cdot \|E_{\hat{F}}(\hat{j},p)\|_F\\
 &\le \lambda_2^{\Delta/2} \|E_{\overline{B(p')}}\|_F\cdot D^2\|\hat{F}\|\sqrt{\|E_{\overline{B}(p)B(p)}\|_F}\\
 &= O(\lambda_2^{\Delta/2})\ ,
\end{align}
as claimed in~\eqref{it:partaclaima2bound}. In the last line, we have again absorbed the constants into the big-$O$-expression. This proves part~\eqref{it:partaclaima2bound} of Claim~\eqref{it:a2bound}.
 
\item[\eqref{it:partbclaima2bound}]
If $j=j'$, then there are at least $\Delta$ sites to the left and right of $j$ which do not belong to~$\mathsf{supp}(F)$. Therefore we may write 
\begin{align}
\bra{\Phi_{j',p'}}F\ket{\Phi_{j,p}}&=
\tr \left(E^s E_{\overline{B(p')}B(p)}E^t E_{F(\tau(j))}\right)\ ,
\end{align}
where $s$ and $t$ are integers greater than $\Delta$, representing the sites surrounding $j$ which are not in the support of $F$.

Applying the Jordan decomposition $E^\Delta=\kket{r}\bbra{\ell}\oplus\tilde{E}^\Delta$
twice (for $E^s$ and $E^t$) then gives four terms 
\begin{align}
\bra{\Phi_{j',p'}}F\ket{\Phi_{j,p}}&=
\bbra{\ell}E_{\overline{B(p')}B(p)}
\kket{r}\bbra{\ell}E_{F(\tau(j))}\kket{r}\\
&+\tr\left(\kket{r}\bbra{\ell}E_{\overline{B(p')}B(p)}\tilde{E}^{s} E_{F(\tau(j))}\right)\\
&+\tr\left(\tilde{E}^{t} E_{\overline{B(p')}B(p)}\kket{r}\bbra{\ell}E_{F(\tau(j))}\right)\\
&+\tr\left(
\tilde{E}^{t} E_{\overline{B(p')}B(p)}\tilde{E}^{s} E_{F(\tau(j))}\right)\ .
\end{align}
Since $s$ and $t$ are both larger than $\Delta$, by the same arguments from before, it is clear that the last three terms can each be bounded by $O(\lambda_2^{\Delta/2})$.
The claim follows since $\bbra{\ell}E_{\overline{B(p')}B(p)}\kket{r}=c_{pp'}$. 
\end{enumerate}

Next, we give the proof of claim~\eqref{it:a3bound}. Let us consider the situation where $j\in\cB^\Delta(\mathsf{supp}(F))$ and $j'\not\in\cB^\Delta(\mathsf{supp}(F))$. The proof of the other setting is analogous.  We consider two cases:
\begin{enumerate}[(a)]
\item[\eqref{it:a3bounda}]
Suppose $j'\not\in\cB^{2\Delta}(\mathsf{supp}(F))$. Let us define the shifted index $\hat{j} = j - \iota(j') + \Delta + 1 \Mod n$. Then we may write
\begin{align}
 \bra{\Phi_{j',p'}}F\ket{\Phi_{j,p}} &= \tr\left(E^s E_{\overline{B(p')}}E^t E_{I^{\otimes \Delta}\otimes F(\tau(j'))\otimes I^{\otimes \Delta}}(\hat{j},p)\right),
\end{align}
where $s$ and $t$ are integers larger than $\Delta$, representing the number of sites adjacent to $j'$ on the left and right which are not in $\cB^\Delta(\supp(F))$. We use the Jordan decomposition $E =\kket{r}\bbra{\ell}\oplus\tilde{E}$ on $E^s$ to get
\begin{align}
    \tr\left(E^s E_{\overline{B(p')}}E^t E_{I^{\otimes \Delta}\otimes F(\tau(j'))\otimes I^{\otimes \Delta}}(\hat{j},p)\right) &= \bbra{\ell}E_{\overline{B(p')}}E^t E_{I^{\otimes \Delta}\otimes F(\tau(j'))\otimes I^{\otimes \Delta}}(\hat{j},p)\kket{r}\\ 
    & \quad + \tr\left(\tilde{E}^s E_{\overline{B(p')}}E^t E_{I^{\otimes \Delta}\otimes F(\tau(j'))\otimes I^{\otimes \Delta}}(\hat{j},p)\right)\\
    &=\tr\left(\tilde{E}^s E_{\overline{B(p')}}E^t E_{I^{\otimes \Delta}\otimes F(\tau(j'))\otimes I^{\otimes \Delta}}(\hat{j},p)\right),
\end{align}
where the first term vanishes due to the gauge condition \eqref{eq:gaugeconditionrewrittenexcit}. From Lemma~\ref{lem:normscalingtransferop}(ii) we have $\|E^t\|_F \le 1$, and repeating the same arguments as before, we get the bound
\begin{align}
    |\bra{\Phi_{j',p'}}F\ket{\Phi_{j,p}}| &= \left|\tr\left(\tilde{E}^s E_{\overline{B(p')}}E^t E_{I^{\otimes \Delta}\otimes F(\tau(j'))\otimes I^{\otimes \Delta}}(\hat{j},p)\right)\right|\\
    &\le \|\tilde{E}^s\|_F \cdot \|E_{\overline{B(p')}}\|_F\cdot \|E^t\|_F \cdot D^2\|F\|\cdot \|E_{\overline{B(p)}B(p)}\|_F\\
    &\le \lambda_2^{s/2} \|E_{\overline{B(p')}}\|_F \cdot D^2\|F\|\sqrt{\|E_{\overline{B(p)}B(p)}\|_F}\ .
\end{align}
Since $s \ge \Delta$, we conclude that
\begin{align}
|\bra{\Phi_{j',p'}}F\ket{\Phi_{j,p}}|&=O\left(\lambda_2^{\Delta/2}\right)\ .
\end{align}

\item[\eqref{it:a3boundb}]
Suppose now that $j'\in\cB^{2\Delta}(\mathsf{supp}(F))$. Then by repeating the argument for case~\eqref{it:a1bound}, with $\Delta$ replaced by~$2\Delta$, we obtain
\begin{align}
\bra{\Phi_{j',p'}}F\ket{\Phi_{j,p}}=
\bbra{\ell}E_{I^{\otimes 2\Delta}\otimes F(\tau(q))\otimes I^{\otimes 2\Delta}}(\hat{j},p,\hat{j'},p')\kket{r}+O(\lambda_2^{2\Delta}),\ 
\end{align}
where  we now have $q\not\in\cB^{4\Delta}(\cF)$. Again, the existence of such a $q$ is guaranteed by the condition $10 \Delta d < n$.

We note that~\eqref{it:a4bound} follows immediately from~\eqref{it:a3bound} by interchanging the roles of $(j,p)$ and $(j',p')$. Note that we can write 
\begin{align}
\bra{\Phi_{j',p'}}F\ket{\Phi_{j,p}} &=
\overline{\overline{\bra{\Phi_{j',p'}}F\ket{\Phi_{j,p}}}}=\overline{\bra{\Phi_{j,p}}F^\dagger\ket{\Phi_{j',p'}}}\ .
\end{align}
The last expression within the parentheses is precisely what we had calculated in \eqref{it:a3bound}, so this implies the following:
\item[\eqref{it:a4bounda}]
If $j\notin \cB^{2\Delta}(\supp(F))$ then
\begin{align}
    \left|\bra{\Phi_{j',p'}}F\ket{\Phi_{j,p}}\right| &= \left|\bra{\Phi_{j,p}}F^\dagger\ket{\Phi_{j',p'}}\right|= O(\lambda_2^{\Delta/2}),
\end{align}
where we note that the exact same bound holds for $F$ and $F^\dagger$ since $\|F\|=\|F^\dagger\|$.

\item[\eqref{it:a4boundb}]
If $j \in \cB^{2\Delta}(\supp(F))$ then
\begin{align}
\bra{\Phi_{j',p'}}F\ket{\Phi_{j,p}}&=\overline{\bra{\Phi_{j,p}}F^\dagger\ket{\Phi_{j',p'}}}\\
&=\overline{\bbra{\ell}E_{I^{\otimes 2\Delta}\otimes F^\dagger(\tau(q))\otimes I^{\otimes 2\Delta}}(\hat{j}',p',\hat{j},p)\kket{r}} + O(\lambda_2^{2\Delta})\\
&=\overline{\bbra{\ell}E_{I^{\otimes 2\Delta}\otimes F^\dagger(\tau(q))\otimes I^{\otimes 2\Delta}}(\hat{j}',p',\hat{j},p)\kket{r}} + O(\lambda_2^{2\Delta})\\
&=\bbra{\ell}E_{I^{\otimes 2\Delta}\otimes F(\tau(q))\otimes I^{\otimes 2\Delta}}(\hat{j},p,\hat{j}',p')\kket{r} + O(\lambda_2^{2\Delta}).
\end{align}
This proves the claim.\footnote{To clarify how the term $\bbra{\ell}E_{I^{\otimes \Delta}\otimes F^\dagger(\tau(q)) \otimes I^{\otimes \Delta}}(\hat{j}',p',\hat{j},p)\kket{r}$ is complex conjugated, first write
\begin{align}
    \bbra{\ell}E_{I^{\otimes \Delta}\otimes F^\dagger(\tau(q)) \otimes I^{\otimes \Delta}}(\hat{j}',p',\hat{j},p)\kket{r} &= \bra{\Phi^L_{\hat{j}',p'}}I\otimes I^{\otimes 2\Delta}\otimes F_{\tau(q)}^\dagger \otimes I^{\otimes 2\Delta} \otimes I\ket{\Phi^L_{\hat{j},p}}, 
\end{align}
where $\ket{\Phi^L_{\hat{j},p}}$ are the states defined by \eqref{eq:phinpaux}, for some appropriate length $L$. Then we can proceed to conjugate the matrix element, giving us
\begin{align}
    \overline{\bra{\Phi^L_{\hat{j}',p'}}I\otimes I^{\otimes 2\Delta}\otimes F_{\tau(q)}^\dagger \otimes I^{\otimes 2\Delta} \otimes I\ket{\Phi^L_{\hat{j},p}}} &= \bra{\Phi^L_{\hat{j},p}}I\otimes I^{2\Delta}\otimes F_{\tau(q)} \otimes I^{2\Delta} \otimes I\ket{\Phi^L_{\hat{j}',p'}}\\
    &= \bbra{\ell}E_{I^{\otimes 2\Delta} \otimes F(\tau(q)) \otimes I^{\otimes 2\Delta}}(\hat{j},p,\hat{j}',p')\kket{r}.
\end{align}}
\end{enumerate}
\end{proof}

Note that in the statement~\eqref{it:partbclaima2bound},
the dependence on $j$ in the expression~$\bbra{\ell} E_{F(\tau(j))}\kket{r}$ can be eliminated as follows:
\begin{lemma}\mylabel{lem:mainmatrixelementexcit}
Suppose $j_1,j_2\not\in\cB^\Delta(\mathsf{supp}(F))$. Then 
\begin{align}
|\bbra{\ell} E_{F(\tau(j_1))}\kket{r}-
\bbra{\ell} E_{F(\tau(j_2))}\kket{r}|&=O(\lambda_2^{\Delta})\ .\mylabel{eq:firstmodifiedxlemma}
\end{align}
In particular, for any fixed $j_0 \not\in\cB^\Delta(\mathsf{supp}(F))$ we have 
\begin{align}
\bra{\Phi_{j,p'}}F\ket{\Phi_{j,p}}
=\bbra{\ell} E_{F(\tau(j_0))}\kket{r}
\cdot c_{pp'}+O(\lambda_2^{\Delta/2})\ ,\quad\textrm{for all }j\not\in\cB^\Delta(\mathsf{supp}(F))\ .\mylabel{eq:secondclaimmodifiedxlemma}
\end{align}
\end{lemma}
\begin{proof}
The claim~\eqref{eq:secondclaimmodifiedxlemma} follows immediately from~\eqref{eq:firstmodifiedxlemma} and claim~\eqref{it:partbclaima2bound} of Lemma~\ref{lem:maincauchyschwarzestimates} since $|c_{pp'}|=O(1)$.

If $\tau(j_1)=\tau(j_2)$, there is nothing to prove. Suppose $\tau(j_1)\neq \tau(j_2)$. Without loss of generality, assume that $\tau(j_1)=0$ and $\tau(j_2)=\xi$. Then we may write
\begin{align}
F(\tau(j_1))&=\sum_{i}F_{i,0}\otimes I^{\otimes a_1}\otimes
F_{i,1}\otimes I^{\otimes a_2}\cdots \otimes I^{\otimes a_{\kappa - 1}}\otimes F_{i,\kappa - 1},\quad\text{and}\\
F(\tau(j_2))&=\sum_{i}F_{i,\xi}\otimes I^{\otimes a_{\xi+1}}\otimes
F_{i,\xi+1}\otimes I^{\otimes a_{\xi+2}}\cdots \otimes I^{\otimes a_{\kappa}}\otimes F_{i,\kappa - 1}\otimes I^{\otimes a_0}\\
&\qquad\otimes F_{i,0}\otimes I^{\otimes a_1}\otimes F_{i,1}\otimes I^{\otimes a_2}\otimes\cdots \otimes F_{i,\xi-1},
\end{align}
where $a_\alpha=|\cA_\alpha|$ for $\alpha\in \{0,\ldots,\kappa\}$. Defining the operators
\begin{align}
\hat{F}_i&=F_{i,\xi}\otimes I^{\otimes a_{\xi+1}}\otimes
F_{i,\xi+1}\otimes I^{\otimes a_{\xi+2}}\cdots \otimes I^{\otimes a_{\kappa - 1}}\otimes F_{i,\kappa - 1}\ ,\\
\hat{G}_i&=F_{i,0}\otimes I^{\otimes a_1}\otimes F_{i,1}\otimes I^{\otimes a_2}\otimes\cdots \otimes F_{i,\xi-1}\ ,
\end{align}
we have 
\begin{align}
F(\tau(j_1))=\sum_i\hat{G}_i\otimes I^{\otimes a_\xi}\otimes\hat{F}_i,\quad\text{and}\quad F(\tau(j_2))=\sum_i\hat{F}_i\otimes I^{\otimes a_0}\otimes \hat{G}_i\ .
\end{align}

\begin{figure}
\begin{align}
F&=\raisebox{-.38\height}{\includegraphics[scale=0.6]{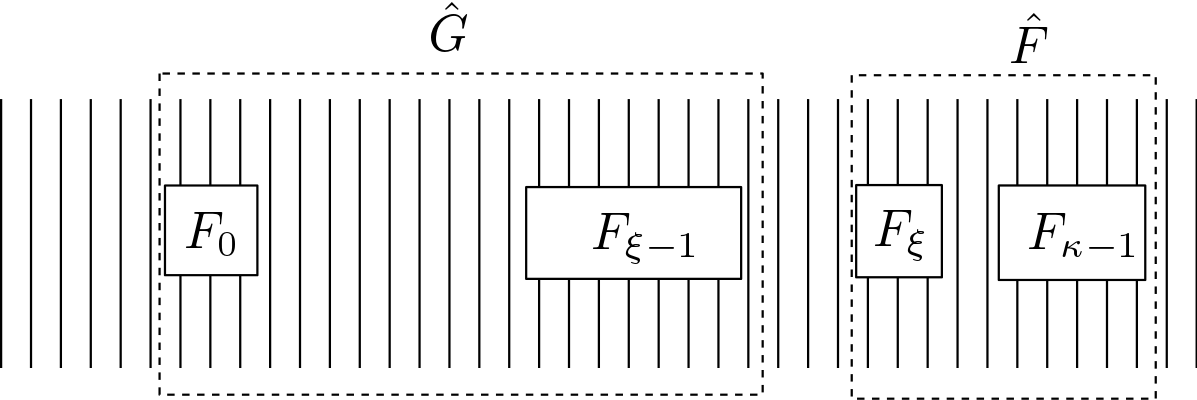}}\ ,\\
F(\tau(j_1))&=\raisebox{-.38\height}{\includegraphics[scale=0.6]{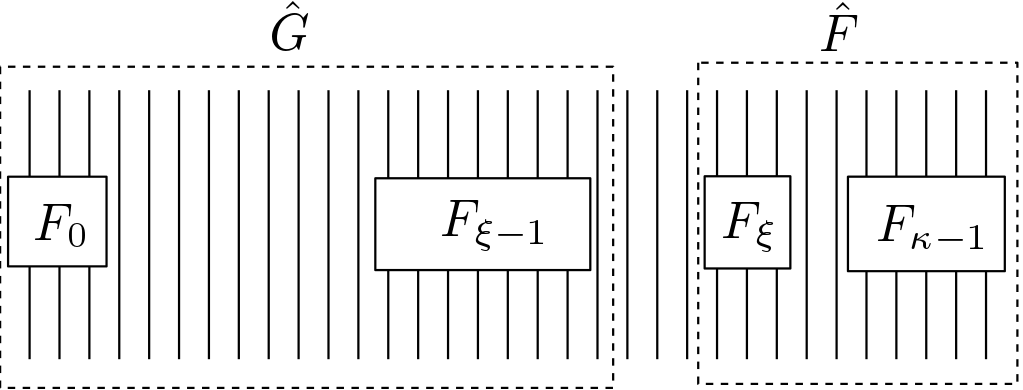}}\ ,\\
F(\tau(j_2))&=\raisebox{-.38\height}{\includegraphics[scale=0.6]{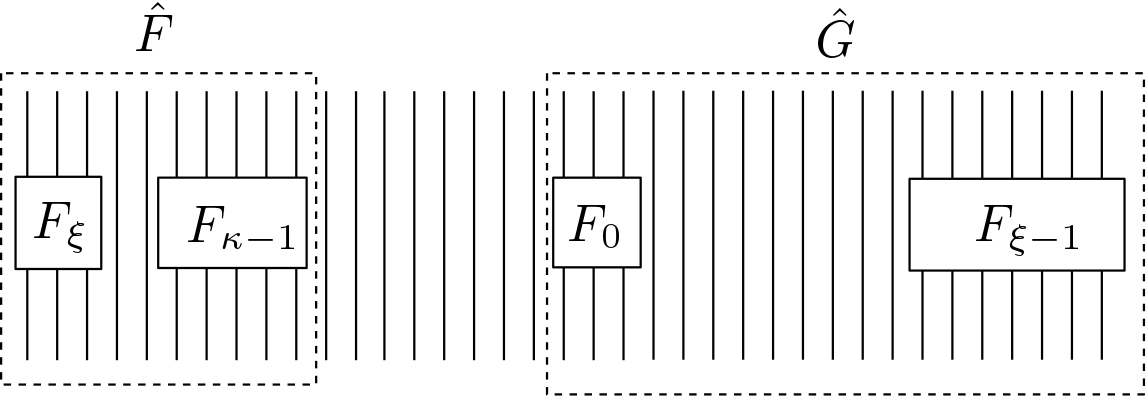}}\ .
 \end{align}
 \caption{Example for the operator $F$ and the corresponding $F(\tau(j_1))$ and $F(\tau(j_2))$.}
 \mylabel{fig:Fcombined}
\end{figure}
\noindent
(We give an example for the operator $F$, $F(\tau(j_1))$ and $F(\tau(j_2))$ in Figure~\ref{fig:Fcombined}.) Therefore we can write 
\begin{align}
\bbra{\ell} E_{F(\tau(j_1))}\kket{r}&=\sum_i\bbra{\ell} E_{\hat{G}_i}E^{a_\xi}E_{\hat{F}_i}\kket{r}\ ,\\
\bbra{\ell} E_{F(\tau(j_2))}\kket{r}&=\sum_i\bbra{\ell} E_{\hat{F}_i}E^{a_0}E_{\hat{G}_i}\kket{r}\ .
\end{align}
Inserting the Jordan decomposition~$E=\kket{r}\bbra{\ell}\oplus\tilde{E}$ gives 
\begin{align}
\bbra{\ell} E_{F(\tau(j_1))}\kket{r}&=\sum_i\left(\bbra{\ell} E_{\hat{G}_i}\kket{r}\bbra{\ell}E_{\hat{F}_i}\kket{r}
+\bbra{\ell} E_{\hat{G}_i}\tilde{E}^{a_\xi}E_{\hat{F}_i}\kket{r}\right)\ ,\\
\bbra{\ell} E_{F(\tau(j_2))}\kket{r}&=\sum_i\left(\bbra{\ell} E_{\hat{F}_i}\kket{r}\bbra{\ell}E_{\hat{G}_i}\kket{r}
+\bbra{\ell} E_{\hat{F}_i}\tilde{E}^{a_0}E_{\hat{G}_i}\kket{r}\right)\ .
\end{align}
Taking the difference, the first terms of the sums cancel, and we are left with
\begin{align}
\big|\bbra{\ell} E_{F(\tau(j_1))}\kket{r}-\bbra{\ell} E_{F(\tau(j_2))}\kket{r}\big|&= \left|\sum_i\bbra{\ell} E_{\hat{G}_i}\tilde{E}^{a_\xi}E_{\hat{F}_i}\kket{r}-\sum_i\bbra{\ell} E_{\hat{F}_i}\tilde{E}^{a_0}E_{\hat{G}_i}\kket{r}\right|\nonumber\\
&\le \left|\sum_i\bbra{\ell}  E_{\hat{G}_i}\tilde{E}^{a_\xi}E_{\hat{F}_i}\kket{r}\right|+\left|\sum_i\bbra{\ell} E_{\hat{F}_i}\tilde{E}^{a_0}E_{\hat{G}_i}\kket{r}\right|\ .\label{eq:secondexpressiontempor}
\end{align}
We can bound the first term $\left|\sum_i\bbra{\ell}  E_{\hat{G}_i}\tilde{E}^{a_\xi}E_{\hat{F}_i}\kket{r}\right|$ as follows. First, we write
\begin{align}
    \left|\sum_i\bbra{\ell}  E_{\hat{G}_i}\tilde{E}^{a_\xi}E_{\hat{F}_i}\kket{r}\right| &= \tr\left(\tilde{E}^{a_\xi}\sum_i E_{\hat{F}_i}\kket{r}\bbra{\ell}E_{\hat{G}_i}\right)\\
    &\le \|\tilde{E}^{a_\xi}\|_F \left\|\sum_i E_{\hat{F}_i}\kket{r}\bbra{\ell}E_{\hat{G}_i}\right\|_F\\
    &\le \lambda_2^\Delta \left\|\sum_i E_{\hat{F}_i}\kket{r}\bbra{\ell}E_{\hat{G}_i}\right\|_F\ ,
\end{align}
where the last inequality comes from the fact that $j_2\not\in\cB^\Delta(\mathsf{supp}(F))$ and $j_2\in\cA_\xi$ implies that $a_\xi\geq 2\Delta$, so Lemma~\ref{lem:normscalingtransferop}(ii) gives $\|\tilde{E}^{a_\xi}\|_F \le \lambda_2^\Delta$. Proceeding as we did in the proof of Lemma~\ref{lem:normboundtransopxnpnprime}, we can write the latter Frobenius norm as
\begin{align}
    \left\|\sum_i E_{\hat{F}_i}\kket{r}\bbra{\ell}E_{\hat{G}_i}\right\|_F^2 = \sum_{\alpha_1,\alpha_2,\beta_1,\beta_2=1}^D\left|\bra{\alpha_1}\bra{\alpha_2}\left(\sum_i E_{\hat{F}_i}\kket{r}\bbra{\ell}E_{\hat{G}_i}\right)\ket{\beta_1}\ket{\beta_2}\right|^2.
\end{align}
The individual terms in the sum can be depicted diagrammatically as
\begin{align}
    \bra{\alpha_1}\bra{\alpha_2}\left(\sum_i E_{\hat{F}_i}\kket{r}\bbra{\ell}E_{\hat{G}_i}\right)\ket{\beta_1}\ket{\beta_2} = \raisebox{-.45\height}{\includegraphics[width=9cm]{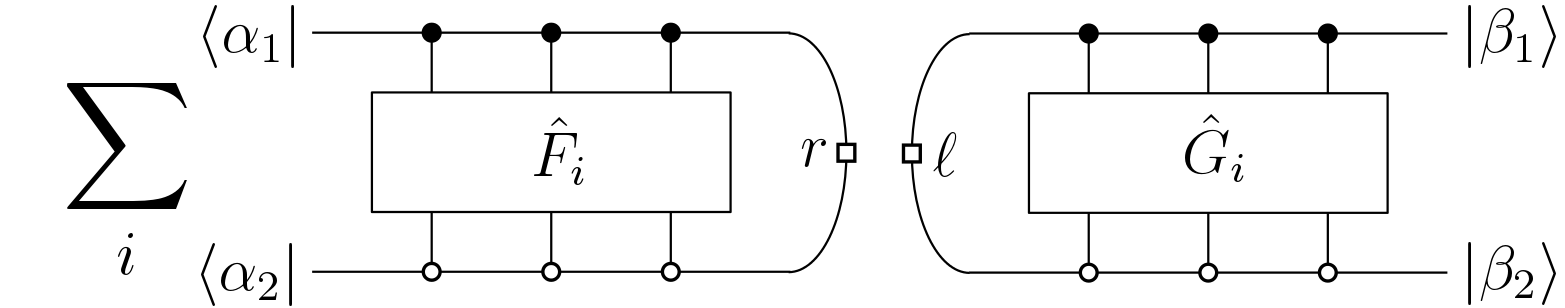}}\ .
\end{align}
Defining the vectors
\begin{align}
    \ket{\Psi(\alpha,\beta)} =\enskip \raisebox{-.45\height}{\includegraphics[scale=0.42]{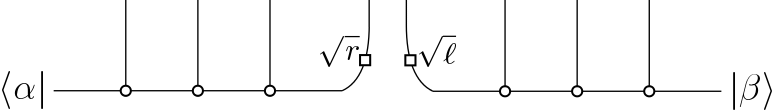}}\ ,
\end{align}
we can then write
\begin{align}
    \bra{\alpha_1}\bra{\alpha_2}\left(\sum_i E_{\hat{F}_i}\kket{r}\bbra{\ell}E_{\hat{G}_i}\right)\ket{\beta_1}\ket{\beta_2} = \bra{\Psi(\alpha_1,\beta_1)}\left(\sum_i \hat{F}_i\otimes I_D\otimes I_D \otimes \hat{G}_i\right) \ket{\Psi(\alpha_2,\beta_2)}.
\end{align}
Applying the Cauchy-Schwarz inequality, we get
\begin{align}
    \left|\bra{\alpha_1}\bra{\alpha_2}\left(\sum_i E_{\hat{F}_i}\kket{r}\bbra{\ell}E_{\hat{G}_i}\right)\ket{\beta_1}\ket{\beta_2}\right|^2 &\le \|\Psi(\alpha_1,\beta_1)\|^2\cdot \|\Psi(\alpha_2,\beta_2)\|^2\cdot \left\|\sum_i \hat{F}_i \otimes I_D\otimes I_D \otimes \hat{G}_i\right\|^2.
\end{align}
The norm of the vector $\ket{\Psi(\alpha,\beta)}$ is given by
\begin{align}
    \|\Psi(\alpha,\beta)\|^2 = \raisebox{-.45\height}{\includegraphics[scale=0.38]{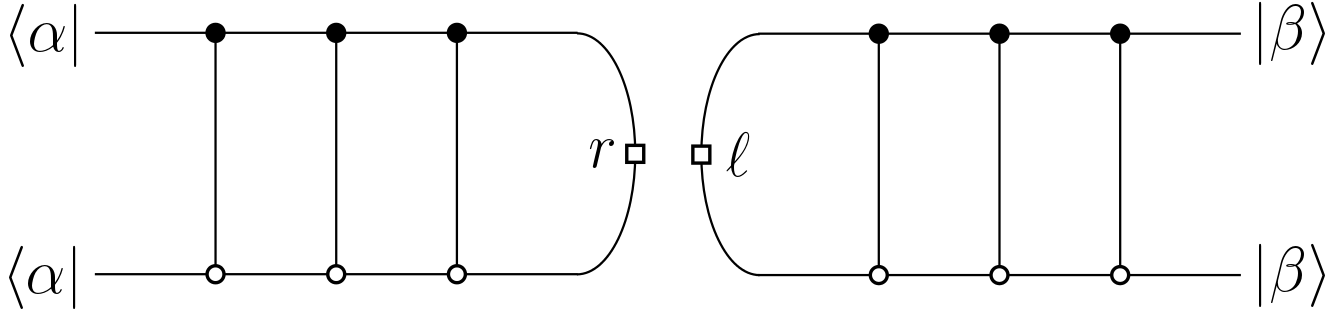}}
    = \raisebox{-.45\height}{\includegraphics[scale=0.40]{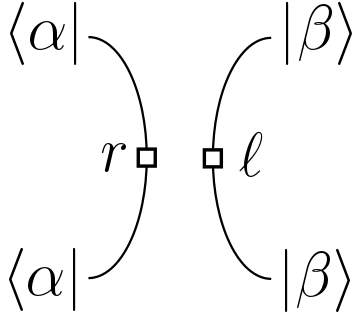}}
    = \bra{\alpha}r\ket{\alpha}\bra{\beta}\ell\ket{\beta}\ ,
\end{align}
where in the second equality we have used the fixed-point equations \eqref{eq:leftrighteveqs}. Therefore we have
\begin{align}
    \left\|\sum_i E_{\hat{F}_i}\kket{r}\bbra{\ell}E_{\hat{G}_i}\right\|_F^2 &\le  \left\|\sum_i \hat{F}_i \otimes I_D\otimes I_D \otimes \hat{G}_i\right\|^2\sum_{\alpha_1,\alpha_2,\beta_1,\beta_2=1}^D \bra{\alpha_1}r\ket{\alpha_1}\bra{\alpha_2}r\ket{\alpha_2}\bra{\beta_1}\ell\ket{\beta_1}\bra{\beta_2}\ell\ket{\beta_2}\\
    &=\left\|\sum_i \hat{F}_i \otimes I_D\otimes I_D \otimes \hat{G}_i\right\|^2\cdot |\tr(r)\tr(\ell)|^2
    = D^2\left\|\sum_i \hat{F}_i \otimes I_D\otimes I_D \otimes \hat{G}_i\right\|^2,
\end{align}
where the last equality follows from the fact that we gauge-fix the left and right fixed-points such that $r = I_{\mathbb{C}^D}$ and $\tr(\ell)=1$. Finally, we note that since the operator norm is multiplicative over tensor products, i.e., $\|A\otimes B\|=\|A\|\cdot \|B\|$, we have
\begin{align}
    \left\|\sum_i \hat{F}_i \otimes I_D\otimes I_D \otimes \hat{G}_i\right\| = \left\|\sum_i \hat{F}_i \otimes \hat{G}_i\right\| = \|F\|.
\end{align}
Therefore, we have
\begin{align}
   \left|\sum_i\bbra{\ell}  E_{\hat{G}_i}\tilde{E}^{a_\xi}E_{\hat{F}_i}\kket{r}\right| \le D\|F\|\lambda_2^\Delta.
\end{align}
The term involving $a_0$ in~\eqref{eq:secondexpressiontempor} can be bounded identically, and so
\begin{align}
    \big|\bbra{\ell} E_{F(\tau(j_1))}\kket{r}-\bbra{\ell} E_{F(\tau(j_2))}\kket{r}\big|\le 2D\|F\|\lambda_2^\Delta\ ,
\end{align}
which proves~\eqref{eq:firstmodifiedxlemma}.
\end{proof}

We also need a different version of statement~\eqref{it:a1bound}, as well as statements~\eqref{it:a3boundb} and~\eqref{it:a4boundb}  derived from it.

\begin{lemma}\mylabel{lem:strengthenedboundsum}
For $\Omega\subset [n]^2$, let us define
\begin{align}
\sigma_{pp'}(\Omega)&=\sum_{(j,j')\in \Omega}e^{i(pj-p'j')}\bra{\Phi_{j',p'}}F\ket{\Phi_{j,p}}\ .\mylabel{eqsigmaomegadef}
\end{align}
Let us write~$\cF:=\mathsf{supp}(F)$ and $\cA^c=[n]\backslash \cA$ for the complement of a subset~$\cA\subset [n]$.  Then: 
\begin{align}
\left|\sigma_{pp'}(\cB^\Delta(\cF)\times \cB^\Delta(\cF))\right|&\leq |\cB^\Delta(\cF)|\cdot \|F\| \sqrt{c_pc_{p'}}+O\left(\sqrt{n}\,\lambda_2^{\Delta/2}\right)\, ,\label{it:firststrengthened}\\
\left|\sigma_{pp'}(\cB^\Delta(\cF)\times \cB^{\Delta}(\cF)^c)\right|
&\leq |\cB^{2\Delta}(\cF)|\cdot \|F\| \sqrt{c_pc_{p'}}+ O\left(n^2\lambda_2^{\Delta/2}\right)\, ,\label{it:secondstrengthened}\\
\left|\sigma_{pp'}(\cB^\Delta(\cF)^c\times \cB^{\Delta}(\cF))\right|
&\leq |\cB^{2\Delta}(\cF)|\cdot \|F\| \sqrt{c_pc_{p'}}+ O\left(n^2\lambda_2^{\Delta/2}\right)\, .\label{it:secondstrengthenedb}
\end{align}
Finally, we have the following: There exists some fixed $j_0\in [n]$ such that for $p=p'$, we have
\begin{align}
\sigma_{pp}(\cB^\Delta(\cF)^c\times\cB^\Delta(\cF)^c)&= |\cB^\Delta(\cF)^c|\cdot \bbra{\ell}E_{F_{\tau(j_0)}}\kket{r}c_p+O\left(n^2\lambda_2^{\Delta/2}\right).\label{eq:sigmappscal}
\end{align}
For $p\neq p'$, we have
\begin{align}
\left|\sigma_{pp'}(\cB^\Delta(\cF)^c\times\cB^\Delta(\cF)^c)\right|
&\leq |\cB^\Delta(\cF)|\cdot \|F\|\sqrt{c_{p}c_{p'}} +O\left(n^2\lambda_2^{\Delta/2}\right)\ . \label{it:thirdstrengthened}
\end{align}
\end{lemma}
We observe that the first expression on the right-hand side of the above bound scales linearly with the support size of~$\cF$ instead of the support size of~$\cF^c$, as may be naively expected.  For~\eqref{it:thirdstrengthened}, this is due to a cancellation of phases, see~\eqref{eq:cancellationphases} below. 
\begin{proof}
For the proof of~\eqref{it:firststrengthened}, let us first define the vectors
\begin{align}
\ket{\Psi(p)}&=\sum_{j\in\cB^{\Delta}(\cF)} e^{ipj}\ket{\Phi_{j,p}}\ .
\end{align}
Then we can write
\begin{align}
|\sigma_{pp'}(\cB^\Delta(\cF)\times \cB^\Delta(\cF))|
&=|\bra{\Psi(p')} F\ket{\Psi(p)}|\leq  \|F\|\cdot \|\Psi(p)\|\cdot \|\Psi(p')\|\ ,\label{eq:upperboundsigmadeltadelta}
\end{align}
where the last inequality follows by Cauchy-Schwarz along with the definition of the operator norm~$\|F\|$.  The vector norm is given by
\begin{align}
\|\Psi(p)\|^2=\sum_{j,j'\in \cB^\Delta(\cF) }e^{ip(j-j')}\langle \Phi_{j',p}|\Phi_{j,p}\rangle,
\end{align}
and together with equation~\eqref{eq:overlapformulajprimepprimejp}, we get
\begin{align}
\|\Psi(p)\|^2&=|\cB^{\Delta}(\cF)|\cdot c_p +O(\lambda_2^{\Delta/2})\ .
\end{align}
Taking the square root and inserting into equation~\eqref{eq:upperboundsigmadeltadelta}, we get
\begin{align}
    |\sigma_{pp'}(\cB^\Delta(\cF)\times \cB^\Delta(\cF))|&=\|F\|\left(\sqrt{|\cB^{\Delta}(\cF)|\cdot c_p} +O(\lambda_2^{\Delta/2})\right)\left(\sqrt{|\cB^{\Delta}(\cF)|\cdot c_{p'}} +O(\lambda_2^{\Delta/2})\right)\\
    &=|\cB^{\Delta}(\cF)|\cdot\|F\| \sqrt{c_pc_{p'}} + O\left(\sqrt{|\cB^\Delta(\cF)|}\cdot \lambda_2^{\Delta/2}\right).
\end{align}
Using the bound $\left|\cB^\Delta(\cF)\right| \le 5d \Delta < n$ gives~\eqref{it:firststrengthened}. 

Next, let us look at~\eqref{it:secondstrengthened}. We have
\begin{align}
\sigma_{pp'}(\cB^\Delta(\cF)\times \cB^{\Delta}(\cF)^c)&=
\sum_{j\in\cB^\Delta(\cF)}\sum_{j'\in\cB^\Delta(\cF)^c}e^{i(pj-p'j')}\bra{\Phi_{j',p'}}F\ket{\Phi_{j,p}}=\Sigma_1+\Sigma_2,
\end{align}
where we define
\begin{align}
    \Sigma_1:=\sum_{j\in\cB^\Delta(\cF)}\sum_{j'\in \cB^{2\Delta}(\cF)\backslash \cB^\Delta(\cF)} e^{i(pj-p'j')}\bra{\Phi_{j',p'}}F\ket{\Phi_{j,p}}\ ,
\end{align}
and
\begin{align}
    \Sigma_2&:=\sum_{j\in\cB^\Delta(\cF)}\sum_{j'\in \cB^{2\Delta}(\cF)^c}e^{i(pj-p'j')}\bra{\Phi_{j',p'}}F\ket{\Phi_{j,p}}\ .
\end{align}
The norm of the second sum can be bounded using Lemma~\ref{lem:maincauchyschwarzestimates}\eqref{it:a3bounda}, giving us
\begin{align}
|\Sigma_2| &\leq \sum_{j\in\cB^\Delta(\cF)}\sum_{j'\in \cB^{2\Delta}(\cF)^c} |\bra{\Phi_{j',p'}}F\ket{\Phi_{j,p}}|\\
&\leq |\cB^\Delta(\cF)|\cdot |\cB^{2\Delta}(\cF)^c|\cdot O(\lambda_2^{\Delta/2})\\
&=O(n^2\lambda_2^{\Delta/2})\ ,\label{eq:normsecondsigma2}
\end{align}
where we again use the trivial bound $\left|\cB^\Delta(\cF)\right|,\left|\cB^{2\Delta}(\cF)^c\right| \le n$ in the last line. Using Lemma~\ref{lem:maincauchyschwarzestimates}~\eqref{it:a3boundb}, we can express the first sum, with some fixed $q\in [n]$, as 
\begin{align}
\Sigma_1&=\sum_{j\in\cB^{\Delta}(\cF)}\sum_{j'\in\cB^{2\Delta}(\cF)\backslash \cB^\Delta(\cF)}
e^{i(pj-p'j')} \bbra{\ell} E_{F(\tau(q))}(\hat{j},p,\hat{j'},p')\kket{r}\\
&\qquad +|\cB^{\Delta}(\cF)|\cdot |\cB^{2\Delta}(\cF)\backslash \cB^{\Delta}(\cF)|\cdot O(\lambda_2^{\Delta/2})\\
&= \sum_{j\in\cB^{\Delta}(\cF)}\sum_{j'\in\cB^{2\Delta}(\cF)\backslash \cB^\Delta(\cF)}
e^{i(pj-p'j')} \bbra{\ell} E_{F(\tau(q))}(\hat{j},p,\hat{j'},p')\kket{r} + O\left(n^2\lambda_2^{\Delta/2}\right),
\end{align}
where the indices $\hat{j}$ and $\hat{j}'$ are defined as in Lemma~\ref{lem:maincauchyschwarzestimates}. To bound the remaining sum, let us introduce the states
\begin{align}
\ket{\Psi_1(p)}&:=\sum_{j\in\cB^{\Delta}(\cF)} e^{ipj}\ket{\Phi^L_{\hat{j},p}}\ ,\qquad\text{and}\\
\ket{\Psi_2(p')}&:=\sum_{j'\in\cB^{2\Delta}(\cF)\backslash \cB^\Delta(\cF)} e^{ip'j'}\ket{\Phi^L_{\hat{j}',p'}}\ ,
\end{align}
where we set $L=|\mathsf{supp}(F(\tau(q)))|$. Here, $\ket{\Phi_{j,p}^L}$ are as defined in~\eqref{eq:phinpaux}. Then we can write
\begin{align}
\Sigma_1 &= \bra{\Psi_2(p')}F(\tau(q))\ket{\Psi_1(p)}+O\left(n^2\lambda_2^{\Delta/2}\right).
\end{align}
By the Cauchy-Schwarz inequality and the orthogonality relations~\eqref{eq:innerproductnpnpprimestatement}, we have 
\begin{align}
|\bra{\Psi_2(p')}F(\tau(q))\ket{\Psi_1(p)}|&\leq \|F\|\cdot \|\Psi_1(p)\|\cdot\|\Psi_2(p')\|\\
&=\|F\|\ \sqrt{c_pc_{p'} |\cB^\Delta(\cF)|\cdot|\cB^{2\Delta}(\cF)\backslash \cB^\Delta(\cF)|},
\end{align}
where we bound the states $\ket{\Psi_{1,2}(p)}$ in exactly the same way as we did in the proof of \eqref{it:firststrengthened}. Using the fact that $|\cB^\Delta(\cF)|,|\cB^{2\Delta}(\cF)\backslash \cB^\Delta(\cF)| \le |\cB^{2\Delta}(\cF)|$, we conclude that
\begin{align}
|\Sigma_1| &\leq |\cB^{2\Delta}(\cF)|\cdot \|F\|\sqrt{c_pc_{p'}}+O\left(n^2\lambda_2^{\Delta/2}\right)\ .
\end{align}
Combining this with~\eqref{eq:normsecondsigma2} gives the claim~\eqref{it:secondstrengthened}. The proof of~\eqref{it:secondstrengthenedb} is analogous, using Lemma~\ref{lem:maincauchyschwarzestimates}\eqref{it:a4bound}.

Finally, consider~\eqref{eq:sigmappscal} and~\eqref{it:thirdstrengthened}. We have 
\begin{align}
\sigma_{pp'}(\cB^\Delta(\cF)^c\times\cB^\Delta(\cF)^c)&=\underbrace{\sum_{\substack{j\in\cB^{\Delta}(\cF)^c\\
\phantom{j\neq j'}\\
}}e^{ij(p-p')}\bra{\Phi_{j,p'}}F\ket{\Phi_{j,p}}}_{=:\Theta_1}+\underbrace{\sum_{\substack{j,j'\in\cB^{\Delta}(\cF)^c\\j\neq j'}}e^{i(pj-p'j')}\bra{\Phi_{j',p'}}F\ket{\Phi_{j,p}}}_{=:\Theta_2}.
\end{align}
Using Lemma~\ref{lem:maincauchyschwarzestimates}\eqref{it:partaclaima2bound}, we have
\begin{align}
|\Theta_2|&\leq \|F\|\cdot O(n^2\lambda_2^{\Delta/2})\ .\label{eq:secondsigma2upperboundxz}
\end{align}
On the other hand,
by Lemma~\ref{lem:maincauchyschwarzestimates}\eqref{it:partbclaima2bound}, or more precisely its refinement in the form of equation~\eqref{eq:secondclaimmodifiedxlemma} from Lemma~\ref{lem:mainmatrixelementexcit}, we have
\begin{align}
\Theta_1&=\left(\sum_{j\in\cB^{\Delta}(\cF)^c}e^{ij(p-p')}\right)\bbra{\ell} E_{F(\tau(j_0))}\kket{r} c_{pp'}+O(n\lambda_2^{\Delta/2})\ .
\end{align}
for some fixed $j_0\in\cB^\Delta(\cF)^c$. For $p'=p$, the sum above is given trivially by $\sum_{j\in\cB^{\Delta}(\cF)^c}1 = \left|\cB^\Delta(\cF)^c\right|$. 
For $p\neq p'$, we have $\sum_{j\in [n]}e^{ij(p-p')}=0$, and hence
\begin{align}
    \left|\sum_{j\in\cB^{\Delta}(\cF)^c}e^{ij(p-p')}\right|=\left|\sum_{j\in\cB^{\Delta}(\cF)}e^{ij(p-p')}\right| \le |\cB^\Delta(\cF)|.\label{eq:cancellationphases}
\end{align}
Therefore, for $p=p'$ we have
\begin{align}
    \Theta_1 = \left|\cB^\Delta(\cF)^c\right|\bbra{\ell} E_{F(\tau(j_0))}\kket{r}c_{p} + O(n\lambda_2^{\Delta/2})
\end{align}
and for $p\neq p'$, we have
\begin{align}
    |\Theta_1|&\le \left|\cB^\Delta(\cF)\right|\bbra{\ell} E_{F(\tau(j_0))}\kket{r}c_{pp'} + O(n\lambda_2^{\Delta/2})\\
    &\le \left|\cB^\Delta(\cF)\right|\cdot \|F\|c_{pp'} + O(n\lambda_2^{\Delta/2}).
\end{align}
Note that we also have $c_{pp'} \le \sqrt{c_pc_{p'}}$ by the Cauchy-Schwarz inequality. Combining these results with~\eqref{eq:secondsigma2upperboundxz} proves claims~\eqref{eq:sigmappscal} and~\eqref{it:thirdstrengthened}.
\end{proof}

\subsection{The parameters of codes based on the excitation ansatz \label{sec:excitationansatzcode}}
Recall that the normalization of the excitation ansatz states $\ket{\Phi_p}\equiv \ket{\Phi_p(B;A)}$ are given by Lemma~\ref{lem:normofexcitationansatzstates} as
\begin{align}
    \|\Phi_p\| = \sqrt{nc_p} + O(n^{3/2}\lambda_2^{n/6}).
\end{align}
In the following, we let $\ket{\phi_p}$ denote the normalized versions of $\ket{\Phi_p}$. In terms of matrix elements, we have
\begin{align}
\bra{\phi_p}F\ket{\phi_{p'}} = \frac{\bra{\Phi_p}F\ket{\Phi_{p'}}}{n\sqrt{c_pc_{p'}}(1 + O(n^2\lambda_2^{\Delta/6}))}
= \frac{\bra{\Phi_p}F\ket{\Phi_{p'}}}{n\sqrt{c_pc_{p'}}}+O(n\lambda_2^{\Delta/6}).\mylabel{eq:matrixelementnorm}
\end{align}
Our main technical result  for the excitation ansatz consists of the following estimates:
\begin{lemma}\label{lem:excitationansatzmainlem}
Let $\nu \in (0,1)$ and $d=n^{1-\nu}$. 
Let $F\in \cB((\mathbb{C}^{\physical})^{\otimes n})$ be a $d$-local operator with unit norm. 
Consider the normalized versions~$\ket{\phi_p}$ and $\ket{\phi_{p'}}$ 
of the excitation ansatz state~\eqref{eq:excitationansatz}.
Then we have
\begin{align}
|\bra{\phi_{p'}}F\ket{\phi_{p}}|&=O(n^{-\nu/2})\qquad\textrm{ for }p\neq p'\ ,\mylabel{eq:phipsmall}
\end{align}
and
\begin{align}
|\bra{\phi_{p}}F\ket{\phi_{p}}
-\bra{\phi_{p'}}F\ket{\phi_{p'}}|&=O(n^{-\nu/2})\qquad\textrm{ for all }p,p'\ .\mylabel{eq:differencephipphipabs}
\end{align}
\end{lemma}
\begin{proof}
By definition of the excitation ansatz states, we have
\begin{align}
\bra{\Phi_{p'}}F\ket{\Phi_{p}}&=
\sum_{j,j'=1}^n e^{i(pj-p'j')}
\bra{\Phi_{j',p'}}F\ket{\Phi_{j,p}}=\sum_{\alpha=1}^4 \sigma_{pp'}(\Omega_\alpha)\ ,\mylabel{eq:matrixelementdecomp}
\end{align}
where 
\begin{align}
\Omega_1 &=\cB^\Delta(\mathsf{supp}(F))\times \cB^\Delta(\mathsf{supp}(F))\ ,\\
\Omega_2&=\cB^\Delta(\mathsf{supp}(F))\times \cB^\Delta(\mathsf{supp}(F))^c\ ,\\
\Omega_3&=\cB^\Delta(\mathsf{supp}(F))^c\times \cB^\Delta(\mathsf{supp}(F))\ ,\\
\Omega_4&=\cB^\Delta(\mathsf{supp}(F))^c\times \cB^\Delta(\mathsf{supp}(F))^c\ ,\ 
\end{align}
is the partition of $[n]^2$ considered in Lemma~\ref{lem:strengthenedboundsum}. Thus, for $p\neq p'$, we obtain
\begin{align}
|\bra{\Phi_{p'}}F\ket{\Phi_{p}}|&\leq 4|\cB^{2\Delta}(\cF)|\sqrt{c_pc_{p'} }+O(n^2\lambda_2^{\Delta/2})\ .
\end{align}
Inserting the expression~\eqref{eq:matrixelementnorm} for the normalized matrix element, we get
\begin{align}
|\bra{\phi_{p'}}F\ket{\phi_{p}}|&\le \frac{\left|\bra{\Phi_p}F\ket{\Phi_{p'}}\right|}{n\sqrt{c_pc_{p'}}}+O(n\lambda_2^{\Delta/6})\\
&\leq \frac{4|\cB^{2\Delta}(\cF)|}{n} + O(n\lambda_2^{\Delta/2})+O(n\lambda^{\Delta/6})\\
&=\frac{4|\cB^{2\Delta}(\cF)|}{n} + O(n\lambda_2^{\Delta/6})
\end{align}
Assume that $\mathsf{supp}(F)$ consists of $\kappa$ disjoint connected components. By definition, we have 
\begin{align}
|\cB^\Delta(\mathsf{supp}(F))|&=\big|\mathsf{supp}(F)\cup \left(\cB^\Delta(\mathsf{supp}(F))\backslash \mathsf{supp}(F)\right)\big|\leq d+2\kappa\Delta \le d(1+2\Delta),
\end{align}
where we use the fact that $\kappa \le d$ in the last inequality. Hence, we have
\begin{align}
|\bra{\phi_{p'}}F\ket{\phi_{p}}|&\leq \frac{4d(1+4\Delta) }{n}+O(n\lambda_2^{\Delta/6})\ .
\end{align}
Let $1>\nu > 0$ be arbitrary. Choosing~$d=n^{1-\nu}$ and $\Delta=6n^{\nu/2}$ gives\footnote{Note that this choice of $d$ and $\Delta$ satisfies the requirement in Lemma~\ref{lem:maincauchyschwarzestimates} for sufficiently large $n$.}
\begin{align}
    \frac{|\cB^{2\Delta}(\supp(F))|}{n} \le \frac{d(1+4\Delta)}{n} = O(n^{-\nu/2}),\mylabel{eq:bigOsupp}
\end{align}
and therefore
\begin{align}
|\bra{\phi_{p'}}F\ket{\phi_{p}}|&=O(n^{-\nu/2})+O(n\lambda_2^{n^{\nu/2}})=O(n^{-\nu/2})\ .
\end{align}
Note that the last equality follows since, for all $\lambda_2 < 1$ and $a,b>0$, we have
$\lim_{n\rightarrow \infty}n^a\lambda_2^{n^b} = 0$. 
This proves claim~\eqref{eq:phipsmall}.  

Next, we prove~\eqref{eq:differencephipphipabs}. Making use of equation~\eqref{eq:matrixelementnorm} and the decomposition~\eqref{eq:matrixelementdecomp}, we have
\begin{align}
|\bra{\phi_{p}}F\ket{\phi_{p}}
-\bra{\phi_{p'}}F\ket{\phi_{p'}}|&\le\left|\frac{\bra{\Phi_p}F\ket{\Phi_{p}}}{nc_p}-\frac{\bra{\Phi_{p'}}F\ket{\Phi_{p'}}}{nc_{p'}}\right|+O(n\lambda_2^{\Delta/6})\\
&\leq \frac{1}{n}\sum_{\alpha=1}^4 \Big|\left(c_p^{-1}\sigma_{pp}(\Omega_\alpha)-c_{p'}^{-1}\sigma_{p'p'}(\Omega_\alpha)\right)\Big| + O(n\lambda_2^{\Delta/6}).
\end{align}
By Lemma~\ref{lem:strengthenedboundsum} we have
\begin{align}
|\sigma_{pp}(\Omega_\alpha)|&\le |\cB^{2\Delta}(\mathsf{supp}(F))|c_p + O(n^2\lambda_2^{\Delta/2})\quad\textrm{ for }\alpha\in \{1,2,3\}\ ,
\end{align}
and so we can write
\begin{align}
\frac{1}{n}\sum_{\alpha=1}^3\left| c_p^{-1}\sigma_{pp}(\Omega_\alpha)-c_{p'}^{-1}\sigma_{p'p'}(\Omega_\alpha)\right|&\le \frac{1}{n}\sum_{\alpha=1}^3\left| c_p^{-1}\sigma_{pp}(\Omega_\alpha)\right|+\frac{1}{n}\sum_{\alpha=1}^3\left|c_{p'}^{-1}\sigma_{p'p'}(\Omega_\alpha)\right|\\
&\le \frac{6|\cB^{2\Delta}(\supp(F))|}{n} + O(n\lambda^{\Delta/2}).
\end{align}
It remains to consider the terms involving $\Omega_4$, whereby using equation~\eqref{eq:sigmappscal} we get
\begin{align}
\Big|c_p^{-1}\sigma_{pp}(\Omega_4)-c_{p'}^{-1}\sigma_{p'p'}(\Omega_4)\Big|
&=O(n^2\lambda_2^{\Delta/2}).
\end{align}
Putting everything together, we have
\begin{align}
|\bra{\phi_{p}}F\ket{\phi_{p}}
-\bra{\phi_{p'}}F\ket{\phi_{p'}}| &\le \frac{6|\cB^{2\Delta}(\supp(F))|}{n} + O(n\lambda_2^{\Delta/6})\\
&= O(n^{-\nu/2}),
\end{align}
where we again use the bound~\eqref{eq:bigOsupp} in the last line. This proves claim~\eqref{eq:differencephipphipabs}.
\end{proof}
With Lemma~\ref{lem:excitationansatzmainlem}, it is straightforward to check
the condition for approximate quantum error-detection from Section~\ref{sec:sufficientconderrordetect}. 
This leads to the following:
\begin{theorem}\label{thm:excitationansatzparams}
Let $\nu \in (0,1)$ and let $\kappa,\Delta>0$ be such that 
\begin{align}
5\kappa+\lambda <\nu\ .
\end{align}
Let $A,B$ be tensors associated with an injective excitation ansatz state~$\ket{\Phi_p(B;A)}$, where $p$ is the momentum of the state.  Then there is a subspace~$\cC\subset (\mathbb{C}^\physical)^{\otimes n}$ spanned by excitation ansatz states~$\{\ket{\Phi_p(B;A)}\}_p$ with different momenta~$p$ such that $\cC$ is an $(\epsilon,\delta)[[n,k,d]]$-AQEDC with parameters
\begin{align}
k&=\kappa \log_\physical n\ ,\\
d&=n^{1-\nu}\ ,\\
\epsilon &=\Theta(n^{-(\nu-(5\kappa+\lambda))})\ ,\\
\delta&=n^{-\lambda}\ .
\end{align}
\end{theorem} 
\begin{proof}
Let us choose an arbitrary set~$\{p_1,\ldots,p_{\physical^k}\}$ of $\physical^k=n^\kappa$ distinct, non-zero momenta and define the space~$\cC$ by
\begin{align}
\cC=\myspan\{\ket{\Phi_{p_j}(B;A)}\}_{j=1}^{\physical^k}\ .
\end{align}
Since momentum eigenstates to different momenta are orthogonal, the states $\{\ket{\phi_{p_j}}\}_{j=1}^{\physical^k}$ form an orthonormal basis of~$\cC$. By Lemma~\ref{lem:excitationansatzmainlem}, we have 
\begin{align}
  |\bra{\phi_{p_r}}F\ket{\phi_{p_s}}-\delta_{r,s} \bra{\phi_{p_1}}F\ket{\phi_{p_1}}|&=O(n^{-\nu/2})\ .
\end{align}
for any $d$-local unit norm operator $F\in\cB((\mathbb{C}^\physical)^{\otimes n})$ and all $r,s\in [\physical^k]$. The sufficient conditions of Corollary~\ref{cor:errordectionapproximate} for approximate error-detection applied with $\gamma=\Theta(n^{-\nu/2})$ show that~$\cC$ is a $(\Theta( \physical^{5k}n^{-\nu}/\delta),\delta)[[n,k,d]]$-AQEDC for any $\delta$ satisfying $\delta>\physical^{5k}n^{-\nu}$. This implies the claim for the given choice of parameters.
\end{proof}

From~\cite{haegeman2013elementary}, we know that isolated energy bands in gapped systems are well approximated, under mild physical conditions, by the Fourier transforms of local operators. In particular, this means that, possibly after blocking, isolated momentum eigenstates of gapped systems are well approximated by some excitation ansatz state, as one would expect.\footnote{In fact, we expect excitation ansatz states to be even better approximations of momentum eigenstates than the constructions considered in~\cite{haegeman2013elementary}. In~\cite{haegeman2013elementary}, the local operators $O$ act on the physical level, whereas the defining tensors $B$ of excitation ansatz states act on the virtual level, and are hence more general.} One consequence of this is that the excitation ansatz codes considered in this section are \it generic \rm among physical systems: essentially any selection of momentum eigenstates from an isolated energy band of a gapped system can be expected to form an error-detecting code with the above parameters.
\newpage
\section{AQEDC at low energies: An integrable model\label{sec:qedcintegrable}}
In this section, we consider the Heisenberg-XXX spin chain. In Section~\ref{sec:xxxmodelmagnoncode}, we introduce the model. The approximate error-detection codes we consider are spanned by eigenstates that we call magnon-states. The latter are particular instances of the algebraic Bethe ansatz, for which a general framework of MPS/MPO descriptions has been introduced in   prior work~\cite{verstraetekorepin}. We review the necessary notation for matrix product operators (MPOs) in Section~\ref{sec:mpointro}. In Section~\ref{sec:mpsdescriptionmagnon}, we give an MPS/MPO description of magnon-states. In Section~\ref{sec:compressedmpsdescriptionmagnon}, we provide a second MPS/MPO description with smaller bond dimension. In Section~\ref{sec:magnonpermutationinvariance}, we consider matrix elements of operators with respect to the magnon-state basis. We show how to relate matrix elements of operators with arbitrary support to matrix elements of operators with connected support. In Section~\ref{sec:jordanstructuremagnon} we analyze the Jordan structure of the transfer operators.
In Section~\ref{sec:matrixelementsmagnon}, we bound matrix elements of local operators in magnon states. 
Finally, in Section~\ref{sec:parametersmagnon}, we determine the parameters of the magnon code.

\subsection{The XXX-model and the magnon code\label{sec:xxxmodelmagnoncode}}
Consider the periodic Heisenberg-XXX spin chain, with Hamiltonian
\begin{align}
H&=-\frac{1}{4}\sum_{m=1}^n \left(\sigma^x_m\sigma^x_{m+1}+\sigma^y_m\sigma^y_{m+1}+\sigma^z_m\sigma^z_{m+1}\right)\label{eq:xxxmodel}
\end{align}
on~$(\mathbb{C}^2)^{\otimes n}$, where we apply periodic boundary conditions, and where $\sigma^x_m,\sigma^y_m,\sigma^z_m$ are the Pauli matrices acting on the $m$-th qubit. The model~\eqref{eq:xxxmodel} is gapless and can be solved exactly using the algebraic Bethe ansatz. Our goal here is to argue that~\eqref{eq:xxxmodel} contains error-detecting codes in its low-energy subspace. More precisely, we consider subspaces spanned by non-zero momentum eigenstates.

The Hamiltonian~\eqref{eq:xxxmodel} may alternatively be expressed as 
\begin{align}
H&=\frac{n}{4}\identityoperator-\frac{1}{2}\sum_{m=1}^n \flip_{m,m+1}\ ,\label{eq:eqxxxmodelasswaps}
\end{align}
where $\flip_{m,m+1}$ is the flip-operator acting on the $m$-th and $(m+1)$-th qubit. Equation~\eqref{eq:eqxxxmodelasswaps} shows that~$H$ commutes with the tensor product representation of the special unitary group~$SU(2)$ on $(\mathbb{C}^2)^{\otimes n}$, hence we may 
restrict to irreducible subspaces (with fixed angular momentum) to diagonalize~$H$. More precisely, let us define, for each qubit~$m$, the
operators 
\begin{align}
\hspin_m^-=\ket{0}\bra{1},\qquad \hspin_m^+=(\hspin_m^-)^\dagger,\qquad \text{and} \qquad\hspin_m^3=\frac{1}{2}(-\proj{0}+\proj{1}) .    
\end{align} 
These satisfy the canonical $\mathfrak{su}(2)$~commutation relations, with $\hspin^+$ and $\hspin^-$ being the raising and lowering operators of the spin-$1/2$ representation, and the basis states $\ket{0}$ and $\ket{1}$ corresponding to $\ket{j,m}=\ket{1/2,-1/2}$ and $\ket{1/2,1/2}$, respectively. The total $z$-angular momentum and raising/lowering-operators for the tensor product representation on~$(\mathbb{C}^2)^{\otimes n}$ are given by
\begin{align}
S_3=\sum_{m=1}^n \hspin_m^3\qquad \text{and}\qquad S_{\pm}=\sum_{m=1}^n \hspin^{\pm}_m\ .
\end{align}
These operators commute with~$H$, and therefore the total Hilbert space splits into a direct sum of spin representations: 
\begin{align}
(\mathbb{C}^2)^{\otimes n}\cong\bigoplus_{j}\cH_j\otimes\mathbb{C}^{m_j}\ 
,
\end{align}
where the direct sum is taken over all irreducible spin represenations (with multiplicity~$m_j$) present in the decomposition of the tensor representation. Each $\cH_j$ defines an irreducible $2j+1$-dimensional angular momentum-$j$ representation, and $H|_{\cH_j}=E_j\identityoperator_{\cH_j}$ is proportional to the identity
on each of these spaces. For instance, the subspace~$\cH_{n/2}$ with maximal angular momentum has highest weight vector~$\ket{1}^{\otimes n}$ and is spanned by ``descendants'' obtained by applying the lowering operator, that is,
\begin{align}
\cH_{n/2}&=\mathsf{span}\left\{ S_-^r\ket{1}^{\otimes n}\ |\ r=0,\ldots,n\ \right\}\ .
\end{align}
It is associated with energy $E_{n/2}=-n/4$, which is the ground state energy of~$H$. 
Clearly, this is the symmetric subspace, containing only permutation-invariant (i.e., zero-momentum) states. Error-correction within this subspace has been considered in~\cite{brandao2017quantum}. Indeed, all the examples  constructed there consist of subspaces of~$\cH_{n/2}$. 

Here we go beyond permutation-invariance. Specifically, we consider the vector
\begin{align}
\ket{\Psi}&=\overline{\omega}\sum_{r=1}^n \omega^r \hspin_r^-\ket{1}^{\otimes n}\qquad\textrm{ where }\qquad \omega=e^{2\pi i/n}\ .\label{eq:onemagnonstatedef}
\end{align}
The factor~$\overline{\omega}$ in front is introduced for convenience. A straightforward calculation shows that $S_+ \ket{\Psi}=0$ and $S_3\ket{\Psi}=
\left(n/2-1\right)\ket{\Psi}$, 
hence this is a highest weight vector for angular momentum~$j=n/2-1$ and
\begin{align}
\cH_{n/2-1}&=\mathsf{span}\left\{ S_-^r\ket{\Psi}\ |\ r=0,\ldots,n-2\right\}\ .\label{eq:subspacenhalfminusone}
\end{align}
The energy of states in this subspace can be computed to be $E_{n/2-1}=-n/4+1-\cos(2\pi/n)=E_{n/2}+O(1/n^2)$. This shows that these states are associated with low-lying excitations, and the system is gapless. Observe also that~\eqref{eq:onemagnonstatedef} is an eigenvector of the cyclic shift with eigenvalue~$\omega$, that is, it has fixed momentum $p=2\pi/n$. As $S_-^r$ commutes with the cyclic shift, the same is true for all states in~$\cH_{n/2-1}$: this is a subspace of fixed momentum and energy. We will argue that $\cH_{n/2-1}$ contains error-detecting codes. Specifically, we consider subspaces spanned by states of the form $\{S^r_-\ket{\Psi}\}_r$ for appropriate choices of magnetization~$r$. The state~\eqref{eq:onemagnonstatedef} is sometimes referred to as a one-magnon state. Correspondingly, we call the corresponding code(s) the {\em magnon-code}. We also refer to the vectors~$\{S^r_-\ket{\Psi}\}_{r}$ (respectively, their normalized versions) as {\em magnon-states}. 
For brevity, let us  denote the $r$-th descendant by
\begin{align}
\ket{\Psi_r}:=S^r_-\ket{\Psi}\qquad\textrm{ for }\qquad r=0,\ldots,n-2\ .\label{eq:descendantdef}
\end{align}
It is clear that the states $\ket{\Psi_r}$ and $\ket{\Psi_s}$ are orthogonal for $r\neq s$ as they have different magnetization, hence they form a basis of the magnon code. It is also convenient to introduce their normalized versions which are given by
\begin{align}
\ket{\psi_r}=\left(\frac{(n-2-r)!}{n(n-2)!r!}\right)^{1/2}S_{-}^r \ket{\Psi}\qquad \textrm{for} \qquad r=0,\ldots,n-2\ ,\mylabel{eq:psipdefinitionbasic}
\end{align}
as follows from the fact that for a normalized highest-weight vector~$\ket{j,j}$ of a spin-$j$-representation, the vectors 
\begin{align}
\ket{j,j-k}=\left(\frac{(2j-k)!}{(2j)!k!}\right)^{1/2}S_-^k \ket{j,j},\end{align} 
with $k=0,\ldots,2j$ form an orthonormal basis.

\subsection{Matrix product operators\label{sec:mpointro}}
Here we briefly review the formalism of matrix product operators (MPO) and introduce the corresponding notation. We only require site-independent MPO. Such an MPO $\mpo\in\cB((\mathbb{C}^\physical)^{\otimes n})$  with bond dimension~$D$ is given by 
\begin{align}
\mpo &=\sum_{\substack{i_1,\ldots,i_n\in [\physical]\\
j_1,\ldots,j_n\in [\physical]}} \tr(O_{i_1,j_1}\cdots
O_{i_n,j_n}X) \ket{i_1}\bra{j_1}\otimes\cdots\otimes\ket{i_n}\bra{j_n}\ 
\end{align}
for a family of local tensors $\{O_{i,j}\}_{i,j\in [\physical]}\subset\cB(\mathbb{C}^D)$, and a boundary operator $X\in \cB(\mathbb{C}^D)$. Alternatively, the MPO~$\mpo$ can also be parametrized by 
the operator $X\in\cB(\mathbb{C}^D)$ together with
family~$\{O^{\alpha,\beta}\}_{\alpha,\beta\in [D]}$ of $\physical\times \physical$-matrices. In this parametrization (illustrated in Figure~\ref{fig:mpoalternative}), the MPO is written as 
\begin{align}
\mpo&=\sum_{\alpha_0,\ldots,\alpha_n\in [D]}
X_{\alpha_n,\alpha_0} O^{\alpha_0,\alpha_1}\otimes O^{\alpha_1,\alpha_2}\otimes\cdots\otimes O^{\alpha_{n-1},\alpha_n}\ . \mylabel{eq:mpofirstparam}
\end{align}
\begin{figure}
\centering
\includegraphics[width=9cm]{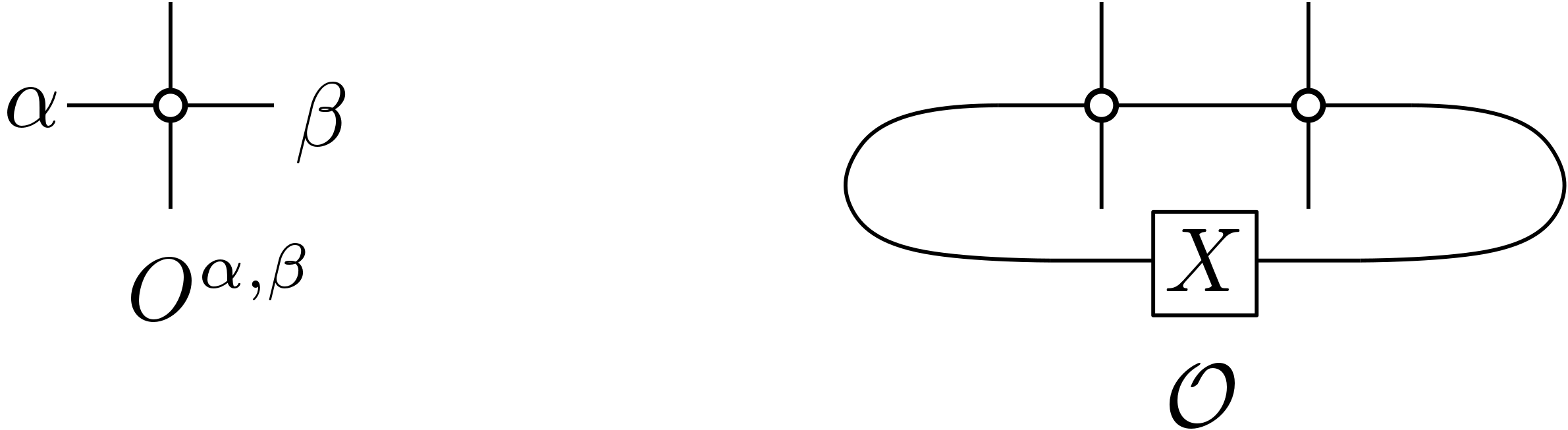}
\caption{Alternative parametrization of an MPO $\mpo$.
\mylabel{fig:mpoalternative}}
\end{figure}
Equation~\eqref{eq:mpofirstparam} shows that the MPO~$\mpo=\mpo(O,X,n)\in \cB((\mathbb{C}^\physical)^{\otimes n})$
is fully specified by three objects: 
\begin{enumerate}[(i)]
\item
a four-index tensor $O$, defined in terms of the collection $\{O_{i,j}\}_{i,j\in [\physical]}$ of matrices acting on the so-called  virtual space~$\mathbb{C}^D$ (alternatively, the collection of matrices~$\{O^{\alpha,\beta}\}_{\alpha,\beta\in [D]}$ acting on the physical space~$\mathbb{C}^\physical$),
\item
a matrix $X\in \cB(\mathbb{C}^D)$ acting on the virtual space, and
\item
an integer $n\in\mathbb{N}$ specifying the number of physical spins.
\end{enumerate}
We refer to the tensor $O$ as a local MPO tensor, and to $X$ as a boundary operator.

It is convenient to introduce the following product on MPO tensors. Suppose $O_1$ and $O_2$ are MPO tensors associated with MPOs having physical dimension~$\physical$, and bond dimensions $D_1$ and $D_2$, respectively. Then $O_1\diamond O_2$ is the MPO tensor of an MPO with physical dimension~$\physical$ and bond dimension~$D_1\cdot D_2$. Its tensor network description is given in Figure~\ref{fig:tensorproductmpoops}. More precisely, if $O_\alpha$ is defined by $\{O^{(x)}_{i,j}\}_{i,j\in [\physical]}$ for $x=1,2$, then $O_1\diamond O_2$ is defined in terms of the matrices
\begin{align}
O_{i,j}&=\sum_{k=1}^{r} (O^{(1)})_{i,k}\otimes (O^{(2)})_{k,j}\in\cB(\mathbb{C}^{D_1}\otimes\mathbb{C}^{D_2})\qquad\textrm{ for }i,j\in [\physical]\ .
\end{align}
This is clearly associative, and allows us to define
$O^{\diamond k}:=O\diamond O^{\diamond (k-1)}$ recursively.

\begin{figure}
\centering
\includegraphics[width=10cm]{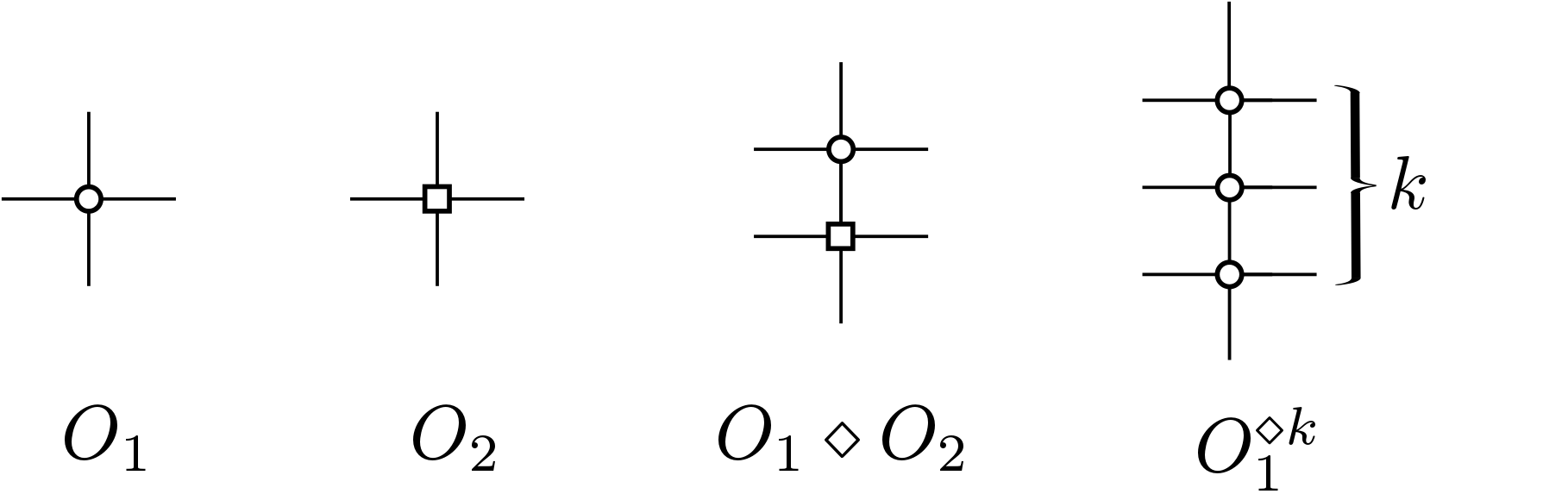}
\caption{The product of two MPO tensors $O_1$ and $O_2$, as well as the power $O_1^{\diamond k}$. 
\mylabel{fig:tensorproductmpoops}}
\end{figure}

Suppose now that an MPO $\mpo=\mpo(O,X,n)$  is given. Observe that for $k\in \mathbb{N}$, the operator~$\mpo^{ k}$ is an MPO 
whose virtual bond space is $(\mathbb{C}^D)^{\otimes k}$
and whose local tensors take the form
\begin{align}
\bra{\alpha_1\cdots \alpha_k}(O^{\diamond k})_{i,j}\ket{\beta_1\cdots \beta_k}&=\hspace{-3ex}
\sum_{s_1,\ldots,s_{k-1}\in [\physical]}\hspace{-3ex} \begin{split}\bra{\alpha_1}O_{i,s_1}\ket{\beta_1}\cdot
\bra{\alpha_2}O_{s_1,s_2}\ket{\beta_2}\qquad\qquad\qquad\qquad\qquad\\
\cdots
\bra{\alpha_{k-1}}O_{s_{k-2},s_{k-1}}\ket{\beta_{k-1}}\cdot 
\bra{\alpha_{k}}O_{s_{k-1},j}\ket{\beta_k}\ ,
\end{split}\mylabel{eq:opowermatrices}
\end{align}
with boundary tensor $X^{\otimes k}$.
\begin{figure}
\centering
\includegraphics[width=8cm]{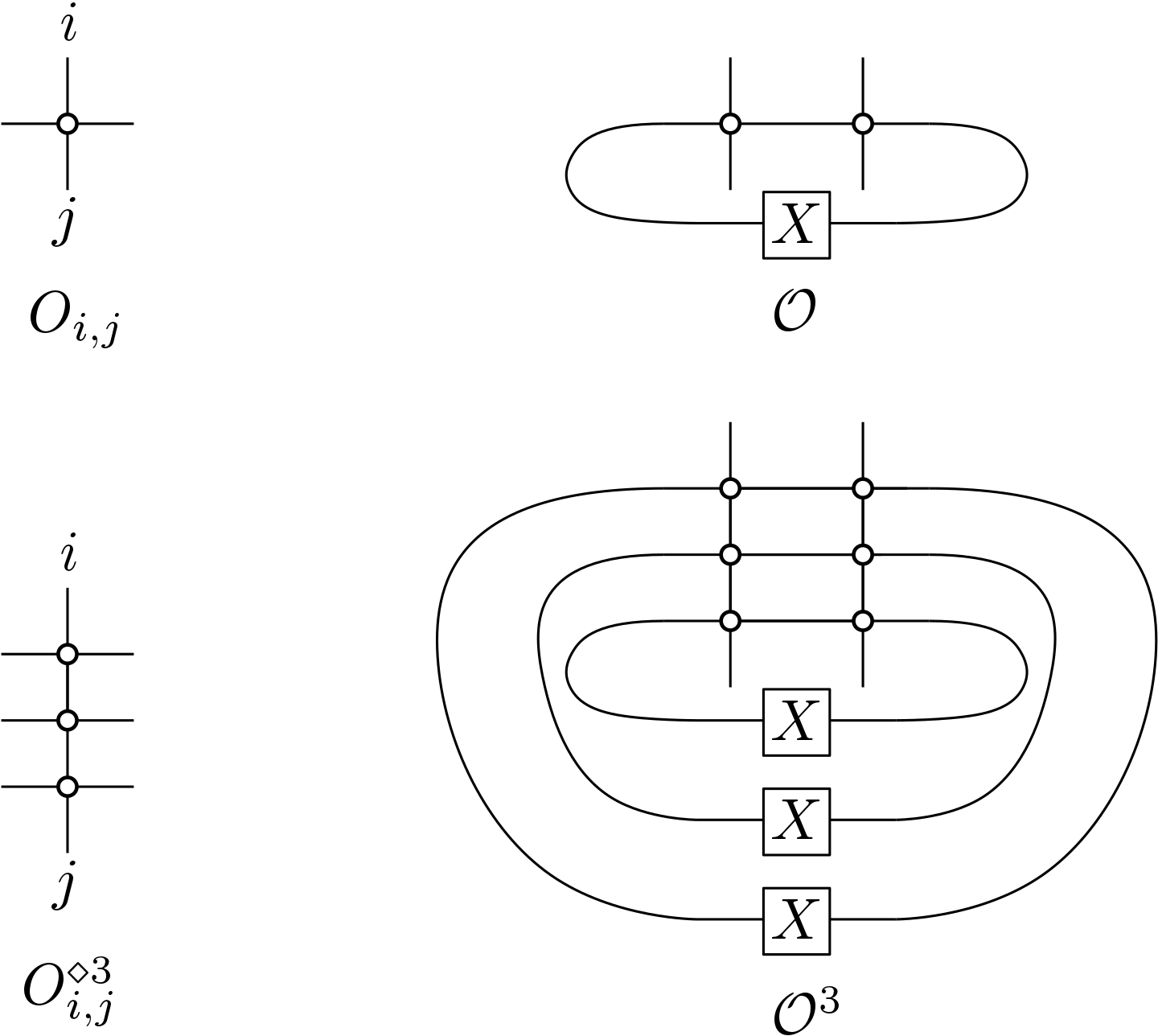}
\caption{
This figure shows an MPO $\mpo=\mpo(O,X,2)$
defined in terms of matrices $\{O_{i,j}\}_{i,j}$  and the matrices $\{O^{\diamond 3}_{i,j}\}_{i,j}$ defining the MPO $\mpo^3$. Left-multiplication by an operator corresponds to stacking a diagram on top.
\mylabel{fig:transferoppower}}
\end{figure}
Thus the MPO $\mpo^k=\mpo(O^{\diamond k}, X^{\otimes k},n)$ is defined by the MPO tensor $O^{\diamond k}$ and the boundary operator~$X^{\otimes k}$. These are visualized in Figure~\ref{fig:transferoppower}, for $k=3$.

Consider an MPS~$\ket{\Psi}=\ket{\Psi(A,X,n)}\in (\mathbb{C}^{\physical})^{\otimes n}$ of bond dimension~$D_1$ and an MPO~$\mpo=\mpo(O,Y,n) \in\cB((\mathbb{C}^\physical)^{\otimes n})$ of bond dimension~$D_2$. Then clearly $\mpo\ket{\Psi}$ is an MPS with bond dimension~$D_1D_2$. We write
\begin{align}
\mpo\ket{\Psi}=\ket{\Psi(O\diamond A,Y\otimes X,n)}\ ,\mylabel{eq:mpoappliedtompsdef}
\end{align}
see Figure~\ref{fig:diamondotdef} for the definition of the MPS tensor $O\diamond T$. 
\begin{figure}
\centering
\includegraphics[width=1.5cm]{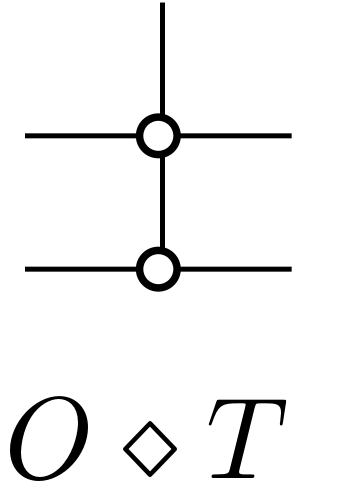}
\caption{Definition of the MPS tensor $O\diamond T$. 
\mylabel{fig:diamondotdef}}
\end{figure}
 
 \begin{figure}
\centering
\includegraphics[width=4cm]{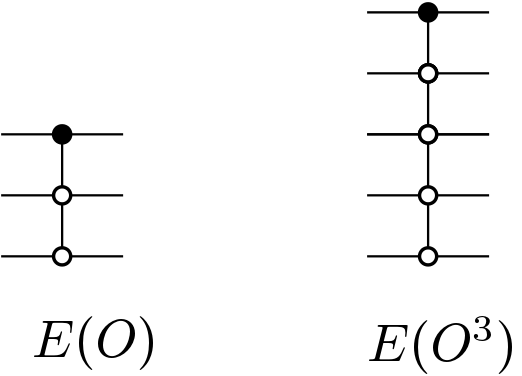}
\caption{
This figure shows the tensor network representations
of $E_O$ and $E_{O^{\diamond 3}}$.\mylabel{fig:eooperators.png}}
\end{figure}

In the following, we  are interested in matrix elements of the form~$\bra{\Psi}\mpo \ket{\Psi}$.  A central object of study is the generalized transfer operator $E_{O}$. If $O$ is specified by matrices $\{O_{i,j}\}_{i,j\in [\physical]}\subset \cB(\mathbb{C}^{D_2})$, this is given by
\begin{align}
E_O& = 
\sum_{s,t\in [D_1]}
\sum_{j,k\in [D_2]}
\bra{s} O_{j,k}\ket{t}
\overline{A_s}\otimes \ket{j}\bra{k}\otimes A_t\in 
\cB(\mathbb{C}^\physical\otimes\mathbb{C}^{D_2}\otimes\mathbb{C}^\physical)\ .
\end{align}
This operator, as well as $E_{O^{\diamond k}}\in\cB(\mathbb{C}^\physical\otimes (\mathbb{C}^{D_2})^{\otimes k}\otimes\mathbb{C}^\physical)$ for  $k=3$ are illustrated in Figure~\ref{fig:eooperators.png}.

\begin{figure}
\centering
\includegraphics[width=6cm]{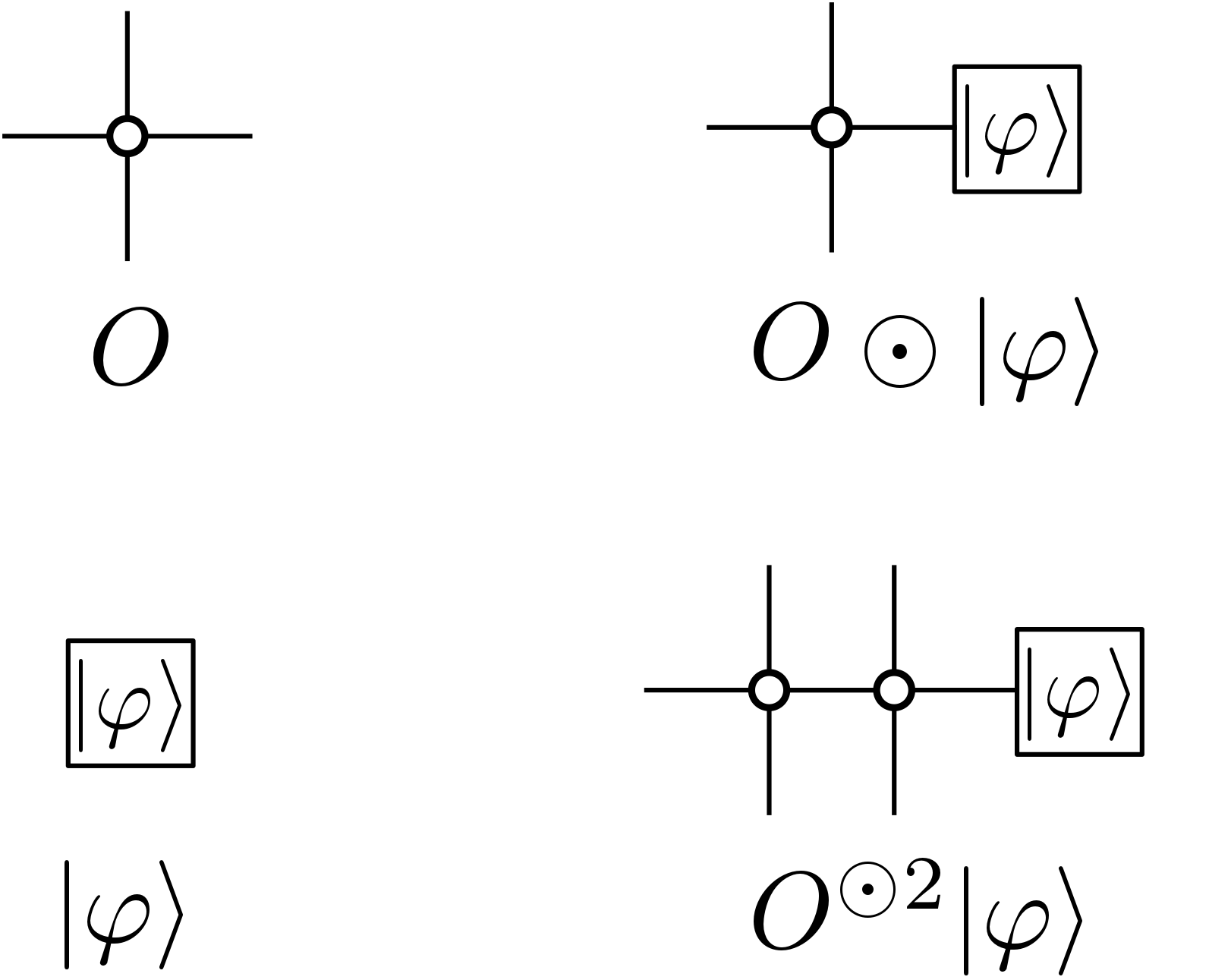}
\caption{This figure illustrates the definition of the product $O\odot \ket{\varphi}$.  \mylabel{fig:specialproductdef}}
\end{figure}
 
 Consider an MPO tensor $O$ with physical space $\cH_p$ and virtual space $\cH_v$.
 Given a vector $\ket{\varphi}\in\cH_v$,
 we can define an element $O\odot\ket{\varphi}\in \cH_v\otimes\cB(\cH_p)$
 by attaching $O$ from the left, see Figure~\ref{fig:specialproductdef}.
 The map $(O,\varphi)\mapsto O\odot\ket{\varphi}$ is bilinear. Hence we can define 
\begin{align}
O\odot (O\odot \ket{\varphi}):=(O\otimes\identityoperator_{\cB(\cH_p)})(O\odot\ket{\varphi})\in \cH_v\otimes\cB(\cH_p)\otimes\cB(\cH_p)\ .
\end{align} 
This is clearly associative. Correspondingly, we also define 
$O^{\odot n} \ket{\varphi}\in\cH_v\otimes\cB(\cH_p)^{\otimes n}$ as the result applying this map $n$~times. Note that an MPO
defined by $(O,X=\ket{\varphi}\bra{\chi})$ can be written as
$\left(\bra{\chi}\otimes\identityoperator^{\otimes n}_{\cB(\cH_p)}\right)O^{\odot n}\ket{\varphi}$. 

Conversely, observe that  a bilinear map $\Gamma:\cH_v\rightarrow\cH_v\otimes\cB(\cH_p)$, together with two states~$\ket{\varphi}, \ket{\chi}\in\cH_v$, defines a site-independent MPO in this fashion.

\subsection{MPS/MPO representation of the magnon states \label{sec:mpsdescriptionmagnon}}
Here we give an MPS/MPO representation of the magnon states that we use throughout our analysis below. We note that more generally,~\cite{verstraetekorepin} discusses such representations for the Bethe ansatz states.

Consider the one-magnon state~$\ket{\Psi}\in(\mathbb{C}^2)^{\otimes n}$ defined by~\eqref{eq:onemagnonstatedef}. It is straightforward to check that an MPS representation of $\ket{\Psi}=\ket{\Psi(\{A_0,A_1,X\})}$  with bond dimension $D=2$ is given by
\begin{align}
    A_0 &=\ket{1}\bra{0}\ ,\\
    A_1 &=\proj{0}+\omega \proj{1}\ ,\mylabel{eq:amatrixmps}\\
    X&=\ket{0}\bra{1}\ ,
\end{align}
where $\omega=e^{2\pi i/n}$, see Figure~\ref{fig:magnon}. 
\begin{figure}
\centering
\includegraphics[width=8.5cm]{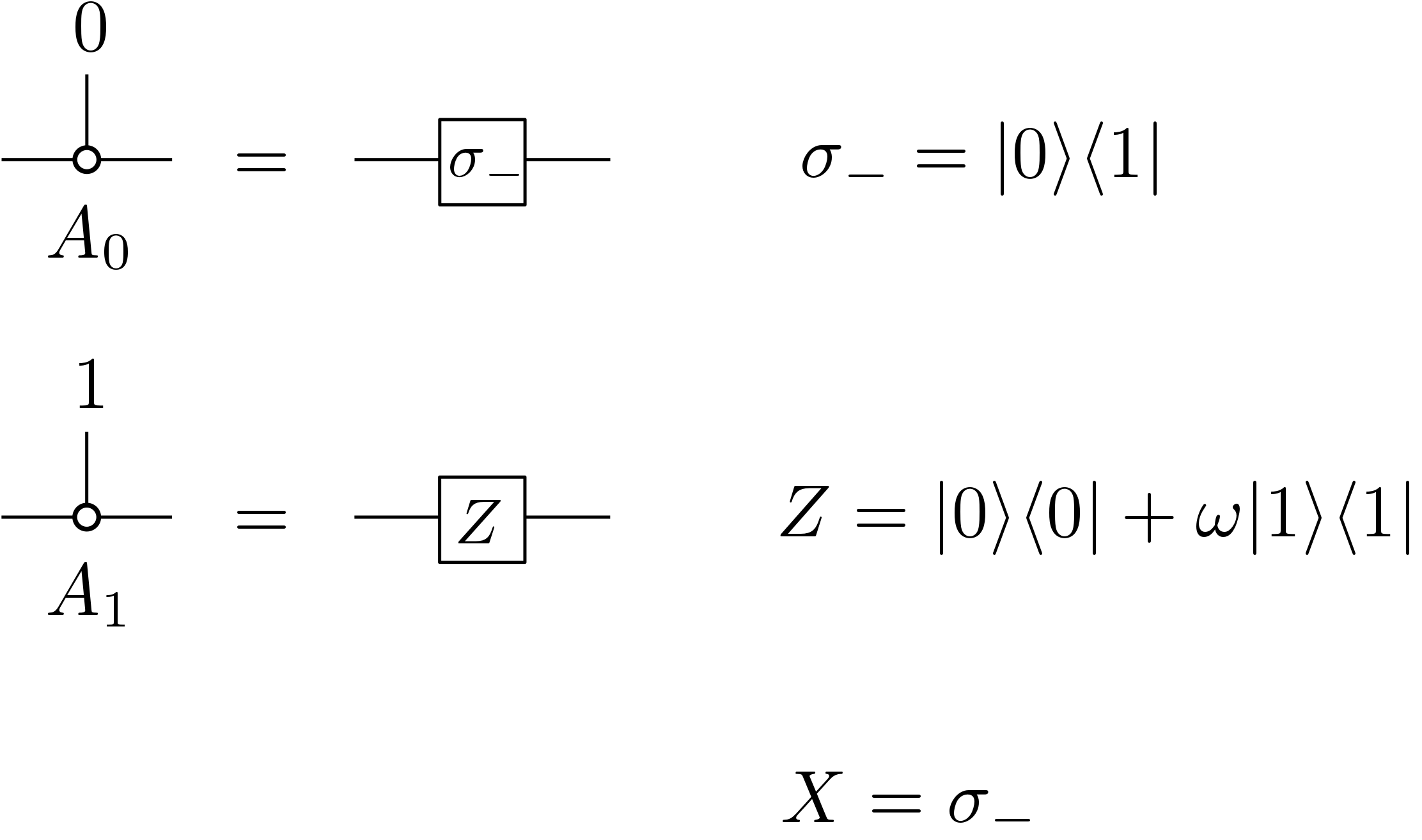}
\caption{An MPS description of the one-magnon state $\ket{\Psi}$  (cf.~\eqref{eq:onemagnonstatedef})\mylabel{fig:magnon}.}
\end{figure}
Next, we consider the descendants~\eqref{eq:descendantdef}. The operator $S_-=\sum_{m=1}^n \hspin^-_m$  can be expressed as a bond dimension $D=2$ MPO, given by 
\begin{align}
S_-=\cO(O^{0,0},O^{0,1},0^{1,0},O^{1,1},X)\in\cB((\mathbb{C}^2)^{\otimes n})\ ,
\end{align} 
where the boundary tensor is $X=\sigma_-:=\ket{0}\bra{1}$, and where the local tensors are defined as
\begin{align}
O^{0,0}=O^{1,1}&=I_{\mathbb{C}^2}\ ,\\
O^{1,0}&=\ket{0}\bra{1}\ ,\mylabel{eq:bmatrixsminus}\\
O^{0,1}&=0\ .
\end{align}
This definition is illustrated in Figure~\ref{fig:magnonmpo}. The adjoint operator $S_+$ has an MPO representation described as in Figure~\ref{fig:magnonmpoadjoint}.

It follows that the  ``descendants'' $\ket{\Psi_s}= S_-^s\ket{\Psi}$ can be represented as in Figure~\ref{fig:descendants}, i.e., they are MPS of the form
\begin{align}
\ket{\Psi_s}&=\ket{\Psi(O^{\diamond s}\diamond A, X^{\otimes (s+1)},n)}\qquad\textrm{ for }s=0,\ldots,n-2\ ,
\end{align}
where  the MPS tensor $O\diamond T$ is defined as in equation~\eqref{eq:mpoappliedtompsdef}. 
\begin{figure}
\centering
\includegraphics[width=8cm]{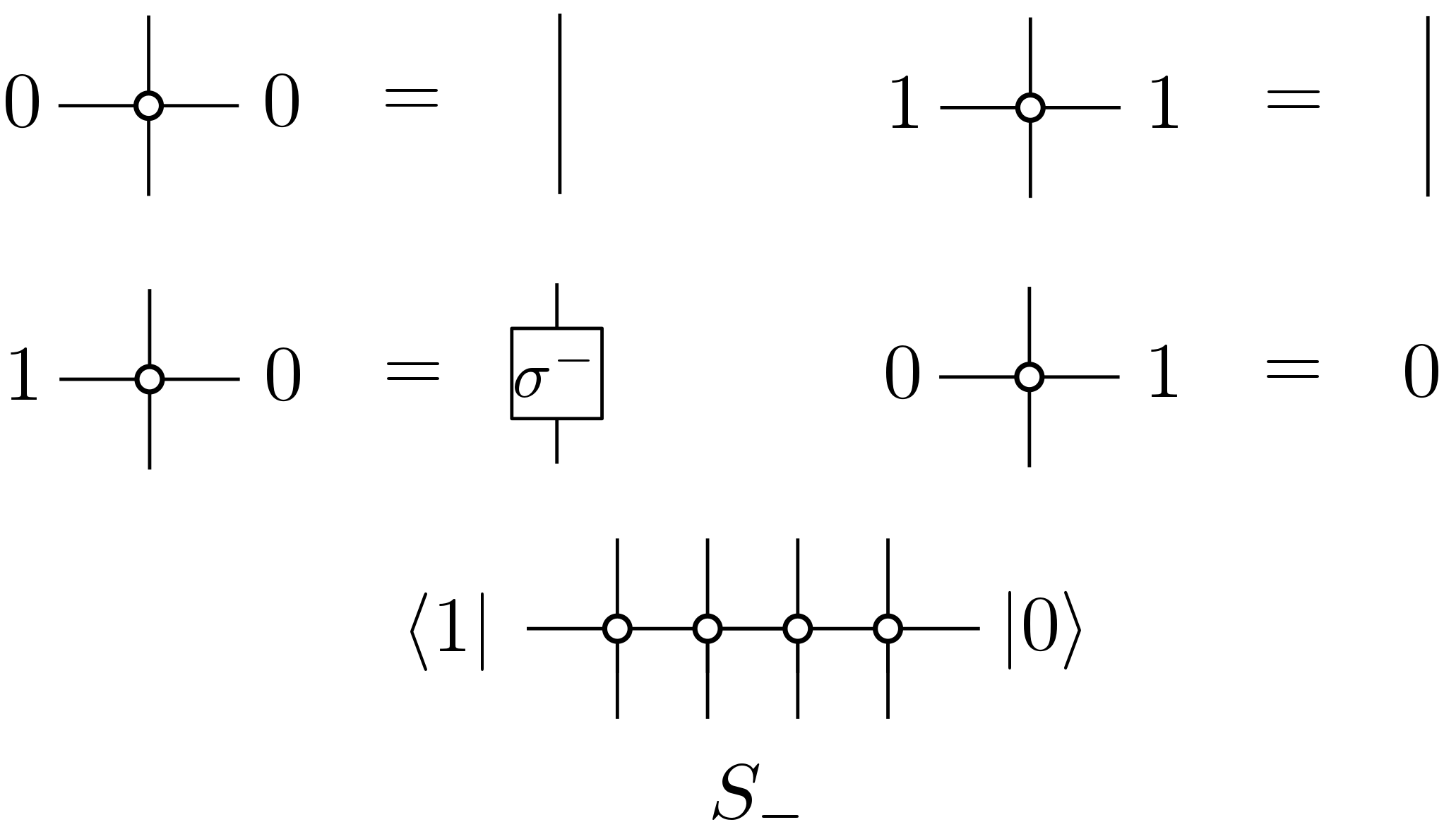}
\caption{An MPO description of the 
lowering operator. The MPO~$S_-=\cO(O,X,n)$ is defined with $O$ as given in the figure and with $X=\sigma_{-}$.
\mylabel{fig:magnonmpo}}
\end{figure}

\begin{figure}
\centering
\includegraphics[width=8cm]{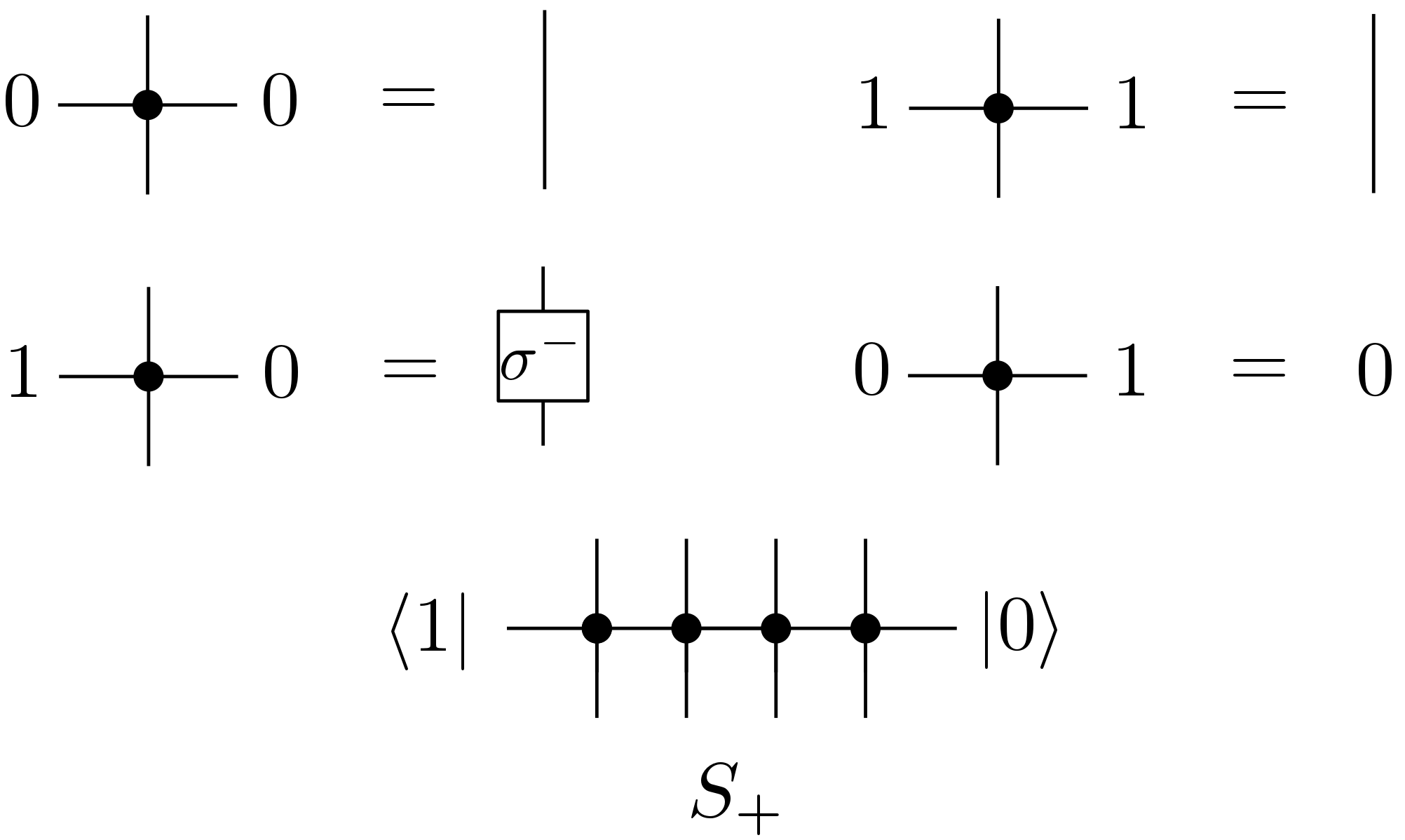}
\caption{An MPO description of the adjoint MPO~$S_+$.  \mylabel{fig:magnonmpoadjoint}}
\end{figure}

\begin{figure}
\centering
\includegraphics[width=8.5cm]{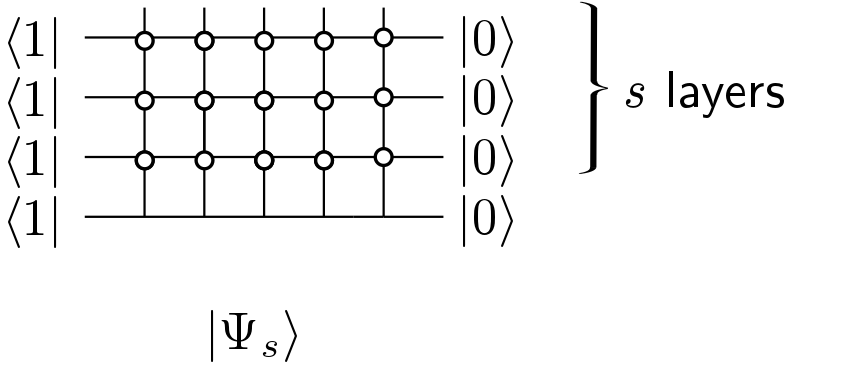}
\caption{An MPS/MPO representation of the vector~$\ket{\Psi_s}=S_-^s\ket{\Psi}$. Seen as an MPS,
this has rank-$1$-boundary tensor $X=(\ket{0}\bra{1})^{\otimes s+1}$.
\mylabel{fig:descendants}}
\end{figure}

\subsection{A compressed MPS/MPO representation of the magnon states \label{sec:compressedmpsdescriptionmagnon}}
Consider the MPO representation~\eqref{eq:bmatrixsminus} of $S_-$. For $s\in [n]$, it implies the MPO representation
\begin{align}
S_{-}^s &=\mpo(O^{\diamond s}, \sigma_-^{\otimes s},n)\label{eq:sminusnoncompressed}
\end{align}
for the $s$-th power of $S_-$, which has bond dimension~$D=2^s$. Below we argue that the MPO~\eqref{eq:sminusnoncompressed} can also be expressed as an MPO with bond dimension~$s+1$. We call this the compressed representation:

\begin{lemma}\label{lem:compressedrep}
Let $s\in [n]$ and consider the operator $S_-^s$, where $S_-=\sum_{m=1}^n \hspin_m^-$. This has the 
bond dimension $D=s+1$-MPO  representation
\begin{align}
S_{-}^s&=\cO(\tilde{O}_s, \tilde{X}_s,n)\ .
\end{align}
Here the virtual space~$\mathbb{C}^{s+1}$ is 
that of a spin-$s/2$ with orthonormal angular momentum eigenstate basis $\{\ket{\textfrac{s}{2},m}\ |\  m=-\frac{s}{2},-\textfrac{s}{2}+1,\ldots,\textfrac{s}{2}\}$. The boundary tensor is 
\begin{align}
\tilde{X}_s=\ket{\textfrac{s}{2},-\textfrac{s}{2}}\bra{\textfrac{s}{2},\textfrac{s}{2}}
\end{align}
and the MPO tensor $\tilde{O}_s$ is defined by the matrices
\begin{align}
(\tilde{O}_s)_{0,0}=(\tilde{O}_s)_{1,1}&=I\ ,\\
(\tilde{O}_s)_{1,0}&=0\ , \label{eq:btildeijdef}\\
(\tilde{O}_s)_{0,1}&=J_+\ ,
\end{align}
where $J_+$ is the usual spin-raising operator.\footnote{With respect to a distinguished orthonormal basis $\{\ket{j,m}\}_{m=-j,-j+1,\cdots,j}$, we have $$J_+\ket{j,m}=\sqrt{j(j+1)-m(m+1)}\ket{j,m+1}$$ 
for all $m=-j,\cdots,j-1$ and $J_+\ket{j,j}=0$.} In particular, the states~$\ket{\Psi_s}=S_-^s\ket{\Psi}$ have an MPS representation of the form 
\begin{align}
\ket{\Psi_s}&=\ket{\Psi(\tilde{O}_s\diamond A,\tilde{X}_s\otimes X,n)}\ ,
\end{align}
with bond dimension~$2(s+1)$. 
\end{lemma}

\begin{proof}
Consider the MPO  tensor $O^{\diamond s}$ associated with the MPO representation~\eqref{eq:sminusnoncompressed} of~$S_-^s$. We express it in terms of matrices $\{O_{i,j}\}_{i,j\in \{0,1\}}\subset\cB((\mathbb{C}^2)^{\otimes s})$ acting on the virtual space of dimension $D=2^s$. The latter has orthonormal basis $\{\ket{\alpha}=\ket{\alpha_1}\otimes\cdots\otimes\ket{\alpha_s}\}_{\alpha=(\alpha_1,\ldots,\alpha_s)\in\{0,1\}^s}$. By definition~\eqref{eq:bmatrixsminus} of   $O$  and the fact that $(O^{1,0})^2=\sigma_-^2=0$, it is easy to see that
\begin{align}
\begin{matrix}
\bra{\alpha}O_{0,0}\ket{\beta}&=&\delta_{\alpha,\beta}\\
\bra{\alpha}O_{1,1}\ket{\beta}&=&\delta_{\alpha,\beta}\\
\bra{\alpha}O_{1,0}\ket{\beta}&=&0
\end{matrix}\ ,\qquad\textrm{ and }\qquad 
\bra{\alpha}O_{0,1}\ket{\beta}&=\begin{cases}
1\qquad &\textrm{ if }\beta \preceq \alpha\\
0\qquad &\textrm{ otherwise }
\end{cases}\ ,
\end{align}
where we write $\beta\preceq\alpha$ for $\alpha, \beta\in \{0,1\}^s$ if and only if there is exactly one $k\in [s]$ such that $\beta_k=0$ and $\alpha_k=1$, and $\alpha_\ell=\beta_\ell$ for all $\ell\neq k$. 
Let us define $\halfj_k^-$ as the operator $\ket{0}\bra{1}$ acting on the $k$-factor in $(\mathbb{C}^2)^{\otimes s}$, and $\halfj_k^+=(\halfj_k^-)^\dagger$ for $k\in [s]$. Then it is easy to check that $\halfj_+:=\sum_{k=1}^s \halfj_k^+$ has the same matrix elements~$\bra{\alpha}O_{0,1}\ket{\beta}$ as $O_{0,1}$.
It follows that
\begin{align}
\begin{matrix}
O_{0,0}&=&I_{(\mathbb{C}^2)^{\otimes s}}\\
O_{1,1}&=&I_{(\mathbb{C}^2)^{\otimes s}}
\end{matrix}\ ,\qquad\textrm{ and }\qquad
\begin{matrix}
O_{1,0}&=&0\\
O_{0,1}&=&\halfj_+
\end{matrix}\ .\label{eq:Bexplicitcompressedxz}
\end{align}
According to the MPO representation~\eqref{eq:sminusnoncompressed} of~$S_-^s$, the matrix elements of this operator can be expressed as 
\begin{align}
\bra{i_1\cdots i_n}S_-^s\ket{j_1\cdots j_n}&=\bra{1}^{\otimes s}O_{i_1 j_1}\cdots O_{i_nj_n}\ket{0}^{\otimes s}
\qquad\textrm{ for all }(i_1,\ldots,i_n), (j_1,\ldots,j_n)\in \{0,1\}^n\ .
\end{align}
Combining this expression with~\eqref{eq:Bexplicitcompressedxz}, it follows that
\begin{align}
\bra{i_1\cdots i_n}S_-^s\ket{j_1\cdots j_n}&=\bra{1}^{\otimes s}(\tilde{O}_s)_{i_1 j_1}\cdots (\tilde{O}_s)_{i_nj_n}\ket{0}^{\otimes s}
\quad\textrm{ for all }(i_1,\ldots,i_n), (j_1,\ldots,j_n)\in \{0,1\}^n\ .
\end{align}
where $(\tilde{O}_s)_{i,j}$ is the restriction of $(O^{\diamond s})_{i,j}$ to the subspace $\mathsf{span}\{(\halfj^+)^r\ket{0}^{\otimes s}\ |\ r=0,\ldots,s\}$. This implies the claim. 
\end{proof}

\subsection{Action of the symmetric group on the magnon states\label{sec:magnonpermutationinvariance}}
The symmetric group $S_n$ acts on $(\mathbb{C}^2)^{\otimes n}$ by permuting the factors, i.e., we have for an orthonormal basis $\{\ket{e_1},\ket{e_2}\}\in\mathbb{C}^2$ that 
\begin{align}
\pi (\ket{e_{i_1}}\otimes\cdots\otimes \ket{e_{i_n}})&=\ket{e_{i_{\pi^{-1}(1)}}}\otimes \cdots \otimes \ket{e_{i_{\pi^{-1}(n)}}}\qquad\textrm{ for all }\pi\in S_n\ ,
\end{align}
and this is linearly extended to all of $(\mathbb{C}^2)^{\otimes n}$. 
Since 
\begin{align}
[\pi,S_-] &=0\qquad\textrm{ for all }\pi\in S_n\ ,\mylabel{eq:permutationinvariancesplus}
\end{align}
the space $\mathsf{span}\{S_-^k\ket{\Psi}\ |\ k\in\mathbb{N}_0 \}$ is invariant under permutations. In the following, we will show that the restriction of the group action to this space has a particularly simple form: every permutation acts as a tensor product of diagonal unitaries. Our main claim (Theorem~\ref{lem:permutations} below) follows from~\eqref{eq:permutationinvariancesplus} and the following statement.

\begin{lemma}\mylabel{lem:permutationaction}
Let $A_0,A_1\in\cB(\mathbb{C}^2)$ be the matrices  defining the MPS $\ket{\Psi}$, cf.~equation~\eqref{eq:amatrixmps}. Then
\begin{align}
A_cA_b &=\omega^c \overline{\omega}^b A_bA_c \qquad\textrm{ for all }\qquad b,c\in \{0,1\}\ .\mylabel{eq:neighboringtransposition}
\end{align}
Consider the MPO tensor $O$ defined by equation~\eqref{eq:bmatrixsminus}
and set 
$\cO_{a,b}=\cO(O,\ket{b}\bra{a},2)\in\cB((\mathbb{C}^2)^{\otimes 2})$ for $a,b\in \{0,1\}$.
Then 
\begin{align}
\cO_{a,b}(Z^{\dagger}\otimes Z)=(Z^{\dagger}\otimes Z)\cO_{a,b}\qquad\textrm{ for all }a,b\in\{0,1\}\ ,\mylabel{eq:mpozcompat}
\end{align}
where $Z=\mathsf{diag}(1,\omega)$. 
\end{lemma}

It is convenient to express the corresponding statements diagramatically. First observe that specializing~\eqref{eq:permutationinvariancesplus} to a neighboring transposition and inserting the MPO description of $S_-$ introduced in Section~~\ref{sec:mpsdescriptionmagnon}, we obtain the diagrammatic identity
\begin{align}
\includegraphics[width=12cm]{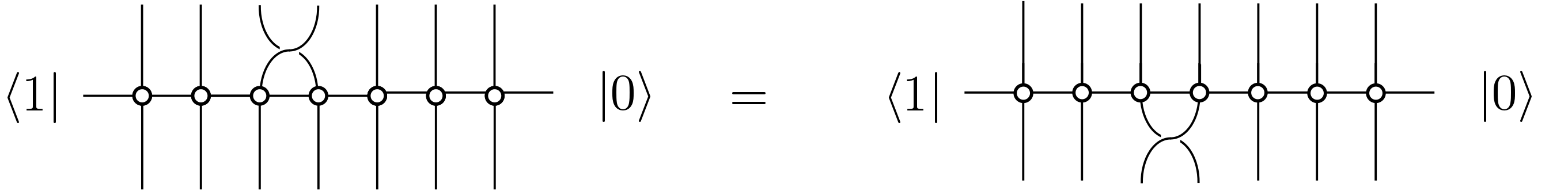}\mylabel{eq:braidingandmpo}
\end{align}
Claim~\eqref{eq:neighboringtransposition} describes the action of a neighboring transposition and can be written as
\begin{align}
\includegraphics[width=8.5cm]{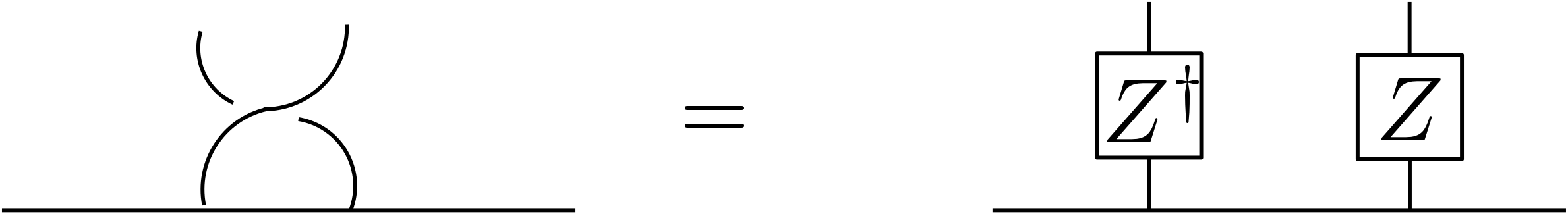}\mylabel{eq:braidingmps}
\end{align}
Claim~\eqref{eq:mpozcompat} can be written as
\begin{align}
\includegraphics[width=8.5cm]{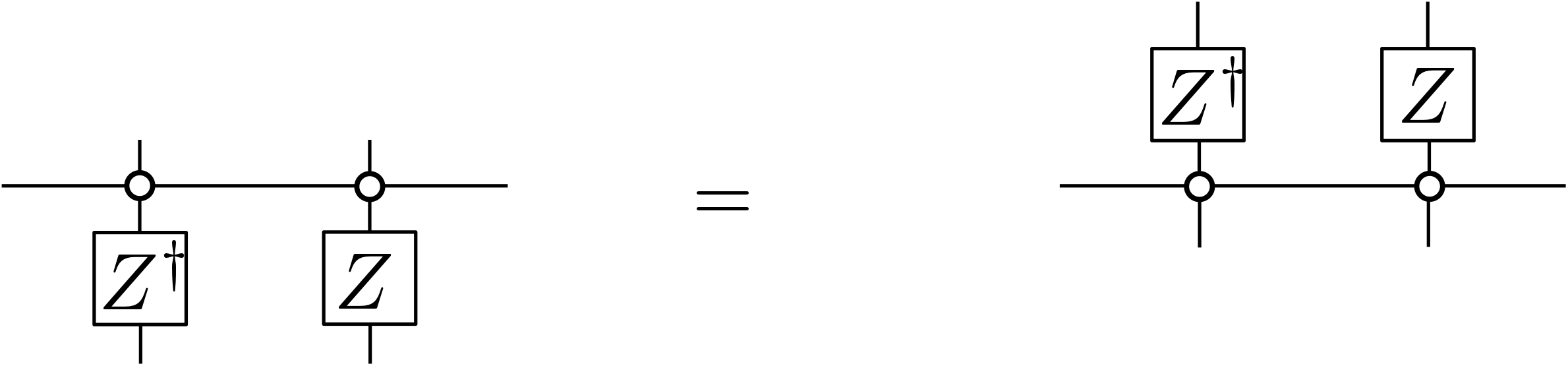}\mylabel{eq:zoperatorsmpolem}
\end{align}
\begin{proof}
Equation~\eqref{eq:neighboringtransposition} can be shown by checking each case:
\begin{center}
\includegraphics[width=16cm]{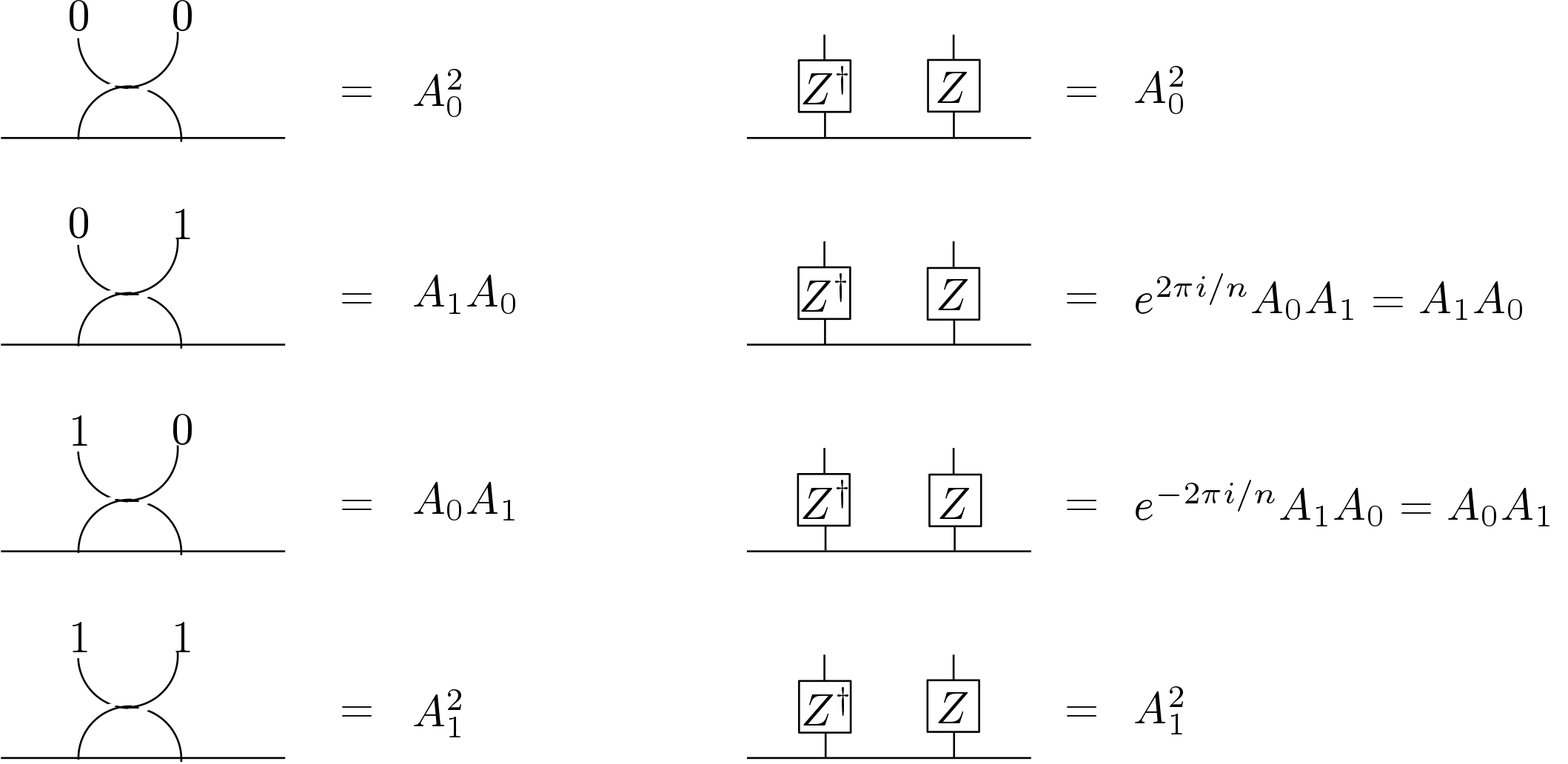}
\end{center}
Similarly,~\eqref{eq:mpozcompat} is shown by direct computation.
\end{proof}

The main feature we need in what follows is the following statement:
\begin{lemma}\mylabel{lem:permutations}
Consider the spin $j=n/2-1$ subspace~$\cH_{n/2-1}\subset (\mathbb{C}^2)^{\otimes n}$ introduced in equation~\eqref{eq:subspacenhalfminusone}. Let $\tau=(k\  k+1)\in S_n$ be an arbitrary transposition of nearest neighbors. Then the restriction of $\tau$ to $\cH_{n/2-1}$ is given by the operator
\begin{align}
\tau|_{\cH_{n/2-1}}&= I^{\otimes k-1}\otimes Z^{\dagger}\otimes Z\otimes I^{\otimes n-k-1}\ ,
\end{align}
where $I=\identityoperator_{\mathbb{C}^2}$.
\end{lemma}
\begin{proof}
It suffices to check that $\tau S_-^s\ket{\Psi}= (I^{\otimes k-1}\otimes Z^{\dagger}\otimes Z\otimes I^{\otimes n-k-1})S_-^s\ket{\Psi}$. This follows immediately from Lemma~\ref{lem:permutationaction}. 
A diagrammatic proof of the steps involved can be given as follows (illustrated for $s=3$):
\begin{align}
\includegraphics[width=11cm]{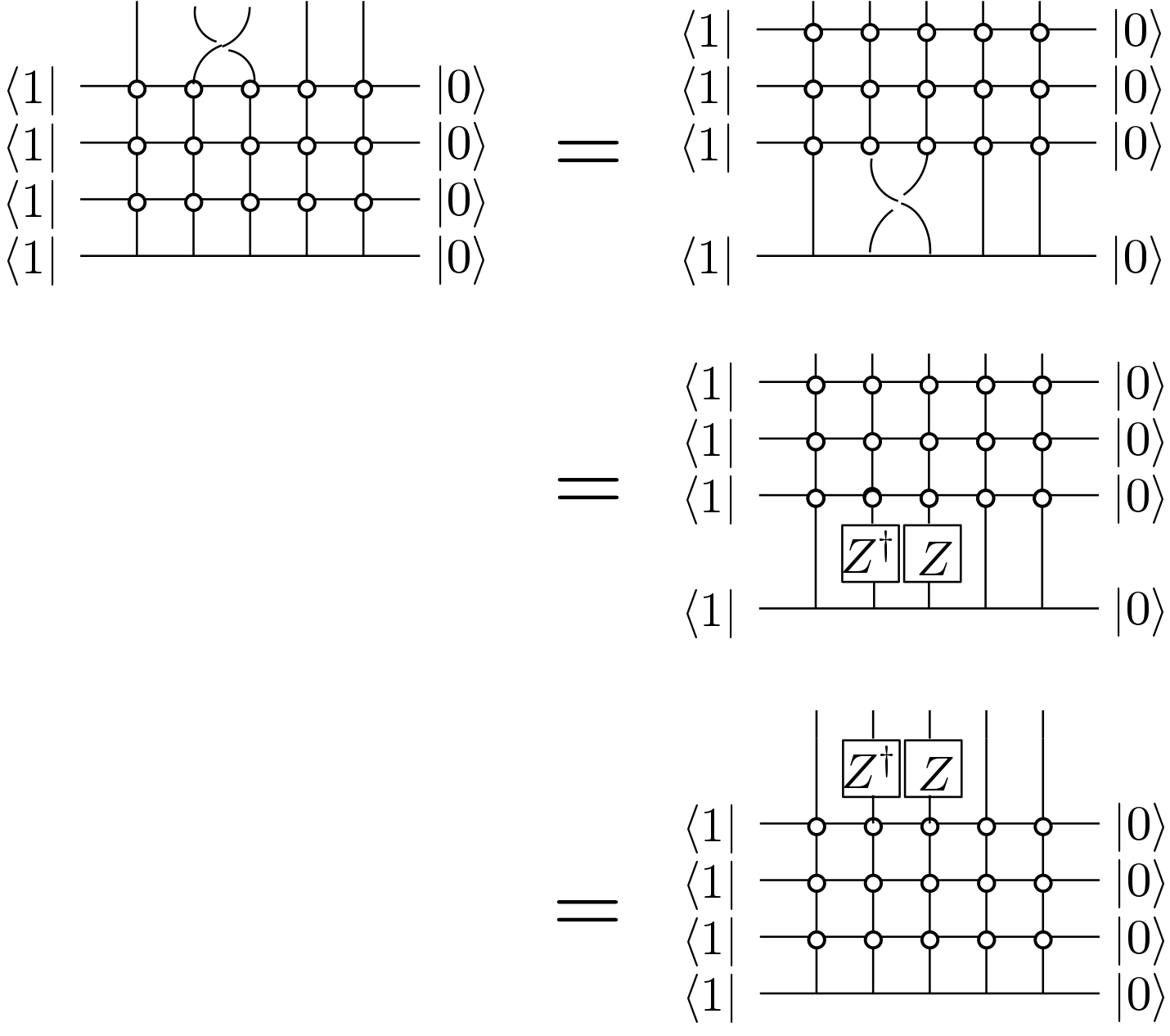}
\end{align}
Here we used~\eqref{eq:braidingandmpo} $s$~times in the first identity, equation~\eqref{eq:braidingmps} in the second identity, and equation~\eqref{eq:zoperatorsmpolem} (applied $s$~times) in the last step. 
\end{proof}

An immediate and crucial consequence of Lemma~\ref{lem:permutations} and the unitarity of $Z$ is the fact that matrix elements of an operator acting on $d$~arbitrary sites can be related to matrix elements of a  local operator on the $d$~first sites.  To express this concisely, we use the following notation: suppose $F=F_1\otimes\cdots \otimes F_d\in \cB((\mathbb{C}^2)^{\otimes d})$ is a tensor product operator and $A=\{a_1<\cdots <a_d\}\subset [n]$ a subset of $d=|A|$ (ordered) sites. Then we write $F_A\otimes \identityoperator_{[n]\backslash A}\in \cB((\mathbb{C}^2)^{\otimes n})$ for the operator acting as $F_k$ on site $a_k$, for $k\in [d]$. By linearity, this definition extends to general (not necessarily product) operators $F\in \cB((\mathbb{C}^2)^{\otimes d})$. Note that if $A=[d]$ are the first $d$ sites, then $F_A\otimes \identityoperator_{[n]\backslash A}=F\otimes I^{\otimes n-d}$. 

\begin{lemma}\mylabel{lem:braidingrelation}
Consider the magnon states $\ket{\Psi_\ell}=S_-^\ell\ket{\Psi}$ and let $r,s\in \{0,\ldots,n-2\}$ be arbitrary.  
Suppose $F\in\cB((\mathbb{C}^2)^{\otimes d})$ acts on a subset $A\subset [n]$ of $d=|A|$ sites. Then
\begin{align}
\bra{\Psi_r}(F_A\otimes \identityoperator_{[n]\backslash{A}})\ket{\Psi_s}&=
\bra{\Psi_r}(\tilde{F}_{[d]}\otimes I^{\otimes n-d})\ket{\Psi_s} \mylabel{eq:permutetobeginning}
\end{align}
where $\tilde{F}\in \cB((\mathbb{C}^2)^{\otimes d})$ is given by 
\begin{align}
\tilde{F}=((Z^{\dagger})^{a}\otimes (Z^{\dagger})^{a+1} \cdots \otimes (Z^{\dagger})^{a+d})  F
(Z^{a}\otimes Z^{a+1} \cdots \otimes Z^{a+d})\ 
\end{align}
where $a=(\min A) -1$. 

More generally, if $B\subset [n]$ is a subset of size $b=|B|$ located ``to the right of $A$'' (i.e., if $\min B>\max A$) and $G\in \cB((\mathbb{C}^2)^{\otimes b})$, then 
\begin{align}
\bra{\Psi_r}(F_A\otimes G_B\otimes  \identityoperator_{[n]\backslash (A\cup B)})\ket{\Psi_s}&=
\bra{\Psi_r}(\tilde{F}_{[d]}\otimes G_B\otimes \identityoperator_{[n]\backslash ([d]\cup B)})\ket{\Psi_s}\mylabel{eq:permutetobeginningwithop}
\end{align}
Furthermore, the analogous statement holds when $G$ is permuted to the right, but with $Z$ replaced by $Z^{\dagger}$.
\end{lemma}
Succintly, equation~\eqref{eq:permutetobeginningwithop} can be represented as follows in the case where $A$ consists of a connected set of sites (and $r=s=0$):
\begin{center}
\includegraphics[width=12cm]{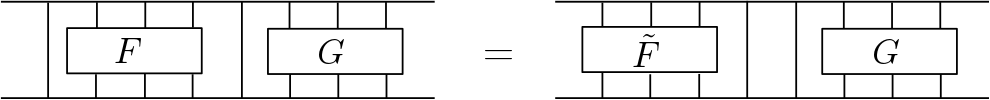}
\end{center}
We emphasize, however, the analogous statement is true for the more general case where $A$ is a union of disconnected components. 

\begin{proof}
The proof of equation~\eqref{eq:permutetobeginningwithop} for a single-site operator $F$ is immediate. We have 
(illustrated for $r=s=0$):
\begin{center}
\includegraphics[width=9cm]{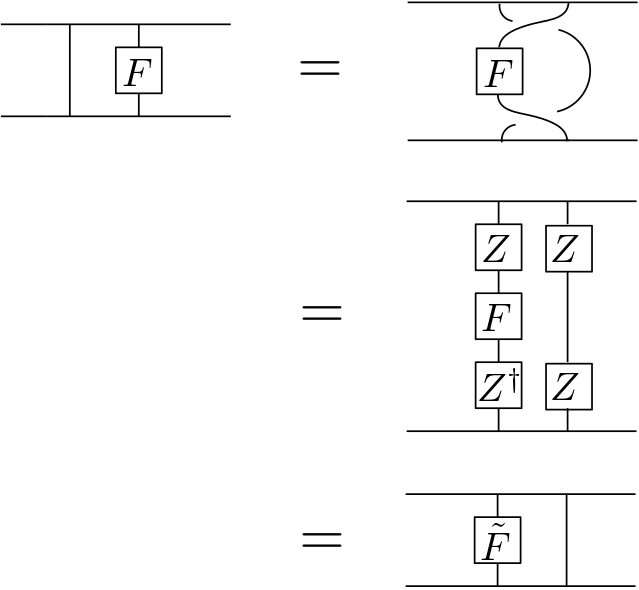}
\end{center}
Applying~\eqref{eq:permutetobeginningwithop}  (with $d=1$) iteratively then shows that the claim~\eqref{eq:permutetobeginningwithop} also holds for any tensor product operator $F=F_1\otimes\cdots\otimes F_d$. The general claim then follows by decomposing an arbitrary operator~$F$ into tensor products and using linearity.
\end{proof}

\subsection{The transfer operator of the magnon states and its Jordan structure\label{sec:jordanstructuremagnon}}
We are ultimately interested in matrix elements $\bra{\Psi_r} (F\otimes I^{\otimes (n-d)})\ket{\Psi_s}$ 
where $F\in\cB((\mathbb{C}^2)^{\otimes d})$ acts on $d$~sites. Using the compressed representation from Lemma~\ref{lem:compressedrep}, we may write these as 
\begin{align}
\bra{\Psi_r} (F\otimes I^{\otimes (n-d)})\ket{\Psi_s}&=
\bra{\Psi(\tilde{O}_r\diamond A,\tilde{X}_r\otimes X,n)}
(F\otimes I^{\otimes (n-d)})
\ket{\Psi(\tilde{O}_s\diamond A,\tilde{X}_s\otimes X,n)}\ .
\end{align}

We are thus interested in the (``overlap'') transfer
operator
\begin{align}
E_{r,s}&=E(\tilde{O}_r\diamond A,\tilde{O}_s\diamond A)\qquad\textrm{ for }r,s\in [n]\ .
\end{align}
For convenience, let us also set 
\begin{align}
E_{0,s}&=E(A,\tilde{O}_s\diamond A)\qquad\textrm{ for }s\in [n]\ ,\\
E_{r,0}&=E(\tilde{O}_r\diamond A,A)\qquad\textrm{ for }r\in [n]\ ,\\
E_{0,0}&=E(A,A)\ .
\end{align}
Observe that $E_{0,0}$ is the transfer operator~$E$ of the MPS~$\ket{\Psi}$, whereas $E_{0,s}$ is the transfer operator of $S_-^s\ket{\Psi}$. We will order the tensor factors such that the virtual spaces of the original MPS are the first two factors. Then $E_{r,s}\in\cB(\mathbb{C}^2\otimes\mathbb{C}^2\otimes \mathbb{C}^{r+1}\otimes\mathbb{C}^{s+1})$. Our main goal in this section is to show the following:
 
\begin{theorem}\label{thm:mainjordanmagnon}
Let $r,s\in \{0,\ldots,n\}$ be arbitrary. 
Then the operator $E_{r,s}\in \cB(\mathbb{C}^2\otimes\mathbb{C}^2\otimes \mathbb{C}^{r+1}\otimes\mathbb{C}^{s+1})$ has spectrum $\mathsf{spec}(E_{r,s})=\{1,\omega,\overline{\omega}\}$ 
(where $1$ has multiplicity $2\cdot (r+1)(s+1)$ and $\omega,\overline{\omega}$ each have multiplicity 
$(r+1)(s+1)$). The size $h^*$ of the largest Jordan block in $E_{r,s}$ is bounded by
\begin{align}
h^*\leq \min\{r,s\}+2\ .\label{eq:hstarupperboundx}
\end{align}
\end{theorem}

 To prove this theorem, we first rewrite the operator $E_{r,s}$.  We have 
\begin{align}
E_{r,s}&=E(\tilde{O}_r\diamond A,\tilde{O}_s\diamond A)=
E(A,\tilde{O}_r^\dagger\diamond \tilde{O}_s\diamond A)\ ,
\end{align}
where $\tilde{O}_r^\dagger$ is obtained from the defining matrices $\{\tilde{O}^{\alpha,\beta}\}_{\alpha,\beta}$ of $\tilde{O}$ by replacing $\tilde{O}^{\alpha,\beta}$ with its adjoint $(\tilde{O}^\dagger)^{\alpha,\beta}$. This amounts to replacing $\sigma_-$ by $\sigma_+$, or alternatively, swapping the indices in the defining matrices $\{(\tilde{O}_s)_{i,j}\}$ (cf.~\eqref{eq:btildeijdef}). That is,
\begin{align}
\begin{matrix}
(\tilde{O}_s)_{0,0}=(\tilde{O}_s)_{1,1}&=&I_{\mathbb{C}^{s+1}}\\
(\tilde{O}_s)_{1,0}&=&0\\
(\tilde{O}_s)_{0,1}&=&J_{+,s}
\end{matrix}\ ,\qquad\textrm{ and }\qquad 
\begin{matrix}
(\tilde{O^\dagger_r})_{0,0}=(\tilde{O^\dagger}_r)_{1,1}&=&I_{\mathbb{C}^{r+1}}\\
(\tilde{O^\dagger_r})_{1,0}&=&J_{+,r}\\
(\tilde{O^\dagger_r})_{0,1}&=&0
\end{matrix}\ .
\end{align}
where $J_{+,s}$ and $J_{+,r}$ are the raising operators in the spin-$s/2$ respectively the spin-$r/2$ representation, respectively. We conclude that
\begin{align}
E_{r,s}&=E\otimes I_{\mathbb{C}^{r+1}}\otimes I_{\mathbb{C}^{s+1}}+E_{\sigma_+}\otimes J_{+,r}\otimes I_{\mathbb{C}^{s+1}}+E_{\sigma_-}\otimes I_{\mathbb{C}^{r+1}}\otimes J_{+,s}+E_{\sigma_+ \sigma_-}\otimes J_{+,r}\otimes J_{+,s} \ ,
\end{align}
where $S_+$ are the  raising operator $S_+$ of the spin-$j$ representation with $j=r/2$ and $j=s/2$, respectively. Here
\begin{align}
E&=\proj{00}+\omega\proj{01}+\overline{\omega}\proj{10}+\proj{11}+\ket{11}\bra{00}\ ,\\
E_{\sigma_-}&=\ket{10}\bra{00}+\omega \ket{11}\bra{01}\ ,\\
E_{\sigma_+}&=\ket{01}\bra{00}+\overline{\omega}\ket{11}\bra{10}\ ,\\
E_{\sigma_+\sigma_-}&=\proj{00}+\overline{\omega}\proj{10}+\omega\proj{01}+\proj{11}\ , 
\end{align}
are the transfer operators of the MPS~$\ket{\Psi}$.  We can write down the transfer matrix $E_{r,s}$ more  explicitly as
\begin{align}
E_{r,s} &= \overline{A_0} \otimes A_0 \otimes I \otimes I +\overline{ A_1} \otimes A_1 \otimes I \otimes I + \overline{A_0}\otimes A_1\otimes I \otimes J_{+,s}\nonumber\\
& \ \ \ \ + \overline{A_1}\otimes A_0 \otimes J_{+,r} \otimes I + \overline{A_1}\otimes A_1 \otimes J_{+,r}\otimes J_{+,s}\ , \label{eq:ersdefinitions}
\end{align}
where
\begin{equation}
    A_0= \sigma_+ = \begin{pmatrix}0 & 0 \\ 1 & 0\end{pmatrix},\ \ \ \ \ \ \ \text{and}\ \ \ \ \ \ \ A_1 = \begin{pmatrix} 1 & 0 \\ 0 & \omega\end{pmatrix}
\end{equation}
are the defining tensors of the original state~$\ket{\Psi}$ (cf.~\eqref{eq:amatrixmps}). With this, we can give the proof of the above theorem as follows.

\begin{proof}[Proof of Theorem~\ref{thm:mainjordanmagnon}]
Observe that in the standard basis of the spin-$j$-representation, the raising operator $J_+$ is strictly lower diagonal. From~\eqref{eq:ersdefinitions} and the definition of $A_0$ and $A_1$, it follows that the transfer operator~$E_{r,s}$ is lower diagonal in the tensor product basis (consisting of these standard bases and the computational basis of~$\mathbb{C}^2$)   since each term in the sum is a tensor product of lower diagonal matrices. In fact, every term except 
\begin{align}
D \equiv \overline{A_1} \otimes A_1 \otimes I \otimes I
\end{align}
is strictly lower diagonal. Therefore, we see that the eigenvalues of $E_{r,s}$ are given by the diagonal entries of $D$, and consist of the eigenvalue $1$ with multiplicity $2(r+1)(s+1)$, and the eigenvalues $\omega$ and $\overline{\omega}$, both with multiplicity $(r+1)(s+1)$. Observe that $A_0$ and $A_1$  commute up to a factor of $\omega$, that is,
\begin{align}
A_1A_0 = \omega A_0A_1\ .\label{eq:acommut}
\end{align}
To shorten some of the expressions, let define
\begin{align}
N_1 =\overline {A_0} \otimes A_0\ ,\qquad N_2 = \overline{A_0} \otimes A_1\ ,\qquad \ N_3 = \overline{A_1} \otimes A_0\ ,  \qquad \text{and}\qquad   A = \overline{A_1} \otimes A_1\ .
\end{align}    
We note that each  $N_i$ is a nilpotent matrix of order $2$, i.e.,
\begin{align}
N_i^2 = 0\qquad\textrm{ for }i=1,2,3\label{eq:nilpotencyreln}
\end{align} since $A_0^2=0$. 
 Moreover, for the same reason and~\eqref{eq:acommut}, we have 
\begin{align}
N_2N_1 =N_1N_2 = N_3N_1 = N_1N_3= 0 \qquad \text{and}\qquad N_2N_3 = \omega^2 N_3 N_2\ .\label{eq:productrls}
\end{align}
equation~\eqref{eq:acommut} also implies that
\begin{align}
N_i A&=q_i(\omega) AN_i\qquad\textrm{ for}\qquad i=1,2,3\ ,\label{eq:comxpowy}
\end{align}
where $q_i(\omega)\in \{1,\omega,\overline{\omega}\}$. 
Now consider the transfer operator with its diagonal term removed, i.e.,
\begin{align} 
E_{r,s} - D = N_1 \otimes I \otimes I + N_2\otimes I \otimes J_{+,s} + N_3 \otimes J_{+,r} \otimes I + A \otimes J_{+,r}\otimes J_{+,s}\ .
\end{align}
Let $\mathbb{C}[\omega,\overline{\omega}]$ be the set of polynomials in $\omega$ and $\overline{\omega}$. Let us define the set\begin{align}
    \cX &= \bigg\{p_1(\omega,\overline{\omega})\, N_1 \otimes I \otimes I + p_2(\omega,\overline{\omega})\, N_2 \otimes I \otimes J_{+,s}\\ 
    &\qquad + p_3(\omega,\overline{\omega})\, N_3 \otimes J_{+,r} \otimes I + p_4(\omega,\overline{\omega})\, A \otimes J_{+,r} \otimes J_{+,s}\enskip \bigg|\enskip p_i\in \mathbb{C}[\omega,\overline{\omega}]\bigg\}\ 
\end{align}
such that $E_{r,s}-D\in\cX$. The key properties of $\cX$ which we need are the following:
\begin{enumerate}[(i)]
\item\label{eq:commutativityxone}
 If $X_1 \in \cX$, then $DX_1 = X_2D$ and $X_1D = DX_3$ for some $X_2,X_3\in \cX$.
 \item\label{eq:commutativityxtwo}
 The product of any $\min\{r,s\}+2$ operators in $\cX$ is equal to zero.
 \end{enumerate}
 Property~\eqref{eq:commutativityxone} follows immediately with the commutation relation~\eqref{eq:comxpowy}
  because $D=A\otimes I\otimes I$. Similarly, property~\eqref{eq:commutativityxtwo} follows from the nilpotency relation~\eqref{eq:nilpotencyreln}, the commutation relation~\eqref{eq:comxpowy} and the fact that 
  \begin{align}
  J_{+,r}^{\min\{r,s\}+2}=  J_{+,s}^{\min\{r,s\}+2}=0\ .
  \end{align}
  We can write these two properties succintly as equalities of sets, that is,
 \begin{align}
 D\cX&=\cX D\ ,\qquad\textrm{ and }\qquad \label{eq:dcxcommutativity}\\
\cX^m&=\{0\}\qquad\textrm{ for all }m\geq \min\{r,s\}+2\ ,\label{eq:exponentialxmvanish}
 \end{align}
 where e.g., $\cX^2=\{X_1X_2\ | X_1,X_2\in\cX\}$.   Let us write $D_\lambda = D-\lambda I$. Then 
 \begin{align}
 E_{r,s}-\lambda I=D_\lambda+(E_{r,s}-D)\in D_\lambda+\cX\ .
 \end{align}
 In particular, for $\ell,m,n\in\mathbb{N}_0$ we have 
 \begin{align}
  (E_{r,s} - I)^\ell(E_{r,s} - \omega I)^m(E_{r,s} - \overline{\omega}I)^n 
& \in \left(D_1 + \cX\right)^\ell\left(D_\omega + \cX\right)^m\left(D_{\overline{\omega}} + \cX\right)^n\\
  &\subseteq \sum_{\substack{a\in \{0,\ldots,\ell\}\\
  b\in \{0,\ldots,m\}\\
  c\in \{0,\ldots,n\}}}D_1^aD_\omega^bD_{\overline{\omega}}^c\cX^{(\ell-a)+(m-b)+(n-c)}\ ,
 \end{align}
 where in the last step, we used
 the binomial expansion,
 the pairwise commutativity of the matrices
 $D_1$, $D_\omega$ and $D_{\overline{\omega}}$, and~\eqref{eq:dcxcommutativity}. 
Since $D_1D_\omega D_{\overline{\omega}} = 0$, the non-zero terms in the expansion must have at least one of $a,b,c$ equal to zero. Choosing
\begin{align}
\ell = m = n = \min\{r,s\}+2\ ,
\end{align}
the exponent $(\ell-a)+(m-b)+(n-c)$ is lower bounded by $\min\{r,s\}+2$ 
for any such triple $(a,b,c)$. We conclude with~\eqref{eq:exponentialxmvanish}
that 
\begin{align}
\sum_{\substack{
  a\in \{0,\ldots,\ell\}\\
  b\in \{0,\ldots,m\}\\
  c\in \{0,\ldots,n\}  
  }} D_1^aD_\omega^bD_{\overline{\omega}}^c\cX^{(\ell-a)+(m-b)+(n-c)}&=\{0\}\ ,
\end{align}
and thus
\begin{align}
(E_{r,s}-I)^{\min\{r,s\}+2}(E_{r,s}-\omega I)^{\min\{r,s\}+2}(E_{r,s}-\overline{\omega}I)^{\min\{r,s\}+2}=0\ .
\end{align}
Therefore the minimal polynomial of $E_{r,s}$ must divide~$p(x)=[(x-1)(x-\omega)(x-\overline{\omega})]^{\min\{r,s\}+2}$. Thus the Jordan blocks of $E_{r,s}$ are bounded above in size by $\min\{r,s\}+2$, as claimed. 
\end{proof}

\subsection{Matrix elements of magnon states\label{sec:matrixelementsmagnon}}
With the established bounds on the Jordan structure of $E_{r,s}$, we can derive  upper bounds on overlaps of magnon states. Recall that $\ket{\Psi_r}=S_-^r\ket{\Psi}$ for $r=0,\ldots,n-2$ and $\ket{\psi_r}$ is its normalized version (cf.~\eqref{eq:psipdefinitionbasic}) .

\begin{theorem}\mylabel{thm:upperboundoverlap}
Let $F\in\cB((\mathbb{C}^2)^{\otimes d})$ be such that $\|F\|\leq 1$. Let $r\neq s$. 
Then 
\begin{align}
|\bra{\psi_r} (F\otimes I^{\otimes (n-d)})\ket{\psi_s}|=O\left(\frac{d}{n^{|s-r|/2}}\right)\ .
\end{align}
\end{theorem}
\begin{proof}
We can take the complex conjugate, effectively interchanging $r$ and $s$. Thus we can without loss of generality assume that $r<s$.  
Recall that $\ket{\Psi_r}$ and  $\ket{\Psi_s}$
can be represented as MPS using bond dimensions
$D_r=2(r+1)$, $D_s=2(s+1)$
such that the associated transfer operators $E_{r}$, $E_s$ and the combined transfer operator~$E_{r,s}$ all have spectrum~$\{1,\omega,\omega\}$ and 
Jordan blocks bounded by $2$,~$2$, and $\min\{r,s\}+2=r+2$, respectively; see Theorem~\ref{thm:mainjordanmagnon}.  Applying Theorem~\ref{thm:basicboundmatrixelements}
  (with $h_1^*=2$, $h_2^*=2$, $h^*=r+2$) 
we get
\begin{align}
|\bra{\Psi_r}(F\otimes I^{\otimes n-d})\ket{\Psi_s}|&\leq 16\cdot 
d(n-d)^{r+1}=O(d\cdot n^{r+1})\ .
\end{align}
Inserting the normalization~\eqref{eq:psipdefinitionbasic}
\begin{align}
    \|\Psi_s\|^2&=\frac{n(n-2)!s!}{(n-2-s)!}  \geq s!\cdot n^{s+1}(1-O(s^2/n))\mylabel{eq:asymptoticscalingpsinorms}
\end{align}
gives
\begin{align}
|\bra{\psi_r}(F\otimes I^{\otimes n-d})\ket{\psi_s}|&=\frac{d n^{r+1}}{(r!s!)^{1/2}n^{(r+s)/2+1}}\cdot (1+O(s^2/n))=O(d\cdot n^{-(s-r)/2})
\end{align}as claimed. 
\end{proof}
We also need estimates for the difference of expectation values of magnon states. Let us first show that the reduced $d$-local operators are all essentially the same.

\begin{lemma}\mylabel{lem:lowerboundtraceoverlap}
Let $\{\ket{\psi_s}\}_{s=0}^{n-2}$ be the normalized magnon-states defined in equation~\eqref{eq:psipdefinitionbasic}. Then
\begin{align}
\bra{1}^{\otimes d}\left(\tr_{n-d}\proj{\psi_s}\right)\ket{1}^{\otimes d}\geq 1-O(ds/n)\ 
\end{align}
for all $s\in [n]$. 
\end{lemma}
\begin{proof}
Let us define
\begin{align}
\ket{\Psi^k_s}&=S_-^s\left(\sum_{j=1}^k \omega^j \sigma_j^-\ket{1}^{\otimes k}\right)\  .
\end{align} 
Observe that for $d<n$
\begin{align}
    \ket{\Psi^n_0}&=
    \ket{\Psi^d_0}\otimes\ket{1}^{\otimes n-d}+
    \omega^d \ket{1}^{\otimes d}\otimes \ket{\Psi^{n-d}_0}\ .
\end{align}
Writing $S_-=S_-^A+S_-^B$
with $S_-^A=\sum_{j=1}^d \sigma_j^-$ and $S_-^B=\sum_{j=d+1}^n \sigma_j^-$
we get
\begin{align}
 \ket{\Psi_s^n}&=\sum_{\ell=0}^s \binom{s}{\ell} (S_-^A)^\ell (S_-^B)^{s-\ell}
 \ket{\Psi_0^n}\\
 &=\omega^d \ket{1}^{\otimes d}\otimes (S_-^B)^s\ket{\Psi_0^{n-d}}+\ket{\Phi}\\
 &=\omega^d \ket{1}^{\otimes d}\otimes\ket{\Psi_s^{n-d}}+\ket{\Phi}\ ,
\end{align}
for a vector $\ket{\Phi}\in(\mathbb{C}^\physical)^{\otimes n}$
 satisfying $(\proj{1}^{\otimes d}\otimes \identityoperator^{\otimes n-d})\ket{\Phi}=0$.
In particular, we have 
\begin{align}
(\proj{1}^{\otimes d}\otimes \identityoperator^{\otimes (n-d)})
\ket{\Psi^n_s}&=\omega^d \ket{1}^{\otimes d}\ket{\Psi^{n-d}_s}\ .    
\end{align}
Tracing out the $(n-d)$ qubits, it follows that 
\begin{align}
    \bra{1}^{\otimes d}\tr_{n-d}\left(\proj{\Psi^n_s}\right)\ket{1}^{\otimes d}
    &=\|\Psi_s^{n-d}\|^2\ . 
    \end{align}
  Rewriting the term using the normalized vector $\psi^n_s = \Psi^n_s/\| \Psi^n_s\|$, we get
    \begin{align}
     \bra{1}^{\otimes d}\left(\tr_{n-d}\proj{\psi^n_s}\right)\ket{1}^{\otimes d}
    &=\frac{\|\Psi_s^{n-d}\|^2}{\|\Psi_s^{n}\|^2}\ .
    \end{align}
    Observe that the norm $\|\Psi^k_s\|^2$
    is a matrix element of the operator $E_{s,s}^k$.     With Lemma~\ref{lem:normscalingtransferop}~\eqref{it:thirdjordan}    we obtain
    \begin{align}
         \bra{1}^{\otimes d}\left(\tr_{n-d}\proj{\psi^n_s}\right)\ket{1}^{\otimes d}
    &=\frac{c\cdot (n-d)^{\ell-1}(1+O((n-d)^{-1}))}{c\cdot n^{\ell-1}(1+O(n^{-1}))}\\
    &=(1-d/n)^{\ell-1}(1+O((n-d)^{-1}))\\
    &\geq (1-d(\ell-1)/n))\cdot (1+O((n-d)^{-1}))\\
    &\geq 1-\left(\frac{d(\ell-1)}{n}+O(1/n)\right)\ .
    \end{align}
    for some $\ell\in \{1,\ldots,h^*\}$, where $h^*$ is the size of the largest Jordan block of the transfer operator~$E_{s,s}$. 
    Since $h^*\leq s+2$ by Theorem~\ref{thm:mainjordanmagnon}, the claim follows.

    \end{proof}

\begin{theorem}\mylabel{thm:matrixelementsdiagonal}
Let $F\in\cB((\mathbb{C}^2)^{\otimes d})$ be such that $\|F\|\leq 1$. 
Fix some $s_0\leq n-2$. 
Then 
\begin{align}
\big|\bra{\psi_{s}}(F\otimes I^{\otimes n-d})\ket{\psi_{s}}-\bra{\psi_{r}}(F\otimes I^{\otimes n-d})\ket{\psi_{r}}\big|&=O(\sqrt{\textfrac{ds_0}{n}})\qquad\textrm{ for all }\qquad r,s\leq s_0\ .
\end{align}
\end{theorem}

\begin{proof}
For any $F\in \cB((\mathbb{C}^2)^{\otimes d})$ with $\|F\|\leq 1$ we have 
\begin{align}
\big|\bra{\psi_{s}}(F \otimes I^{\otimes n-d} )\ket{\psi_{s}}-\bra{1}^{\otimes d} F\ket{1}^{\otimes d}\big| &\leq \| \tr _{n-d} \ket{\psi_s} \bra{\psi_s} - \ket 1 \bra 1 ^{\otimes n} \| \\
&\leq \sqrt{1-\bra{1}^{\otimes d} \left(\tr_{n-d} \ket{\psi_s} \bra{\psi_s}\right) \ket{1}^{\otimes d}}\ ,
\end{align}
using the Fuchs - van de Graaf inequality
$\frac{1}{2}\|\rho-\proj{\varphi}\|_1\leq \sqrt{1-\bra{\varphi}\rho\ket{\varphi}}$~\cite{fuchsgraaf}.
With Lemma \ref{lem:lowerboundtraceoverlap} we get 
\begin{align}
\big|\bra{\psi_{s}}(F \otimes I^{\otimes n-d} )\ket{\psi_{s}}-\bra{1}^{\otimes d} F\ket{1}^{\otimes d}\big| \leq O(\sqrt{\textfrac{ds_0}{n}})\ .
\end{align}
Using the triangle inequality, the claim follows.
\end{proof}

\subsection{The parameters of the magnon code\label{sec:parametersmagnon}}

Our main result is the following:

\begin{theorem}[Parameters of the magnon-code]\label{thm:magnoncodemain}
Let $\nu \in (0,1)$ and $\lambda,\kappa>0$ be such that
\begin{align}
6\kappa+\lambda<\nu\ .
\end{align}
Then there is a subspace~$\cC$ spanned by descendant states $\{S_-^r\ket{\Psi}\}_r$ with magnetization~$r$ pairwise differing by at least~$2$ such that  $\cC$ is an $(\epsilon,\delta)[[n,k,d]]$-AQEDC with parameters
\begin{align}
k&=\kappa \log_2 n\ ,\\
d&=n^{1-\nu}\ ,\\
\epsilon &=\Theta(n^{-(\nu-(6\kappa+\lambda))})\ ,\\
\delta&=n^{-\lambda}\ .
\end{align}
\end{theorem}
\begin{proof}
 We claim that the subspace
\begin{align}
\cC=\myspan \{\psi_s\ |\ s\textrm{ even  and } s\leq 2n^\kappa\}\ 
\end{align}
spanned by a subset of magnon-states has the claimed property. Clearly, $\dim \cC=n^\kappa=2^{k}$. 

Let $F$ be an arbitrary $d$-local operator on $(\mathbb{C}^2)^{\otimes n}$ of unit norm.
According to Lemma~\ref{lem:braidingrelation}, the following considerations 
concerning matrix elements $\bra{\psi_q}F\ket{\psi_p}$ of  magnon  states do  not 
depend on the location of the support of $F$ as we are interested in the supremum over~$d$-local operators~$F$ and unitary conjugation does not change the locality or the norm. Thus, we can assume that $F=\tilde{F}\otimes I^{\otimes n-d}$. That is, we have
\begin{align}
\sup_{\substack{F \textrm{ $d$-local}\\
\|F\|\leq 1} } |\bra{\psi_r} F\ket{\psi_s}| &=
\sup_{\substack{\tilde{F}\in\cB((\mathbb{C}^2)^{\otimes d})
\\
\|\tilde{F}\|\leq 1}}
 |\bra{\psi_r} (\tilde{F}\otimes I^{\otimes n-d})\ket{\psi_s}|
=O(\textfrac{d}{n^{|r-s|/2}})\qquad\textrm{ for }r,s\leq 2n^\kappa\ .
\end{align}
by Theorem~\ref{thm:upperboundoverlap}. 
In particular, if $|r-s|\geq 2$, then this is bounded by $O(d/n)$.  Similarly, 
\begin{align}
\sup_{\substack{F \textrm{ $d$-local}\\
\|F\|\leq 1}}
\big|\bra{\psi_{s}}F\ket{\psi_{s}}-\bra{\psi_{r}}F\ket{\psi_{r}}\big|
&=\sup_{\substack{\tilde{F}\in\cB((\mathbb{C}^2)^{\otimes d})\\
\|\tilde{F}\|\leq 1}}
\big|\bra{\psi_{s}}(\tilde{F}\otimes I^{\otimes(n-d)})\ket{\psi_{s}}-\bra{\psi_{r}}(\tilde{F}\otimes I^{\otimes (n-d)})\ket{\psi_{r}}\big|\\
&=O(\sqrt{\textfrac{dn^{\kappa}}{n}})=O(\sqrt{d n^{\kappa-1}})\qquad\textrm{ for }r,s\leq 2n^\kappa\ .
\end{align}
by Theorem~\ref{thm:matrixelementsdiagonal}.  Since $d/n=O(\sqrt{d n^{\kappa-1}})$, we conclude that
for all $d$-local operators $F$ of unit norm, we have 
\begin{align}
\big|\bra{\psi_r}F\ket{\psi_s}-\delta_{r,s}\bra{\psi_0}F\ket{\psi_0}|=O(d^{1/2}n^{(\kappa-1)/2})\qquad\textrm{ for all }r,s\textrm{ even  with } r,s\leq 2n^\kappa\ .
\end{align}
The sufficient conditions of Corollary~\ref{cor:errordectionapproximate} for approximate error-detection,
applied with $\gamma=\Theta(d^{1/2}n^{(\kappa-1)/2})$, thus imply that~$\cC$ is an $(\Theta(2^{5k}d n^{\kappa-1})/\delta,\delta)[[n,k,d]]$-AQEDC for any $\delta$ satisfying
\begin{align}
\delta>\Theta(2^{5k} dn^{\kappa-1})=\Theta(n^{6\kappa-\nu})\ .
\end{align}
for the choice $d=n^{1-\nu}$. With $\delta=n^{-\lambda}$, the claim follows. 
\end{proof}

\newpage

\myacknowledgements
We thank Ahmed Almheiri, Fernando Brand\~ao, Elizabeth Crosson, Spiros Michalakis, and  John Preskill  for discussions. We thank the Kavli Institute for Theoretical Physics for their hospitality as part of a follow-on program, as well as the coordinators of the  QINFO17 program, where this work was initiated; this research was supported in part by the National Science Foundation under Grant No. PHY-1748958.

RK  acknowledges support by the Technical University of Munich -- Institute of Advanced Study, funded by the German Excellence Initiative and the European Union Seventh Framework Programme under grant agreement no.~291763 and by the German Federal Ministry of Education through the funding program Photonics Research Germany, contract no.~13N14776~(QCDA-QuantERA). BS acknowledges the support from the Simons Foundation through It from Qubit collaboration; this work was supported by a grant from the Simons Foundation/SFARI (385612, JPP). ET acknowledges the support of the Natural Sciences and Engineering Research Council of Canada (NSERC), PGSD3-502528-2017. BS and ET also acknowledge funding provided by the Institute for Quantum Information and Matter, an NSF Physics Frontiers Center (NSF Grant No. PHY- 1733907).
\newpage

\appendix
\section{Canonical form of excitation ansatz states\label{app:canonicalexcitation}}

For the reader's convenience, we include here a proof of the Lemma~\ref{lem:canonicalformexcitationansatz}  following~\cite{haegeman2014geometry}. 

\begin{replemma}{lem:canonicalformexcitationansatz}
Let $\ket{\Phi_p(B;A)}$ be an injective excitation ansatz state and assume that $A$ is normalized such that the transfer operator has spectral radius~$1$. Let $\ell$ and $r$ be the corresponding left- and right- eigenvectors corresponding to eigenvalue $1$. Assume $p\neq 0$. Then there exists a tensor $\tilde{B}$
such that $\ket{\Phi_p(B;A)}=\ket{\Phi_p(\tilde{B};A)}$, and such that
\begin{align}
\bbra{\ell} E_{\tilde{B}(p)}&=0\qquad\textrm{ and }\qquad \bbra{\ell} E_{\overline{\tilde{B}(p)}}=0\ .\tag{\eqref{eq:gaugeconditionrewrittenexcit}}
\end{align}
\end{replemma}
\begin{proof}
We note that the equations~\eqref{eq:gaugeconditionrewrittenexcit} can be written as
\begin{align}
      \sum_{i\in [\physical]} A_i^{\dagger}\ell \tilde{B}_i=0\ ,\qquad\text{and}\qquad \sum_{i\in [\physical]}\tilde{B}_i^{\dagger}\ell A_i =0\ .\mylabel{eq:toprovenormalformx}
\end{align}
Diagrammatically, they take the form
\begin{align}
\raisebox{-.38\height} {\includegraphics[scale=0.05]{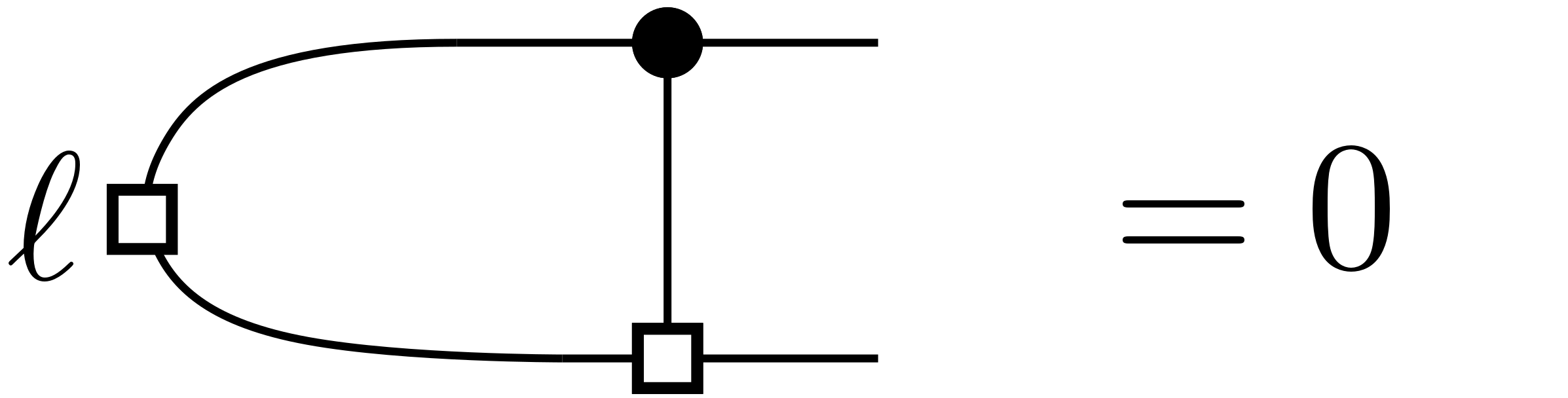}}\mylabel{eq:leftgauge},\\
\raisebox{-.38\height} {\includegraphics[scale=0.05]{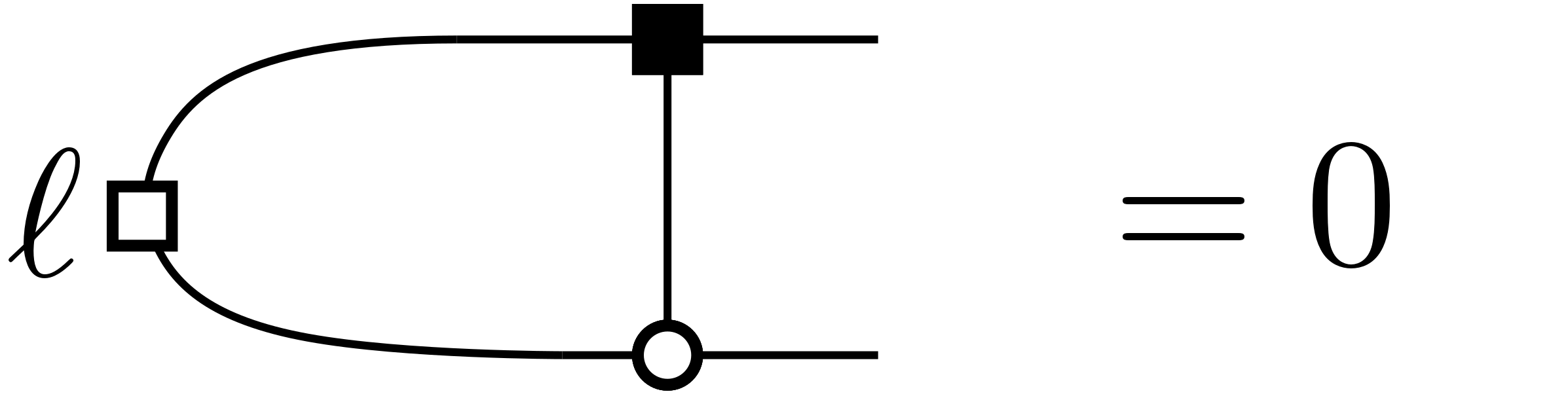}} ,\end{align}
where square and round boxes correspond to $\tilde{B}$ and $A$, respectively.

Let the original MPS tensors be~$A=\{A_j\}_{j=1}^{\physical}$ 
and $B=\{B_j\}_{j=1}^{\physical} \subset \cB(\mathbb{C}^\bond\otimes\mathbb{C}^\bond)$. Suppose $X\in \cB(\mathbb{C}^{\bond})$ is invertible.
Define the MPS tensor $C=\{\tilde{B}_j\}_{j=1}^{\physical} \subset\cB(\mathbb{C}^\bond)$ by
\begin{align}
C_j&= A_j X-e^{-ip}XA_j\qquad\textrm{ for }j=1,\ldots,\physical\ .
\end{align}
that is,
\begin{align}
\raisebox{-.42\height} {\includegraphics[scale=0.12]{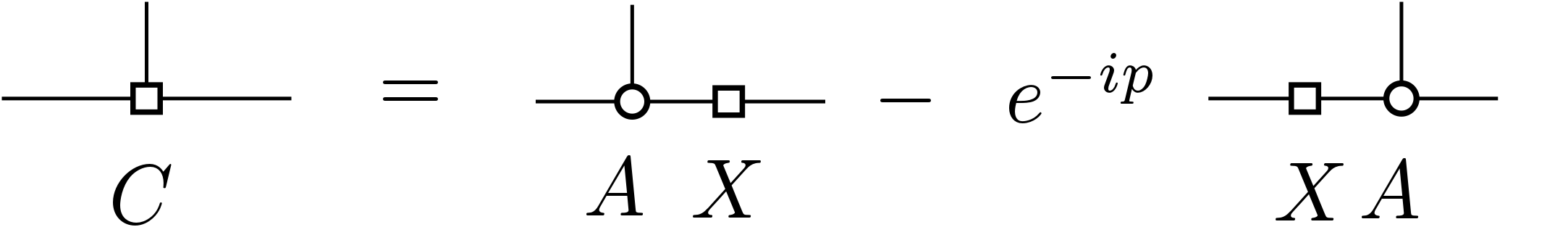}}.\mylabel{eq:gaugefreedom}
\end{align}
It is then easy to check that 
\begin{align}
\ket{\Phi_p(B;A)}&=
\ket{\Phi_p(B+C;A)}\ ,
\end{align}
where $B+C$ is the MPS tensor
obtained by setting $(B+C)_j=B_j+C_j$ for each $j=1,\ldots,\physical$. Indeed, the difference of these two vectors is
\begin{align}
&\qquad\ket{\Phi_p(B+C;A)}-
\ket{\Phi_p(B;A)}\\
&=\sum_{i_1,\ldots,i_n\in [\physical]}\sum_{k=1}^n e^{ipk}\tr(A_{i_1}\cdots A_{i_{k-1}}C_{i_k}A_{i_{k+1}}\cdots A_{i_n})\ket{i_1\cdots i_n}\\
&=\sum_{i_1,\ldots,i_n\in [\physical]}
\left(
\sum_{k=1}^n e^{ipk}\tr\big[A_{i_1}\cdots A_{i_{k-1}}(A_{i_k}X-e^{-ip}XA_{i_{k}})A_{i_{k+1}}\cdots A_{i_n}\big]
\right)\ket{i_1\cdots i_n}\\
&=0\ ,
\end{align}
since the terms in the square brackets vanish because of the cyclicity of the trace (alternatively, this  can be seen by substituting each square box (corresponding to $B$) in Figure~\ref{fig:excitationansatz}
by a formal linear combination of a square box ($B$) and diagram~\eqref{eq:gaugefreedom}).

Observe that the second equation in~\eqref{eq:toprovenormalformx} can be obtained from the first by taking the adjoint since $\ell$ is a selfadjoint operator.  It thus suffices to show that
there is an MPS tensor~$\tilde{B}$ with the desired property 
$\ket{\Phi_p(\tilde{B};A)}=\ket{\Phi_p(B;A)}$ such that
\begin{align}
    \sum_{i\in [\physical]} A_i^{\dagger}\ell \tilde{B}_i&=0\ .\mylabel{eq:tosatisfyequationx}
\end{align}
It turns out that setting $\tilde{B}=B+C$ for an appropriate choice of~$X$ (and thus~$C$) suffices. 
equation~\eqref{eq:tosatisfyequationx} then amounts to the identity
\begin{align}
        \sum_{j\in [\physical]} A_j^{\dagger}\ell 
    (B_j+A_j X-e^{-ip} XA_j)    
        &=0\ , \mylabel{eq:tosatisfyequationxyz}
\end{align}
or diagrammatically,
\begin{align}
\raisebox{-.42\height} {\includegraphics[scale=0.1]{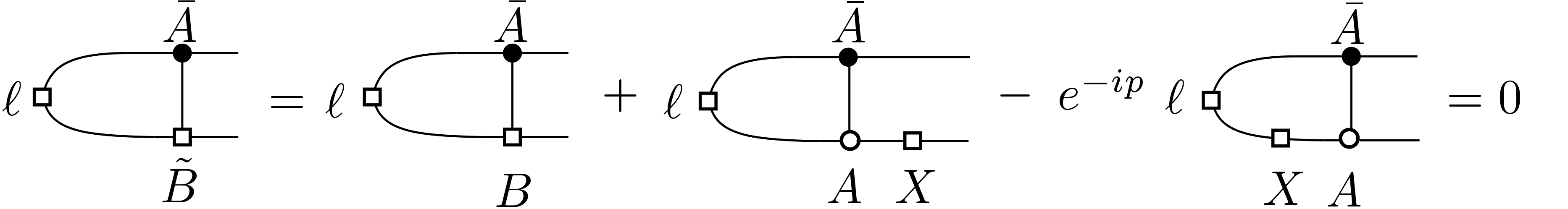}}. 
\end{align}
Because $\ell$ is the unique eigenvector of $\cE^\dagger(\rho)=\sum_{j \in [\physical]}  A_j^{\dagger}\rho A_j$ to eigenvalue~$1$, equation~\eqref{eq:tosatisfyequationxyz} simplifies to
\begin{align}
    \sum_{j\in [\physical]} A_j^{\dagger}\ell B_j+ \ell X-e^{-ip}\sum_{j\in [\physical]} A_j^{\dagger} \ell X A_j=0\ ,\mylabel{eq:tosatisfyxzyb}
\end{align}
or 
\begin{align}
\raisebox{-.42\height} {\includegraphics[scale=0.1]{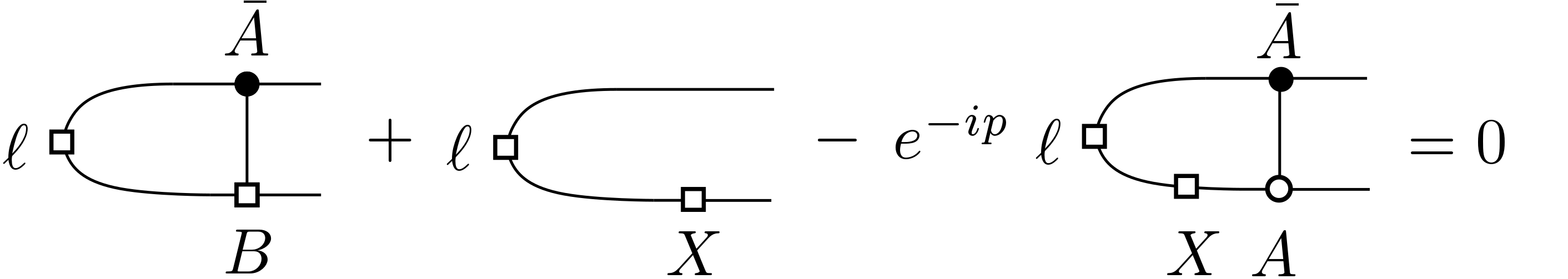}}.
\end{align}
Since $\ell$ is  full rank, we may substitute $X=\ell^{-1}Y$. Then~\eqref{eq:tosatisfyequationxyz} is satisfied if
\begin{align}
\raisebox{-.42\height} {\includegraphics[scale=0.1]{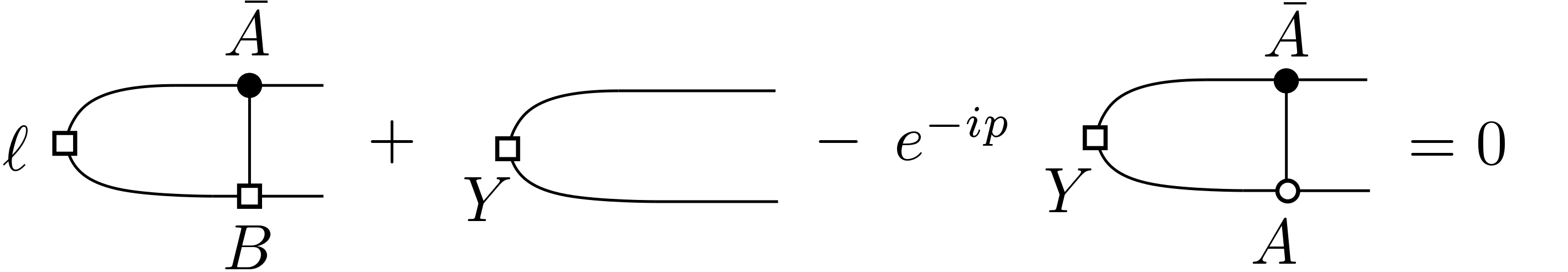}},
\end{align}
or
\begin{align}
    \sum_{j\in [\physical]} A_j^{\dagger}\ell B_j+(\identityoperator- e^{-ip}\cE)(Y)&=0\ .
\end{align}
Because $1$ is the unique eigenvalue of magnitude~$1$
of $\cE$, the map $(\lambda \identityoperator-e^{-ip}\cE)$ is
invertible under the assumption that $p\neq 0$, and we obtain
the solution
\begin{align}
    X&=\ell^{-1}Y\\
    &=-\ell^{-1}\left(\identityoperator-e^{-ip}\cE\right)^{-1}
    \left(\sum_{j\in [\physical]}A_j^{\dagger}\ell B_j\right)
\end{align}
to equation~\eqref{eq:tosatisfyequationxyz}, proving the claim for $p\neq 0$.

\end{proof}
\newpage

\nocite{*}
\bibliographystyle{IEEEtran}
\bibliography{mybib}

\begin{thebibliography}{10}
\providecommand{\url}[1]{#1}
\csname url@samestyle\endcsname
\providecommand{\newblock}{\relax}
\providecommand{\bibinfo}[2]{#2}
\providecommand{\BIBentrySTDinterwordspacing}{\spaceskip=0pt\relax}
\providecommand{\BIBentryALTinterwordstretchfactor}{4}
\providecommand{\BIBentryALTinterwordspacing}{\spaceskip=\fontdimen2\font plus
\BIBentryALTinterwordstretchfactor\fontdimen3\font minus
  \fontdimen4\font\relax}
\providecommand{\BIBforeignlanguage}[2]{{%
\expandafter\ifx\csname l@#1\endcsname\relax
\typeout{** WARNING: IEEEtran.bst: No hyphenation pattern has been}%
\typeout{** loaded for the language `#1'. Using the pattern for}%
\typeout{** the default language instead.}%
\else
\language=\csname l@#1\endcsname
\fi
#2}}
\providecommand{\BIBdecl}{\relax}
\BIBdecl

\bibitem{shor1995scheme}
P.~W. Shor, ``Scheme for reducing decoherence in quantum computer memory,''
  \emph{Physical review A}, vol.~52, no.~4, p. R2493, 1995.

\bibitem{calderbank1996good}
A.~R. Calderbank and P.~W. Shor, ``Good quantum error-correcting codes exist,''
  \emph{Physical Review A}, vol.~54, no.~2, p. 1098, 1996.

\bibitem{kitaev1997quantum}
A.~Y. Kitaev, ``Quantum computations: algorithms and error correction,''
  \emph{Russian Mathematical Surveys}, vol.~52, no.~6, pp. 1191--1249, 1997.

\bibitem{knill1997theory}
E.~Knill and R.~Laflamme, ``Theory of quantum error-correcting codes,''
  \emph{Physical Review A}, vol.~55, no.~2, p. 900, 1997.

\bibitem{knill1998resilient}
E.~Knill, R.~Laflamme, and W.~H. Zurek, ``Resilient quantum computation,''
  \emph{Science}, vol. 279, no. 5349, pp. 342--345, 1998.

\bibitem{kitaev2003fault}
A.~Y. Kitaev, ``Fault-tolerant quantum computation by anyons,'' \emph{Annals of
  Physics}, vol. 303, no.~1, pp. 2--30, 2003.

\bibitem{freedman2003topological}
M.~Freedman, A.~Kitaev, M.~Larsen, and Z.~Wang, ``Topological quantum
  computation,'' \emph{Bulletin of the American Mathematical Society}, vol.~40,
  no.~1, pp. 31--38, 2003.

\bibitem{ogburn1999topological}
R.~W. Ogburn and J.~Preskill, ``Topological quantum computation,'' in
  \emph{Quantum computing and quantum communications}.\hskip 1em plus 0.5em
  minus 0.4em\relax Springer, 1999, pp. 341--356.

\bibitem{dennis2002topological}
E.~Dennis, A.~Kitaev, A.~Landahl, and J.~Preskill, ``Topological quantum
  memory,'' \emph{Journal of Mathematical Physics}, vol.~43, no.~9, pp.
  4452--4505, 2002.

\bibitem{nayak2008non}
C.~Nayak, S.~H. Simon, A.~Stern, M.~Freedman, and S.~D. Sarma, ``Non-abelian
  anyons and topological quantum computation,'' \emph{Reviews of Modern
  Physics}, vol.~80, no.~3, p. 1083, 2008.

\bibitem{raussendorf2006fault}
R.~Raussendorf, J.~Harrington, and K.~Goyal, ``A fault-tolerant one-way quantum
  computer,'' \emph{Annals of physics}, vol. 321, no.~9, pp. 2242--2270, 2006.

\bibitem{stern2013topological}
A.~Stern and N.~H. Lindner, ``Topological quantum computation - from basic
  concepts to first experiments,'' \emph{Science}, vol. 339, no. 6124, pp.
  1179--1184, 2013.

\bibitem{terhal2015quantum}
B.~M. Terhal, ``Quantum error correction for quantum memories,'' \emph{Reviews
  of Modern Physics}, vol.~87, no.~2, p. 307, 2015.

\bibitem{kitaev2006anyons}
A.~Kitaev, ``Anyons in an exactly solved model and beyond,'' \emph{Annals of
  Physics}, vol. 321, no.~1, pp. 2--111, 2006.

\bibitem{levin2005string}
M.~A. Levin and X.-G. Wen, ``String-net condensation: A physical mechanism for
  topological phases,'' \emph{Physical Review B}, vol.~71, no.~4, p. 045110,
  2005.

\bibitem{kitaev2009periodic}
A.~Kitaev, ``Periodic table for topological insulators and superconductors,''
  in \emph{AIP Conference Proceedings}, vol. 1134, no.~1.\hskip 1em plus 0.5em
  minus 0.4em\relax AIP, 2009, pp. 22--30.

\bibitem{fidkowski2011topological}
L.~Fidkowski and A.~Kitaev, ``Topological phases of fermions in one
  dimension,'' \emph{Physical review b}, vol.~83, no.~7, p. 075103, 2011.

\bibitem{chen2011classification}
X.~Chen, Z.-C. Gu, and X.-G. Wen, ``Classification of gapped symmetric phases
  in one-dimensional spin systems,'' \emph{Physical review b}, vol.~83, no.~3,
  p. 035107, 2011.

\bibitem{chen2011complete}
------, ``Complete classification of one-dimensional gapped quantum phases in
  interacting spin systems,'' \emph{Physical review b}, vol.~84, no.~23, p.
  235128, 2011.

\bibitem{chen2011two}
X.~Chen, Z.-X. Liu, and X.-G. Wen, ``Two-dimensional symmetry-protected
  topological orders and their protected gapless edge excitations,''
  \emph{Physical Review B}, vol.~84, no.~23, p. 235141, 2011.

\bibitem{chen2013symmetry}
X.~Chen, Z.-C. Gu, Z.-X. Liu, and X.-G. Wen, ``Symmetry protected topological
  orders and the group cohomology of their symmetry group,'' \emph{Physical
  Review B}, vol.~87, no.~15, p. 155114, 2013.

\bibitem{almheiri15}
A.~Almheiri, X.~Dong, and D.~Harlow, ``Bulk locality and quantum error
  correction in ads/cft,'' \emph{JHEP 1504:163,arXiv:1411.7041v3}, 2015.

\bibitem{harlow2017ryu}
D.~Harlow, ``The {R}yu--{T}akayanagi formula from quantum error correction,''
  \emph{Communications in Mathematical Physics}, vol. 354, no.~3, pp. 865--912,
  2017.

\bibitem{fannes1992finitely}
M.~Fannes, B.~Nachtergaele, and R.~F. Werner, ``Finitely correlated states on
  quantum spin chains,'' \emph{Communications in mathematical physics}, vol.
  144, no.~3, pp. 443--490, 1992.

\bibitem{white1992density}
S.~R. White, ``Density matrix formulation for quantum renormalization groups,''
  \emph{Physical review letters}, vol.~69, no.~19, p. 2863, 1992.

\bibitem{white1993density}
------, ``Density-matrix algorithms for quantum renormalization groups,''
  \emph{Physical Review B}, vol.~48, no.~14, p. 10345, 1993.

\bibitem{vidal2003efficient}
G.~Vidal, ``Efficient classical simulation of slightly entangled quantum
  computations,'' \emph{Physical review letters}, vol.~91, no.~14, p. 147902,
  2003.

\bibitem{vidal2004efficient}
------, ``Efficient simulation of one-dimensional quantum many-body systems,''
  \emph{Physical review letters}, vol.~93, no.~4, p. 040502, 2004.

\bibitem{perez2006matrix}
D.~Perez-Garcia, F.~Verstraete, M.~M. Wolf, and J.~I. Cirac, ``Matrix product
  state representations,'' \emph{Quantum Information and Computation}, vol.~7,
  no. 401, 2007.

\bibitem{verstraete2004renormalization}
F.~Verstraete and J.~I. Cirac, ``Renormalization algorithms for quantum-many
  body systems in two and higher dimensions,'' \emph{arXiv preprint
  cond-mat/0407066}, 2004.

\bibitem{verstraete2008matrix}
F.~Verstraete, V.~Murg, and J.~I. Cirac, ``Matrix product states, projected
  entangled pair states, and variational renormalization group methods for
  quantum spin systems,'' \emph{Advances in Physics}, vol.~57, no.~2, pp.
  143--224, 2008.

\bibitem{vidal2008class}
G.~Vidal, ``Class of quantum many-body states that can be efficiently
  simulated,'' \emph{Physical review letters}, vol. 101, no.~11, p. 110501,
  2008.

\bibitem{verstraete2006projected}
F.~Verstraete, M.~Wolf, D.~P{\'e}rez-Garc{\'\i}a, and J.~I. Cirac, ``Projected
  entangled states: Properties and applications,'' \emph{International Journal
  of Modern Physics B}, vol.~20, no. 30n31, pp. 5142--5153, 2006.

\bibitem{perez2007peps}
D.~Perez-Garcia, F.~Verstraete, J.~I. Cirac, and M.~M. Wolf, ``P{E}{P}{S} as
  unique ground states of local hamiltonians,'' \emph{arXiv preprint
  arXiv:0707.2260}, 2007.

\bibitem{kraus2010fermionic}
C.~V. Kraus, N.~Schuch, F.~Verstraete, and J.~I. Cirac, ``Fermionic projected
  entangled pair states,'' \emph{Physical Review A}, vol.~81, no.~5, p. 052338,
  2010.

\bibitem{schwarz2012preparing}
M.~Schwarz, K.~Temme, and F.~Verstraete, ``Preparing projected entangled pair
  states on a quantum computer,'' \emph{Physical review letters}, vol. 108,
  no.~11, p. 110502, 2012.

\bibitem{fishman2018faster}
M.~Fishman, L.~Vanderstraeten, V.~Zauner-Stauber, J.~Haegeman, and
  F.~Verstraete, ``Faster methods for contracting infinite two-dimensional
  tensor networks,'' \emph{Physical Review B}, vol.~98, no.~23, p. 235148,
  2018.

\bibitem{buerschaper2009explicit}
O.~Buerschaper, M.~Aguado, and G.~Vidal, ``Explicit tensor network
  representation for the ground states of string-net models,'' \emph{Physical
  Review B}, vol.~79, no.~8, p. 085119, 2009.

\bibitem{gu2009tensor}
Z.-C. Gu, M.~Levin, B.~Swingle, and X.-G. Wen, ``Tensor-product representations
  for string-net condensed states,'' \emph{Physical Review B}, vol.~79, no.~8,
  p. 085118, 2009.

\bibitem{konig2009exact}
R.~K{\"o}nig, B.~W. Reichardt, and G.~Vidal, ``Exact entanglement
  renormalization for string-net models,'' \emph{Physical Review B}, vol.~79,
  no.~19, p. 195123, 2009.

\bibitem{schuch2010peps}
N.~Schuch, I.~Cirac, and D.~P{\'e}rez-Garc{\'\i}a, ``P{E}{P}{S} as ground
  states: Degeneracy and topology,'' \emph{Annals of Physics}, vol. 325,
  no.~10, pp. 2153--2192, 2010.

\bibitem{buerschaper2014twisted}
O.~Buerschaper, ``Twisted injectivity in projected entangled pair states and
  the classification of quantum phases,'' \emph{Annals of Physics}, vol. 351,
  pp. 447--476, 2014.

\bibitem{csahinouglu2014characterizing}
M.~B. {\c{S}}ahino{\u{g}}lu, D.~Williamson, N.~Bultinck, M.~Mari{\"e}n,
  J.~Haegeman, N.~Schuch, and F.~Verstraete, ``Characterizing topological order
  with matrix product operators,'' \emph{arXiv preprint arXiv:1409.2150}, 2014.

\bibitem{williamson2016matrix}
D.~J. Williamson, N.~Bultinck, M.~Mari{\"e}n, M.~B. {\c{S}}ahino{\u{g}}lu,
  J.~Haegeman, and F.~Verstraete, ``Matrix product operators for
  symmetry-protected topological phases: Gauging and edge theories,''
  \emph{Physical Review B}, vol.~94, no.~20, p. 205150, 2016.

\bibitem{bultinck2017anyons}
N.~Bultinck, M.~Mari{\"e}n, D.~J. Williamson, M.~B. {\c{S}}ahino{\u{g}}lu,
  J.~Haegeman, and F.~Verstraete, ``Anyons and matrix product operator
  algebras,'' \emph{Annals of Physics}, vol. 378, pp. 183--233, 2017.

\bibitem{TN-TQFT}
M.~B. \c{S}ahino\u{g}lu, M.~Walter, and D.~J. Williamson, ``A tensor network
  framework for topological order in higher dimensions,'' 2016, not yet
  published, a version can be found at Chapter 2 of MBS' PhD thesis, available
  at the url: http://othes.univie.ac.at/43085/.

\bibitem{swingle2012entanglement}
B.~Swingle, ``Entanglement renormalization and holography,'' \emph{Physical
  Review D}, vol.~86, no.~6, p. 065007, 2012.

\bibitem{pastawski2015holographic}
F.~Pastawski, B.~Yoshida, D.~Harlow, and J.~Preskill, ``Holographic quantum
  error-correcting codes: Toy models for the bulk/boundary correspondence,''
  \emph{Journal of High Energy Physics}, vol. 2015, no.~6, p. 149, 2015.

\bibitem{hayden2016holographic}
P.~Hayden, S.~Nezami, X.-L. Qi, N.~Thomas, M.~Walter, and Z.~Yang,
  ``Holographic duality from random tensor networks,'' \emph{Journal of High
  Energy Physics}, vol. 2016, no.~11, p.~9, 2016.

\bibitem{akers2018holographic}
C.~Akers and P.~Rath, ``Holographic {R}{\'e}nyi entropy from quantum error
  correction,'' \emph{arXiv preprint arXiv:1811.05171}, 2018.

\bibitem{dong2018flat}
X.~Dong, D.~Harlow, and D.~Marolf, ``Flat entanglement spectra in fixed-area
  states of quantum gravity,'' \emph{arXiv preprint arXiv:1811.05382}, 2018.

\bibitem{brandao2017quantum}
F.~G. Brandao, E.~Crosson, M.~B. {\c{S}}ahino{\u{g}}lu, and J.~Bowen, ``Quantum
  error correcting codes in eigenstates of translation-invariant spin chains,''
  \emph{arXiv preprint arXiv:1710.04631}, 2017.

\bibitem{Kim2017}
\BIBentryALTinterwordspacing
I.~H. Kim and M.~J. Kastoryano, ``Entanglement renormalization, quantum error
  correction, and bulk causality,'' \emph{Journal of High Energy Physics}, vol.
  2017, no.~4, p.~40, Apr 2017. [Online]. Available:
  \url{https://doi.org/10.1007/JHEP04(2017)040}
\BIBentrySTDinterwordspacing

\bibitem{pastawski2017towards}
F.~Pastawski, J.~Eisert, and H.~Wilming, ``Towards holography via quantum
  source-channel codes,'' \emph{Physical review letters}, vol. 119, no.~2, p.
  020501, 2017.

\bibitem{toriccodequbitloss}
\BIBentryALTinterwordspacing
T.~M. Stace, S.~D. Barrett, and A.~C. Doherty, ``Thresholds for topological
  codes in the presence of loss,'' \emph{Phys. Rev. Lett.}, vol. 102, p.
  200501, May 2009. [Online]. Available:
  \url{https://link.aps.org/doi/10.1103/PhysRevLett.102.200501}
\BIBentrySTDinterwordspacing

\bibitem{pollatsek2004permutationally}
H.~Pollatsek and M.~B. Ruskai, ``Permutationally invariant codes for quantum
  error correction,'' \emph{Linear algebra and its applications}, vol. 392, pp.
  255--288, 2004.

\bibitem{ouyang2014permutation}
Y.~Ouyang, ``Permutation-invariant quantum codes,'' \emph{Physical Review A},
  vol.~90, no.~6, p. 062317, 2014.

\bibitem{bravyihastingsmichalakis}
\BIBentryALTinterwordspacing
S.~Bravyi, M.~B. Hastings, and S.~Michalakis, ``Topological quantum order:
  Stability under local perturbations,'' \emph{Journal of Mathematical
  Physics}, vol.~51, no.~9, p. 093512, 2010. [Online]. Available:
  \url{https://doi.org/10.1063/1.3490195}
\BIBentrySTDinterwordspacing

\bibitem{hastings11}
\BIBentryALTinterwordspacing
M.~B. Hastings, ``Topological order at nonzero temperature,'' \emph{Phys. Rev.
  Lett.}, vol. 107, p. 210501, Nov 2011. [Online]. Available:
  \url{https://link.aps.org/doi/10.1103/PhysRevLett.107.210501}
\BIBentrySTDinterwordspacing

\bibitem{gentlemeasurementlemma}
T.~Ogawa and H.~Nagaoka, ``A new proof of the channel coding theorem via
  hypothesis testing in quantum information theory,'' in \emph{Proceedings IEEE
  International Symposium on Information Theory,}, June 2002, pp. 73--.

\bibitem{aguadovidal08}
\BIBentryALTinterwordspacing
M.~Aguado and G.~Vidal, ``Entanglement renormalization and topological order,''
  \emph{Phys. Rev. Lett.}, vol. 100, p. 070404, Feb 2008. [Online]. Available:
  \url{https://link.aps.org/doi/10.1103/PhysRevLett.100.070404}
\BIBentrySTDinterwordspacing

\bibitem{hastings2006spectral}
M.~B. Hastings and T.~Koma, ``Spectral gap and exponential decay of
  correlations,'' \emph{Communications in mathematical physics}, vol. 265,
  no.~3, pp. 781--804, 2006.

\bibitem{brandao2013area}
F.~G. Brand{\~a}o and M.~Horodecki, ``An area law for entanglement from
  exponential decay of correlations,'' \emph{Nature Physics}, vol.~9, no.~11,
  p. 721, 2013.

\bibitem{brandao2015exponential}
------, ``Exponential decay of correlations implies area law,''
  \emph{Communications in mathematical physics}, vol. 333, no.~2, pp. 761--798,
  2015.

\bibitem{schuch2011classifying}
N.~Schuch, D.~P{\'e}rez-Garc{\'\i}a, and I.~Cirac, ``Classifying quantum phases
  using matrix product states and projected entangled pair states,''
  \emph{Physical review b}, vol.~84, no.~16, p. 165139, 2011.

\bibitem{hastings2006solving}
M.~B. Hastings, ``Solving gapped {H}amiltonians locally,'' \emph{Physical
  review b}, vol.~73, no.~8, p. 085115, 2006.

\bibitem{haegeman2013elementary}
J.~Haegeman, S.~Michalakis, B.~Nachtergaele, T.~J. Osborne, N.~Schuch, and
  F.~Verstraete, ``Elementary excitations in gapped quantum spin systems,''
  \emph{Physical review letters}, vol. 111, no.~8, p. 080401, 2013.

\bibitem{verstraetekorepin}
V.~Murg, V.~E. Korepin, and F.~Verstraete, ``Algebraic {B}ethe ansatz and
  tensor networks,'' \emph{Physical Review B}, 2012.

\bibitem{crepeauetal05}
C.~Cr{\'e}peau, D.~Gottesman, and A.~Smith, ``Approximate quantum
  error-correcting codes and secret sharing schemes,'' in \emph{Advances in
  Cryptology -- EUROCRYPT 2005}, R.~Cramer, Ed.\hskip 1em plus 0.5em minus
  0.4em\relax Berlin, Heidelberg: Springer Berlin Heidelberg, 2005, pp.
  285--301.

\bibitem{barnumknillnielsen}
H.~Barnum, E.~Knill, and M.~A. Nielsen, ``On quantum fidelities and channel
  capacities,'' \emph{IEEE Transactions on Information Theory}, vol.~46, no.~4,
  pp. 1317--1329, July 2000.

\bibitem{errorcorrectionoldintro}
E.~Knill, R.~Laflamme, A.~Ashikhmin, H.~Barnum, L.~Viola, and W.~Zurek,
  ``Introduction to quantum error correction,'' \emph{arXiv quant-ph/0207170},
  2002.

\bibitem{benyoreshkov}
\BIBentryALTinterwordspacing
C.~B\'eny and O.~Oreshkov, ``General conditions for approximate quantum error
  correction and near-optimal recovery channels,'' \emph{Phys. Rev. Lett.},
  vol. 104, p. 120501, Mar 2010. [Online]. Available:
  \url{https://link.aps.org/doi/10.1103/PhysRevLett.104.120501}
\BIBentrySTDinterwordspacing

\bibitem{evanshoegh}
\BIBentryALTinterwordspacing
D.~E. Evans and R.~Hoegh-Krohn, ``Spectral properties of positive maps on
  {C}*-algebras,'' \emph{Journal of the London Mathematical Society}, vol.
  s2-17, no.~2, pp. 345--355, 1978. [Online]. Available:
  \url{https://londmathsoc.onlinelibrary.wiley.com/doi/abs/10.1112/jlms/s2-17.2.345}
\BIBentrySTDinterwordspacing

\bibitem{ruizthesis}
\BIBentryALTinterwordspacing
M.~S. Ruiz, ``Tensor networks in condensed matter,'' Ph.D. dissertation,
  Technische Universität München, 2011. [Online]. Available:
  \url{https://mediatum.ub.tum.de/doc/1070963/1070963.pdf}
\BIBentrySTDinterwordspacing

\bibitem{wolfchannel}
\BIBentryALTinterwordspacing
M.~M. Wolf, ``Quantum channels \& operations,'' 2012. [Online]. Available:
  \url{https://www-m5.ma.tum.de/foswiki/pub/M5/Allgemeines/MichaelWolf/QChannelLecture.pdf}
\BIBentrySTDinterwordspacing

\bibitem{bohdanowicz2018good}
T.~C. Bohdanowicz, E.~Crosson, C.~Nirkhe, and H.~Yuen, ``Good approximate
  quantum ldpc codes from spacetime circuit hamiltonians,'' \emph{arXiv
  preprint arXiv:1811.00277}, 2018.

\bibitem{haegemanetal13}
J.~Haegeman, T.~Osborne, and F.~Verstraete, ``Post-matrix product state
  methods: To tangent space and beyond,'' \emph{Phys. Rev. B 88, 075133,
  arXiv:1305.1894v1}, 2013.

\bibitem{haegeman2014geometry}
J.~Haegeman, M.~Mari{\"e}n, T.Osborne, and F.~Verstraete, ``Geometry of matrix
  product states: Metric, parallel transport, and curvature,'' \emph{Journal of
  Mathematical Physics}, vol.~55, no.~2, p. 021902, 2014.

\bibitem{fuchsgraaf}
C.~A. Fuchs and J.~van~de Graaf, ``Cryptographic distinguishability measures
  for quantum mechanical states,'' \emph{arXiv e-prints quant-ph/9712042},
  1997.

\end{thebibliography}

\end{document}